\RequirePackage{etoolbox}
\csdef{input@path}{%
 {sty/}
 {img/}
}%
\csgdef{bibdir}{bib/}

\documentclass[12pt]{book}
\pdfoutput=1

\usepackage{miseEnPage}
\usepackage{couverture}

\usepackage{mathtools}
\usepackage{subcaption}

\usepackage{amssymb,amsmath,amsthm,amsfonts}
\usepackage{a4wide}
\usepackage{caption}
\usepackage{comment}
\usepackage{color}
\usepackage{enumitem} 
\usepackage[T1]{fontenc}
\usepackage{graphicx}
\usepackage{geometry}
\usepackage{minitoc}
\usepackage[numbers]{natbib}
\usepackage{palatino}
\usepackage{placeins}
\usepackage{pgf,pgfarrows,pgfnodes,pgfautomata,pgfheaps,pgfshade}
\usepackage{subfloat}

\usepackage{comment}

\usepackage{tabularx}
\usepackage{float}
\usepackage{wrapfig}
\floatstyle{boxed}
\usepackage{tikz}
\usetikzlibrary{shapes,backgrounds}
\usetikzlibrary{patterns}
\usetikzlibrary{decorations.pathreplacing}
\usetikzlibrary{arrows,decorations.pathmorphing,backgrounds,fit,positioning,shapes.symbols,chains}

\usepackage[resetlabels]{multibib}
\newcites{Me}{Bibliographie personnelle}
\bibliographystyleMe{plain}

\newcommand{\poly}{\mathrm{poly}}
\newcommand{\CK}{\mathrm{C}}

\theoremstyle{plain}

\newtheorem{proposition}{Proposition}

\newtheorem{theorem}{Théorème}
 \makeatletter
\newcommand{\neutralize}[1]{\expandafter\let\csname c@#1\endcsname\count@}
\makeatother

\newenvironment{theorembis}[1]
  {\renewcommand{\thetheorem}{\ref*{#1}$'$}%
   \neutralize{theorem}\phantomsection
   \begin{theorem}}
  {\end{theorem}} 

\newtheorem{corollary}{Corollaire}

\newtheorem{lemma}{Lemme}

\newtheorem{claim}{Fait}

\newtheorem*{theorem*}{Théorème}

\theoremstyle{remark}
\newtheorem{remark}{Remarque} 
\newtheorem{definition}{Définition}
\newtheorem{example}{Exemple} 

\newtheorem*{example*}{Exemple} 

\newtheorem{question}{Question} 

\newenvironment{examplebis}[1]
  {\renewcommand{\theexample}{\ref*{#1}$'$}%
   \neutralize{example}\phantomsection
   \begin{example}}
 {\end{example}}

\newenvironment{examplerev}[1]
  {\renewcommand{\theexample}{\ref*{#1} \em revisité}%
   \neutralize{example}\phantomsection
   \begin{example}}
 {\end{example}}

\newenvironment{examplebis2}[1]
  {\renewcommand{\theexample}{\ref*{#1}$''$}%
   \neutralize{example}\phantomsection
   \begin{example}}
 {\end{example}}

  \newtheoremstyle{dotless}{}{}{}{}{\bfseries}{}{ }{}

  \theoremstyle{dotless}

\newtheorem*{problem*}{Problème ouvert}

\usepackage[normalem]{ulem}

\selectcolormodel{RGB}
\makeatletter
\DeclareRobustCommand{\format@sec@number}[2]{{\normalfont\upshape#1}#2}
\renewcommand{\chaptermark}[1]{%
  \markboth{\format@sec@number{\ifnum\c@secnumdepth>\m@ne\@chapapp\ \thechapter. \fi}{#1}}{}}
\renewcommand{\sectionmark}[1]{%
  \markright{\format@sec@number{\ifnum\c@secnumdepth>\z@\thesection. \fi}{#1}}}
\makeatother

\begin{document}
\pagenumbering{roman}
\NoAutoSpacing

\frontmatter

\title{\protect\parbox{\textwidth}{\protect\centering Étude de la complexité algorithmique et du caractère sofique des shifts en dimension deux}}
\author{Julien Destombes}

\directorA{Andrei Romashchenko}

\date{\rule{0mm}{0cm} 17 décembre 2021}

\thispagestyle{empty}

\university{UNIVERSITE DE MONTPELLIER}
\univlogo{logo_UM.png}
\univwallpaper{wallpaper.png}
\doctoral{I2S (Information, Structures et Syst\`emes)}
\researchunit{ESCAPE}
\specialisation{Informatique}
\jury{Mathieu Sablik}{Professeur}{Université Paul Sabatier, Toulouse}{Rapporteur}
\jury{Pascal Vanier}{Professeur}{Université de Caen Normandie}{Rapporteur}
\jury{Christophe Fiorio}{Professeur}{Université de Montpellier}{Examinateur}
\jury{Enrico Formenti}{Professeur}{Université Côte d'Azur, Nice}{Examinateur}
\jury{Pierre Guillon}{Chargé de recherche}{Institut de Mathématiques de Marseille}{Examinateur}
\jury{Svetlana Puzynina}{Associate professor}{St. Petersburg State University, Russie}{Examinatrice}
\jury{Gwenaël Richomme}{Professeur}{Université Paul-Valéry et LIRMM, Montpellier}{Examinateur}
\jury{Andrei Romashchenko}{Chargé de recherche}{LIRMM, Montpellier}{Directeur de thèse}

\maketitle

\chapter*{Résumé}
\addcontentsline{toc}{chapter}{Résumé}

Nous étudions les propriétés des shifts multidimensionnels. Plus précisément, nous nous intéressons aux conditions nécessaires ou suffisantes pour qu'un shift soit sofique ; autrement dit, nous étudions la frontière qui sépare les shifts sofiques des shifts effectifs et non sofiques. Nous montrons que plusieurs versions de la complexité algorithmique (ou complexité de Kolmogorov) sont utiles pour formuler de telles conditions.

Ainsi, nous proposons dans un premier temps des conditions nécessaires pour qu'un shift multidimensionnel soit sofique, formulées en termes de complexité de Kolmogorov à ressources bornées. En utilisant cette technique nous construisons un exemple de shift effectif et non sofique sur $\mathbb{Z}^2$ avec une très faible complexité combinatoire : le nombre de motifs globalement admissibles de taille $N \times N$ ne croît que de manière polynomiale en $N$. Nous développons ensuite d'autres outils permettant de prouver qu'un shift n'est pas sofique, également basés sur la complexité de Kolmogorov. Nous montrons en outre que la technique développée par S.~Kass et K.~Madden dans~\cite{kass-madden} correspond à un cas particulier de ceux-ci.

Dans un second temps, nous définissons une large classe de shifts multidimensionnels sofiques, à savoir les shifts possédant un très faible nombre de cases noires sur un océan de cases blanches. Plus précisément, pour tout $\epsilon < 1$ calculable, le shift sur l'alphabet $\{ \square, \blacksquare \}$ dont, pour chaque configuration, les carrés de taille $N \times N$ ne contiennent pas plus de $N^\epsilon$ cases noires, est sofique. De plus, les sous-shifts effectifs de ce shift sont également sofiques. La preuve de ce résultat est principalement basée sur la construction d'un shift à point-fixe, en utilisant quelques autres ingrédients : un modèle de calcul \textit{ad hoc} permettant des calculs massivement parallèles, basé sur les automates cellulaires non déterministes, ainsi qu'un résultat sur les flots parcourant un type de graphes spécifiques. 

\tableofcontents

\mainmatter

\listoffigures

\chapter{Introduction}\label{ch:introduction}

\section{Les shifts}

\subsection{Les espaces des shifts : origines et contexte moderne}\label{ss:1-dimension}

La dynamique symbolique apparaît originellement en mathématiques comme une branche de la théorie des systèmes dynamiques qui étudie les systèmes dynamiques réguliers ou topologiques en discrétisant l'espace étudié. Depuis la fin des années 1930, la dynamique symbolique devient un domaine indépendant de recherche, voir \cite{symbolic-dynamics-1938,symbolic-dynamics-1940}. Un \emph{système dynamique} classique est un espace (ou un ensemble d'états) $\cal S$ muni d'une fonction $F$ agissant sur cet espace ; cette fonction représente la `` règle d'évolution'', i.e., la dépendance au temps d'une configuration dans cet espace. La notion centrale de la théorie des systèmes dynamiques est celle de trajectoire --- une séquence de configurations obtenues en itérant la règle d'évolution, 
\[
x, F(x), F(F(x)),\ldots, F^{(n)}(x),\ldots
\]

En dynamique symbolique l'espace des états est réduit à un ensemble fini (un \emph{alphabet}). Les trajectoires sont représentées par une séquence infinie (ou bi-infinie) de lettres de cet alphabet, et la ``règle d'évolution'' est l'opérateur du shift (ou ``décalage'' en français) agissant sur ces séquences.
La dynamique symbolique se focalise sur les \emph{espaces de shift} --- Un ensemble de séquences bi-infinies de lettres (sur un alphabet fini) qui est invariant pour cet opérateur du shift.

Plus précisément, un \emph{espace de shift} (aussi appelé un \emph{shift}) sur un alphabet $\Sigma$ est un ensemble non-vide de séquences bi-infinies sur l'alphabet $\Sigma$, invariant par translation et clos dans la topologie standard de l'espace de Cantor (Cette topologie est générée par les ensembles de base suivants : pour un motif fini fixé, nous considérons toutes les configurations le contenant).  

Tout shift peut être défini en terme de \emph{motifs interdits de taille finie} : nous fixons un ensemble de mots (finis) $\cal F$ et disons qu'une configuration (une séquence bi-infinie) appartient au shift $S_{\cal F}$ si et seulement s'il ne contient aucun motif de ~$\cal F$. Formellement, nous avons l'équivalence suivante : 

\begin{proposition}\label{p:motif-interdits}
Soit $\Sigma$ un alphabet (un ensemble fini de lettres). 

(a) Pour tout ensemble de motifs finis $\cal F$ sur $\Sigma$, l'ensemble $S_{\cal F}$ de tous les  $x\in {\Sigma}^\mathbb{Z}$ qui ne contiennent aucun motif de $\cal F$ est un espace de shift.

(b) Pour tout espace de shift $S$ il existe un ensemble de motifs finis $\cal F$ tel qu'une configuration $x$ appartient à $S$ si et seulement si $x$ ne contient aucun motif de $\cal F$.
\end{proposition}

\begin{proof}
Soit un shift $S$ sur l'alphabet $\Sigma$ : $S$ est un ensemble non-vide de séquences bi-infinies sur l'alphabet $\Sigma$, invariant par translation et clos dans la topologie standard de l'espace de Cantor. Un cylindre de $\Sigma$ est un sous-ensemble de $\Sigma^\mathbb{Z}$ où un nombre fini de positions sont associées à des lettres de $\Sigma$ ; le motif obtenu est appelé la base du cylindre. Le reste des positions de $\mathbb{N}$ sont libres. L'ensemble des cylindres de $\Sigma$ forme une base de $\Sigma^\mathbb{N}$ pour la topologie standard. Comme $S$ est clos dans cette topologie, $\bar{S}=\Sigma^\mathbb{N} - S$ est un ouvert. Par conséquent $\bar{S}$ est une union (potentiellement infinie) de cylindres de $\Sigma$. 
$S$ est invariant par translation, donc $\bar{S}$ aussi. Ainsi, si un cylindre $C$ est un sous-ensemble de $\bar{S}$, les cylindres dont la base est une translation de $C$ sont également des sous-ensembles de $\bar{S}$. L'ensemble des bases des cylindres dont l'union est $\bar{S}$ est exactement l'ensemble des motifs interdits de $S$ pour la définition alternative d'un shift.

Réciproquement, si $S$ est défini par un ensemble de motifs interdits, nous considérons les cylindres dont les bases sont ces motifs interdits et leurs translations. L'union de ces cylindres est un ouvert (puisque les cylindres forment une base de $\Sigma^\mathbb{N}$), invariant par translation. Le complémentaire de cette union, égal à $S$, est un ensemble clos, également invariant par translation. 
\end{proof}

Nous gardons à l'esprit que des ensembles différents de motifs interdits peuvent induire le même shift.

Nous allons maintenant considérer plusieurs exemples de shifts. Nous définissons un premier shift $S_1$, sur l'alphabet $\{ \blacksquare, \square \}$, par un seul motif interdit : $\blacksquare \blacksquare$. Une partie d'une configuration de $S_1$ peut ainsi être $\square \square \blacksquare \square \blacksquare \square \square \square \square \blacksquare \square \square$ ou $\square \square \square \square \square \square \square \square \square \blacksquare \square \square$, mais pas $\square \square \blacksquare \square \blacksquare \blacksquare \square \square \square \square \square \square$.

Un deuxième shift $S_2$, toujours sur l'alphabet binaire, a pour motifs interdits les paires de cases noires séparées par un nombre impair de cases blanches (nous avons donc comme motifs interdits $\blacksquare \square \blacksquare$, $\blacksquare \square \square \square \blacksquare$, $\blacksquare \square \square \square \square \square \blacksquare$, $\ldots$). Une partie d'une configuration de $S_2$ peut ainsi être $\square \blacksquare \blacksquare \square \square \blacksquare \square \square \square \square \blacksquare \blacksquare$, mais pas $\square \blacksquare \square \blacksquare \square \square \square \blacksquare \blacksquare \blacksquare \square \square$.

Enfin, un troisième exemple est le shift miroir, sur l'alphabet $\{ \blacksquare, \square, {\color{red} \blacksquare} \}$. Les motifs interdits sont d'une part les suites de taille $n>1$ où apparaissent deux carrés rouges, et d'autres part, pour $n>1$, les suites de taille $2n+1$, pour lesquelles le $(n+1)$-ème carré est rouge et où les $n$ premiers carrés ne sont pas symétriques, par rapport au carré rouge, avec les $n$ derniers. Par exemple, $\square \blacksquare \blacksquare \blacksquare \square \blacksquare \square \square \square \blacksquare \blacksquare \blacksquare$ ou $\blacksquare \blacksquare \square \square \blacksquare \square \square {\color{red} \blacksquare} \square \square \blacksquare \square$ sont des parties de configurations du shift miroir, mais pas $\square {\color{red} \blacksquare} \blacksquare \blacksquare \blacksquare \square \square \blacksquare \square \blacksquare {\color{red} \blacksquare} \blacksquare$ ou $\blacksquare \blacksquare \square \square \blacksquare \square \square {\color{red} \blacksquare} \blacksquare \square \blacksquare \square$.

Les propriétés d'un shift dépendent fortement de l'ensemble des motifs interdits le définissant. Les trois grandes classes de shifts suivantes jouent un rôle important en dynamique symbolique et en théorie de la calculabilité :
\begin{itemize}
\item \emph{les shifts de type fini}, qui peuvent être définis par un ensemble fini de motifs finis interdits ;
\item \emph{les shifts sofiques}, shifts (introduits dans \cite{weiss}), qui peuvent être défini par un ensemble de motifs interdits qui est un langage régulier ;
\item \emph{les shifts effectifs} (ou shifts \emph{effectivement clos}), qui peuvent être définis par un ensemble calculable de motifs finis interdits. Observons que nous pouvons demander à ce que l'ensemble des motifs finis interdits soit \emph{décidable} ou \emph{récursivement énumérable}. Bien que d'un point de vue formel ces définitions soient différentes, elles définissent la même classe de shifts. Ainsi, si un shift est défini par un ensemble récursivement énumérable de motifs finis interdits, alors le même shift peut également être défini par un ensemble décidable de motifs finis interdits.
\end{itemize}

La classe des shifts effectifs est souvent comprise comme la classe des shifts ``constructibles explicitement'', et est donc d'un intérêt important au-delà du champ de la théorie de la calculabilité. Ces trois classes sont distinctes :
\[
[\text{les shifts de type fini}] \subsetneqq [\text{les shifts sofiques}] \subsetneqq  [\text{les shifts effectifs}].
\]

Clairement, le shift $S_1$ défini ci-dessus est un shift de type fini (il est définit ci-dessus par un seul motif interdit, et donc par un nombre fini de motifs interdits).

Le shift $S_2$ est un shift sofique, puisque l'ensemble des motifs interdits le définissant ci-dessus est un langage régulier. Nous pouvons en outre montrer qu'il n'est pas de type fini. Supposons par l'absurde qu'il le soit, et considérons que $S_2$ soit défini par un ensemble fini de motifs interdits. Notons $l$ la largeur maximale de ces motifs. Considérons alors deux configurations de $\{\square,\blacksquare \}^{\mathbb{Z}}$, ne contenant chacune que deux cases noires, séparées par un nombre de cases blanches égal à respectivement $l$ et $(l+1)$. Les deux configurations possèdent les mêmes motifs de tailles égales ou inférieures à $l$, et donc soit au moins un motif interdit apparaît dans les deux configurations, soit aucun motif interdit n'apparaît dans aucune des deux configurations. Or l'une d'elles appartient à $S_2$, et l'autre non ; ceci est une contradiction.

Les shifts sofiques peuvent également être définis comme la projection coordonnée par coordonnée des configurations d'un shift de type fini :

\begin{definition}\label{d:sofique}
Un shift ${\cal S}$ sur un alphabet $\Sigma$ est un \emph{shift sofique} s'il existe un shift de type fini ${\cal S}'$ sur un alphabet $\Sigma'$ et une projection :

$\pi : \Sigma' \to \Sigma$, telle que ${\cal S}$ consiste en la projection coordonnée par coordonnée 
\[
(\ldots \pi( y_{-1}) \pi (y_0 ) \pi(y_1) \pi(y_2)\ldots)
\]
de toutes les configurations $(\ldots  y_{-1} y_0  y_1 y_2\ldots)$ de ${\cal S}'$.
\end{definition}

Ces deux définitions sont équivalentes :

\begin{proposition}[{voir par exemple \cite[Théorème~3.2.1]{sft}}]
En dimension 1, les shifts pouvant être définis par un ensemble de motifs interdits qui est un langage régulier, et les shifts définis dans la Définition~\ref{d:sofique}, sont une et même classe de shifts, la classe des shifts sofiques. 
\end{proposition}

Nous pouvons illustrer cette définition alternative en explicitant, pour le shift sofique $S_2$, un shift de type fini $S_2'$ et une projection $\pi$ tels que $\pi(S_2')=S_2$ . Le shift $S_2'$ est défini sur l'alphabet $\{ \square,{\color{black!60}\blacksquare },\blacksquare \}$, avec pour motifs interdits l'ensemble $\{ \square \square, {\color{black!60}\blacksquare } {\color{black!60}\blacksquare }, \blacksquare {\color{black!60}\blacksquare },\square \blacksquare\}$ (une suite de cases blanches et grises alterne entre cases blanches et grises, commence par une case blanche et finit par une case grise). $S_2'$ est bien un shift de type fini. Par exemple, la suite $\blacksquare \blacksquare \square {\color{black!60}\blacksquare } \square {\color{black!60}\blacksquare } \square {\color{black!60}\blacksquare } \blacksquare \square {\color{black!60}\blacksquare } \blacksquare $ peut appartenir à une configuration de $S_2'$, mais pas la suite $\blacksquare \blacksquare \square {\color{black!60}\blacksquare } \square \blacksquare \blacksquare \square {\color{black!60}\blacksquare } \square {\color{black!60}\blacksquare } \blacksquare$. Soit la projection $\pi$ projetant la case blanche et la case grise de $S_2'$ sur la case blanche de $S_2$, et la case noire de $S_2'$ sur la case noire de $S_2$. Nous pouvons facilement vérifier que $\pi(S_2') = S_2$. 

\label{p:successors-1d}
Il existe une caractérisation simple des shifts sofiques. Nous disons que deux mots $w_1$ et $w_2$ sont \emph{équivalents} dans un shift $\cal S$, s'ils ont le même \emph{ensemble de successeurs}, i.e., si exactement les mêmes configurations infinies vers la droite apparaissent dans $\cal S$ immédiatement à droite de $w_1$ et à droite de $w_2$. Un shift est sofique si et seulement si les motifs finis du shift sont divisés en un nombre fini de classes d'équivalence (voir \cite[Théorème~3.2.10]{sft}). Intuitivement, ce critère garantit que quand nous lisons une configuration de la gauche vers la droite et vérifions que celle-ci appartient au shift, nous n'avons besoin de nous souvenir que d'une information finie.

Pour notre exemple $S_2$, il y a donc trois classes d'équivalence : 
\begin{itemize}
	\item celle des mots finissant par une case noire : l'ensemble de leurs successeurs est l'ensemble des configurations ;
	\item celle des mots finissant par un nombre \emph{pair} de cases blanches : l'ensemble de leurs successeurs sont les mots commençant par un nombre \emph{pair} de cases blanches ;
	\item celle des mots finissant par un nombre \emph{impair} de cases blanches : l'ensemble de leurs successeurs sont les mots commençant par un nombre \emph{impair} de cases blanches.
\end{itemize}
	
Revenons maintenant au shift miroir. C'est un shift effectif, puisque l'ensemble des motifs interdits le définissant ci-dessus est clairement énumérable. Il n'est par ailleurs pas sofique (et donc a fortiori pas de type fini). Nous pouvons le démontrer en utilisant la caractérisation des shifts sofiques basée sur les classes des ensembles de successeurs décrite ci-dessus, en considérant, pour tout entier $n$, les $2^n$ mots de taille $n+1$ composés d'une suite de $n$ carrés, noirs ou blancs, suivie d'un carré rouge. Pour un $n$ donné, chacun de ces $2^n$ mots a un ensemble de successeurs différents : comme $n$ est non borné, le nombre des classes d'équivalence des ensembles de successeurs est infini, et le shift miroir n'est donc pas sofique.

Les shifts de type fini, et même les shifts sofiques, sont des classes de shifts assez restrictives avec des propriétés très particulières. Sans surprise, beaucoup d'exemples importants de shifts effectifs ne sont pas sofiques. Prouver qu'un shift n'est pas sofique se fait habituellement avec une version du ``pumping lemma'' de la théorie des automates. Le pumping lemma est par exemple présenté dans la section 4.1 et 7.2 de \cite{hopcroft2001introduction} ; il est par exemple utilisé dans \cite{pumping-lemma} pour montrer que le shift des nombre premiers n'est pas sofique, ou dans \cite[Exemple~3.1.7]{sft}.

\subsection{Shifts à deux dimensions}\label{ss:2-dimension}

Le formalisme des shifts peut être étendu de manière naturelle aux grilles de $\mathbb{Z}^d$ pour $d>1$. La Proposition~\ref{p:motif-interdits} s'étend ainsi aux dimensions supérieures. Un shift de $\mathbb{Z}^d$ (sur un alphabet fini $\Sigma$) est défini comme un ensemble de configurations de dimensions $d$,  i.e., $f \ :\ \mathbb{Z}^d\to \Sigma$ qui sont (i)~invariantes par translations
(pour des translations dans toutes les directions) et (ii)~clos dans la topologie de Cantor. Comme pour le cas en dimension 1, les shifts peuvent être définis en terme de motifs finis interdits. 

Les définitions des \emph{shifts effectifs} (l'ensemble des motifs interdits est énumérable) et des \emph{shifts de type fini} (l'ensemble des motifs finis interdits est fini) s'appliquent directement aux espaces de shift à plusieurs dimensions, sans ajustement nécessaires. Comme en dimension 1, on peut définir un shift effectif par un ensemble décidable de motifs interdits (et pas seulement énumérable). Les \emph{shifts sofiques} sur $\mathbb{Z}^d$ sont définis par la Définition~\ref{d:sofique} ci-dessus (comme la projection coordonnée par coordonnée des shifts de type fini). Nous renvoyons le lecteur à \cite{hochman2009dynamics} pour une discussion exhaustive sur les relations entre shifts de type fini, shifts sofiques, et shifts effectifs sur $\mathbb{Z}^d$.

\begin{example}
Un premier exemple de shift de type fini est le shift sur l'alphabet $\{ \square, \blacksquare \}$, où apparaissent des rectangles noirs sur un fond blanc, voir Fig~\ref{f:shift-rectangle}. Le shift contient également la configuration entièrement blanche et la configuration entièrement noire, les configurations avec un quart de plan noir, et les configurations avec un demi-plan noir (ces configurations ne peuvent être évitées, par compacité de l'espace des shifts). Ce shift peut se définir par un ensemble fini de motifs interdits ${\cal M}$ (voir figure ci-dessous).

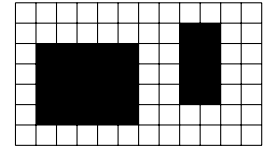
\begin{figure}
\center
\begin{tikzpicture}[scale=0.27]
\node at (-4,1) { ${\cal M}=$ \{ };
\node [scale = 2] at (-2,1) { \{ };

\fill[black] (0,0) rectangle (1,1);
\fill[black] (1,1) rectangle (2,2);
\draw (0,1) rectangle (1,2);
\draw (1,2) -- (2,2);\draw (0,0) -- (0,1);
\node at (3,0) { , };

\fill[black] (4,0) rectangle (5,1);
\fill[black] (5,1) rectangle (6,2);
\draw (5,0) rectangle (6,1);
\draw (4,0) -- (5,0);\draw (6,1) -- (6,2);
\node at (7,0) { , };

\fill[black] (8,1) rectangle (9,2);
\fill[black] (9,0) rectangle (10,1);
\draw (8,0) rectangle (9,1);
\draw (9,0) -- (10,0);\draw (8,1) -- (8,2);
\node at (11,0) { , };

\fill[black] (12,1) rectangle (13,2);
\fill[black] (13,0) rectangle (14,1);
\draw (13,1) rectangle (14,2);
\draw (12,2) -- (13,2);\draw (14,0) -- (14,1);

\node [scale = 2] at (15,1) { \} };
\node at (16,0) { . };

\node at (25,1) { Une configuration : };

\draw (34,0) grid (46,7);
\fill[black] (35,1) rectangle (40,5);
\fill[black] (42,2) rectangle (44,6);
\end{tikzpicture}
\caption{Le shift de type fini dont les configurations sont des rectangles noirs sur une mer de cases blanches.}\label{f:shift-rectangle}
\end{figure}

\end{example}

\paragraph{Un formalisme particulier pour les shifts de type fini : les tuiles de Wang.}
Les tuiles de Wang furent initialement introduites dans \cite{wang} pour étudier des fragments de la logique du premier ordre.
 
\begin{definition}
Un jeu de tuiles de Wang $\tau$ est constitué d'un ensemble de carrés, appelés tuiles, dont les couleurs des bords appartiennent à un ensemble de couleurs fini (les tuiles ne peuvent pas être tournées). Un pavage valide de $\tau$ est une configuration de $\tau^{\mathbb{Z}^2}$ où, pour chaque paire de tuiles voisines, les couleurs de leur bord commun correspondent.
\end{definition}

\begin{example}
Un jeu de tuiles apériodique (i.e., dont toutes les configurations sont apériodiques), minimal en nombre de tuiles et en nombre de couleurs (voir Fig~\ref{f:wang} et \cite{minimal-aperiodic-tile-set})~:
\begin{figure}[h]
	\centering
	\begin{subfigure}{.3\textwidth}
		\centering
		\includegraphics[scale = 0.3]{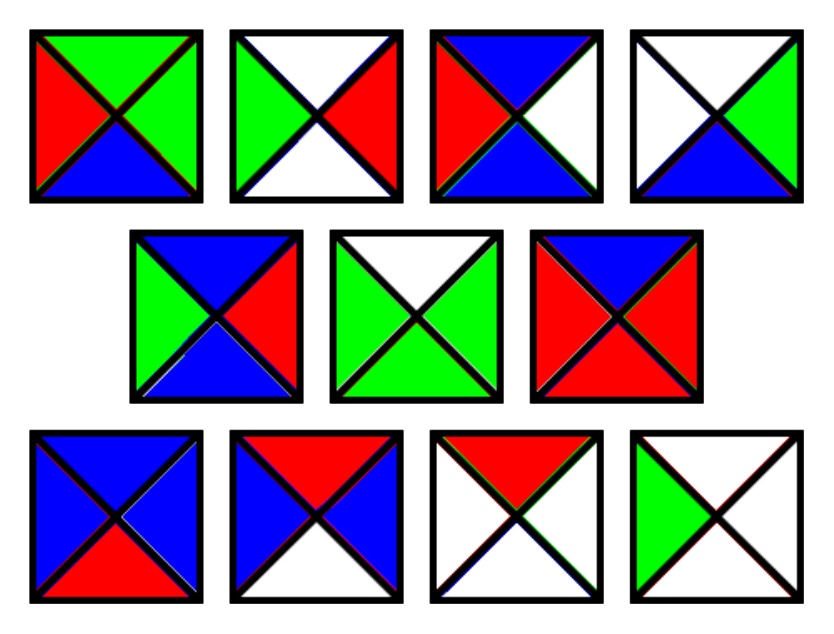}
		\caption{Le jeu de tuiles.}
	\end{subfigure}
	\begin{subfigure}{.3\textwidth}
		\centering
		\includegraphics[scale = 0.3]{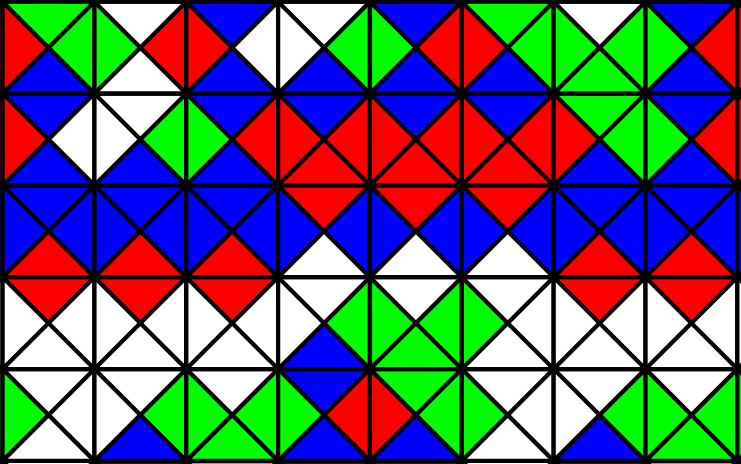}
		\caption{Une configuration.}
	\end{subfigure}
	\caption{Un jeu de tuiles apériodique de 11 tuiles.}\label{f:wang}
\end{figure}
\end{example}

Clairement, un jeu de tuiles de Wang est un shift de type fini (les lettres de l'alphabet sont les tuiles, et les motifs interdits, de taille $1\times 2$ ou $2 \times 1$, sont les paires de tuiles voisines dont les couleurs ne correspondent pas). La proposition suivante montre que les jeux de tuiles sont en réalité équivalents aux shifts de type fini.

\begin{proposition}\label{p:shift-tuiles}
Pour tout shift de type fini $S$, il existe un jeu de tuiles $\tau$ tel que l'ensemble des configurations de $S$ est isomorphe à l'ensemble des pavages de $\tau$.
\end{proposition}

Un isomorphisme entre deux shifts est une bijection, continue dans le sens topologique, et qui commute avec l'opération du shift. 

\begin{proof}
Soit un shift de type fini $S$ sur un alphabet $A$. Nous considérons un ensemble fini ${\cal I}$ de motifs interdits définissant $S$. Soit $l$ la largeur maximale des motifs interdits, $h$ leur hauteur maximale, et $c$ le maximum entre $l$ et $h$. Nous pouvons définir un autre ensemble fini ${\cal I}'$ de motifs interdits définissant $S$, où tous les motifs interdits sont de taille $c \times c$. Nous pouvons alors également définir $S$ par l'ensemble fini ${\cal A} = A^{c^2} - {\cal I}'$ des motifs autorisés de taille $c \times c$. 

Nous définissons alors un jeu de tuiles $\tau$ dont l'ensemble des couleurs des bords est l'ensemble ${\cal A}$. Intuitivement, une tuile située à la position $p=(x,y)$ sera chargée de ``vérifier'' que dans le carré $C$ de taille $(c+1) \times (c+1)$ de $S$ dont le coin en bas à gauche est à la position $p$, le sous-carré de taille $c\times c$ situé en bas à gauche de $C$ est compatible avec les sous-carrés de $C$ de même taille situé en haut à gauche et en bas à droite.

Ainsi, pour chaque motif carré $C$ de taille $(c+1) \times (c+1)$ de $S$, nous ajoutons une tuile $t$ au jeu de tuiles $\tau$, où les couleurs de gauche et du bas de $t$ correspondent au carré de taille $c \times c$ situé en bas à gauche de $C$, et les couleurs du haut et de droite de correspondent aux carrés de taille $c \times c$ situés respectivement en haut à gauche et en bas à droite de $C$. Nous définissons également une fonction bijective $\phi$ avec $\phi(C)=t$.

Nous pouvons alors définir $\pi$ de l'ensemble des configurations de $S$ vers l'ensemble des pavages de $\tau$ de la manière suivante : pour chaque position $p=(x,y)$ d'une configuration ${\cal C}$ de $S$, nous considérons le carré $C$ de ${\cal C}$ dont le coin en bas à gauche est à la position $p$, et associons à cette position $p$ la tuile $\phi(C)$ (réciproquement, à un pavage ${\cal P}$ de $\tau$ la fonction $\pi^{-1}$ associe à chaque tuile $t$ à une position $(x,y)$ dans ${\cal P}$ la lettre en bas à gauche du carré $\phi^{-1}(t)$). Nous pouvons alors vérifier que $\pi$ est bien un isomorphisme des configurations de $S$ vers les pavages de $\tau$ ($\pi$ commute bien avec l'opération du shift).
\end{proof}

\paragraph{Des shifts sofiques qui ne sont pas de type fini :}
\begin{example}
Un exemple de shift sofique qui n'est pas de type fini est le shift $S$ sur l'alphabet $\{ \square, \blacksquare \}$, dont les configurations sont des carrés noirs sur un fond blanc (plus de nouveau la configurations entièrement blanche, celle entièrement noire, les quart de plan noirs et les demi-plans noirs), voir Fig~\ref{f:shift-carrés}. 

Pour prouver que le shift $S$ est sofique, nous définissons un shift $S'$ de type fini et une projection $\pi$, tels que $\pi(S')=S$. Intuitivement, la définition de $S$ repose sur le fait qu'un carré de taille $n \times n$ (avec $n\geq 3$) est constitué d'un carré de taille $(n-2) \times (n-2)$ entouré par une bordure de largeur 1. Si $(n-2)>3$, nous pouvons répéter l'opération, et ainsi de suite jusqu'à obtenir un carré central de taille $2 \times 2$ ou constitué d'une seule case. Remarquons que si l'on applique le même procédé à partir d'un rectangle qui n'est pas un carré, nous obtenons au centre de celui-ci un rectangle de taille $2 \times h$ ou $l \times 2$, avec $l>2$ et $h>2$.     

Pour forcer cette structure, l'alphabet de $S'$ est formé de 4 lettres pour les 4 coins des bordures ($\ulcorner, \urcorner, \llcorner$ et $\lrcorner$), de 2 lettres pour les bords horizontaux et verticaux ($-, |$) et de la lettre $\blacksquare$ pour le carré central constitué d'une seule case ; toutes ces lettres sont envoyées par $\pi$ sur la lettre $\blacksquare$ de l'alphabet de $S$. Enfin, l'alphabet de $S'$ contient la lettre $\square$, avec $\pi(\square )=\square$. Les motifs interdits de $S'$ sont de tailles au plus $4 \times 4$, et garantissent la structure des bordures imbriquées les unes dans les autres, avec au centre un des deux carrés possibles. En effet, nous pouvons montrer qu'un ensemble de cases connexes qui ne sont pas blanches est soit un carré, soit un demi-plan soit un quart de plan, en partant de la bordure de l'ensemble de case et en considérant les couches successives ; des motifs interdits de taille $4 \times 4$ sont suffisants pour interdire les autres configurations.

\begin{figure}
\center

\begin{tikzpicture}[scale=0.27]
\draw (0,0) grid (20,12);

\draw [ultra thick] (2.5,2.5) rectangle (8.5,8.5);
\draw [ultra thick] (3.5,3.5) rectangle (7.5,7.5);
\draw [ultra thick] (4.5,4.5) rectangle (6.5,6.5);
\fill [black] (5,5) rectangle (6,6);

\draw [ultra thick] (12.5,4.5) rectangle (17.5,9.5);
\draw [ultra thick] (13.5,5.5) rectangle (16.5,8.5);
\draw [ultra thick] (14.5,6.5) rectangle (15.5,7.5);

\fill [black] (13,1) rectangle (14,2);
\draw [ultra thick] (16.5,1.5) rectangle (17.5,2.5);

\node at (10,-2) { Une configuration de $S'$ };

\draw[->,ultra thick] (22.5,6) -- (27.5,6);
\node at (25,7) {$\pi$};

\draw (30,0) grid (50,12);
\fill [black] (32,2) rectangle (39,9);
\fill [black] (42,4) rectangle (48,10);
\fill [black] (43,1) rectangle (44,2);
\fill [black] (46,1) rectangle (48,3);
\node at (40,-2) { La configuration de $S$ corresondante };

\end{tikzpicture}
\caption{Le shift sofique dont les configurations sont des carrés noirs sur une mer de cases blanches.}\label{f:shift-carrés}
\end{figure}
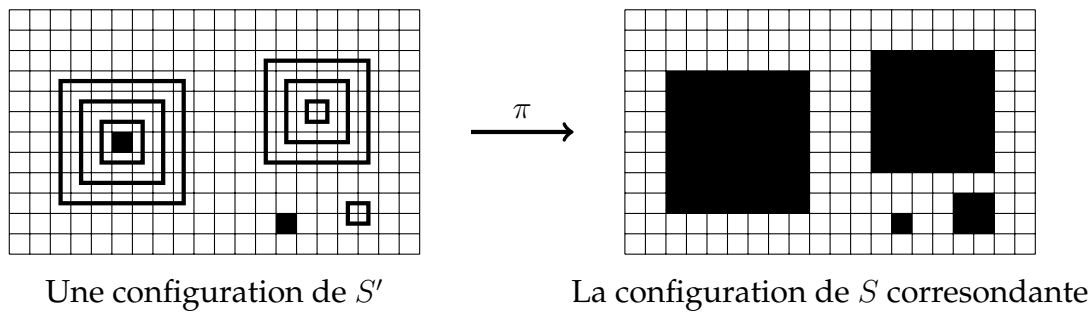

Montrons maintenant que $S$ n'est pas un shift de type fini, d'une manière similaire à celle utilisée pour le shift unidimensionnel $S_2$ de la sous-partie précédente. Nous supposons par l'absurde que $S$ est de type fini ; soit ${\cal M}$ un ensemble fini de motifs interdits définissant $S$. Soit $l$ et $h$ respectivement la largeur et la hauteur maximales des motifs interdits de ${\cal M}$, et $c$ le maximum entre $l$ et $h$ (tous les motifs interdits de $S$ sont donc inclus dans le carré de taille $c \times c$. Nous considérons alors deux configurations entièrement blanches, à l'exception pour l'une d'un carré de taille $c\times c$, et pour l'autre d'un rectangle de taille $(c+1)\times c$. Ces deux configurations possèdent les mêmes motifs inclus dans le carré de taille $c \times c$ ; par conséquent, les deux configurations appartiennent à $S$ ou aucune des deux n'y appartient, ce qui  est absurde.
\end{example}

En deux dimensions, les shifts sofiques sont une classe étonnamment large, et montrer qu'un shift est sofique peut nécessiter des arguments subtils. \label{p:ex-westrick} Nous pouvons citer en exemple le shift sofique présenté dans \cite{westrick2017seas}, qui est un sous-shift du shift de l'exemple précédent pour lequel les carrés noirs sont tous de tailles différentes (ou même dont les tailles appartiennent à un ensemble $\Pi_0^1$). Ce résultat peut paraître surprenant, dans la mesure ou deux carrés de même taille peuvent apparaître arbitrairement loin.

Un autre exemple de shifts sofiques sont ceux obtenus à partir de shifts effectifs unidimensionnels, en répétant chaque lettre sur une bande verticale, voir \cite{drs,aubrun-sablik}. 

\subsection{Montrer qu'un shift multidimensionnel n'est pas sofique}

Le sujet principal de ce manuscrit est d'étudier la limite de séparation, pour les shifts de dimensions deux, entre les \emph{shifts sofiques} et \emph{les shifts effectifs et non sofiques}. Comme nous l'avons mentionné ci-dessus, pour les shifts multidimensionnels la classe des shifts sofiques est très riche. Bien sûr, seuls les shifts effectifs peuvent être sofiques. Cependant, il peut être difficile de montrer qu'un shift n'est pas sofique. Nous commençons notre étude en présentant les techniques classiques qui peuvent être utilisées pour montrer qu'un shift \emph{n'est pas sofique}.

La manière usuelle de montrer qu'un shift n'est pas sofique est d'étudier le ``flux d'information'' ou la ``densité d'information'' associés à ce shift. Cette idée est similaire à la technique des ``ensembles de successeurs'' mentionnée ci-dessus à la page~\pageref{p:successors-1d}. Intuitivement, nous voulons mesurer l'information requise pour certifier qu'un motif $P$ de taille $N \times N$ est globalement admissible, ou certifier qu'un motif $F$ défini sur l'ensemble des cases en dehors de $P$ et le motif $P$ sont compatibles l'un avec l'autre.

\begin{example}
Précisons cette notion de flux d'information à travers deux exemples de shifts très basiques, sur l'alphabet $\{ \square, \blacksquare\}$. Tout d'abord, considérons le shift possédant au plus \emph{une} case noire. Dans ce cas, le flux d'information entre un motif $P$ de taille $N \times N$ et le motif $F$ correspondant, nécessaire pour certifier que les deux motifs sont compatibles, est juste un bit, indiquant si $P$ contient une case noire ou non.

Considérons maintenant le cas du shift où deux cases noires ne peuvent être voisines, et de nouveau un motif $P$ de taille $N\times N$ et le motif $F$ correspondant à $P$. Pour garantir que $P$ et $F$ sont compatibles, il faut et il suffit de connaître la position des cases noires situées sur la bordure de $P$, et le flux d'information est ici en $O(N)$ bits (pour chacune des $4(-1)N$ cases de la bordure nous pouvons indiquer par un bit si elle est noire ou blanche).

En outre, ces deux shifts contiennent la configuration entièrement blanche, et donc des motifs de taille arbitrairement grande avec une densité d'information nulle.

Il n'est pas difficile de voir que ces deux exemples sont des shifts sofiques (le second est même de type fini). Cela n'est pas surprenant, dans la mesure où la ``densité d'information'' et le ``flux d'information'' (du moins dans leur sens intuitif) encodés dans ces shifts sont faibles.
\end{example}

Intuitivement, les idées derrière les preuves classiques qu'un shift n'est pas sofique reposent sur l'observation suivante : pour les shifts de types finis ainsi que pour les shifts sofiques, la valeur du ``flux d'information'' traversant la bordure d'un motif $N\times N$, ou encore la quantité d'information encodée dans un motif de taille $N \times N$, ne peuvent être supérieurs à $O(N)$. Nous illustrons cette idée avec deux exemples bien connus : le shift miroir (en deux dimensions) et le shift de grande complexité.

\begin{figure}
\vspace{-15pt}
\begin{center}
  \begin{tikzpicture}[scale=0.27,x=1cm,baseline=2.125cm,ultra thick/.style= {line width=3.0pt},]
      \foreach \x in {0,...,15} 
    {
       \path[fill=red!41,draw=black] (\x,8) rectangle ++ (1,1);

    }

 \pgfmathsetseed{4}
    \foreach \x in {0,...,15} \foreach \y in {0,...,7}
    {
        \pgfmathparse{mod(int(random*23),2) ? "black!66" : "black!7"}
        \edef\colour{\pgfmathresult}
        \path[fill=\colour,draw=black] (\x,\y) rectangle ++ (1,1);
        \path[fill=\colour,draw=black] (\x,8+8-\y) rectangle ++ (1,1);
    }
    
    \draw [draw=blue,ultra thick]  (3,2) rectangle (7,6);
    \draw [draw=blue,ultra thick]  (3,11) rectangle (7,15);
       
\end{tikzpicture}
\caption[Une configuration avec une symétrie en miroir par rapport à une ligne rouge horizontale.]{Une configuration avec une symétrie en miroir par rapport à une ligne rouge horizontale.
Les carrés bleus délimitent deux motifs noir et blanc symétriques.} 
\label{f:mirror}
\end{center}
\vspace{-15pt}
\end{figure}
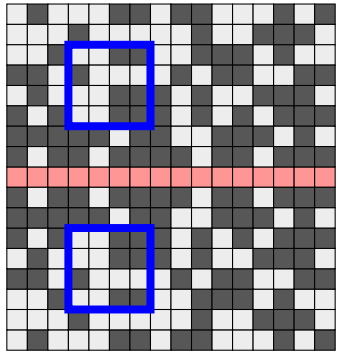

\begin{example}[le shift miroir]\label{ex:mirror}
Un des exemples standards de shift non-sofique est le shift des configurations symétriques en miroir sur $\mathbb{Z}^2$. Cet exemple illustre l'idée intuitive que dans un shift sofique la valeur du ``flux d'information'' traversant la bordure d'un motif doit être proportionnelle à la longueur du contour (et donc que si cette propriété n'est pas respectée le shift n'est pas sofique).

Soit $\Sigma = \{ \blacksquare, \square, {\color{red} \blacksquare} \}$ un alphabet de trois lettres ; les configurations du shift sont toutes les configurations constituées de cases noires et blanches (sans aucune case rouge), et les configurations avec une ligne de cases rouges infinie horizontale, et au-dessus et en-dessous d'elle deux demi-plans de cases noires et blanches symétriques, voir Fig.~\ref{f:mirror}. Nous appelons ce shift $S_{\text{mirror}}$.

Il est facile de voir que ce shift est effectif (les motifs interdits sont ceux où les cases rouges ne sont pas alignées, et ceux où les zones au-dessus et en-dessous de la ligne de cases rouges ne sont pas symétriques). Par ailleurs, le shift n'est pas sofique. L'explication intuitive de ce résultat est la suivante : considérons une paire de motifs noir et blanc symétriques de taille $N\times N$, au-dessus et en-dessous de la ligne rouge horizontale (voir les carrés bleus dans la Fig.~\ref{f:mirror}). Pour être sûr que la configuration appartienne au shift, nous devons ``comparer'' ces deux motifs l'un par rapport à l'autre.
Pour cela, nous devons transmettre l'information à propos d'un motif de taille $N^2$ à travers sa bordure (de taille $O(N)$). Cependant, dans un shift sofique, le ``flux d'information'' qui traverse un contour de taille $O(N)$ est limité par $O(N)$ bits, et cette contradiction implique que le shift ne soit pas sofique.
Pour un argument plus formel, voir par exemple \cite{mirror} et \cite{jeandel-hdr}, ou un exemple similaire dans \cite[Exemple~2.4]{kass-madden}.
Nous revisitons cet exemple dans la Sous-Partie~\ref{s:plain-epitomes} et l'utilisons comme une illustration simple de notre technique.
\end{example}

\begin{example}[Le shift de grande complexité]
\label{ex:complexity}
Cet exemple illustre l'idée intuitive que dans un shift sofique la densité d'information encodée dans un motif globalement admissible de taille $N\times N$ ne peut être plus grande que $O(N)$ (et donc que si cette propriété n'est pas respectée, le shift n'est pas sofique).

Soit $S$ l'ensemble de toutes les configurations binaires de $\mathbb{Z}^2$ où pour chaque motif $P$ de taille $N\times N$ sa complexité de Kolmogorov est quadratique, i.e. $C(P)=\Omega(N^2)$. Techniquement, cela signifie qu'aucun motif globalement admissible ne peut être produit par un programme de taille inférieure à $c \cdot N^2$, pour un certain facteur $c>0$ (voir la définition formelle de la complexité de Kolmogorov ci-dessous).

Donc, par construction, dans ce shift tous les motifs globalement admissibles de taille $N\times N$ contiennent $\Omega(N^2)$ bits d'information dans le sens de la complexité de Kolmogorov. Cela est beaucoup plus que $O(N)$, la quantité d'information qui peut être réalisée par un shift sofique ; par conséquent, ce shift n'est pas sofique. Formellement, cela découle de deux faits (prouvés dans \cite{dls}) :
\begin{itemize}
\item[(i)] Pour un certain $c<1$, le shift défini ci-dessus n'est pas vide.
\item[(ii)] Dans tout shift sofique non vide de $\mathbb{Z}^2$, il existe une configuration où la complexité de Kolmogorov de tous les motifs de taille $N\times N$ est bornée par $O(N)$.
\end{itemize}
Par ailleurs, le shift est clairement effectif : nous pouvons énumérer par un algorithme les différents motifs dont la complexité de Kolmogorov est plus petite que le seuil défini. Dans la Section~\ref{s:grande-complexité-bornée} nous étudions une construction similaire basée sur une version moins usuelle de la complexité de Kolmogorov. 
\end{example}

Nous pouvons remarquer que les shifts non sofiques des deux exemples précédents ont une entropie positive (le nombre de motifs globalement admissibles de taille $N \times N$ croît en $2^{\Omega(n^2)}$). Cela n'est pas surprenant : pour montrer que ces shifts ne sont pas sofiques, les preuves sont basées sur l'intuition qu'une trop grande quantité d'``information'' est présente dans un motif (soit dans le sens de Kolmogorov, soit dans le sens de flux d'information entre le motif et son extérieur). Ces types d'arguments peuvent être adaptés à des shifts pour lesquels le nombre de motifs globalement admissibles de taille $N\times N$ croît plus lentement que  $2^{\Omega(N^2)}$, mais néanmoins plus vite que $2^{O(N)}$ (voir par exemple \cite[Exemple~2.5]{kass-madden}).

Comme remarqué dans \cite{westrick2017seas}, ``\emph{tous les exemples connus par l'auteure de shifts effectifs mais non sofiques sont obtenus, dans un certain sens, en encodant trop d'information dans une petite zone.}''

Dans les exemples ci-dessus nous avons illustré l'intuition derrière les preuves standards qu'un shift multidimensionnel n'est pas sofique. Jusqu'à présent nous n'avons pas donné de preuves formelles (nous y reviendrons plus tard, dans le Chapitre~\ref{c:non-sofique}). Contrairement au cas des shifts unidimensionnels, les preuves en dimensions supérieures ne sont pas aussi faciles et directes. En effet, pour le cas des shifts unidimensionnels l'intuition de ``flux d'information'' peut être formalisée par le langage des ensembles d'extensions, voir page~\pageref{p:successors-1d} ; pour les shifts multidimensionnels, une telle approche échoue. Nous discutons plus en détail des limitations de la technique des ensembles d'extension page~\pageref{p:extenders-do-not-work}. Cependant, il y a plusieurs tentatives fructueuses pour convertir cette intuition en une méthode mathématique formelle. Par exemple, ce type d'argument a été formalisé comme des conditions assez générales permettant de montrer qu'un shift n'est pas sofique, dans \cite{pavlov} et \cite{kass-madden}. Les théorèmes de Kass et Madden (\cite[Theorem~3.2.10]{kass-madden}) et de Pavlov (\cite[Theorem~1.1]{pavlov}) ne s'appliquent que dans les cas des shifts en deux dimensions où le nombre de motifs globalement admissibles de taille $N \times N$ est plus grand que $2^{O(N)}$. 

Cependant, il n'y a pas de raison de penser que ce soit une condition nécessaire pour qu'un shift ne soit pas sofique (voir, par exemple, la discussion dans \cite[Section~1.2.2]{jeandel-hdr}). Il est instructif d'observer que les shifts non effectifs (et donc non sofiques) peuvent avoir une complexité combinatoire (i.e. un nombre de motifs de taille $N \times N$, voir la partie ci-après) très petite, voir \cite{kass-madden,ormes-pavlov}.

Nous avons vu dans cette première sous-partie que l'étude des shifts s'est d'abord faite en dimension un, et que dans ce cas ceux-ci sont assez bien compris ; en particulier, la classe des shifts sofiques y est bien caractérisée. Il y a cependant un fossé entre la dimension un et la dimension deux : la plupart des résultats valables en dimension un ne s'appliquent plus, et en particulier il est difficile de caractériser les shifts sofiques. Ce manuscrit a pour but d'étudier la frontière séparant les shifts effectifs et non-sofiques des shifts sofiques. Nous revisitons l'intuition de ``densité d'information'' et de ``flux d'information'' dans les shifts multidimensionnels. Nous étudions des outils mathématiques formalisant cette intuition et construisons de nouvelles méthodes formelles pour prouver qu'un shift est sofique ou non sofique.

Dans ce but, nous allons utiliser plusieurs notions de complexité qui peuvent être attachées à un shift, que nous présentons dans la sous-partie suivante.

\section{Différentes notions de complexité}

Dans nos preuves nous utilisons à de nombreuses reprises plusieurs mesures formelles correspondant intuitivement à la notion de ``quantité d'information''. Ces mesures capturent chacune dans un certain sens l'idée de ``complexité'' d'un shift. Pour rendre plus clair le contexte général, nous rappelons différentes approches existantes pour mesurer la ``complexité'' qui sont utilisées en dynamique symbolique.

\paragraph{1. La complexité logique des ensembles de motifs interdits.} La classification de tous les shifts en \emph{shifts de type fini}, \emph{shifts sofiques}, \emph{shifts effectifs}, et finalement \emph{shifts non effectifs} peut être vue comme une mesure de leur complexités. En effet, cette hiérarchie classifie à quel point l'ensemble de motifs interdits définissant un shift peut être complexe. Rappelons que : $[\text{les shifts de type fini}] \subsetneqq [\text{les shifts sofiques}] \subsetneqq  [\text{les shifts effectifs}]$.

\paragraph{2. La complexité combinatoire des motifs globalement admissibles.} La complexité en termes de motifs globalement admissibles, ou complexité combinatoire. Dans ce cas nous ne nous focalisons pas sur la complexité ``logique'' des ensembles des motifs admissibles, mais plutôt sur le nombre de motifs globalement admissibles de différentes tailles. Dans cette approche, la question typique est combien de motifs globalement admissibles de taille $N \times N$ existent pour un shift donné. Si le nombre de motifs globalement admissibles est sous-exponentiel, ou encore polynomial, ou même linéaire, alors la complexité combinatoire d'un shift est basse. Si le nombre de motifs globalement admissibles croît aussi rapidement que $2^{h N^2+ o(N^2)}$, alors la complexité combinatoire du shift est élevée. Le coefficient $h$ est appelé l'entropie du shift, et pour les shift de basse complexité l'entropie est nulle. 

Il existe des shifts pour lesquels le nombre de motifs globalement admissibles a une croissance intermédiaire, i.e., elle peut être en $2^{f(N)}$ pour une fonction telle que $N\ll f(N) \ll N^2$. Ils ont une entropie de nulle, mais peuvent néanmoins avoir une grande diversité de motifs globalement admissibles.

Le fait pour un shift d'avoir une faible complexité combinatoire est relié au fait d'être périodique ou apériodique. La conjecture de Nivat (voir \cite{nivat}) affirme ainsi que si une configuration d'un shift de type fini est de faible complexité combinatoire (i.e. il y apparaît moins de $M \cdot N$ rectangles différents de taille $M \times N$), alors cette configuration est périodique. Nous pouvons remarquer qu'ici la complexité combinatoire est basée sur des rectangles et non sur des carrés, comme ce sera le cas dans ce manuscrit.  

La conjecture de Nivat n'est toujours pas résolue. Néanmoins, des progrès significatifs ont été réalisés récemment. En particulier, il a été montré (voir \cite{moutot-kari} et \cite{moutot}) que si une configuration ${\cal C}$ en dimension deux a une faible complexité par bloc (il y a un nombre de motif de taille $M \times N$ inférieur à $M \cdot N$), alors il existe une configuration périodique ${\cal C}'$ qui appartient à la clôture topologique de l'orbite de ${\cal C}$ (l'orbite est l'ensemble de toutes les translations de ${\cal C}$). Il s'ensuit que si le nombre de motifs rectangulaires de taille $M \times N$ globalement admissibles dans un shift $S$ à deux dimensions est inférieur à $M \cdot N$, alors le shift $S$ contient au moins une configuration périodique. En d'autres termes, une version de la conjecture de Nivat pour les shifts est vraie.

\paragraph{3. La complexité dans le sens de la théorie de la calculabilité : que peut-on dire à propos des degrés de non-calculabilité des configurations d'un shift ?} Il y a eu une longue série de travaux sur les degrés de non-calculabilité pour les configurations des shifts en dimension un (voir \cite{Turing-degree-spectra,word-computational}). Une étude assez complète des degrés Turing pour les shifts de type fini multidimensionnels a été réalisée dans \cite{Pi01}. La recherche des degrés Turing des shifts multidimensionnels a été faite dans \cite{Turing-degree-spectra}. Nous n'entrerons pas dans les détails de ces études puisque cela n'est pas le sujet de la thèse. Néanmoins, cette approche a partiellement motivé l'approche basée sur la complexité de Kolmogorov, qui est le principal outil technique de ce travail.

\paragraph{4. Complexité algorithmique des motifs globalement admissibles.}
Une approche de la complexité assez nouvelle, étudiée dans ce manuscrit : la complexité en termes de description algorithmique (ou complexité de Kolmogorov). Dans un certain sens, cette approche combine les intuitions derrière les approches (2) et (3). Nous étudions la complexité de Kolmogorov d'un motif $P$ globalement admissible de taille $N\times N$. Intuitivement, la complexité de Kolmogorov de  $P$ est la taille du programme minimal qui peut produire $P$. Dans cette approche nous nous intéressons à la complexité de Kolmogorov minimale ou maximale des motifs globalement admissibles de taille $N\times N$. D'une certaine manière, cette approche peut sembler similaire à celle de la complexité combinatoire. Dans une direction la relation est évidente : si la complexité de Kolmogorov de tous les motifs globalement admissibles est basse, alors tous les motifs peuvent être générés par des programmes courts ; puisqu'il y a au plus $2^k$ programmes de taille $k$, nous pouvons conclure que la complexité combinatoire est également basse. Par exemple, si la complexité de Kolmogorov de tous les motifs globalement admissibles de taille $N \times N$ est bornée par $O(\log N)$, alors la complexité combinatoire du shift est bornée par $2^{O(\log N)} = {\rm poly}(N)$.

Cependant, l'approche basée sur la complexité de Kolmogorov semble être bien plus flexible : nous pouvons considérer la complexité de Kolmogorov simple, ou la complexité de Kolmogorov avec un oracle particulier, ou la complexité de Kolmogorov pour des algorithmes ayant des restrictions sur les ressources (en temps et en espace), et ainsi de suite.

Ces 4 ``axes'' de complexité sont connectés les uns aux autres. Dans ce travail nous étudions principalement l'approche basée sur les différentes versions de la complexité algorithmique, et sur les connections entre celles-ci et les mesures de complexité plus classiques (1) et (2).

\section{Principaux résultats}

Nous étudions la ligne de démarcation entre les shifts sofiques et les shifts effectifs mais non sofiques. Le principal défi est de trouver une caractérisation permettant de distinguer les shifts multidimensionnels sofiques ou non sofiques. Nous pensons que la complexité algorithmique (la complexité d'un shift mesurée en terme de complexité de Kolmogorov de ses motifs globalement admissibles) peut s'avérer utile dans cette tâche. Une intuition naïve suggère que nous pourrions avoir une caractérisation en deux parties :

\begin{itemize}
\item si un shift à deux dimensions est sofique, alors il est effectif et l'information ``essentielle'' dans chaque motif globalement admissible de taille $N\times N$ doit avoir une description algorithmique courte et efficace ;
\item si un shift à deux dimensions est effectif et que pour l'information ``essentielle'' dans chaque motif globalement admissible de taille $N\times N$ nous avons une description algorithmique suffisamment courte et efficace, alors ce shift est sofique.
\end{itemize}

Nous somme encore très loin d'avoir une caractérisation complète de ce genre, nous ne pouvons pas prouver un théorème ``si et seulement si'' de ce type. Cependant, nous avons fait quelques progrès en approchant par les deux directions la caractérisation voulue (dans la partie ``si'' et dans la partie ``seulement si'').

\textbf{Chapitre 2} : dans ce chapitre nous démontrons plusieurs résultats, qui intuitivement affirment que dans tout shift à deux dimensions sofique, tout motif globalement admissible de taille $N \times N$ doit être ``simple''. Les résultats de ce chapitre sont publiés dans~\citeMe{stacs} ; la version longue de l'article a été soumise à un journal et est disponible sur arXiv (\citeMe{stacs-submitted}). Les figures de ce chapitre sont reprises de ces articles.

Nous proposons plusieurs énoncés techniques de cette idée générale. En particulier, nous prouvons que si un shift est sofique, alors il contient une configuration pour laquelle tous les motifs carrés sont de faible complexité de Kolmogorov à ressources bornées (i.e., il existe un petit programme calculant un de ces motifs rapidement, voir la définition formelle page~\pageref{d:kolmogorov-bornée}). Nous utilisons alors cette propriété des shifts sofiques pour construire un shift effectif et non sofique (car ne respectant pas cette propriété), et ayant une complexité combinatoire très faible (seulement polynomiale). Techniquement, les motifs de taille $N\times N$ du shift construit auront à la fois une complexité de Kolmogorov simple très basse, en $O(\log N)$, et une complexité de Kolmogorov à ressources bornées élevée, en $\Omega(N^{2-\epsilon})$.

Le principal outil de ce chapitre est la complexité de Kolmogorov, et la principale nouveauté de notre approche est l'utilisation de la technique de la complexité de Kolmogorov à ressources bornées.

D'autres outils pour prouver qu'un shift n'est pas sofique sont développés dans le reste du chapitre.

\textbf{Chapitre 3} : dans ce chapitre, nous prouvons plusieurs résultats affirmant qu'un shift effectif dont tous les motifs globalement admissibles ont une très faible complexité descriptive est sofique. En particulier,
nous montrons que pour $\epsilon <1$, le shift $S$ ayant pour alphabet les lettres blanches et noires, et pour contrainte que le nombre de lettres noires dans un carré de taille $N \times N$ est inférieur à $N^{\epsilon}$ est sofique (sous réserve que $\epsilon$ soit calculable par en haut). De plus, tout sous-shift effectif de $S$ est également sofique. Nos résultats sont motivés par le théorème de L.B.~Westrick (voir \cite{westrick2017seas}), qui a prouvé que pour un ensemble $X\in \Pi_1^0$, le shift $S$ dont les configurations sont constituées de carrés noirs sur une mer de cases blanches, dont les tailles des carrés sont distinctes et appartiennent à $X$, est sofique. De même, tout sous-shift effectif de $S$ est sofique.

Le principal outil technique de ce chapitre est la technique des pavages à point-fixe (introduite et développée dans \cite{fixed-point}). Nous utilisons également la technique de flots dans des graphes (le principe ``flot maximal/coupe minimale'') et un modèle de calcul en parallèle relativement exotique basé sur les automates cellulaire non déterministes.

\section{Définitions basiques}

\subsection{Dynamique symbolique}

Nous avons déjà introduit ci-dessus les principales notions de la théorie des shifts. Pour le confort du lecteur, nous résumons dans cette sous-partie les principales définitions concernant les espaces des shifts.

Les shifts étudiés dans ce manuscrit étant de dimension deux, nous donnons les définitions formelles pour les shifts dans cette dimension.

Un \emph{alphabet} est un ensemble fini d'éléments, appelés \emph{lettres} ou \emph{couleurs}. 
Une \emph{configuration} sur un alphabet $A$ est un élément de $A^{\mathbb{Z}^2}$.
Un \emph{support} $F$ est une partie de la grille $\mathbb{Z}^2$, qui peut être finie ou infinie : $F \subset \mathbb{Z}^2$.
Un \emph{motif} sur un alphabet $A$ et de support $F$, est un élément de $A^F$.

Soit $(x,y)\in \mathbb{Z}^2$ des coordonnées sur la grille, et $\vec{u}$ la translation associée. Soit ${\cal C}$ une configuration, et ${\cal C}_{\vec{u}}$ la configuration ${\cal C}$ translatée selon $\vec{u}$.
Nous disons qu'un motif $P$, de support $F$, apparaît dans ${\cal C}$ à la position $(x,y)$ si la restriction de ${\cal C}_{\vec{u}}$ à $F$ est égale à $P$ : $\left. {\cal C}_{\vec{u}} \right| _F=P$. Plus généralement, nous disons qu'un motif $P$ de support $F$ \emph{apparaît} dans une configuration ${\cal C}$ s'il existe une translation $\vec{u}$ telle que $\left. {\cal C}_{\vec{u}} \right| _F=P$. 
De manière similaire, soient $P$ et $Q$ deux motifs de supports respectivement $F$ et $H$, $\vec{u}$ une translation et $Q_{\vec{u}}$ le motif de support $H_{\vec{u}}$ obtenu en translatant $Q$ selon $\vec{u}$. Alors si $\left. P \right| _{H_{\vec{u}}}=Q_{\vec{u}}$, nous disons que $Q$ \emph{apparaît} dans $P$. 

Un \emph{shift} $S$, sur un alphabet $A$, est défini par un ensemble fini ou infini ${\cal M}$, constitué de motifs dont les supports sont finis, appelé l'ensemble des motifs interdits. Le shift est l'ensemble des configurations sur l'alphabet $A$ pour lesquelles aucun motif interdit de ${\cal M}$ n'apparaît.

Si $S$ et $S'$ sont des shifts tels que $S' \subset S$, on dit que $S'$ est un sous-shift de $S$.

Soient $S$ un shift sur l'alphabet $A$, défini par l'ensemble des motifs interdits ${\cal M}$, et $P$ un motif sur l'alphabet $A$. Si aucun motif de ${\cal M}$ n'apparaît dans $P$, alors nous disons que $P$ est un motif localement admissible de $S$. De plus, si il existe une configuration de $S$ où $P$ apparaît, $P$ est un motif globalement admissible de $S$. 

\subsection{Automates cellulaires}

Pour une introduction aux automates cellulaires, voir par exemple \cite{delorme1998cellular}.

Nous nous intéresserons aux automates cellulaires en dimension un. Comme pour les shift, un automate est défini sur un \emph{alphabet} $A$, et une \emph{configuration} d'un automate est un élément de $A^{\mathbb{Z}}$. L'\emph{état} d'une case de $\mathbb{Z}$, ou \emph{cellule}, à un temps $t+1$ dépend de l'état de la cellule ainsi que de ses voisines de gauche et de droite (nous considérons des automates cellulaires avec un voisinage de taille 1). Ainsi, un tel automate cellulaire en une dimension sur un alphabet $A$ est défini
par une \emph{règle} :
\begin{itemize}
	\item soit par une fonction $f : A^3 \rightarrow A$, et dans ce cas l'automate est déterministe ;
	\item soit par une fonction $f : A^3 \rightarrow {\cal P}(A)$, et dans ce cas l'automate est non déterministe (pour au moins une entrée l'ensemble retourné par la fonction doit être de cardinalité au moins deux).
\end{itemize} 

La notion d'automate cellulaire déterministe est largement étudiée. L'idée d'un automate cellulaire non déterministe est moins commune. Néanmoins, les automates cellulaires non déterministes ont aussi été étudiés depuis différents points de vue et dans différents contextes. Nous pouvons mentionner plusieurs contributions à ce domaine de recherche.
Une étude des automates non déterministes utilisant une approche purement topologique peut être trouvée dans \cite{di2014nondeterministic} ; des propriétés topologiques de ceux-ci sont présentées dans \cite{di2020topological}. Dans \cite{ozhigov1999computations}, l'auteur étudie le compromis pour les automates cellulaires non déterministes, entre d'une part leur dimensions, et d'autre part le temps et l'espace nécessaires pour effectuer des calculs.
Dans \cite{kutrib2012non}, l'auteur s'intéresse à la puissance des automates non déterministes si le nombre de cellules ayant un comportement non déterministe est borné. Les automates cellulaires non déterministes sur un groupe avec un alphabet potentiellement infini peuvent être caractérisés avec la notion de continuité uniforme, voir \cite{furusawa2017uniform}. Pour les automates cellulaires, un ``Jardin d'Éden'' est une configuration n'ayant pas de prédécesseur : Richardson donne dans \cite{richardson1972tessellations} une condition suffisante pour qu'un automate cellulaire n'en possède pas. À l'inverse, dans \cite{yaku1976surjectivity} l'auteur montre que le problème de savoir s'il existe un Jardin d'Éden n'est pas décidable en général. 

Pour plus de références sur l'histoire des automates cellulaires déterministes et non déterministes, voir \cite{delorme1998cellular,di2020topological}.

\subsection{Complexité de Kolmogorov}

La notion de complexité algorithmique (ou de Kolmogorov) capture certains aspects de la ``quantité d'information'' d'un objet fini. De manière informelle, la complexité de Kolmogorov d'un objet est la longueur du plus petit programme qui peut retourner cet objet (éventuellement, avec certaines restrictions sur les ressources computationnelles disponibles). Bien entendu, pour être plus formel nous devons spécifier le langage de programmation utilisé (et si nécessaire comment sont mesurées les ressources computationnelles --- temps et mémoire utilisés par l'algorithme).

On peut retracer les origines de la complexité de Kolmogorov, ou complexité en termes de description algorithmique, dans les travaux de Solomonoff (\cite{solomonoff1960preliminary,solomonoff1964formalI,solomonoff1964formalII}), et de Chaitin (\cite{chaitin1969simplicity}). Déjà dans \cite{kolmogorov-three-approaches} Kolmogorov mentionnait que les ressources computationnelles pourraient être prises en compte dans la définition de la complexité algorithmique. De manière plus systématique, la complexité de Kolmogorov à ressources bornées a été introduite dans \cite{hartmanis1983generalized,sipser1983complexity,ko1986notion,allender1989some}.

Le livre \cite{shen2017kolmogorov} reprend les principaux sujets et résultats sur la complexité de Kolmogorov, d'un point de vue technique et philosophique. Une étude très approfondie de ce domaine peut être trouvée dans \cite{li2008introduction}.

Nous appelons $\{0,1\}^*$ l'ensemble des suites binaires de taille finie (i.e., une suite, finie, de 0 et de 1). La complexité de Kolmogorov est une fonction qui associe à chaque suite binaire un entier.

Soit $U:\{0,1\}^* \rightarrow \{0,1\}^*$ une fonction (partielle) calculable, appelée \emph{décompresseur}. La complexité d'une suite binaire $x \in \{0,1\}^*$  par rapport à $U$ est définie par 
 $$
 \CK_U(x) :=  \min \{  \ |p| \ : \ U(p) = x\ \}.
 $$
Si une suite binaire $p$ telle que $U(p) = x$ n'existe pas, nous définissons alors $\CK_U(x)=\infty$.
Ici $U$ doit être vue comme l'interpréteur d'un langage de programmation et $p$ comme le programme qui retourne $x$ ; la complexité de $x$ est alors la taille de l'un des plus petits programmes qui retournent $x$ (sur l'entrée vide).

Le problème évident avec cette définition de la complexité est sa dépendance par rapport à $U$. La théorie de la complexité de Kolmogorov est rendue possible par le théorème d'invariance :

\begin{theorem*}[Kolmogorov, \cite{kolmogorov-three-approaches}]
\label{t:kolmogorov-invariance}
Il existe une fonction calculable $U$ telle que pour toute autre fonction calculable $V$, il existe une constante $c$ telle que :
 $$
   \CK_U(x) \le \CK_V(x)+c
 $$
pour toute suite binaire $x$. 
\end{theorem*}

Cette fonction $U$ est appelée un \emph{décompresseur optimal}. Nous fixons un décompresseur optimal $U$, et dans la suite nous omettons l'indice $U$ dans $\CK_U(x)$.
La valeur $\CK(x)$ est appelée la complexité (simple) de Kolmogorov de $x$. Nous pouvons parler de la complexité de Kolmogorov conditionnelle de $x$ étant donné $y$, notée $\CK(x|y)$, si la fonction $U$ reçoit également $y$ en entrée.

De nombreuses propriétés de la complexité de Kolmogorov sont connues ; par exemple, la règle de la chaîne est respectée à un terme logarithmique près, comme montré dans \cite{zvonkin1970complexity} :
$$
	\CK(x,y) = \CK(y) + \CK(x|y) + O(\log \CK(x,y)).
$$

D'une façon similaire à la version simple, nous définissons la complexité de Kolmogorov en termes de programmes avec des ressources bornées (le temps et l'espace disponibles pour le calcul). Pour cela, nous devons préciser un modèle de calcul pour la fonction calculable $U$ ; nous choisissons par exemple le modèle classique des machines de Turing à une tête de lecture sur un ruban. Pour une telle machine de Turing $T$, nous définissons $\CK^{t}_{T}(x)$ comme la longueur de $p$, l'un des plus courts programmes retournant $x$ en au plus $t$ étapes. Techniquement, nous pouvons supposer que $T$ possède au moins trois symboles : 0,1 et ``Fin''. Initialement le ruban de $T$ contient les bits de $p$, suivis du symbole Fin. Quand $T$ entre dans son état d'arrêt, son ruban contient les bits de $x$ suivis du symbole Fin. 

Il existe une \emph{machine de décompression optimale} $T$ dans le sens suivant : pour toute machine de Turing $V$, il existe un polynôme tel que
\begin{equation*}
	\forall t, \CK_{T}^{\poly(t)}(x)\le \CK_{V}^{t}(x)+O(1).
\end{equation*}

Pour une machine de Turing à plusieurs rubans, un résultat légèrement plus fort peut être prouvé :
\begin{theorem*}[voir~\cite{li2008introduction} ; la preuve utilise la technique de simulation de~\cite{hennie-stearns}]
Il existe une machine de Turing \textup(à plusieurs ruban\textup) $T$ de \emph{décompression optimale} dans le sens suivant :
pour toute machine de Turing $V$ à plusieurs rubans, il existe une constante $c$ telle que :
$$
\forall t, \CK_{T}^{ct\log t}(x)\le \CK_{V}^{t}(x)+c
$$
pour toute suite binaire $x$.       
\end{theorem*}
\label{d:kolmogorov-bornée}
Nous fixons une telle machine $T$, et par la suite nous utilisons pour la version de la complexité de Kolmogorov à ressources bornées la notation $\CK^{t}(x)$ au lieu de $\CK_{T}^{t}(x)$.

Sans perte de généralité, nous pouvons supposer que $\CK(x)\le \CK^{t}(x)$ pour tout $x$ et pour tout $t$.

Nous fixons une énumération calculable des motifs finis (sur un alphabet fini) qui assigne une suite binaire (un \emph{code}) à chaque motif de dimension deux. Par la suite nous prendrons la liberté de parler de la complexité de Kolmogorov de motifs finis de dimension deux (en faisant référence à la complexité de Kolmogorov des \emph{codes} de ces motifs).

\chapter{Des conditions nécessaires pour qu'un shift soit sofique}
\label{c:non-sofique}

\section{Une grande complexité de Kolmogorov à ressources bornées est compatible avec une faible complexité par bloc}
\label{s:grande-complexité-bornée}

\subsection{Une condition nécessaire pour être sofique}
\label{s:condition-sofique}

Le théorème suivant a été prouvé implicitement dans \cite{dls}:
\begin{theorem}\label{th:dls}
Dans tout shift sofique non vide $S$ il existe une configuration $x$ telle que pour tout motif $P$ de taille $n\times n$  dans $x$, on a
\[
\CK^{T(n)}(P) = O(n)
\]
pour une fonction de seuil $T(n)=2^{O(n^2)}$. 
\end{theorem}

Dans \cite{dls} une version plus faible de ce théorème est proposée : il est seulement démontré que la complexité de Kolmogorov \emph{simple} des motifs de taille $n\times n$ est $O(n)$.
Cependant, cet argument venant de \cite{dls} implique une limite pour la complexité de Kolmogorov à  \emph{ressources bornées}. 
Dans un but d'autosuffisance, nous donnons dans ce qui suit une preuve de ce théorème.

\begin{proof}

Il est suffisant de prouver ce théorème pour les shifts de type fini (puisqu'une configuration d'un shift sofique est une projection coordonnée par coordonnée d'un shift de type fini).
Nous fixons un shift de type fini (où les supports de tous les motifs interdits sont des paires de lettres voisines) et montrons que pour chaque $k$ il existe un motif $P_k$ de taille $(2^{k}+1) \times (2^{k}+1)$ localement admissible possédant la propriété requise suivante :
tout motif carré de taille $n\times n$ à l'intérieur de $P_k$ a une complexité de Kolmogorov en $O(n)$.

Dans ce but, nous choisissons un coloriage arbitraire de la bordure d'un carré de taille $(2^{k}+1) \times (2^{k}+1)$ tel que ce coloriage puisse être étendu en au moins un coloriage localement admissible du carré entier, voir Fig.~\ref{f:recursuve-construction}~(a). 

Puisque l'ensemble de tous les coloriages de ce carré est fini, nous pouvons trouver algorithmiquement (par une recherche exhaustive naïve) le premier coloriage dans l'ordre lexicographique des lignes centrales horizontales et verticales du carré qui sont compatibles avec le coloriage fixé de la bordure (i.e., l'union du coloriage de la bordure et du coloriage des lignes centrales peut être étendue en un motif localement admissible du carré entier, voir Fig.~\ref{f:recursuve-construction}~(b). On peut observer que les lignes centrales coupent le carré en quatre carrés de taille  $(2^{k-1}+1) \times (2^{k-1}+1)$.

\begin{figure}[H]
\vspace{5pt}
\begin{center}
  \begin{tikzpicture}[scale=0.14,x=1cm,baseline=2.125cm]
 
 \begin{scope}[shift={(0,0)}]
    \foreach \x in {1,...,17} \foreach \y in {1,...,17}
    {
        \path[draw=black] (\x,\y) rectangle ++ (1,1);
    }

   \foreach \x in {1,...,17} \foreach \y in {1,17}
    {
        \path[fill=blue!60,draw=black] (\x,\y) rectangle ++ (1,1);
        \path[fill=blue!60,draw=black] (\y,\x) rectangle ++ (1,1);
    }
                      
\end{scope}

 \begin{scope}[shift={(27,0)}]
 
    \foreach \x in {1,...,17} \foreach \y in {1,...,17}
    {
        \path[draw=black] (\x,\y) rectangle ++ (1,1);
    }

   \foreach \x in {1,...,17} \foreach \y in {1,9,17}
    {
        \path[fill=blue!60,draw=black] (\x,\y) rectangle ++ (1,1);
        \path[fill=blue!60,draw=black] (\y,\x) rectangle ++ (1,1);
    }

\end{scope}
 \begin{scope}[shift={(54,0)}]
 
    \foreach \x in {1,...,17} \foreach \y in {1,...,17}
    {
        \path[draw=black] (\x,\y) rectangle ++ (1,1);
    }

   \foreach \x in {1,...,17} \foreach \y in {1,5,9,13,17}
    {
        \path[fill=blue!60,draw=black] (\x,\y) rectangle ++ (1,1);
        \path[fill=blue!60,draw=black] (\y,\x) rectangle ++ (1,1);
    }

  \end{scope}

  \node at (9.5,-2) {(a)};
  \node at (36.5,-2) {(b)};
  \node at (63.5,-2) {(c)};  
 
   \node at (23,9)  [thick, font=\fontsize{20}{20}\selectfont, thick]  {$\Rightarrow$ };
   \node at (50,9)  [thick, font=\fontsize{20}{20}\selectfont, thick]  {$\Rightarrow$ };
   \node at (77,9)  [thick, font=\fontsize{20}{20}\selectfont, thick]  {$\Rightarrow$ };
   \node at (83,9)  [thick, font=\fontsize{20}{20}\selectfont, thick]  {...}; 
              
\end{tikzpicture}

\caption{Trois étapes de la procédure récursive de coloriage pour un motif de taille  $(2^n+1)\times (2^n+1)$.} \label{f:recursuve-construction}
\end{center}
\end{figure}
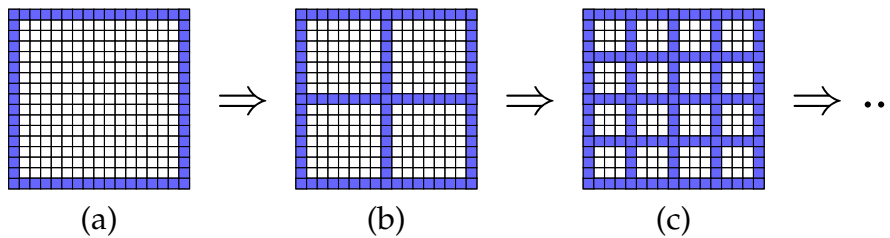

Nous répétons alors la même procédure récursivement : pour chaque carré de taille $(2^{k-1}+1) \times (2^{k-1}+1)$, nous trouvons le premier coloriage pour l'ordre lexicographique de leurs lignes centrales horizontale et verticale qui soit compatible avec la bordure de chacun de ces carrés
(voir Fig.~\ref{f:recursuve-construction}~(c)), et ainsi de suite.
À la dernière étape, nous terminons avec le premier coloriage valide pour l'ordre lexicographique des motifs isolés de taille $1\times 1$ (qui doivent être cohérents avec le coloriage de leur voisinage choisi en amont). Cela conclut la construction des $P_k$.

Dans la procédure expliquée ci-dessus, le carré initial de taille $(2^{k}+1) \times (2^{k}+1)$ est découpé en quatre carrés de taille $(2^{k-1}+1) \times (2^{k-1}+1)$, puis en seize carrés de taille $(2^{k-2}+1) \times (2^{k-2}+1)$, etc. Cette procédure récursive trouve le coloriage de chacun de ces carrés en positions ``standards'' étant donné (comme une sorte de contrainte de frontière) le coloriage de la bordure de ce carré. 

Ainsi, pour trouver les couleurs fixées d'un carré de taille $(2^{m}+1) \times (2^{m}+1)$ dans une position ``standard'', la procédure récursive a besoin de ne connaître que le coloriage de la bordure du carré (ce qui requiert $O(2^{m})$ lettres et donc $O(2^{m})$ bits d'information);
Il n'est pas difficile de voir que les appels récursifs s’exécutent en temps :

\[
 2^{O\big((2^{m}+1) \times (2^{m}+1)\big)}+   4\times 2^{O\big((2^{m-1}+1) \times (2^{m-1}+1)\big)}+ 16\times 2^{O\big((2^{m-2}+1) \times (2^{m-2}+1)\big)}+\ldots
=  2^{O(2^{2m})}
\]
(c'est un calcul récursif de profondeur $m$, avec une recherche exhaustive et quatre appels récursifs à chaque niveau de la hiérarchie).

Un carré quelconque de taille $n\times n$ (potentiellement dans une position non-standard) est recouvert par au plus quatre carrés ``standards'' (de taille au plus deux fois plus grande que $n$), voir Fig.~\ref{f:recursuve-construction-covering}.

\begin{figure}[H]
\begin{center}
  \begin{tikzpicture}[scale=0.20,x=1cm,baseline=2.125cm]

 \begin{scope}[shift={(27,0)}]
 
     \foreach \x in {1,...,17} \foreach \y in {1,...,17}
    {
        \path[draw=black] (\x,\y) rectangle ++ (1,1);
    }

   \foreach \x in {1,...,17} \foreach \y in {1,9,17}
    {
        \path[fill=blue!60,draw=black] (\x,\y) rectangle ++ (1,1);
        \path[fill=blue!60,draw=black] (\y,\x) rectangle ++ (1,1);
    }

    \draw [pattern=crosshatch, pattern color=red!90, thick, dashed]  (7,6) rectangle (14,13) ;
\end{scope}

\end{tikzpicture}

\caption[Un motif de taille $7\times 7$ recouvert par un quadruplé de motifs ``standards'' de taille $9\times 9$.]{Un motif de taille $7\times 7$ (hachuré en rouge) recouvert par un quadruplé de motifs ``standards'' de taille $9\times 9$.} \label{f:recursuve-construction-covering}
\end{center}
\end{figure}
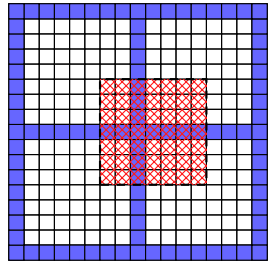

Par conséquent, pour identifier un motif de taille $n\times n$ à l'intérieur de $P_k$, il suffit de décrire le quadruplet de carrés standards recouvrant ce motif et, par ailleurs, la position (les coordonnées des coins) de ce motif par rapport aux carrés standards le recouvrant. Nous savons que chacun des carrés standards est, par construction, déterminé par sa bordure.

Partant, tout carré de taille $n\times n$ à l'intérieur de $P_k$ peut être spécifié par seulement $O(n)$ bits d'information;
de plus, il peut être recalculé à partir de sa description en temps $2^{O(n^2)}$ (en premier lieu, nous reconstruisons les quatre carrés standards à partir de leur bordure, et ensuite nous coupons le motif avec les coordonnées des coins fournies).

Pour conclure la preuve de ce théorème, il reste à observer que les motifs $P_k$ sont définis pour des $k$ arbitrairement grands, et l'argument de compacité nous donne une coloration valide de tout $\mathbb{Z}^2$ avec la propriété requise.
\end{proof}

\subsection{Existence d'un shift effectif non-sofique de faible complexité par bloc}
\label{s:effectif-non-sofique}

Le Théorème~\ref{th:dls} garantit qu'un shift n'est pas sofique si dans toute configuration nous pouvons trouver une séquence de motifs de tailles croissantes avec une complexité de Kolmogorov super-linéaire. Dans le théorème suivant nous prouvons l'existence de tels shifts effectifs.

\begin{theorem}\label{th:deep-shift}
Pour tout $\epsilon>0$ et pour toute fonction calculable $T(n)$ il existe un shift effectif sur  $\mathbb{Z}^2$ tel que pour tout $n$ et 
pour tout motif $P$ de taille $n\times n$ globalement admissible, on ait
 \begin{itemize}
  \item[(i)] $\CK(P) = O(\log n)$, et
  \item[(ii)] $\CK^{T(n)}(P) = \Omega(n^{2-\epsilon})$.
 \end{itemize}
\end{theorem}
Nous différons la preuve du Théorème~\ref{th:deep-shift} à la fin de ce chapitre. Dans ce qui suit, nous proposons une version légèrement plus faible de ce théorème, qui est cependant assez puissante pour nos applications principales : 
\begin{theorembis}{th:deep-shift}
\label{th:deep-shift-simple}
Pour toute fonction calculable $T(n)$ il existe un shift effectif, minimal et transitif (dans le sens de la Définition~\ref{d:trans})sur $\mathbb{Z}^2$ tel que 
 \begin{itemize}
  \item[(i)]  pour tout $n$ et pour tout motif $P$ de taille $n\times n$ globalement admissible, nous avons
    $\CK(P) = O(\log n),$ et 
    \item[(ii)] pour une infinité de $n$ et pour tout motif $P$ de taille $n\times n$ globalement admissible, nous avons
    $\CK^{T(n)}(P) = \Omega(n^{1.5}).$
 \end{itemize}
\end{theorembis}

\begin{definition}\label{d:trans}
Un shift est dit minimal s'il ne contient pas de sous-shift propre.

Un shift $S$ est dit transitif si pour tout ouvert non vide $U, V \subset X$, il existe un entier $n \in \mathbb{N}$ tel que $\sigma^n(U) \cap V$ n'est pas vide, où $\sigma$ est l'opérateur du shift.
\end{definition}

D'après le Théorème~\ref{th:dls} et le Théorème~\ref{th:deep-shift-simple} nous déduisons le corollaire suivant :
\begin{corollary}
Il existe un shift effectif mais non sofique sur $\mathbb{Z}^2$ avec une complexité par bloc polynomiale en $n$,
i.e., le nombre de motifs de taille $n\times n$ globalement admissibles est borné par un polynôme.
\end{corollary}
\begin{proof}
Nous considérons le shift du Théorème~\ref{th:deep-shift-simple} en supposant que la fonction de seuil $T(n)$ est bien supérieure à $2^{\Omega(n^2)}$ (par exemple, nous pouvons fixer $T(n)=2^{n^3}$, avec $2^{n^3}$ qui devient bien plus grand que $2^{cn^2}$ pour n'importe quel $c$.). D'une part, la propriété~(ii) du Théorème~\ref{th:deep-shift-simple} et le Théorème~\ref{th:dls} garantissent que le shift n'est pas sofique. D'autre part, la propriété~(i) du Théorème~\ref{th:deep-shift-simple} implique que le nombre de motifs de taille $n\times n$ n'est pas supérieur à $2^{O(\log n)}$.
\end{proof}
\begin{remark}\label{r:improved-upper-bound}
Notre preuve du Théorème~\ref{th:deep-shift-simple} implique une limite plus forte que celle de la propriété~(i). En fait, au lieu de la limite
$\CK(P) = O(\log n)$,
nous pouvons prouver que pour une (suffisamment grande) fonction calculable $\hat T(n)$, pour tout motif $P$ de taille $n\times n$ globalement admissible nous avons
\begin{equation}
\CK^{\hat T(n)}(P) =O( \log n). \nonumber 
\end{equation}
\end{remark}

\begin{remark}\label{r:paradoxical-shift}
Pour toute fonction $\hat T(n)$ et pour tout nombre $\lambda>0$, soit $S_{\lambda,\hat T}$ le shift suivant défini sur l'alphabet $\{0,1\}$:  une configuration $x$ appartient à $S_{\lambda,\hat T}$ si et seulement si tout motif $P$ de $x$ de taille $n\times n$ vérifie $\CK^{\hat T(n)}(P) \le \lambda \log n$. De manière évidente, pour un entier $\lambda$ et une fonction calculable $\hat T(n)$, ce shift est effectif.

La Remarque~\ref{r:improved-upper-bound} peut être reformulée de la manière suivante : il existe un entier $\hat \lambda$, une fonction calculable $\hat T(n)$ et une séquence de motifs $P_n$ de taille $n\times n$ ($n=1,2,\ldots$) globalement admissibles dans $S_{\lambda,\hat T}$, tels que $\CK^{2^{n^3}}(P_n) = \Omega (n^{1.5})$. Nous posons $S_{0} = S_{\lambda,\hat T}$.

 tels que pour le shift $S_{\hat \lambda,\hat T}$ correspondant il existe  ; par ailleurs, pour tout motif de taille $n\times n$ globalement admissible nous avons $\CK(P) \le \CK^{\hat T(n)}(P) =O( \log n)$.

Le shift construit dans le Théorème~\ref{th:deep-shift-simple} est un sous-shift propre de $S_{0}$. En effet, en plus de toutes les configurations du shift du Théorème~\ref{th:deep-shift-simple} (Ces configurations ne doivent avoir que des motifs de grande complexité de Kolmogorov), $S_{0}$ possède aussi des motifs avec une très petite complexité de Kolmogorov (par exemple, la configuration avec seulement des $0$ et la configuration avec seulement des $1$).
Dans la partie suivante nous utiliserons $S_{0}$ pour construire d'autres exemples de shifts effectifs qui ne sont pas sofiques. 
\end{remark}

\begin{proof}[Preuve du Théorème~\ref{th:deep-shift-simple}]
Dans cette preuve, nous construisons le shift requis de manière explicite. Fixons une séquence $(n_i)$ où  $n_0$ est un nombre entier suffisamment grand, et  
\begin{equation}
\label{eq:th2-c}
 n_{i+1}: =(n_0\cdot \ldots \cdot  n_i)^c\ \mbox{ for } i=0,1,2,\ldots,
\end{equation}
où $c\ge 3$ est une constante. 
Fixons $N_i := n_0 \cdot \ldots \cdot n_i$.
Dans ce qui suit nous construisons pour chaque $i$ une paire de motifs binaires \emph{standards} $Q_i^0$  et $Q^1_i$ de taille $N_i \times N_i$  tels que
  \begin{itemize}
  \item la complexité de Kolmogorov simple des motifs standards $\CK(Q_i^0)$ et $\CK(Q_i^1)$ ne soit pas plus grande que $O(\log N_i)$, et
  \item la complexité de Kolmogorov à ressources bornées $\CK^{T(N_i)}(Q_i^0)$ et $\CK^{T(N_i)}(Q_i^1)$ ne soit pas plus petite que $\Omega(N_i^{1.5})$.
  \end{itemize}
La construction est hiérarchique : à la fois $Q_i^0$ et $Q^1_i$ sont définis comme des matrices de taille  $n_i\times n_i$ composées de motifs  
$Q_{i-1}^0$  et $Q^1_{i-1}$; pour chaque $i$ les motifs $Q_i^0$  et $Q^1_i$ sont les inversions bits à bits l'un de l'autre.

Quand les motifs standards $Q_i^0$ et $Q^1_i$ sont construits pour chaque $i$, 
nous définissons le shift comme la clôture de ces motifs : nous disons qu'un motif fini est globalement admissible si et seulement s'il apparaît dans un motif $Q_i^j$, ou au moins dans un bloc de taille  $2\times 2$ composé de $Q_i^0$  et de $Q^1_i$  (pour un $i$ donné).
\begin{remark}\label{r:low-level}
Étant donné la structure hiérarchique des motifs standards, nous pouvons garantir que tout motif $P$ de taille $N_i\times N_i$ globalement admissible apparaît dans un bloc de taille $2\times 2$ composé de  $Q_i^0$  et de $Q^1_i$ 
(il n'est pas nécessaire d'essayer les motifs $Q_{s}^j$ pour $s>i$). 

Il est nécessaire de considérer les blocs de taille $2 \times 2$ pour prendre en compte les cas dégénérés. Dans ces configurations un motif de $P$ peut ne jamais être contenu dans un motif standards $Q_i^0$ ou $Q_i^1$ : pour tout $i \in \mathbb{N}$, $P$ est à cheval sur quatre blocs pouvant être $Q_i^0$ ou $Q_i^1$. 
\end{remark}   
Comme la construction des $Q_i^j$ est explicite, le shift obtenu est effectif. Les Propriétés~(i) et~(ii) du théorème seront des conséquences directes des propriétés des motifs standards.

Dans ce qui suit, nous expliquons la construction par induction de $Q_i^0$  et de $Q^1_i$. Soit $Q_0^0$  et $Q^0_1$ des carrés composés respectivement entièrement de $0$s et entièrement de $1$. De plus, pour chaque $i$ nous considérons la matrice binaire $R_i$ de taille $n_i\times n_i$, qui est la première pour l'ordre lexicographique telle que
\begin{equation}
\label{eq:incompressible}
  \CK^{t_i}(R_i) \ge n_i^2
\end{equation}
(la limite de ressource $t_i$ sera fixée par la suite, quand nous aurons identifié les propriétés que $t_i$ devra avoir).
Nous affirmons qu'une telle matrice existe. En effet, il existe une matrice binaire de taille $n_i\times n_i$ qui est incompressible au sens de la complexité de Kolmogorov simple. La complexité de Kolmogorov à ressources bornées ne peut être que plus grande que la complexité de Kolmogorov simple.
Par conséquent, il existe au moins une matrice binaire satisfaisant \eqref{eq:incompressible}.
Si $t_i$ est une fonction calculable de $i$, étant donné $i$ nous pouvons trouver $R_i$ algorithmiquement.

Maintenant nous substituons dans $R_i$ chaque zéro et chaque un respectivement par une copie de    $Q_{i-1}^0$  et par une copie de $Q^1_{i-1}$. Par exemple, 
\[
R_i=\left(
\begin{array}{cccccc}
0&0&0&0&1\\
0&1&0&0&1\\
1&1&1&1&0\\
0&1&1&0&0\\
0&1&0&1&0
\end{array}
\right) 
\ \Longrightarrow\ 
Q_i^0:=\left(
\begin{array}{c|c|c|c|c}
Q_{i-1}^0&Q_{i-1}^0&Q_{i-1}^0&Q_{i-1}^0&Q_{i-1}^1\\
\hline
Q_{i-1}^0&Q_{i-1}^1&Q_{i-1}^0&Q_{i-1}^0&Q_{i-1}^1\\
\hline
Q_{i-1}^1&Q_{i-1}^1&Q_{i-1}^1&Q_{i-1}^1&Q_{i-1}^0\\
\hline
Q_{i-1}^0&Q_{i-1}^1&Q_{i-1}^1&Q_{i-1}^0&Q_{i-1}^0\\
\hline
Q_{i-1}^0&Q_{i-1}^1&Q_{i-1}^0&Q_{i-1}^1&Q_{i-1}^0
\end{array}
\right) 
\]

La matrice obtenue (de taille $N_i\times N_i$) est appelée $Q_i^0$. La matrice $Q_i^1$ est définie comme l'inversion bit à bit de $Q_i^0$.

Pour tout $i$ assez grand, le motif incompressible $R_i$ contient des copies de tous les $2^4$ motifs binaires de taille $2\times 2$.
Par conséquent, nous pouvons garantir que tous les motifs standards $Q_i^j$ contiennent tous les motifs globalement admissibles de taille $N_{i-1}\times N_{i-1}$. 
Il s'ensuit que le shift construit dans le Théorème~\ref{th:deep-shift-simple} est transitif et même minimal.

\smallskip
\noindent
\emph{Fait 1}. En supposant que $t'_i\ll t_i$  (dans ce qui suit nous discuterons plus précisément du choix de $t_i'$ ) nous avons
\[
  \CK^{t'_i}(Q_i^0)  = \Omega(N_i^{1.5}) \mbox{ et }   \CK^{t'_i}(Q_i^1)  = \Omega(N_i^{1.5}). 
\]

\begin{proof}[Preuve du Fait~1:]
Étant donné $Q_i^j $ (pour $j=0,1$) nous pouvons recalculer la matrice $R_i$ (ce calcul peut être implémenté en temps polynomial). Ainsi, pour toute limite de temps $t$ 
\[
\CK^{t+ \poly(N_i)} (R_i) \le \CK^{t}(Q_i^j) + O(1).
\]

Donc, si $t_i > t_i' + \poly(N_i)$ on a
\[
 n_i^2 \le \CK^{t_i}(R_i) \le  \CK^{t_i'}(Q_i^j). 
\]
Il reste à observer que notre choix de paramètres dans l'Équation~\eqref{eq:th2-c} avec $c\ge 3$ implique
$ n_i^{1/2} \geq \left(n_0 \cdot \ldots \cdot n_{i-1} \right)^{3/2}$, d'où
\[
n_i^2 \geq \left(n_0 \cdot \ldots \cdot n_i\right)^{1.5} = (N_i)^{1.5}.
\]
Par conséquent, nous obtenons que $\CK^{t'_i}(Q_i^j)  \ge  (N_i)^{1.5} - O(1)$, et le fait est prouvé.
\end{proof}
\begin{remark}
En choisissant une constante $c$ plus large dans l'Équation~\eqref{eq:th2-c}, nous pouvons avoir une limite plus petite $  \CK^{t'_i}(Q_i^j) = \Omega(n^{2-\epsilon})$ pour tout $\epsilon>0$.
\end{remark}
\smallskip
\noindent
\emph{Fait 2}. Pour tout motif $P$ de taille $N_i\times N_i$ globalement admissible 
(et pas seulement pour les motifs standards, comme c'était le cas dans le Fait 1)
sa complexité de Kolmogorov à ressources bornées $\CK^{T(N_i)}(P) $  est $\Omega(n^{1.5})$ (en supposant que $T(N_i) \ll t_i'$).

\smallskip

\begin{figure}
\begin{center}

  \begin{tikzpicture}[scale=0.25,x=1cm,baseline=2.125cm]

  
    \draw [fill=black!30, thick, dashed]  (3,6) rectangle (11,14);
 
    \foreach \x in {0,...,15} \foreach \y in {0,...,15}
    {
        \path[draw=black] (\x,\y) rectangle ++ (1,1);
    }

      \draw [ultra thick]  (0,0) rectangle (8,8);
      \draw [ultra thick]  (0,8) rectangle (8,16);
      \draw [ultra thick]  (8,0) rectangle (16,8);
      \draw [ultra thick]  (8,8) rectangle (16,16);


 \begin{scope}[shift={(20,0)}]
  
    \draw [fill=black!30, thick, dashed]  (3,6) rectangle (11,14);

    \draw [fill=red!30]  (3,6) rectangle (8,8);
    \draw [pattern=crosshatch, pattern color=red!91]  (11,6) rectangle (16,8);

    \draw [fill=blue!30]  (3,8) rectangle (8,14);
    \draw [pattern=north east lines, pattern color=blue!90]  (11,0) rectangle (16,6);

    \draw [fill=orange!30]  (8,8) rectangle (11,14);
    \draw [pattern=north west lines, pattern color=orange!90]  (8,0) rectangle (11,6);

    \foreach \x in {0,...,15} \foreach \y in {0,...,15}
    {
        \path[draw=black] (\x,\y) rectangle ++ (1,1);
    }

      \draw [ultra thick]  (0,0) rectangle (8,8);
      \draw [ultra thick]  (0,8) rectangle (8,16);
      \draw [ultra thick]  (8,0) rectangle (16,8);
      \draw [ultra thick]  (8,8) rectangle (16,16);

   \draw[-latex, ultra thick,black,dashed](6.5,12.5)node{ }  to[out=0,in=90] (14.5,2.5);
   \draw[-latex, ultra thick,black,dashed](9.1,9.5)node{ }  to[out=-135,in=135] (9.1,1.5);
   \draw[-latex, ultra thick,black,dashed](4.5,7.0)node{ }  to[out=-90,in=-135] (12.5,6.5);

\end{scope}
  
   \node at (8.0,-2.0) {(a)};
   \node at (28.0,-2.0) {(b)};

\end{tikzpicture}

\caption[Un motif de taille $N_k\times N_k$, recouvert par un quadruplé de motifs standards de même taille, contient assez d'information pour reconstruire un motif standard.]{Un motif de taille $N_k\times N_k$ (dessiné en gris dans la fig.~(a)), recouvert par un quadruplé de motifs standards de même taille, contient assez d'information pour reconstruire un motif standard (fig.~(b)).} \label{f:non-standard-block}
\end{center}
\end{figure}
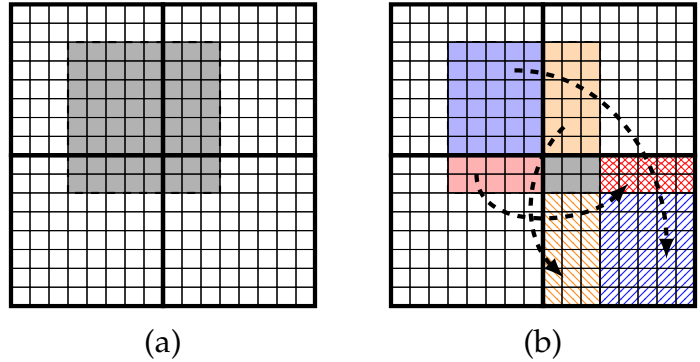

\begin{proof}[Preuve du Fait~2 :]
Si un motif $P$ de taille $N_i\times N_i$ est globalement admissible alors il est recouvert par un quadruplé de motifs standards de rang $i$, voir Remarque~\ref{r:low-level} ci-dessus.
Alors $P$ peut être divisé en quatre rectangles qui correspondent aux ``coins'' des motifs standards de rang $i$, voir Fig.~\ref{f:non-standard-block}~(a). 
Puisque les motifs standards $Q_i^0$ et $Q_i^1$ sont les inversions l'un de l'autre, ces quatre coins  (avec une inversion bit à bit si nécessaire) forment ensemble le motif standard en entier, comme montré dans la Fig.~\ref{f:non-standard-block}~(b).
Par conséquent, nous pouvons reconstruire $Q_i^j$ à partir de $P$ étant donné (a)~la position de $P$ par rapport à la grille des motifs standards (ce qui nécessite $O(\log N_i)$ bits) 
et (b)~les quatre bits identifiant les blocs standards recouvrant $P$ (nous devons savoir lesquels d'entre eux sont des copies de $Q_i^0$ et lesquels sont des copies de $Q_i^1$).

La reconstruction de $Q_i^j$ à partir de $P$ ne requiert que $\poly(N_i)$ étapes de calcul (auxquelles s'ajoute le temps nécessaire pour calculer $P$).
Alors le fait est la conséquence directe de la limite pour la complexité de Kolmogorov à ressource bornée des motifs standards $Q_i^0$ et $Q_i^1$.
\end{proof}

\smallskip
\noindent
\emph{Fait 3}. Pour tout motif de taille $k\times k$ dans $Q_i^0$ ou $Q_i^1$, sa complexité de Kolmogorov simple est au plus de $O(\log k)$.
\begin{proof}[Preuve du Fait~3 :]
Tout d'abord, nous pouvons observer que les motifs standards $Q_i^0$ et $Q_i^1$ peuvent être calculés étant donné $i$. Par conséquent, 
  $
  \CK(Q_i^0) = O(\log i)\mbox{ et } \CK(Q_i^1) = O(\log i).
  $

Tout motif de taille $k\times k$ globalement admissible peut être recouvert par au plus quatre motifs standards $Q_i^0$ ou $Q_i^1$ avec
\[
N_{i-1} <k \le N_{i},
\]
voir Remarque~\ref{r:low-level}. Ainsi, pour obtenir un motif $P$ de taille $k\times k$ globalement admissible nous devons construire un quadruplé de motifs standards de taille $N_{i}\times N_i$ et ensuite spécifier la position de $P$ par rapport à la grille des motifs standards. Cette description consiste en seulement 
 $
 O(\log N_i)
 $
bits, et nous pouvons conclure que  
 $
 \CK(P) = O(\log k).
 $
\end{proof}
Il reste à fixer la limite de temps $t_i$.
Étant donné une fonction de seuil calculable $T(N_i)$, nous choisissons un $t_i'\gg T(N_i)$ qui convient, et ensuite une fonction $t_i\gg t_i'$ qui convient. Nous nous écrivons pas d'expression explicite de $t_i$ et de $t_i'$ car la construction est flexible et nous pourrions avoir à changer leur expression dans un contexte similaire mais légèrement différent.
Le théorème découle de Fait~2 et de Fait~3.
\end{proof}

\label{s:no-go}

Il existe un shift effectif non vide sur $\mathbb{Z}^2$ pour lequel la complexité de Kolmogorov de tous les motifs de taille $n\times n$ est $\Omega(n^2)$ (voir \cite{dls} et \cite{rumyantsev-ushakov}). Une question naturelle se pose alors : 
peut-on améliorer le Théorème~\ref{th:deep-shift-simple} et renforcer la condition~(ii) en
$\CK^{T(n)}(P) = \Omega(n^{2})$ ? La réponse est négative : nous ne pouvons demander une complexité de Kolmogorov à ressources bornées en $\Omega(n^{2})$,
même pour une version plus faible de la propriété~(i) pour la complexité simple :
\begin{proposition}\label{p:non-quadratic}
Pour toute limite de ressources assez grande $T(n)$ calculable, il n'existe pas de shift sur $\mathbb{Z}^2$ tel que  
 \begin{itemize}
  \item[(i)]
  pour tout motif $P$ de taille $n\times n$ globalement admissible, nous avons $\CK(P) = o(n^2)$, et 
    \item[(ii)] 
   pour une infinité de $n$ et pour tout motif $P$ de taille $n\times n$ globalement admissible, nous avons $\CK^{T(n)}(P) = \Omega(n^{2}).$
 \end{itemize}
\end{proposition}
\begin{proof}
Supposons par l'absurde qu'un tel shift existe. Fixons un nombre réel $\varepsilon>0$.
La condition~(i) implique qu'il existe un entier $k=k(\varepsilon)$ tel que le nombre de motif de taille  $k\times k$ globalement admissible de ce shift est 
$
 L_k \le 2^{\varepsilon k^2}. 
$

Estimons la complexité de Kolmogorov d'un motif de taille $(Nk) \times (Nk)$ globalement admissible.
Pour tout $N$, tout motif $P$ de taille  $(Nk)\times (Nk)$ globalement admissible peut être spécifié par 
 \begin{itemize}
 \item la liste de tous les motifs de taille $k\times k$ globalement admissibles, ce qui requiert  $L_k \cdot k^2$ bits (avec  $O(L_k \cdot k^2$) bits nous pouvons même fournir une description auto-délimitée de la liste), et
 \item un tableau de taille $N\times N$ avec les indices des motifs de taille $k\times k$ qui constituent $P$ (ce qui nécessite $N^2 \cdot \log L_k$ bits).
 \end{itemize}
Nous pouvons combiner ces deux parties de la description en un programme. Clairement, $P$ peut être reconstruit à partir d'une telle description en temps polynomial. Donc, il s'ensuit que pour un certain polynôme $\poly_1(n)$
\begin{eqnarray}
\CK^{\poly_1(Nk)} (P) \le  O(L_k \cdot k^2) + N^2 \cdot \epsilon_1 k^2. \label{eq:Nk}
\end{eqnarray}

Nous pouvons alors estimer la complexité de Kolmogorov pour un motif $P$ de taille $n\times n$ globalement admissible dans le cas où $n$ n'est pas divisible par $k$. 
Soit $N$ un entier tel que 
\[
(N-1)k  < n \le N k.
\]
Le motif $P$ peut être inclus dans un motif plus grand $P'$ globalement admissible de taille  $ (Nk)\times  (Nk)$.  Pour décrire $P$, nous devons décrire $P'$ et ensuite spécifier la position de $P$ par rapport à la bordure de $P'$.
La distance entre $P$ et la bordure de $P'$ ne peut être supérieure à $k$. Par conséquent, pour obtenir une limite supérieure pour la complexité de Kolmogorov de $P$, nous n'avons besoin d'ajouter que $O(\log k)$ bits à \eqref{eq:Nk}.
Ainsi, pour un certain polynôme $\poly_2(n)$
\begin{eqnarray}
\CK^{\poly_2(n)} (P)& \le&
O(L_k \cdot k^2) + N^2 \cdot \epsilon_2 k^2 + O(\log k) \nonumber\\
&\le& \epsilon_2 (n+k)^2 + O(L_k \cdot k^2) \nonumber\\
&\le& \epsilon_2 n^2 + 2 \epsilon_2 nk+ \epsilon_2 k^2  + O(L_k \cdot k^2). \label{eq:n^2}
\end{eqnarray}
Pour un $k$ fixé et un $n$ suffisamment grand , le premier terme dans \eqref{eq:n^2} domine les termes restant, et nous obtenons  
\[
\CK^{\poly_2(n)} (P) \le   2\epsilon_2 n^2. 
\]
Rappelons que nous pouvons démontrer cette limite pour tout $\epsilon_2>0$ et pour tout $n$ suffisamment grand.
Cela contredit la condition~(ii) de la proposition.
\end{proof}

\subsection{Preuve du Théorème~\ref{th:deep-shift}}
\label{s:preuve-theoreme}

\begin{proof}

Les grandes lignes de la preuve du Théorème~\ref{th:deep-shift} sont similaires à celles de la preuve du Théorème~\ref{th:deep-shift-simple}.
Nous commençons par la construction de motifs ``standards'' de taille croissante de sorte que tous les motifs présents dans ces motifs ont une grande complexité de Kolmogorov à ressources bornées. Par la suite, nous définissons notre shift comme la clôture de ces motifs standards.

\underline{Étape 1: construction des motifs standards}.
Nous fixons les paramètres $n_i$, $i=0,1,2,\ldots$
\begin{equation}
\label{eq:def-n_0}
   n_0 \gg 1,\  n_{i+1} = (n_0\cdot \ldots \cdot n_i)^c
\end{equation}
(la constante $c$ sera spécifiée par la suite)
 et
\[
 N_ i :=  n_0\cdot \ldots \cdot n_i.
\]

Nous ne pouvons pas préciser $n_0$ car il dépend du choix du décompresseur optimal dans la définition de la complexité de Kolmogorov).
Dans ce qui suit, nous définissons par induction les  \emph{motifs standards} $Q_i^j$ de taille $N_i\times N_i$
 pour $i=0,1,2,\ldots$ La différence majeure entre cette preuve et la preuve du Théorème~\ref{th:deep-shift-simple} est que pour chaque $i$ nous définissons un ensemble $\ell_i = n_{i+1}^2$ de motifs standards de niveau $i$
\[
   Q_i^1, \ldots, Q_i^{\ell_i}
\]
(au lieu de seulement deux motifs standards comme dans la preuve précédente).
La construction est hiérarchique : chaque motif standard de rang $i$ est défini comme un tableau de taille $n_i\times n_i$ constitué de motifs standards de niveau $(i-1)$.
Lors de cette construction par induction, nous nous assurons que pour chaque $i$ la propriété suivante soit vérifiée :
 
  \smallskip
  
 \noindent
 \emph{Propriété principale}:   Pour tout $i>0$,  pour la liste de motifs standards $Q_i^{1}, \ldots, Q_i^{\ell_i}$, nous avons
\begin{equation}
\label{eq:main-p}
  \CK^{t_i} (Q_i^{1}, \ldots, Q_i^{\ell_i} ) \ge  \ell_i \log (\ell_{i-1}!)  = (n_{i+1})^2  \cdot \log \big((n_i^2)!\big)
\end{equation}
 (la fonction de seuil $t_i$ sera spécifiée plus tard).
 
 \smallskip

Notre construction est explicite, et les motifs standards $Q_i^{j}$ sont calculables (étant donnés $i$ et $j$) même si le temps requis pour le calcul peut être très important.
 
\begin{remark}
Nous allons donner ici l'intuition sur laquelle repose notre construction.
Nous construisons par induction les motifs standards vérifiant les deux propriétés suivantes.
D'une part,tous les motifs suffisamment grands apparaissant dans chaque motif standard de taille $N_i\times N_i$ doivent avoir une grande complexité de Kolmogorov à ressources bornées à cause de la ``structure locale'': chacun des motifs standards de rang $(i-1)$ (de taille $N_{i-1}\times N_{i-1}$) est complexe,
et ces motifs sont indépendants les uns des autres.
D'autre part, le motif standard de taille $N_i\times N_i$ dans son intégralité doit avoir une grande complexité de Kolmogorov à ressources bornées (avec même un seuil plus grand) à cause de l'information encodée dans sa ``structure globale'' (la manière d'arranger les motifs de rang $(i-1)$ dans un motif de rang $i$).
Pour mener à bien l'argument par induction, nous construisons pour chaque $i$ un grand nombre de motifs standards de taille $N_i\times N_i$, et ces motifs doivent être dans un  certain sens ``indépendants.''
 \end{remark}

\emph{Base de l'induction :}
Nous définissons les motifs $Q_0^{1}, \ldots, Q_0^{\ell_0} $ comme des matrices binaires de taille $n_0\times n_0$ arbitraires
(nous pouvons supposer que $2^{n_0^2}> \ell_0 = n_1^2$). 
\smallskip

\emph{Étape d'induction :} Étant donnés un ensemble de motifs standards $Q_i^{1}, \ldots, Q_i^{\ell_i} $ de taille $N_i\times N_i$, nous construisons les motifs standards de rang suivant, $Q_{i+1}^{1}, \ldots, Q_{i+1}^{\ell_i} $ de taille $N_{i+1}\times N_{i+1}$. Chaque nouveau motif $Q_{i+1}^{j}$
est défini par un tableau de taille $n_{i+1}\times n_{i+1}$ composé de motifs $Q_i^{1}, \ldots, Q_i^{\ell_i}$. 

Comme $\ell_i = n_{i+1}^2$, nous pouvons imposer que chaque motif $Q_{i}^{k}$ soit utilisé exactement une fois dans chaque $Q_{i+1}^{j}$. En d'autres termes,
chaque bloc standard $Q_{i+1}^{j}$ peut être vu comme une permutation 
 \[
   \pi_ j \ : \ \{1,\ldots, \ell_i\} \to \{1,\ldots, \ell_i\} 
 \]
sur l'ensemble des motifs standards du rang précédent.

Il y a $(\ell_i \,!)$ possibilités pour choisir chaque permutation $\pi_ j$ 
et $(\, \ell_i\,! \, )^{\ell_{i+1}}$  possibilités pour choisir toutes les permutations  $ \pi_ 1 ,\ldots,  \pi_ {\ell_{i+1}}$. 
Par un argument de comptage trivial il s'ensuit qu'il existe un n-uplet de permutations $ \pi_ 1 ,\ldots,  \pi_ {\ell_{i+1}}$ tel que
\[
  \CK( \pi_ 1 ,\ldots,  \pi_ {\ell_{i+1}}) \ge \ell_{i+1} \log (\ell_i!)
\]
Par conséquent, pour toute fonction de seuil calculable $h_i$, nous pouvons trouver algorithmiquement un n-uplet de permutations tel que
 \[
  \CK^{h_{i+1}}( \pi_ 1 ,\ldots,  \pi_ {\ell_{i+1}}) \ge \ell_{i+1} \log (\ell_i!).
 \]
Puisque chaque $\pi_j$ est déterminée de manière unique par les motifs correspondants $Q_{i+1}^j$, nous obtenons
\[
  \CK^{h_{i+1}-\mathrm{gap}_{i+1}}( Q_{i+1}^{1} ,\ldots,  Q_{i+1}^{\ell_{i+1}}) \ge \ell_{i+1} \log (\ell_i!),
\]
où $\mathrm{gap}_{i+1}$ est le temps nécessaire pour extraire les permutations $\pi_j$ à partir des motifs correspondants $Q_{i+1}^j$. 

Ainsi, si la fonction de seuil $h_{i+1}$ est bien plus grande que $t_i$, nous pouvons conclure que 
\[
  \CK^{t_i}( Q_{i+1}^{1} ,\ldots,  Q_{i+1}^{\ell_{i+1}}) \ge \ell_{i+1} \log (\ell_i!),
\]
 i.e., la \emph{Propriété Principale} reste vraie au niveau $i+1$.

\begin{lemma}\label{l:1}
Pour tout $i>0$ et chaque n-uplet de $k$ motifs standards (différents deux à deux) $Q_{i}^{{j_1}} ,\ldots,  Q_{i}^{j_k}$, nous avons
\[
 \CK^{t_i'}(Q_{i}^{{j_1}} ,\ldots,  Q_{i}^{j_k}) \ge \frac12 k \log (\ell_{i-1}!)
\]
(la fonction de seuil $t_i'$ sera spécifiée par la suite). 
\end{lemma}
\begin{remark}
Le facteur $\frac12$ dans ce lemme peut être changé par n'importe quelle constante inférieure à $1$.
\end{remark}
\begin{proof}[Preuve du lemme :] Supposons par l'absurde que pour certains motifs $Q_{i}^{{j_1}} ,\ldots,  Q_{i}^{j_k}$ nous avons
\[
 \CK^{t_i'}(Q_{i}^{{j_1}} ,\ldots,  Q_{i}^{j_k}) < \frac12 k \log (\ell_{i-1}!)
\]
Alors nous pouvons fournir la description de la liste intégrale des motifs de rang~$i$ suivante :
 \begin{itemize}
 \item une description de taille $\frac12 k \log (\ell_{i-1}!)$ pour ces $k$ motifs particuliers,
 \item $k\log \ell_{i}$ bits pour spécifier les indices $j_1,\ldots,j_k$ de ces $k$ motifs particuliers,
 \item une description directe des autres motifs standards de rang $i$, ce qui nécessite $(\ell_{i}-k)\log (\ell_{i-1}!)$ bits.
 \end{itemize}
Nous devons ajouter un en-tête de taille $O(\log k+\log i )$ pour unir ces différentes parties en une description.
Il reste à observer que la taille du résumé de ces données est inférieure à $ \ell_{i+1} \log (\ell_i!)$ 
(la partie droite de l'équation \eqref{eq:main-p}).

Estimons le temps nécessaire pour calculer la liste des motifs standards de rang~$i$ à partir de cette description. Tout d'abord, nous calculons la liste de tous les motifs standards de rang inférieur $Q_{i-1}^j$ (nous devons supposer que le temps nécessaire pour calculer ces motifs est bien inférieur à ~$t_i$). Ensuite en un temps $t_i'$ nous obtenons les $k$ motifs ``particuliers'' ayant une petite description. En $\poly(N_{i})$ étapes de calcul nous pouvons générer le reste des $(\ell_i -k)$ motifs standards, et en $\poly(N_{i})$ étapes de calcul supplémentaires nous assemblons ces motifs en une liste unique.
En supposant que $t_i$ soit suffisamment grand, nous obtenons une contradiction avec la \emph{Propriété Principale} des motifs standards.
\end{proof}

\begin{lemma}\label{l:2}
Pour chaque motif standards de rang $i$, la complexité de Kolmogorov simple des $Q_i^j$ est très petite :
\[
 \CK(Q_i^j) = O(\log N_i).
\]

\end{lemma}
\begin{proof}
Il y a $\ell_i$ motifs standards de rang $i$, et la liste de tous les motifs standards peut être calculée étant donné $i$. Par conséquent,
\[
\CK(Q_i^j) = O(\log \ell_i).
\]
Pour les paramètres choisis nous avons $\log \ell_i = \log (n_{i+1}^2) = O(\log N_i)$.
\end{proof}
Nous pouvons observer que la construction des motifs standards est explicite, et que nous pouvons choisir une fonction calculable $\hat T(n)$ telle que la totalité des motifs standards soit calculable à partir de $i$ en temps $\hat T(n)$. Cette observation implique la validité de la Remarque~\ref{r:improved-upper-bound}.

\smallskip

\underline{Étape 2 : complexité des motifs contenus dans les motifs standards}.
Soit $P$ un tableau de motifs standards (deux à deux disjoints) de rang $i$ de taille $m\times m$. 
Nous pouvons observer que la taille de $P$ (mesurée en lettres individuelles) est $k\times k$, avec $k=mN_i$.

D'après le Lemme~\ref{l:1}, nous avons 
\[
 \CK^{t_i'}(P) \ge \frac12 m^2 \log (\ell_{i-1}!) = \Omega(m^2 \ell_{i-1}) = \Omega\big((mn_i)^2\big).
\]
En d'autres termes, la complexité de Kolmogorov à ressources bornées de ce motif de taille $k\times k$ est supérieure à 
\[
 \Omega\big((mn_i)^2\big) =  \Omega\big(k^2 / (n_0\cdot \ldots \cdot n_{i-1})^2\big).
\]

\smallskip

\underline{Étape 3: la clôture des motifs standards.}
Nous définissons le shift comme la clôture des motifs standards : un motif fini est globalement admissible si et seulement s'il apparaît dans un tableau de taille $2\times 2$ composé de motifs standards (d'un rang donné $i$).
La construction des motifs standards étant hiérarchique, nous pouvons conclure que un motif de taille $N_i\times N_i$ est un motif \emph{non} globalement admissible s'il n'apparaît jamais dans un tableau de taille $2\times 2$ composé de motifs standards de rang $i$.

Soit $P$ un motif globalement admissible dans le shift ci-dessus de taille $k\times k$. Soit $2N_i \le k < 2N_{i+1}$ pour un certain $i>0$.
Alors $P$ doit pouvoir être recouvert par un  tableau de motifs standards de taille $2\times 2$ de rang $i+1$. Par conséquent, une fraction constante de $P$
peut être représentée comme un tableau de motifs standards de rang $i$ (à l'intérieur d'un motif standard de rang $(i+1)$). 
Nous pouvons appliquer la limite de l'Étape 2 et conclure que
\begin{equation}
\label{eq:lower-bound}
 \CK^{T_i}(P) =  \Omega(k^2/(n_0\cdot \ldots \cdot n_{i-1})^2)
\end{equation}
(en supposant que l'écart entre $t'_i$ et $T_i$ est suffisamment important).

Par ailleurs, la complexité de Kolmogorov simple de la totalité des motifs globalement admissibles est très petite. En effet, pour décrire un motif $P$ de taille $k$ globalement admissible, il est suffisant de spécifier quatre motifs standards recouvrant $P$ et la position de $P$ par rapport à la grille des motifs standards. D'après le Lemme~\ref{l:2}, nous obtenons
\[
 \CK(P) =  O(\log k). 
\]

\smallskip

\underline{Étape 4: le choix des paramètres}.
Nous choisissons la constante $c$ dans l'équation \eqref{eq:def-n_0} telle que pour $k=\Omega(N_i)$ l'égalité \eqref{eq:lower-bound} se réécrive en 
\[
 \CK^{T_i}(P) =  \Omega(k^{2-\epsilon}).
\]

Il reste à discuter le choix des limites des ressources. La fonction de seuil $T_i$ est donnée dans le théorème. Étant donné $T_i$, nous choisissons une
fonction de seuil $t'_i$ adéquate (l'écart entre $t'_i$ et $T_i$ doit être suffisamment important, pour que l'argument de l'Étape 3 fonctionne).
Alors, étant donné $t_i'$ nous choisissons  une fonction de seuil adéquate $t_i$ (l'écart entre $t_i$ et $t'_i$ doit être suffisamment important, pour que la preuve du Lemme~\ref{l:1} soit correcte).
Le choix de $t_i$ détermine la valeur de $h_i$ à l'Étape~1. La définition de ces séquences est inductive : pour définir $t_i$ nous devons connaître $t_{i-1}$ et le temps nécessaire pour calculer les motifs standards de rang $(i-1)$ (voir la preuve du Lemme~\ref{l:1}).

Comme $T_i$ est une fonction calculable de $i$, nous pouvons choisir des séquences calculables $t''_i$, $t'_i$, $t_i$, and $h_i$. Cela conclut la preuve.
\end{proof}

\section{La notion de ``résumés''}
\label{s:epitomes}
La technique de la Partie~\ref{s:grande-complexité-bornée} ne s'applique pas aux shifts contenant des configurations très simples (pour lesquelles tous les motifs ont une petite complexité de Kolmogorov). En particulier, il ne s'applique pas à l'Exemple~\ref{ex:mirror} de l'introduction.
Dans cette partie nous proposons une technique différente (également basée sur la complexité de Kolmogorov) qui puisse s'appliquer à ces shifts.
L'intuition derrière cette technique est la suivante : nous essayons de résumer l'information ``essentielle'' présente dans chaque motif (en éliminant les données non pertinentes) et mesurons alors la complexité de Kolmogorov simple ou à ressources bornées d'un ``résumé'' extrait à partir de cette information essentielle.
Nous montrons que pour un shift sofique, si l'on définit soigneusement la notion de ``résumé de l'information essentielle'' d'un motif, alors la complexité du résumé doit être inférieure à la longueur de la bordure de ce motif. Ainsi, si cette condition n'est pas respectée, alors le shift n'est pas sofique.

Fixons quelques notations.
Nous appelons $B_n$ l'ensemble $\{0,\ldots, n-1\}^2\subset \mathbb{Z}^2$ et $F_n$ son complément, $F_n:=\mathbb{Z}^2\setminus B_n$. Nous disons que deux motifs de support disjoint sont \emph{compatibles} (pour un shift $S$) si l'union de ces motifs est globalement admissible pour $S$.
En particulier, un motif fini $P$ de support $B_n$ et un motif infini $R$ de support $F_n$ sont compatibles, si l'union de ces motifs est une configuration valide de ce shift. 

Dans la Sous-Partie~\ref{s:plain-epitomes} nous commençons par une définition simple et plus restrictive des  \emph{résumés simples}. Ensuite, dans la Sous-Partie~\ref{s:ordered-epitomes} nous considérons une définition plus générale de \emph{résumés ordonnés} et présentons plusieurs exemples où nous utilisons cette technique pour prouver qu'ils ne sont pas sofiques. Dans la Sous-Partie~\ref{s:km} nous comparons cette technique des résumés avec la technique des \emph{ensembles d'extensions formant une chaîne d'union croissante} proposée par Kass et Madden dans \cite{kass-madden}; nous montrons que la méthode de Kass et Madden est équivalente à un cas particulier de la technique des résumés ordonnés.

\subsection{Résumés simples}\label{s:plain-epitomes}
Dans cette sous-partie nous proposons une approche générale permettant de prouver qu'un shift n'est pas sofique ; cette approche est basée sur l'intuition qui guide la preuve standard de l'Exemple~\ref{ex:mirror}. Nous commençons par la définition des résumés ``simples''.

Un résumé est une fonction qui extrait l'information essentielle d'un motif carré. Dans la définition formelle, nous préférons définir cette fonction de résumé séparément pour chacune des tailles des motifs carrés.
\begin{definition}\label{d:epitome}
Pour tout entier $n>0$, soit
\[
{\cal E}_n \ : \ [\text{motif de taille }n\times n] \mapsto [\text{mot binaire}]
\]
une fonction {partielle}. Nous appelons ${\cal E}_n$ une famille de \emph{résumés} pour un shift $S$, si pour chaque motif $P$ de support $B_n$ globalement admissible tel que ${\cal E}_n(P)$ est défini, il existe un motif $R$ sur $F_n$
tel que
\begin{itemize}
\item[(i)] $R$ est compatible avec $P$, i.e., l'union de $P$ et $R$ forme une configuration valide de $S$, et
\item[(ii)]  pour tout motif $P'$  de support $B_n$ compatible avec $R$, si ${\cal E}_n(P')$ est défini alors
$
{\cal E}_n(P')  = {\cal E}_n(P) 
$
\textup(i.e., la configuration $R$ sur le complément $B_n$ détermine de manière unique la valeur des résumés ${\cal E}_n$ pour tous les motifs valides $P'$\textup).  
\end{itemize}
Nous appelons une famille de résumés \emph{calculable} s'il existe un  algorithme qui calcule les fonctions ${\cal E}_n$. 
Précisons que le calcul est uniforme : il existe un seul algorithme qui calcule ${\cal E}_n$ pour tous les entiers $n>0$.
Comme il est courant en théorie de la calculabilité, nous supposons qu'un algorithme calculant une fonction partielle ne retourne aucun résultat sur les entrées en dehors de son domaine.

Si, de plus, il existe un algorithme qui calcule ${\cal E}_n$ en temps $2^{O(n^2)}$, nous appelons cette famille de résumés \emph{calculable en temps exponentiel}. 

De façon similaire, nous appelons une famille de résumés \emph{calculable avec un oracle $\cal O$} s'il existe un algorithme qui calcule ${\cal E}_n$ étant donné un accès à l'oracle $\cal O$. Nous pouvons aussi parler de résumés \emph{calculables en temps exponentiel} avec un oracle.
\end{definition}

\begin{remark}
Si une famille de résumés est calculable en temps exponentiel, alors le domaine de ${\cal E}_n$ est décidable en temps exponentiel :
il suffit de lancer l'algorithme calculant ${\cal E}_n$ et de s'interrompre si le calcul ne s'est pas arrêté dans le temps imparti. 
\end{remark}

\begin{remark}
De manière formelle, les valeurs de ${\cal E}_n$ sont des chaînes de caractères binaires --- des objets pour lesquels nous avons défini la complexité de Kolmogorov. En pratique, il peut être commode d'utiliser des nombres, des vecteurs ou des motifs finis multidimensionnels sur un alphabet fini comme valeur des résumés. Tous ces objets peuvent être représentés par des chaînes de caractères binaires finies (avec un encodage naturel) ; par conséquent dans certains cas nous abusons des notations et parlons de résumés ayant pour valeur des entiers (ou des vecteurs, des motifs finis, et ainsi de suite) et de leur complexité de Kolmogorov.
\end{remark}

\begin{proposition}\label{p:epitomes} 
(a) Pour tout shift sofique avec une famille de résumés ${\cal E}_n$, pour tout motif $P$ de taille  $n\times n$ globalement admissible 
tel que ${\cal E}_n(P)$ est défini, nous avons
 $
 \CK({\cal E}_n(P)) = O(n).
 $
 
(b) Pour tout shift sofique avec une famille de résumés calculable en temps exponentiel ${\cal E}_n$, pour chaque motif $P$ de taille $n\times n$ globalement admissible tel que ${\cal E}_n(P)$ est défini,  nous avons
 $
 \CK^{T(n)}({\cal E}_n(P)) = O(n)
 $
Pour une fonction de seuil $T(n)=2^{O(n^2)}$.

(c) Les énoncés (a) et (b) relativisés. C'est à dire, pour tout shift sofique avec une famille de résumés ${\cal E}_n$ qui est \emph{calculable avec un oracle} $\cal O$ (ou \emph{calculable en temps exponentiel avec un oracle}  $\cal O$), pour tout motif $P$ de taille $n\times n$ globalement admissible tel que ${\cal E}_n(P)$ est défini, nous avons
 $
 \CK^{\cal O}({\cal E}_n(P)) = O(n)
 $
 (ou, respectivement,   $\CK^{T(n), \cal O}({\cal E}_n(P)) = O(n)$ pour une fonction de seuil  $T(n)=2^{O(n^2)}$).

\end{proposition}

\begin{proof}
(a) Supposons que $S$ est un shift sofique avec un Shift de type fini couvrant $\hat S$ ($S$ est une projection coordonnée par coordonnée $\pi$ de $\hat S$).
Soit $P$ un motif de $S$ de support $B_n$  et $R$ un motif sur le complément de $B_n$ qui impose la valeur du résumé ${\cal E}_n$ de $P$ (comme spécifié dans la Définition~\ref{d:epitome}). Appelons $Y$ une configuration de $\hat S$ dont la projection par $\pi$ donne l'union de $P$ et $R$.
Soit $Q$ un motif de taille $n\times n$ dans $Y$ tel que $P$ est la projection coordonnée par coordonnée de $Q$.
Soit $\partial Q$ la ``bordure'', i.e., le voisinage immédiat de $Q$ (la bordure de $Q$ n'est pas incluse dans $Q$; si $Q$ est un carré de taille $n\times n$, alors la bordure $\partial Q$ de $Q$ consiste en $4(n+1)$ lettres, voir Fig.~\ref{f:sft2sofique}).

Nous supposons sans perte de généralité que les contraintes locales de $\hat S$ n'impliquent que des paires de nœuds voisins de $\mathbb{Z}^2$. Alors, tout motif $Q'$ de taille $n\times n$ localement admissible qui est localement compatible avec la bordure $\partial Q$, est également compatible avec le reste de la configuration $Y$. Par conséquent, les projections de ces motifs $Q'$ par$\pi$ sont  compatibles avec $R$.
Nous appliquons la définition des résumés ; si ${\cal E}_n$ est défini sur les projections par $\pi$ de tels $Q'$, alors les valeurs du résumé sont égales à ${\cal E}_n(P)$.

\begin{figure}
\begin{tikzpicture}[scale=.6,on grid]

%
   \begin{scope}[
           yshift=-95,every node/.append style={
           yslant=0.5,xslant=-1},yslant=0.5,xslant=-1
           ]
       \draw[step=4mm, black] (0,0) grid (5,5);
       \draw[black,thick] (0,0) rectangle (5,5);
              \foreach \x in {-1,...,4} \foreach \y in {-1,...,4} 
           \fill[black!50] (1.65+0.4*\x,1.95+0.4*\y) rectangle (1.95+0.4*\x,1.65+0.4*\y); 


   \end{scope}

   \begin{scope}[
           xshift=0.0,yshift=0.5,every node/.append style={
           yslant=0.5,xslant=-1},yslant=0.5,xslant=-1
           ]
       \fill[white,fill opacity=0.80] (0,0) rectangle (5,5);
       \draw[step=4mm, black] (0,0) grid (5,5); 
       \draw[black,thick] (0,0) rectangle (5,5);

       \foreach \x in {-1,...,4} \foreach \y in {-1,...,4} 
           \fill[blue!70] (1.65+0.4*\x,1.95+0.4*\y) rectangle (1.95+0.4*\x,1.65+0.4*\y);

    \foreach \x in {-1,...,6} 
       \fill[red!41] (1.25+0.4*\x,1.55-0.4) rectangle (1.55+0.4*\x,1.25-0.4);
    \foreach \x in {-1,...,6} 
       \fill[red!41] (1.25+0.4*\x,1.55+6*0.4) rectangle (1.55+0.4*\x,1.25+6*0.4);
    \foreach \x in {-1,...,6} 
       \fill[red!41] (1.25-0.4+0*\x,1.55+\x*0.4) rectangle (1.55-0.4+0*\x,1.25+\x*0.4);
    \foreach \x in {-1,...,6} 
       \fill[red!41] (1.25+2+0.4,1.55+\x*0.4) rectangle (1.55+2+0.4,1.25+\x*0.4);
   \end{scope}

%
%
 
   \draw[-latex,ultra thick,black](-3,5)node[left]{motif $Q$ de taille $n\times n$}
       to[out=0,in=90] (0,2.6);
   \draw[-latex,ultra thick,black](-4,4)node[left]{bordure $\partial Q$ }
       to[out=0,in=90] (-.8,1.4);


   \draw[thick,black!30,dashed] (2.0,2.4)  -- (2.0,1.0);
   \draw[-latex,thick,black,dashed](2.0,1.0)  to (2.0,-1.0);
   \node[color=black] at (2.4,0.85) {$\pi$};
   \node [draw,circle,inner sep=0.5pt,fill] at (2.0,2.4) {};

   \draw[-latex,ultra thick,black](-4,-3)node[left]{motif $P$ de taille $n\times n$}
       to[out=0,in=200] (-0.5,-1.3);
   \draw[thick,black](6.4,4.5) node {une configuration dans un SFT};
   \draw[thick,black](7.3,-2.8) node {une configuration dans un shift sofique};

\end{tikzpicture}
\caption{Projection d'un motif de taille $n\times n$ d'un shift de type fini vers un shift sofique.} \label{f:sft2sofique}
\end{figure}
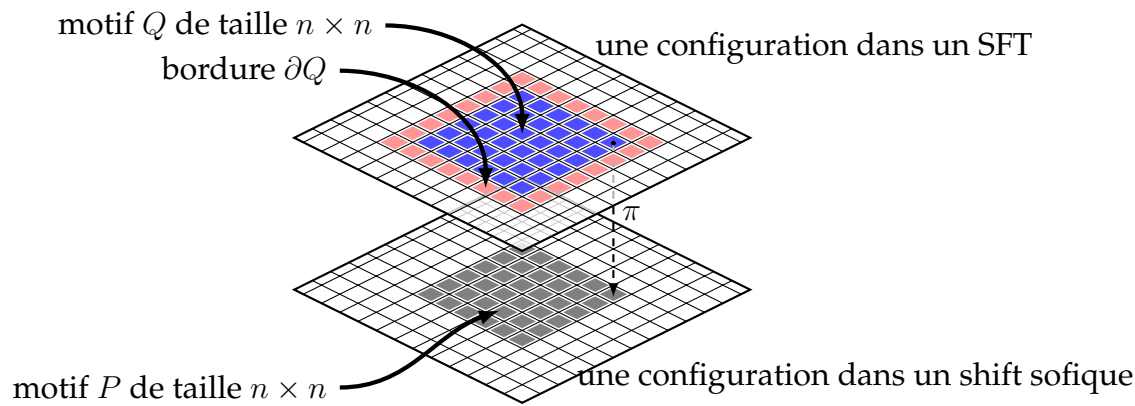

Il s'ensuit que ${\cal E}_n(P)$ peut être calculé étant donné le coloriage de $\partial Q$:  nous utilisons une recherche exhaustive pour trouver tous les motifs $Q'$ localement compatibles avec $\partial Q$, appliquons la projection $\pi$ sur ceux-ci, et calculons en parallèle pour chacun d'eux leur résumé. Dès que nous avons trouvé au moins un $Q'$ compatible avec  $\partial Q$ tel que le résumé de sa projection par $\pi$ est défini, nous avons fini :
même si la projection $\pi(Q')$ peut être différente de $P$, la valeur du résumé de $\pi(Q')$  coïncide avec celle du résumé de $P$. Puisque la taille de $\partial Q$ est linéaire en $n$, nous concluons que 
$
\CK({\cal E}_n(P)) = O(n).
$

Nous pouvons remarquer que pour une bordure $\partial Q$ dans $\hat S$, il peut exister plusieurs motifs $Q$ différents (dans $\hat S$), qui peuvent être projeté sur des motifs $P$ de $S$ différents ; néanmoins, tous ces motifs $P$ auront le même résumé. C'est pourquoi il est suffisant de trouver le premier motif $Q$ compatible avec la bordure $\partial Q$, puis de calculer sa projection.

Pour prouver le fait~(b), il nous suffit d'observer que la recherche exhaustive de tous les motifs de taille $n\times n$ s’exécute en temps $2^{O(n^2)}$, nous avons donc
$
\CK^{2^{O(n^2)}}({\cal E}_n(P)) = O(n).
$
Il est facile de vérifier que la preuve se relativise (chaque étape de l'argumentation reste valide avec un oracle $\cal O$ fixé), et donc nous obtenons le fait~(c).
\end{proof}

Nous utiliserons les résumés avec oracle dans la Section~\ref{s:km}.

La Proposition~\ref{p:epitomes} donne une condition nécessaire pour être sofique. Pour montrer qu'un shift n'est pas sofique, nous devons trouver une famille de résumés calculable (ou calculable en temps exponentiel) avec une grande complexité de Kolmogorov (ou, respectivement, de grande complexité  de Kolmogorov à ressources bornées).

\begin{remark}
Quand nous donnons à un algorithme un accès à un oracle, sa puissance de calcul peut augmenter très fortement. Ainsi, quand nous affirmons que 
 $\CK^{ \cal O}(w) = O(n)$ pour un oracle $\cal O$ (non spécifié), nous n'affirmons pas grand chose sur la chaîne de caractères binaires $w$. Cependant, pour tout oracle nous pouvons toujours affirmer qu'il ne peut y avoir qu'au plus $2^{O(n)}$ objets de complexité $O(n)$. Donc quand nous affirmons que $\CK^{ \cal O}({\cal E}_n(P)) = O(n)$ pour tout $P$, cela équivaut à dire qu'il existe $2^{O(n)}$ valeurs pour les résumés ${\cal E}_n$. Autrement dit, les limites sur la complexité de Kolmogorov avec un oracle (et sans limite de ressources) sont seulement un langage différent pour les traditionnels arguments de comptage.
En revanche, le sens de la complexité de Kolmogorov devient plus subtil quand nous restreignons les ressources de calcul à disposition des algorithmes ; les énoncés $\CK^{T(n)}(w) = O(n)$  ou $\CK^{T(n),{\cal O}}(w) = O(n)$ ne peuvent être réduits à un simple argument de comptage.
\end{remark}

Dans la suite nous présentons deux exemples simples d'application de cette technique. Nous commençons par l'exemple du shift symétrique en miroir.

\begin{examplerev}{ex:mirror}
\label{ex:mirror-revisited}

Soit $S_{\text{mirror}}$ le shift de l'Exemple~\ref{ex:mirror} du chapitre d'introduction (les configurations symétriques en miroir).
Nous ne pouvons appliquer le Théorème~\ref{th:dls} directement et conclure que le shift n'est pas sofique. L'obstacle réside dans le fait que ce shift admet des motifs avec une très petite complexité de Kolmogorov à ressources bornées; par exemple, le shift admet la configuration avec une ligne horizontale  infinie en rouge et uniquement des lettres blanches au-dessus et en-dessous de celle-ci.
Pour cet exemple nous définissons les fonctions de résumé ${\cal E}_n$ de la manière suivante :
\begin{itemize}
\item Si un motif $P$ de taille $n\times n$ contient uniquement des lettres noires et blanches, alors ${\cal E}_n(P)$ est égal à $P$ (de manière informelle, ${\cal E}_n$ ne compresse pas les motifs en noir et blanc);
\item  ${\cal E}_n$ n'est pas défini pour les motifs contenant au moins une lettre rouge.
\end{itemize}
La famille ${\cal E}_n$ définie satisfait la définition d'une famille de résumés calculable, puisque la partie de la configuration au-dessus de la ligne rouge détermine de manière unique tous les motifs noir et blanc situés en-dessous de celle-ci. 
Plus précisément, pour tout motif $P$ de taille $P$ qui contient uniquement des lettres noires et blanches, 
soit $R$ le motif infini défini de la manière suivante :
\begin{itemize}
\item le support de $R$ est le complément du support de $P$ de taille $n\times n$,
\item la ligne horizontale dans $R$ située une ligne au-dessus de $P$ est constituée uniquement de lettres rouges,
\item les positions symétriques (par rapport à la ligne rouge) aux positions des lettres noires de $P$ sont noires,
\item toutes les autres lettres de $R$ sont blanches,
\end{itemize}
comme montré dans la Fig.~\ref{f:p-and-r}. Il est clair que $R$ est compatible avec $P$ et n'est pas compatible avec d'autres motifs de taille $n\times n$, comme demandé dans la définition des résumés.

Puisqu'il y a $2^{n^2}$ motifs de taille $n\times n$ avec des lettres noires et blanches, pour certains motifs de taille $n\times n$ nous avons
 $
 \CK({\cal E}_n(P)) \ge n^2.
 $
Par conséquent, nous pouvons appliquer la Proposition~\ref{p:epitomes}~(a) et conclure que le shift n'est pas sofique.
\begin{figure}[H]

  \centering
  \begin{tikzpicture}[scale=0.25,x=1cm,baseline=2.125cm]

  \pgfmathsetseed{1}
    \foreach \x in {1,...,12} \foreach \y in {1,...,7}
    {
        \pgfmathparse{mod(int(random*23),2) ? "black!10" : "black!66"}
        \edef\colour{\pgfmathresult}
        \path[fill=black!10,draw=black] (\x,7-\y) rectangle ++ (1,1);
        \path[fill=black!10,draw=black] (\x,7+\y) rectangle ++ (1,1);
    }

    \foreach \x in {6,...,10} \foreach \y in {1,...,5}
    {
        \pgfmathparse{mod(int(random*23),2) ? "black!10" : "black!66"}
        \edef\colour{\pgfmathresult}
        \path[fill=\colour,draw=black] (\x-20,7-\y) rectangle ++ (1,1);
        \path[fill=\colour,draw=black] (\x,7-\y) rectangle ++ (1,1);
        \path[fill=\colour,draw=black] (\x,7+\y) rectangle ++ (1,1);
    }

     \foreach \x in {1,...,12} 
    {
        \path[fill=red!41,draw=black] (\x,7) rectangle ++ (1,1);
    }

\path[fill=white,draw=blue, ultra thick] (6,5-3) rectangle ++ (5,5);
\path[draw=orange, ultra thick, dashed] (6,4-3+7+0.1) rectangle ++ (5,5);
\path[draw=blue, ultra thick] (6-20,5-3) rectangle ++ (5,5);

    \draw[-latex,ultra thick,blue!95](-18,10.5)node[left]{$P$}    to[out=0,in=90] (-11.5,7.3);    
    \draw[-latex,ultra thick,blue!95](20,3.5)node[below]{$R$}    to[out=90,in=0] (13.2,7.5); 

    \draw[-latex,ultra thick,blue!95, dashed](-11.5,3.8)node[below]{}    to[out=-30,in=-135] (8.0,4.2);

\end{tikzpicture}

\caption[Un motif $P$ de taille $n\times n$ composé de lettres noires et blanches et le motif $R$ correspondant sur le domaine complémentaire.]{Un motif $P$ de taille $n\times n$ composé de lettres noires et blanches et le motif $R$ correspondant sur le domaine complémentaire. Un cadre orange indique les lettres de $R$ qui sont symétriques en miroir à celles de $P$.} 
\label{f:p-and-r}
\vspace{-5pt}

\end{figure}
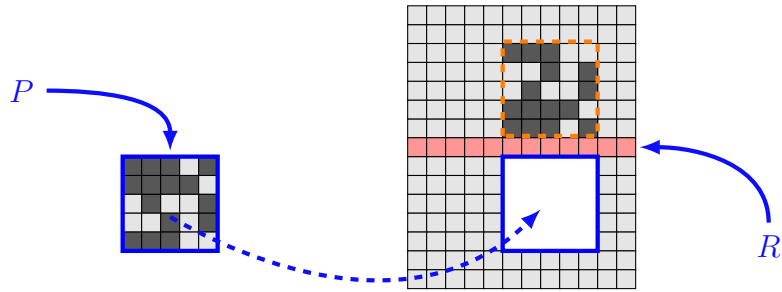
\end{examplerev}

\begin{example}[Le shift miroir avec une petite complexité de Kolmogorov simple]
\label{ex:mirror-low-complexity}
Soit le shift $S_{\text{mirror}}$ de l'Exemple~\ref{ex:mirror}:
Nous n'admettons toujours que les configurations symétriques, mais nous demandons en outre que les motifs $P$ de taille $n\times n$ noir et blanc soient globalement admissibles pour le shift ${S}_{0}$  défini dans la Remarque~\ref{r:paradoxical-shift},
p.~\pageref{r:paradoxical-shift}.
Une configuration non dégénérée de ce nouveau shift se présente ainsi :
il y a une ligne horizontale infinie composée de lettres rouges, et deux demi-plans (constitués de lettres noires et blanches) au-dessus et en-dessous de cette ligne, tels que tout motif $P$ de taille $n\times n$ contenu dans un demi-plan vérifie
\[
 \CK^{\hat T(n)}(P) \le \hat \lambda \log n
\]
(le choix de $\hat \lambda$ et de $ \hat T(n)$ est expliqué dans la Remarque~\ref{r:paradoxical-shift}).
Le nouveau shift est effectif, et le nombre de motifs globalement admissibles est au plus
$
2^{O(\log n)}=\poly(n).
$
Nous savons aussi que pour \emph{quelques} motifs de taille $n\times n$ globalement admissibles on a
\begin{equation}\label{eq:super-linear-lower-bound}
\CK^{2^{n^3}}(P) =\Omega(n^{1.5}). 
\end{equation}

La Proposition~\ref{p:epitomes}~(a) ne peut pas s'appliquer, car la complexité de Kolmogorov simple de tout motif globalement admissible (et, donc, de tout résumé calculable) est logarithmique. Cependant, dans ce cas nous pouvons utiliser les résumés calculables en temps exponentiel.

Les fonctions ${\cal E}_n$ définies dans l'Exemple~\ref{ex:mirror} (voir p.~\pageref{ex:mirror-revisited}) fournissent une famille de résumés \emph{calculables en temps exponentiel} pour ce shift. 
Pour certains motifs $P$ de taille $n\times n$ (bien que pas pour tous) nous avons \eqref{eq:super-linear-lower-bound},
et donc il s'ensuit par la Proposition~\ref{p:epitomes}~(b) que le shift n'est pas sofique.
\end{example}

\begin{remark}
Dans la définition des résumés nous avons autorisé que les fonctions ${\cal E}_n$ puissent être définies sur les motifs non admissibles. 
Dans l'exemple suivant nous montrons que cette possibilité peut être utile.
\end{remark}
\begin{figure}
\centering

  \begin{tikzpicture}[scale=0.27,x=1cm,baseline=2.125cm]
  
  \pgfmathsetseed{5}
    \foreach \x in {1,...,8} \foreach \y in {1,...,8}
    {
        \pgfmathparse{mod(int(random*23),2) ? "black!10" : "black!66"}
        \edef\colour{\pgfmathresult}
        \path[fill=\colour,draw=black] (\x,\y) rectangle ++ (1,1);
    }
    
     \foreach \x in {1,...,8}
    {
    \path[fill=red!41,draw=black] (\x,0) rectangle ++ (1,1);
    \path[fill=red!41,draw=black] (\x,9) rectangle ++ (1,1);    
    }    

     \foreach \y in {0,...,9}
    {
    \path[fill=red!41,draw=black] (0,\y) rectangle ++ (1,1);
    \path[fill=red!41,draw=black] (9,\y) rectangle ++ (1,1);    
    }

\end{tikzpicture}
  \caption[Un exemple de motif spécial de taille $n\times n$ de l'Exemple~\ref{ex:busy-beaver}.]{Un exemple de motif spécial de taille $n\times n$ de l'Exemple~\ref{ex:busy-beaver} : la bordure (incluse dans le motif) est constituée de lettres rouges, et les $(n-2)\times (n-2)$ lettres restantes peuvent être noires ou blanches.}
  \label{fig:special}
\end{figure}
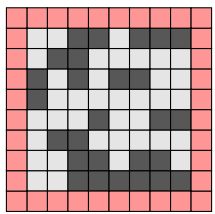  
\begin{example}
\label{ex:busy-beaver}
Dans cet exemple nous utilisons à nouveau un alphabet de trois lettres (\emph{rouge}, \emph{noir}, et \emph{blanc}). Avant de définir un shift sur cet alphabet, nous allons introduire quelques notations.
Nous appelons un motif $P$ de taille $n\times n$ \emph{special} s'il ne contient que des lettres rouges sur sa bordure, et seulement des lettres noires et blanches dans son intérieur, comme montré dans la Fig.~\ref{fig:special}. Il est commode d'imposer qu'un motif spécial contienne au moins une lettre noire, et il y a donc $2^{(n-2)^2}-1$ motifs spéciaux de taille $n\times n$.
Pour tout $n$ nous fixons une bijection entre ces motifs spéciaux de taille $n\times n$ et les chaînes de caractères binaires de taille inférieure à $(n-2)^2$. Nous appelons $x_P$ la chaîne de caractère binaire correspondant au motif $P$.

Soit $U$ une fonction calculable de la définition de la complexité de Kolmogorov (voir la Section~\ref{t:kolmogorov-invariance}). Soit un algorithme calculant $U$. Il est connu que la fonction  $U$ est partielle, donc l'algorithme diverge sur certaines entrées. 
Pour tout entier $m> 0$ nous appelons $\hat p_m$ la chaîne de caractères binaire de longueur inférieure à  $m$ telle que le calcul de $U(\hat p_m)$ converge mais prend plus de temps que le calcul de $U(p)$ pour tout autre $p$ avec $|p|<m$. En cas d'égalité, (s'il y a plusieurs $p$ nécessitant le même nombre d'étapes de calcul) nous choisissons comme $\hat p_m$ le premier pour l'ordre lexicographique.

Nous définissons alors le shift sur $\mathbb{Z}^2$ par la règle suivante : pour tout $n$, un motif spécial $P$ de taille $n\times n$ est globalement admissible si $U(x_P)$ n'est pas défini ou si $x_P = \hat p_{(n-2)^2}$ (i.e., si le motif correspond au résultat du plus long calcul de $U$ parmi ceux qui terminent pour toutes les entrées de taille inférieure à $(n-2)^2$). Il n'y a pas d'autre contrainte : une configuration appartient au shift si et seulement si elle ne contient aucun motif spécial interdit par la règle.

Le shift est effectif : nous devons interdire les motifs de complexité faible. Il suffit donc pour cela d'énumérer les motifs ayant une description courte. Comme l'ensemble des entrées $p$ (les descriptions) pour lesquelles le calcul de $U(p)$ converge est récursivement énumérable, l'ensemble des motifs spéciaux interdits l'est aussi. 

Nous allons maintenant montrer que le shift n'est pas sofique.
Pour cela, nous définissons une famille de résumés ${\cal E}_n$ de la manière suivante : 
\begin{itemize}
\item pour tout motif spécial $P$ de taille $n$ nous fixons ${\cal E}_n(P)= P$ si $U(x_P)$ est défini;
\item  ${\cal E}_n$ n'est pas défini pour les autres motifs.
\end{itemize}
La famille de résumés ${\cal E}_n$ est clairement calculable.
On peut noter que ${\cal E}_n$ est défini pour de nombreux motifs non admissibles et pour un et un seul motif globalement admissible de taille $n\times n$ (à savoir le motif correspondant au plus long calcul de $U$). Il est clair que la condition pour être une famille de résumés est satisfaite.

Étant donné $\hat p_{(n-2)^2}$, nous pouvons trouver le nombre maximal d'étapes de calcul pour que $U(p)$  converge pour $|p| <(n-2)^2$.
Cette information nous permet alors de trouver tous les calculs de $U(p)$ qui convergent pour $|p| <(n-2)^2$, et donc d'obtenir toutes les chaînes de caractères binaires $x$ ayant une complexité de Kolmogorov strictement inférieure à $(n-2)^2$. Nous pouvons alors renvoyer la première chaîne de caractères binaires incompressible de taille $(n-2)^2$. 
Autrement dit, étant donné $\hat p_{(n-2)^2}$ nous pouvons trouver un objet avec une complexité de  Kolmogorov d'au moins $(n-2)^2$. Par conséquent, la complexité de Kolmogorov de $\hat p_{(n-2)^2}$ elle-même est proche de $(n-2)^2$.
Ainsi, pour le seul motif $P$ de taille $n\times n$ globalement admissible pour lequel ${\cal E}_n(P)$ est défini, nous avons
\[
\CK({\cal E}_n(P)) = \Omega(n^2).
\]
Il s'ensuit d'après la Proposition~\ref{p:epitomes}~(a) que le shift n'est pas sofique. 

Nous pouvons noter pour conclure que dans cet exemple nous avons utilisé un résumé défini sur de nombreux motifs non admissibles, et que l'union des domaines de ${\cal E}_n$ pour $n\ge 0$ est récursivement énumérable mais non décidable. Ces propriétés semblent être essentielles dans cette preuve.
\end{example}

\subsection{Résumés ordonnés}\label{s:ordered-epitomes}
Dans l'article \cite{kass-madden}, S.Kass et K.Madden développent une technique générale permettant de prouver qu'un shift n'est pas sofic. En particulier, ils appliquent celle-ci sur un nouveau shift, dans l'Exemple~2.5 de leur article (nous allons présenter cet exemple par la suite).

Les preuves basées sur la Définition~\ref{d:epitome} ne s'appliquent pas à cet exemple et aux shifts similaires. Pour prouver que ces shifts ne sont pas sofiques, nous introduisons une définition un peu plus générale des résumés :
\begin{definition}
\label{d:ordered-epitome}
Soit $E_n$ un ensemble fini muni d'un ordre partiel $\preccurlyeq_n$, et
\[
{\cal E}_n \ : \ [\text{motif de taille }n\times n] \mapsto [\text{élément de }E_n]
\]
une fonction \emph{partielle}, pour tout entier $n>0$. Nous appelons $({\cal E}_n, \preccurlyeq_n)$ une famille de \emph{résumés ordonnés} pour un shift $S$, 
si pour tout motif $P$ de support $B_n$ globalement admissible tel que ${\cal E}_n(P)$ soit défini, il existe un motif $R$ sur $F_n$
tel que
\begin{itemize}
\item[(i)] $R$ est compatible avec $P$, i.e., l' union de $P$ et $R$ constitue une configuration valide de $S$, et
\item[(ii)]  pour tout motif $P'$ de $B_n$ compatible avec $R$, si ${\cal E}_n(P')$ est défini alors
\[
{\cal E}_n(P')  \preccurlyeq_n  {\cal E}_n(P) 
\]
\end{itemize}
\textup(i.e., cette configuration $R$ sur le complément de $B_n$ détermine le maximum des résumés ${\cal E}_n$ pour tout motif $P'$ valide\textup).  

Nous appelons une famille de résumés ordonnés \emph{calculable} 
s'il existe des algorithmes qui calculent les relations $\preccurlyeq_n$ et les fonctions ${\cal E}_n$. 
Précisons que le calcul est uniforme : il existe deux algorithmes qui calculent respectivement $\preccurlyeq_n$ et ${\cal E}_n$ pour tout entier $n>0$.

Si, de plus, ${\cal E}_n$ et $\preccurlyeq_n$ sont calculables en temps $2^{O(n^2)}$, nous appelons cette famille de résumés ordonnés \emph{calculable en temps exponentiel}.
Comme il est courant, un algorithme calculant une fonction partielle calculable ne doit renvoyer aucun résultat pour une entrée en dehors de son domaine.

De manière similaire, nous appelons une famille de résumés ordonnés \emph{calculable avec un oracle $\cal O$} s'il existe deux algorithmes qui 
calculent respectivement les relations $\preccurlyeq_n$  et les fonctions ${\cal E}_n$ étant donné un accès à l'oracle $\cal O$. 
Nous définissons aussi les résumés ordonnés \emph{calculable en temps exponentiel} avec un oracle.
\end{definition}

Les résumés simples de la Définition~\ref{d:epitome} peuvent être vus comme un cas particulier de la Définition~\ref{d:ordered-epitome}.
Si ${\cal E}_n$ est une famille de résumés calculable (ou calculable en temps exponentiel) dans le sens de la Définition~\ref{d:epitome}
et si $\preccurlyeq_n$ est un ordre calculable (calculable en temps exponentiel) arbitraire sur les résumés ${\cal E}_n$, alors $({\cal E}_n,\preccurlyeq_n)$ est une famille de résumés ordonnés calculable (ou respectivement calculable en temps exponentiel) dans le sens de la Définition~\ref{d:ordered-epitome}.
Dans la Définition~\ref{d:epitome} le voisinage $R$ détermine la valeur exacte de ${\cal E}_n(P')$ pour tous les motifs $P'$ compatible avec $R$, 
alors que dans la  Définition~\ref{d:ordered-epitome} $R$ détermine seulement le maximum des ${\cal E}_n(P')$.

\begin{proposition}\label{pbis:epitomes} 
(a) Pour tout shift sofique avec une famille de résumés ordonnés calculable $({\cal E}_n,\preccurlyeq_n)$, pour tout motif $P$ de taille $n\times n$ globalement admissible tel que ${\cal E}_n(P)$ est défini, nous avons
 $
 \CK({\cal E}_n(P)) = O(n).
 $
 
(b) Pour tout shift sofique avec une famille de résumés ordonnés calculable en temps exponentiel $({\cal E}_n,\preccurlyeq_n)$, pour tout motif $P$ de taille $n\times n$ globalement admissible tel que ${\cal E}_n(P)$ est défini, nous avons
 $
 \CK^{T(n)}({\cal E}_n(P)) = O(n)
 $
pour une fonction de seuil $T(n)=2^{O(n^2)}$.

(c) Les propositions (a) et (b) relativisées. C'est-à-dire que, pour tout shift sofique avec une famille de résumés ordonnés $({\cal E}_n,\preccurlyeq_n)$ \emph{calculable avec un oracle} $\cal O$ (ou \emph{calculable en temps exponentiel avec un oracle}  $\cal O$), pour tout motif $P$ de taille $n\times n$ globalement admissible tel que ${\cal E}_n(P)$ est défini, nous avons
 $
 \CK^{\cal O}({\cal E}_n(P)) = O(n)
 $
 (ou, respectivement,   $\CK^{T(n), \cal O}({\cal E}_n(P)) = O(n)$ pour une fonction de seuil $T(n)=2^{O(n^2)}$).

\end{proposition}
\begin{proof}
Nous réutilisons les notations de la preuve de la Proposition~\ref{p:epitomes}: soit $S$ un shift sofique, $\hat S$ un shift de type fini couvrant $S$, et $\pi$ une projection coordonnée par coordonnée de $\hat S$ sur $S$. Nous supposons que les contraintes locales de $\hat S$ n'impliquent que des paires de nœuds voisins dans $\mathbb{Z}^2$.
Pour un motif $P$ de taille $n\times n$ de $S$, nous appelons $Q$ un motif de taille $n\times n$ de $\hat S$ tel que son projeté par $\pi$ est $P$, et $\partial Q$ sa \emph{bordure} (les plus proches voisins) autour de $Q$ constituée de $4(n+1)$ lettres, comme montré dans la Fig.~\ref{f:sft2sofique}. Nous appelons de tels motifs $\partial Q$ des \emph{bordures} de taille $n$.

\smallskip

Nous commençons par la preuve de (b), qui est très similaire à la preuve de la Proposition~\ref{p:epitomes}. Dans la preuve précédente, il était suffisant de trouver \emph{au moins un} motif $Q'$ compatible avec la bordure $\partial Q$ tel que le résumé de $\pi(Q')$ est défini; alors en calculant le résumé de $\pi(Q')$ nous obtenions le résumé de $P$. Ici, nous devons trouver \emph{tous}  les motifs $Q'$ compatibles avec $\partial Q$, appliquer à chacun d'eux la projection $\pi$, calculer leur résumé (pour les motifs pour lesquels ${\cal E}_n$ est défini), et ensuite prendre le maximum des résultats obtenus. 
Le fait que les ${\cal E}_n$ soient des fonctions partielles ne pose pas de problème. En effet, nous avons supposé que les résumés sont calculables en temps exponentiel ; nous pouvons donc stopper les calculs qui n'ont pas convergé dans le temps imparti.
Il reste à observer que la recherche exhaustive sur l'ensemble des motifs de taille $n\times n$  nécessite l'examen de $2^{O(n^2)}$ possibilités, ce qui peut être réalisé en temps exponentiel.

\smallskip

Nous ne pouvons prouver~(a) de la même manière car le calcul des résumés peut nécessiter un temps arbitrairement long, et nous ne pouvons pas trouver de manière algorithmique la liste finale des ${\cal E}_n(P')$ pour tous les motifs $P'$ compatibles avec $\partial Q$. Dans ce cas nous avons besoin d'un argument plus subtil.

Nous allons utiliser de nouveau le fait suivant : 
\begin{claim}\label{claim:1}
Le nombre de bordures $\partial Q$ de taille $n$ globalement admissibles de $\hat S$ n'est pas plus grand que $2^{O(n)}$.
\end{claim}

La preuve de ce fait est assez trivial :
\begin{proof}
Le domaine de ces bordures consiste en $4(n+1)$ nœuds, et le nombre de motifs de cette taille (admissibles ou non) est au plus de $|\Sigma|^{4(n+1)}$, où $\Sigma$ est l'alphabet de $\hat S$.
\end{proof}

\begin{claim}\label{claim:2}
Le nombre de valeurs de ${\cal E}_n(P)$ pour les motifs $P$ globalement admissibles de $S$ n'est pas plus grand que le nombre de toutes les bordures $\partial Q$ de taille $n$ qui sont globalement admissibles dans $\hat S$.
\end{claim}
\begin{proof}
Par la définition des résumés ordonnés, pour toute valeur de ${\cal E}_n(P)$ il existe un motif $R$ sur $F_n$ qui détermine implicitement ${\cal E}_n(P)$ : le maximum des ${\cal E}_n(P')$ pour tous les motifs $P'$ compatibles avec $R$ est égal à ${\cal E}_n(P)$. Bien que $R$ soit infini, les ``interactions'' entre $P$ et $R$ passent à travers un contour de taille linéaire. 
Comme nous l'avons vu dans la preuve de la Proposition~\ref{p:epitomes}, nous pouvons choisir une bordure $\partial Q$ de taille $n$ (globalement admissible dans $\hat S$) 
avec la propriété suivante :
si nous prenons l'ensemble des motifs $Q$ de taille $n\times n$ localement compatibles avec $\partial Q$ dans $\hat S$, et considérons les projections par $\pi$ de ces motifs, nous obtenons le motif recherché  $P$ et potentiellement quelques autres motifs $P'$ (globalement admissibles pour $S$) tels que ${\cal E}_n(P')$, si défini, est inférieur ou égal à ${\cal E}_n(P)$ dans le sens de $\preccurlyeq_n$.
Par conséquent, le choix de la bordure $\partial Q$ de taille $n$ détermine de manière unique la valeur de ${\cal E}_n(P)$. 

On peut noter que $\partial Q$ détermine ${\cal E}_n(P)$ seulement dans un sens abstrait ; nous ne pouvons pas calculer ${\cal E}_n(P)$ à partir de $\partial Q$ de manière algorithmique. En effet, puisque ${\cal E}_n$ est une fonction partielle (et que l'algorithme calculant ${\cal E}_n$ ne s'arrête pas sur les entrées pour lesquelles la fonction n'est pas définie), nous ne pouvons trouver de manière algorithmique la valeur maximale des ${\cal E}_n(P')$ pour tous les $P'$ obtenus à partir de $\partial Q$. Cependant, nous pouvons affirmer que le nombre de valeurs des ${\cal E}_n(P)$ pour tous les motifs $P$ globalement admissibles (dans le shift $S$) n'est pas plus grand que le nombre de bordures $\partial Q$ globalement admissibles (dans le shift $\hat S$).
\end{proof}

L'ensemble des bordures de taille $n$ globalement admissibles d'un shift de type fini (en fait, d'un shift énumérable) est co-énumérable: nous pouvons examiner tous les motifs rangés par taille croissante, détecter ceux qui contiennent des motifs interdits, et ainsi obtenir une par une les bordures inadmissibles (celles dont toutes les extensions contiennent des motifs interdits).  
En utilisant cette procédure nous finirons par trouver tous les motifs non admissibles ; par contre, nous ne pouvons \emph{à priori} pas savoir quand la dernière bordure de taille $n$ inadmissible est découverte.

Notons $N_n$ le nombre de bordures de taille $n$ globalement admissibles pour $\hat S$. Étant donné ce nombre (sa représentation binaire) nous pouvons déterminer le moment où l'algorithme décrit ci-dessus finalise la liste des bordures de taille $n$ globalement admissibles (i.e., quand toutes les bordures de taille $n$ inadmissibles ont été trouvées. Ainsi, étant donné la liste de toutes les bordures de taille $n$ non éliminées (globalement admissibles), nous pouvons trouver tous les motifs de taille $n\times n$ de $\hat S $ qui sont localement compatibles avec au moins une des bordures de taille $n$ localement admissibles (ces motifs sont également globalement admissibles). Nous pouvons alors appliquer à chacun d'eux la projection $\pi$ et obtenir les motifs de taille $n\times n$ qui sont globalement admissibles dans $S$. Enfin, nous appliquons l'algorithme calculant ${\cal E}_n$ à chaque motif obtenu en exécutant ces calculs en parallèle. Chacune des valeurs de ${\cal E}_n(P)$ pour les motifs $P$ globalement admissibles sera tôt ou tard trouvée par cette procédure.

Ainsi, pour décrire une valeur spécifique ${\cal E}_n$, nous devons donc décrire la procédure d'énumération expliquée ci-dessus et spécifier le nombre ordinal de la valeur ${\cal E}_n$ recherchée dans cette énumération (son ordre d'arrivée). Pour pouvoir mener à bien ce processus d'énumération, nous devons connaître les nombres $n$ et $N_n$. Par le Fait~\ref{claim:1}, nous savons que la représentation binaire de $N_n$ est constituée de seulement $O(\log N_n) = O(n)$ bits.
Le nombre de valeurs différentes de ${\cal E}_n$ pour des motifs globalement admissibles est borné par  $N_n$ (Fait~\ref{claim:2}) ; donc le nombre ordinal d'un élément de cette énumération peut être spécifié avec $\log N_n = O(n)$ bits (à nouveau par le Fait~\ref{claim:1}). Par conséquent, toute valeur de ${\cal E}_n$ peut être décrite par un algorithme de taille $O(n)$, ce qui est l'énoncé~(a) du théorème.

\smallskip

Les arguments utilisés dans les preuves de (a) et (b) peuvent être relativisés, et nous obtenons l'énoncé~(c) pour les calculs avec un oracle.
\end{proof}

La Proposition~\ref{pbis:epitomes} nous fournit une nouvelle condition nécessaire pour être sofique. Pour prouver qu'un shift n'est pas sofique, il est donc suffisant de construire une famille de résumés ordonnés calculables (ou calculables en temps exponentiel) avec une complexité de Kolmogorov simple (respectivement une complexité de Kolmogorov à ressources bornées) sur-linéaire.

\begin{example}[semi-mirror shift]
\label{ex:semi-mirror}
Dans cet exemple, nous définissons un shift qui étend celui défini dans l'Exemple~\ref{ex:mirror}.
Soit $\Sigma$ l'alphabet avec trois lettres (par exemple, \emph{noir}, \emph{blanc}, et \emph{rouge}) ;
les configurations admissibles pour ce shift sont toutes les configurations en noir et blanc (sans lettre rouge) et les configurations avec une ligne rouge horizontale infinie et de deux demi-plans au-dessus et en-dessous d'elle qui sont \emph{semi-symétriques} dans le sens suivant :
si la lettre dans la $i$ème colonne à distance $k$ \emph{en-dessous} de la ligne rouge horizontale est noire, alors
la lettre dans la même $i$ème colonne à distance $k$ \emph{au-dessus} de la ligne rouge doit également être noire.
Autrement dit, nous pouvons prendre deux demi-plans symétriques par rapport à la ligne rouge, et ensuite transformer des lettres noires du demi-plan du dessous en lettres blanches, voir Fig.~\ref{f:semi-mirror}. 
Nous appelons ce shift $S_{\text{semi-mirror}}$.

\begin{figure}[H]
\begin{center}

  \centering
  \begin{tikzpicture}[scale=0.25,x=1cm,baseline=2.125cm]

  \pgfmathsetseed{2}
    \foreach \x in {1,...,12} \foreach \y in {1,...,7}
    {
        \pgfmathparse{mod(int(random*23),2) ? "black!10" : "black!66"}
        \edef\colour{\pgfmathresult}
        \path[fill=\colour,draw=black] (\x,7-\y) rectangle ++ (1,1);
        \path[fill=\colour,draw=black] (\x,7+\y) rectangle ++ (1,1);
    }
    
     \foreach \x in {1,...,12} 
    {
        \path[fill=red!41,draw=black] (\x,7) rectangle ++ (1,1);
    }

        \path[fill=black!10,draw=blue,ultra thick] (1,7-4) rectangle ++ (1,1);
        \path[fill=black!10,draw=blue,ultra thick] (3,7-3) rectangle ++ (1,1);
        \path[fill=black!10,draw=blue,ultra thick] (8,7-2) rectangle ++ (1,1);
        \path[fill=black!10,draw=blue,ultra thick] (9,7-6) rectangle ++ (1,1);

        \path[fill=black!66,draw=blue,ultra thick] (1,7+4) rectangle ++ (1,1);
        \path[fill=black!66,draw=blue,ultra thick] (3,7+3) rectangle ++ (1,1);
        \path[fill=black!66,draw=blue,ultra thick] (8,7+2) rectangle ++ (1,1);
        \path[fill=black!66,draw=blue,ultra thick] (9,7+6) rectangle ++ (1,1);

\end{tikzpicture}

\caption[Une configuration du shift $S_{\text{semi-mirror}}$.]{Une configuration du shift $S_{\text{semi-mirror}}$ (nous indiquons en bleu les positions où les demi-plans du dessous et du dessous diffèrent).} \label{f:semi-mirror}
\vspace{-5pt}

\end{center}

\end{figure}
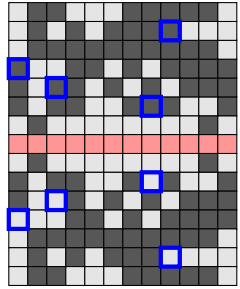

Il est facile de se convaincre que le shift est effectif : les motifs interdits sont ceux où les lettres rouges ne sont pas alignées, et ceux pour lesquels les zones au-dessus et en-dessous de la ligne rouge ne respectent pas la propriété de semi-symétrie.
Nous allons maintenant montrer que ce shift n'est pas sofique. Pour cela nous allons utiliser la technique des résumés ordonnés. 

Nous définissons la famille de résumés ${\cal E}_n$ sur les motifs de taille $n\times n$ qui ne contiennent que des lettres noires et blanches.
Pour un tel motif $P$ le résultat de ${\cal E}_n$ est le motif $P$ lui-même (ou l'encodage en binaire de ce motif de taille $n\times n$). 
${\cal E}_n$ n'est pas défini pour les motifs constitués d'au moins une lettre rouge.
Donc, par définition, il existe $2^{n^2}$ valeurs possibles pour ${\cal E}_n$.

Nous définissons l'ordre partiel $\preccurlyeq_n$ de la manière suivante : ${\cal E}_n(P_1) \preccurlyeq_n {\cal E}_n(P_2)$ si pour chaque position avec une lettre noire dans $P_1$
la position correspondante dans $P_2$ contient aussi une lettre noire, 
voir Fig.~\ref{f:p1-le-p2}.

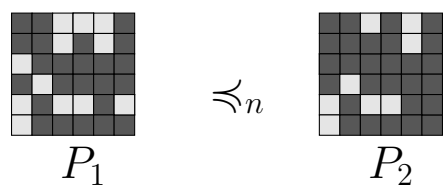
\begin{figure}[H]
\vspace{5pt}
  \centering
  \begin{tikzpicture}[scale=0.27,x=1cm,baseline=2.125cm]
  
  \pgfmathsetseed{1}
    \foreach \x in {1,...,6} \foreach \y in {1,...,6}
    {
        \pgfmathparse{mod(int(random*23),2) ? "black!10" : "black!66"}
        \edef\colour{\pgfmathresult}
        \path[fill=\colour,draw=black] (\x,7-\y) rectangle ++ (1,1);
        \path[fill=\colour,draw=black] (\x+15,7-\y) rectangle ++ (1,1);
    }

        \path[fill=black!10,draw=black] (3,5) rectangle ++ (1,1);
        \path[fill=black!66,draw=black] (3+15,5) rectangle ++ (1,1);

        \path[fill=black!10,draw=black] (4,6) rectangle ++ (1,1);
        \path[fill=black!66,draw=black] (4+15,6) rectangle ++ (1,1);

        \path[fill=black!10,draw=black] (1,4) rectangle ++ (1,1);
        \path[fill=black!66,draw=black] (1+15,4) rectangle ++ (1,1);

        \path[fill=black!10,draw=black] (6,2) rectangle ++ (1,1);
        \path[fill=black!66,draw=black] (6+15,2) rectangle ++ (1,1);

           \draw[thick,black](12,3) node {\Large $\preccurlyeq_n$};

           \draw[thick,black](4.5,-0.5) node {\Large $P_1$};
           \draw[thick,black](19.5,-0.5) node {\Large $P_2$};

\end{tikzpicture}

\caption[Une paire de motifs $P_1$ et $P_2$ avec ${\cal E}_n(P_1)\preccurlyeq_n {\cal E}_n(P_2)$.]{Une paire de motifs $P_1$ et $P_2$: la relation ${\cal E}_n(P_1)\preccurlyeq_n {\cal E}_n(P_2)$ signifie que l'ensemble des position des lettres noires de $P_1$ est inclus dans l'ensemble des positions des lettres noires de $P_2$.} \label{f:p1-le-p2}
\vspace{-5pt}

\end{figure}

Vérifions que $({\cal E}_n, \preccurlyeq_n)$ satisfait la définition des résumés ordonnés. Pour tout motif $P$ de taille $n\times n$ constitué de lettres noires et blanches nous définissons le motif $R$ correspondant exactement de la même manière que dans l'Exemple~\ref{ex:mirror} de la p.~\pageref{ex:mirror-revisited}, voir Fig.~\ref{f:p-and-r}.
Il est clair que quand nous faisons l'union de $P$ avec son motif $R$ correspondant, nous obtenons une configuration admissible du shift. Par ailleurs, si nous faisons l'union de ce motif $R$ avec un autre motif $P'$ (comme dans la Fig.~\ref{f:p-and-p'}), nous obtenons une configuration admissible si et seulement si l'union de $P'$ et de $R$ forme  une configuration semi-symétrique, i.e., une lettre de $P'$ peut être noire \emph{seulement si} la lettre à la position symétrique de $R$ est également noire. Par conséquent, $R$ est compatible avec $P'$ si et seulement si ${\cal E}_n(P') \preccurlyeq_n {\cal E}_n(P)$, ce qui est exactement la propriété requise dans la définition des résumés ordonnés.

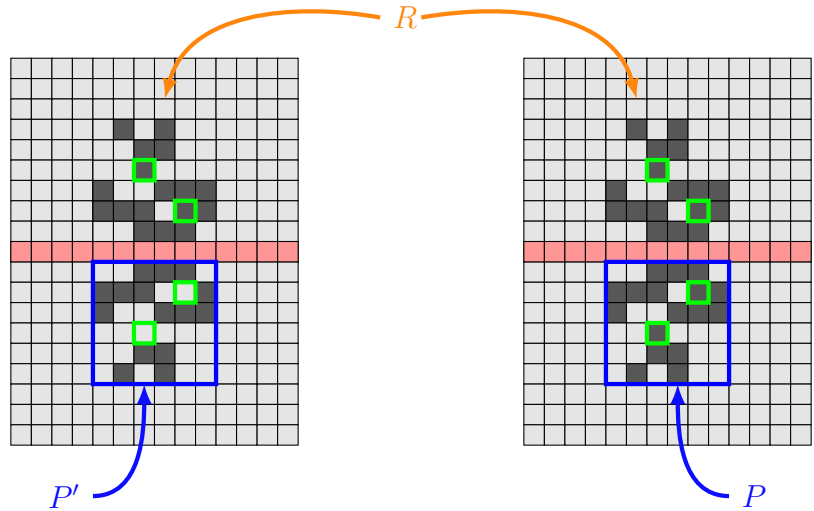
\begin{figure}[H]

  \centering
  \begin{tikzpicture}[scale=0.27,x=1cm,baseline=2.125cm]
  
 \pgfmathsetseed{2}

\foreach \x in {-3,...,10} \foreach \y in {-9,...,9}
{
  \path[fill=black!10,draw=black] (\x,7-\y) rectangle ++ (1,1);
  \path[fill=black!10,draw=black] (\x+25,7-\y) rectangle ++ (1,1);
}
     \foreach \x in {-3,...,10} 
    {
        \path[fill=red!41,draw=black] (\x,7) rectangle ++ (1,1);
        \path[fill=red!41,draw=black] (\x+25,7) rectangle ++ (1,1);
    }

    \foreach \x in {1,...,6} \foreach \y in {1,...,6}
    {
        \pgfmathparse{mod(int(random*23),2) ? "black!10" : "black!66"}
        \edef\colour{\pgfmathresult}
        \path[fill=\colour,draw=black] (\x,7-\y) rectangle ++ (1,1);
        \path[fill=\colour,draw=black] (\x,7+\y) rectangle ++ (1,1);

        \path[fill=\colour,draw=black] (\x+25,7-\y) rectangle ++ (1,1);
        \path[fill=\colour,draw=black] (\x+25,7+\y) rectangle ++ (1,1);

    }

 \path[fill=black!10,draw=green,ultra thick] (3,7-4) rectangle ++ (1,1);
 \path[fill=black!10,draw=green,ultra thick] (5,7-2) rectangle ++ (1,1);
 \path[fill=black!66,draw=green,ultra thick] (3,7+4) rectangle ++ (1,1);
 \path[fill=black!66,draw=green,ultra thick] (5,7+2) rectangle ++ (1,1);

 \path[fill=black!66,draw=green,ultra thick] (3+25,7-4) rectangle ++ (1,1);
 \path[fill=black!66,draw=green,ultra thick] (5+25,7-2) rectangle ++ (1,1);
 \path[fill=black!66,draw=green,ultra thick] (3+25,7+4) rectangle ++ (1,1);
 \path[fill=black!66,draw=green,ultra thick] (5+25,7+2) rectangle ++ (1,1);

 \path[draw=blue,ultra thick] (1,1) rectangle ++ (6,6); %
 \path[draw=blue,ultra thick] (1+25,1) rectangle ++ (6,6);

                                                      
    \draw[-latex,ultra thick,orange!95](15,19)node[right]{$R$}    to[out=170,in=75] (4.5,15.0);                                               
    \draw[-latex,ultra thick,orange!95](17,19)node[left]{}    to[out=10,in=+105] (27.5,15.0);     

    \draw[-latex,ultra thick,blue!95](1,-4.5)node[left]{$P'$}    to[out=0,in=-90] (3.5,1.0);    
    \draw[-latex,ultra thick,blue!95](32,-4.5)node[right]{$P$}    to[out=180,in=-90] (29.5,1.0);

\end{tikzpicture}
\caption[Deux motifs $P$ et $P'$ tels que ${\cal E}_n(P')\preccurlyeq_n {\cal E}_n(P)$.]{Sur la droite : l'union d'un motif $P$ noir et blanc (dans le cadre bleu) avec son motif $R$ correspondant (à l'extérieur du cadre bleu). Sur la gauche : l'union d'un autre motif $P'$ noir et blanc (tel que  ${\cal E}_n(P')\preccurlyeq_n {\cal E}_n(P)$) avec le même motif $R$. Nous indiquons en vert les positions $P$ et $P'$ qui diffèrent, et les positions dans $R$ correspondantes (symétriques par rapport à la ligne rouge).} \label{f:p-and-p'}
\vspace{-5pt}

\end{figure}

Nous avions fait la remarque que pour tout $n$ nous avons $2^{n^2}$ valeurs possibles pour ${\cal E}_n$. Par un argument de comptage nous pouvons affirmer que pour tout $n$ il existe au moins un motif $P$ de taille $n\times n$ tel que $\CK({\cal E}_n(P)) \ge n^2 \gg n$. La propriété $\CK({\cal E}_n(P)) =O(n)$ ne peut donc être vérifiée pour ce shift. D'après la Proposition~\ref{pbis:epitomes}~(a) il s'ensuit que le shift n'est pas sofique.
\end{example}

Une configuration non dégénérée du shift  $S_{\text{semi-mirror}}$ défini dans l'Exemple~\ref{ex:semi-mirror} consiste en deux demi-plans ``semi-sy\-mét\-riques'' séparés par une ligne rouge horizontale infinie.  
La propriété de semi-symétrie signifie que le demi-plan du bas peut être obtenu à partir du demi-plan du haut par la composition de la symétrie en miroir classique suivie de la recoloration de lettres noires du demi-plan du dessous en lettres blanches. Dans cette définition nous ne posons pas de contraintes sur le nombre de lettres noires qui sont converties en lettres blanches :
les deux demi-plans peuvent exactement le symétrique l'un de l'autre, ou, inversement, le demi-plan du bas ne peut être constitué que de lettres blanches. 
Il est intéressant d'étudier les sous-shifts de $S_{\text{semi-mirror}}$ où l'on ajoute des contraintes sur le nombre de positions qui diffèrent entre les deux demi-plans. Dans les deux exemples suivant nous étudions brièvement deux exemples extrêmes : un sous-shift de $S_{\text{semi-mirror}}$  où la symétrie peut ne pas être respectée par au plus une paire de lettres, et un autre sous-shift où la symétrie n'est pas respectée pour l'ensemble des lettres noires sauf éventuellement une.

\begin{examplebis}{ex:semi-mirror}[shift en miroir avec une très petite différence entre les deux demi-plans]
\label{ex:semi-mirror-bis}
Définissons un nouveau shift  $S_{\text{semi-mirror}}' $ (tel que $S_{\text{semi-mirror}}' \subset  S_{\text{semi-mirror}}$) en ajoutant à la définition de $ S_{\text{semi-mirror}}$ une nouvelle condition, imposant que la symétrie entre les deux demi-plans (au-dessus et en-dessous de la ligne rouge) n'est pas respectée par \emph{au plus une paire de lettres}, comme montré dans la Fig~\ref{f:semi-mirror-prim}. 
Il est facile de vérifier que le shift est effectif. Exactement le même argument que dans l'Exemple~\ref{ex:semi-mirror} implique que $S_{\text{semi-mirror}}' $ n'est lui aussi pas sofique.
\end{examplebis}

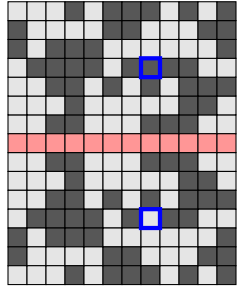
\begin{figure}[H]

  \centering
  \begin{tikzpicture}[scale=0.25,x=1cm,baseline=2.125cm]

  \pgfmathsetseed{5}
    \foreach \x in {1,...,12} \foreach \y in {1,...,7}
    {
        \pgfmathparse{mod(int(random*23),2) ? "black!10" : "black!66"}
        \edef\colour{\pgfmathresult}
        \path[fill=\colour,draw=black] (\x,7-\y) rectangle ++ (1,1);
        \path[fill=\colour,draw=black] (\x,7+\y) rectangle ++ (1,1);
    }
    
     \foreach \x in {1,...,12} 
    {
        \path[fill=red!41,draw=black] (\x,7) rectangle ++ (1,1);
    }

         \path[fill=black!10,draw=blue,ultra thick] (8,7-4) rectangle ++ (1,1);
 
         \path[fill=black!66,draw=blue,ultra thick] (8,7+4) rectangle ++ (1,1);

\end{tikzpicture}

\caption[Une configuration du shift $S_{\text{semi-mirror}}'$.]{Une configuration du shift $S_{\text{semi-mirror}}'$ (nous indiquons en bleu la paire unique des positions symétriques où le demi-plan du dessus diffère du demi-plan du dessous.} \label{f:semi-mirror-prim}
\vspace{-5pt}

\end{figure}

\begin{examplebis2}{ex:semi-mirror}[Le shift miroir avec une très grande différence entre les deux demi-plans]
\label{ex:semi-mirror-bis2}
Soit le shift $S_{\text{semi-mirror}}'' $ (de nouveau $S_{\text{semi-mirror}}'' \subset  S_{\text{semi-mirror}}$) 
avec la condition que le demi-plan en dessous de la ligne rouge horizontale contienne \emph{au plus une lettre noire}, voir Fig.~\ref{f:semi-mirror-prim2}.
Ce shift est également effectif.
Avec cette nouvelle restriction, l'argument de l'Exemple~\ref{ex:semi-mirror} ne s'applique plus.
En effet, comme nous avons au plus une lettre noire dans le demi-plan du dessous, la complexité de Kolmogorov des résumés ${\cal E}_n(P)$ définis dans l'Exemple~\ref{ex:semi-mirror} se réduit à $O(\log n)$, et nous ne pouvons pas obtenir de contradiction avec la Proposition~\ref{pbis:epitomes}~(a). 
Notre argument servant à montrer que le shift n'est pas sofique ne s'applique pas. Cela n'est pas problématique :
il n'est pas difficile de vérifier que le shift $S_{\text{semi-mirror}}'' $ est en réalité sofique. Nous étudierons cet exemple plus tard dans la Sous-partie~\ref{s:km}, lorsque nous verrons les ensembles des extensions d'un motif.
\end{examplebis2}

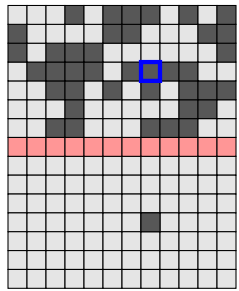
\begin{figure}[H]

  \centering
  \begin{tikzpicture}[scale=0.25,x=1cm,baseline=2.125cm]

  \pgfmathsetseed{5}
    \foreach \x in {1,...,12} \foreach \y in {1,...,7}
    {
        \pgfmathparse{mod(int(random*23),2) ? "black!10" : "black!66"}
        \edef\colour{\pgfmathresult}
        \path[fill=black!10,draw=black] (\x,7-\y) rectangle ++ (1,1);
        \path[fill=\colour,draw=black] (\x,7+\y) rectangle ++ (1,1);
    }
    
     \foreach \x in {1,...,12} 
    {
        \path[fill=red!41,draw=black] (\x,7) rectangle ++ (1,1);
    }

         \path[fill=black!66,draw=black] (8,7-4) rectangle ++ (1,1);
 
         \path[fill=black!66,draw=blue,ultra thick] (8,7+4) rectangle ++ (1,1);
                                                      
\end{tikzpicture}

\caption[Une configuration du shift $S_{\text{semi-mirror}}''$.]{Une configuration du shift $S_{\text{semi-mirror}}''$. Nous indiquons en bleu la position dans le demi-plan du haut qui est symétrique à l'unique lettre noire du demi-plan du bas.} \label{f:semi-mirror-prim2}
\vspace{-5pt}

\end{figure}

Dans le prochain exemple nous étudierons un shift avec une très petite complexité par bloc, pour lequel on ne pourra donc utiliser que la version des résumés ordonnés avec une complexité de Kolmogorov à ressources bornées pour prouver qu'il n'est pas sofique.

\begin{example}[semi-mirror shift with low block complexity] 
\label{ex:semi-mirror-sofique}
Pour cette exemple nous utilisons de nouveau le shift ${S}_{0}$ défini dans la Remarque~\ref{r:paradoxical-shift},
p.~\pageref{r:paradoxical-shift}.
Nous définissons un nouveau shift $S_{\text{semi-mirror}}^{0}$ (tel que $S_{\text{semi-mirror}}^{0} \subset S_{\text{semi-mirror}}'$) en ajoutant à la définition de $S_{\text{semi-mirror}}'$ la condition suivante : tout motif noir et blanc situé en-dessous de la ligne rouge horizontale doit être globalement admissible pour $S_{0}$.

Il est facile de vérifier que $S_{\text{semi-mirror}}^{0}$ est lui-aussi effectif. Par construction, pour tout motif $P$ de taille $n\times n$ globalement admissible pouvant apparaître en-dessous de la ligne rouge horizontale, nous avons $\CK(P)  = O(\log n)$. La propriété reste vraie pour les motifs apparaissant au-dessus de la ligne rouge horizontale, puisque effectuer la symétrie par miroir et ajouter une lettre noire ne peut augmenter la complexité de Kolmogorov d'au plus $O(\log n)$. Celle-ci reste également vraie pour les motifs globalement admissibles où apparaissent des lettres rouges, puisque de tels motifs peuvent être décrits par leurs parties au-dessus et en-dessous de la ligne rouge horizontale, et ces deux parties ont une complexité de Kolmogorov en $O(\log n)$. Par conséquent, pour \emph{tout} motif $P$ de taille $n\times n$ globalement admissible nous avons $\CK(P)  = O(\log n)$. Ainsi, il existe au plus $2^{O(\log n)} = \poly(n)$ motifs globalement admissibles, i.e., le shift a une complexité par bloc seulement polynomiale.

La technique des résumés ordonnés avec la complexité de Kolmogorov simple ne s'applique pas puisque pour tout motif $P$ globalement admissible la valeur de $\CK(P) $ est très petite, et la valeur de  $\CK({\cal E}_n(P)) $ est donc elle aussi sous-linéaire ; partant, nous ne pouvons pas obtenir de contradiction en utilisant la Proposition~\ref{pbis:epitomes}~(a).
Cependant, nous pouvons prouver que ce shift n'est pas sofique à l'aide de la complexité de Kolmogorov à ressources bornées. En effet, le shift est défini de telle manière que pour certains motifs $P$ de taille  $n\times n$ apparaissant au-dessous de la ligne rouge horizontale, nous avons $\CK^{2^{n^3}}({\cal E}_n(P)) = \Omega(n^{1.5})$. 
Nous pouvons donc répéter l'argumentation utilisée dans l'Exemple~\ref{ex:semi-mirror} avec la même famille de résumés ordonnés ${\cal E}_n$ (Ils sont calculables en temps exponentiel) mais cette fois en utilisant la complexité de Kolmogorov à ressources bornées avec une fonction de seuil $T=2^{\omega(n^2)}$.
En appliquant la Proposition~\ref{pbis:epitomes}~(b), nous pouvons conclure que le shift n'est pas sofique.

\end{example}

Dans l'exemple suivant nous étudions une instance intéressante du shift proposé par Kass et Madden dans \cite[Example 2.5]{kass-madden}. Nous reformulons la preuve que ce shift n'est pas sofique décrite dans  \cite{kass-madden} dans le langage de la complexité de Kolmogorov, en utilisant la technique des résumés ordonnés. Dans ce cas, la définition des résumés est moins directe que dans les exemples précédents.

\begin{example}[le shift sans carré rouge et noir caché]\label{ex:km}
Soit $\Sigma$ l'alphabet de trois lettres (par exemple, \emph{noir}, \emph{blanc}, et \emph{rouge}). Nous définissons les motifs interdits comme l'ensemble des motifs carrés (de taille quelconque) où la ligne du haut est constituée de lettres rouges, et celle du bas de lettres noires (\emph{les carrés rouge et noir cachés}), comme montrés dans la Fig.~\ref{f:012-forbidden}.

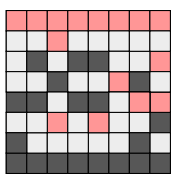
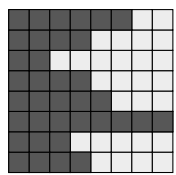
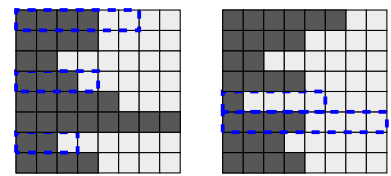
\begin{figure}[H]
\centering
\begin{subfigure}{0.2\linewidth}
  \centering
  \begin{tikzpicture}[scale=0.27,x=1cm,baseline=2.125cm]
  
  \pgfmathsetseed{1}
    \foreach \x in {0,...,7} \foreach \y in {0,...,7}
    {
        \pgfmathparse{mod(int(random*23),2) ? "red!41" : "black!66"}
        \edef\colour{\pgfmathresult}
        \path[fill=\colour,draw=black] (\x,\y) rectangle ++ (1,1);
    }
    
     \foreach \x in {0,...,7} 
    {
        \path[fill=red!41,draw=black] (\x,7) rectangle ++ (1,1);
    }
     \foreach \x in {0,...,7} 
    {
        \path[fill=black!66,draw=black] (\x,0) rectangle ++ (1,1);
    }

        \path[fill=black!7,draw=black] (3,4) rectangle ++ (1,1);
        \path[fill=black!7,draw=black] (3,2) rectangle ++ (1,1);
        \path[fill=black!7,draw=black] (6,5) rectangle ++ (1,1);
        \path[fill=black!7,draw=black] (1,1) rectangle ++ (1,1);
        \path[fill=black!7,draw=black] (5,3) rectangle ++ (1,1);
        \path[fill=black!7,draw=black] (6,3) rectangle ++ (1,1);
        \path[fill=black!7,draw=black] (5,2) rectangle ++ (1,1);
        \path[fill=black!7,draw=black] (1,6) rectangle ++ (1,1);
        \path[fill=black!7,draw=black] (0,5) rectangle ++ (1,1);
        \path[fill=black!7,draw=black] (2,3) rectangle ++ (1,1);
        \path[fill=black!7,draw=black] (2,1) rectangle ++ (1,1);
        \path[fill=black!7,draw=black] (4,1) rectangle ++ (1,1);
        \path[fill=black!7,draw=black] (5,1) rectangle ++ (1,1);                        
        \path[fill=black!7,draw=black] (0,2) rectangle ++ (1,1);
        \path[fill=black!7,draw=black] (3,6) rectangle ++ (1,1); 
        \path[fill=black!7,draw=black] (4,6) rectangle ++ (1,1); 
        \path[fill=black!7,draw=black] (5,6) rectangle ++ (1,1);       
        \path[fill=black!7,draw=black] (6,6) rectangle ++ (1,1);       
        \path[fill=black!7,draw=black] (7,6) rectangle ++ (1,1);       
        \path[fill=black!7,draw=black] (0,6) rectangle ++ (1,1);       
        \path[fill=black!7,draw=black] (7,1) rectangle ++ (1,1);  
        \path[fill=black!7,draw=black] (6,2) rectangle ++ (1,1);  
        \path[fill=black!7,draw=black] (7,4) rectangle ++ (1,1);  
        \path[fill=black!7,draw=black] (2,5) rectangle ++ (1,1); 
        \path[fill=black!7,draw=black] (0,4) rectangle ++ (1,1); 
        \path[fill=black!7,draw=black] (1,4) rectangle ++ (1,1); 
        \path[fill=black!7,draw=black] (3,1) rectangle ++ (1,1); 
        \path[fill=black!7,draw=black] (5,5) rectangle ++ (1,1); 
        \path[fill=black!7,draw=black] (1,2) rectangle ++ (1,1); 
        \path[fill=black!7,draw=black] (4,4) rectangle ++ (1,1); 
        \path[fill=black!7,draw=black] (5,3) rectangle ++ (1,1); 

        \path[fill=black!66,draw=black] (0,1) rectangle ++ (1,1); 
        \path[fill=black!66,draw=black] (1,5) rectangle ++ (1,1); 
        \path[fill=black!66,draw=black] (6,3) rectangle ++ (1,1); 
        \path[fill=black!66,draw=black] (6,1) rectangle ++ (1,1);
        \path[fill=red!41,draw=black] (2,6) rectangle ++ (1,1);  
        \path[fill=red!41,draw=black] (6,3) rectangle ++ (1,1); 
        \path[fill=red!41,draw=black] (4,2) rectangle ++ (1,1);
                                      
\end{tikzpicture}
  \caption{Un motif interdit~: un carré avec une ligne du haut rouge et une ligne du bas noire.}
  \label{f:012-forbidden}
\end{subfigure}
\quad
\begin{subfigure}{0.35\linewidth}
\begin{center}
  \begin{tikzpicture}[scale=0.27,x=1cm,baseline=2.125cm]
 
     
    \foreach \x in {0,...,7} \foreach \y in {0,...,7}
    {
          \path[fill=black!7,draw=black] (\x,\y) rectangle ++ (1,1);
    }

 \edef\y{0}
 \edef\xmax{3}
 \foreach \x in {0,...,\xmax}
 \path[fill=black!66,draw=black] (\x,\y) rectangle ++ (1,1);

 \edef\y{1}
 \edef\xmax{2}
 \foreach \x in {0,...,\xmax}
 \path[fill=black!66,draw=black] (\x,\y) rectangle ++ (1,1);

 \edef\y{2}
 \edef\xmax{7}
 \foreach \x in {0,...,\xmax}
 \path[fill=black!66,draw=black] (\x,\y) rectangle ++ (1,1);

 \edef\y{3}
 \edef\xmax{4}
 \foreach \x in {0,...,\xmax}
 \path[fill=black!66,draw=black] (\x,\y) rectangle ++ (1,1);

 \edef\y{4}
 \edef\xmax{3}
 \foreach \x in {0,...,\xmax}
 \path[fill=black!66,draw=black] (\x,\y) rectangle ++ (1,1);

 \edef\y{5}
 \edef\xmax{1}
 \foreach \x in {0,...,\xmax}
 \path[fill=black!66,draw=black] (\x,\y) rectangle ++ (1,1);

 \edef\y{6}
 \edef\xmax{3}
 \foreach \x in {0,...,\xmax}
 \path[fill=black!66,draw=black] (\x,\y) rectangle ++ (1,1);

 \edef\y{7}
 \edef\xmax{5}
 \foreach \x in {0,...,\xmax}
 \path[fill=black!66,draw=black] (\x,\y) rectangle ++ (1,1);


\end{tikzpicture}
\end{center}
\vspace{-10pt}
	\caption{Un motif pour lequel le résumé ${\cal E}_n$ est défini : chaque ligne commence avec des lettres noires à gauche suivies par des lettres blanches à droite.}
 \label{f:012-standard}
\end{subfigure}
\quad
\begin{subfigure}{0.35\linewidth}
\vspace{-25pt}
\begin{center}
  \begin{tikzpicture}[scale=0.27,x=1cm,baseline=2.125cm]
 
     
    \foreach \x in {0,...,7} \foreach \y in {0,...,7}
    {
          \path[fill=black!7,draw=black] (\x,\y) rectangle ++ (1,1);
    }

 \edef\y{0}
 \edef\xmax{3}
 \foreach \x in {0,...,\xmax}
 \path[fill=black!66,draw=black] (\x,\y) rectangle ++ (1,1);

 \edef\y{1}
 \edef\xmax{1}
 \foreach \x in {0,...,\xmax}
 \path[fill=black!66,draw=black] (\x,\y) rectangle ++ (1,1);
  \path[draw=blue,line width=0.5mm,dashed] (0,\y) rectangle ++ (3,1);

 \edef\y{2}
 \edef\xmax{7}
 \foreach \x in {0,...,\xmax}
 \path[fill=black!66,draw=black] (\x,\y) rectangle ++ (1,1);

 \edef\y{3}
 \edef\xmax{4}
 \foreach \x in {0,...,\xmax}
 \path[fill=black!66,draw=black] (\x,\y) rectangle ++ (1,1);

 \edef\y{4}
 \edef\xmax{2}
 \foreach \x in {0,...,\xmax}
 \path[fill=black!66,draw=black] (\x,\y) rectangle ++ (1,1);
  \path[draw=blue,line width=0.5mm,dashed] (0,\y) rectangle ++ (4,1);

 \edef\y{5}
 \edef\xmax{1}
 \foreach \x in {0,...,\xmax}
 \path[fill=black!66,draw=black] (\x,\y) rectangle ++ (1,1);

 \edef\y{6}
 \edef\xmax{3}
 \foreach \x in {0,...,\xmax}
 \path[fill=black!66,draw=black] (\x,\y) rectangle ++ (1,1);

 \edef\y{7}
 \edef\xmax{3}
 \foreach \x in {0,...,\xmax}
 \path[fill=black!66,draw=black] (\x,\y) rectangle ++ (1,1);
  \path[draw=blue,line width=0.5mm,dashed] (0,\y) rectangle ++ (6,1);


 \end{tikzpicture}
 \quad \begin{tikzpicture}[scale=0.27,x=1cm,baseline=2.125cm]
 
     
    \foreach \x in {0,...,7} \foreach \y in {0,...,7}
    {
          \path[fill=black!7,draw=black] (\x,\y) rectangle ++ (1,1);
    }

 \edef\y{0}
 \edef\xmax{3}
 \foreach \x in {0,...,\xmax}
 \path[fill=black!66,draw=black] (\x,\y) rectangle ++ (1,1);

 \edef\y{1}
 \edef\xmax{2}
 \foreach \x in {0,...,\xmax}
 \path[fill=black!66,draw=black] (\x,\y) rectangle ++ (1,1);

 \edef\y{2}
 \edef\xmax{3}
 \foreach \x in {0,...,\xmax}
 \path[fill=black!66,draw=black] (\x,\y) rectangle ++ (1,1);
  \path[draw=blue,line width=0.5mm,dashed] (0,\y) rectangle ++ (8,1);

 \edef\y{3}
 \edef\xmax{0}
 \foreach \x in {0,...,\xmax}
 \path[fill=black!66,draw=black] (\x,\y) rectangle ++ (1,1);
  \path[draw=blue,line width=0.5mm,dashed] (0,\y) rectangle ++ (5,1);

 \edef\y{4}
 \edef\xmax{3}
 \foreach \x in {0,...,\xmax}
 \path[fill=black!66,draw=black] (\x,\y) rectangle ++ (1,1);

 \edef\y{5}
 \edef\xmax{1}
 \foreach \x in {0,...,\xmax}
 \path[fill=black!66,draw=black] (\x,\y) rectangle ++ (1,1);

 \edef\y{6}
 \edef\xmax{3}
 \foreach \x in {0,...,\xmax}
 \path[fill=black!66,draw=black] (\x,\y) rectangle ++ (1,1);

 \edef\y{7}
 \edef\xmax{5}
 \foreach \x in {0,...,\xmax}
 \path[fill=black!66,draw=black] (\x,\y) rectangle ++ (1,1);

\end{tikzpicture}
\end{center}
\vspace{-10pt}
\caption{Une paire de motifs incomparables.} \label{f:012-patterns}
\end{subfigure}
\caption{Un motif interdit, un motif pour lequel ${\cal E}_n$ est défini, et deux motifs incomparables.}
\end{figure}

\begin{proposition}[\cite{kass-madden}]\label{prop-kass-maden}
Le shift sur $\mathbb{Z}^2$ défini par l'ensemble des motifs interdits spécifiés ci-dessus n'est pas sofique.
\end{proposition}
Dans \cite{kass-madden} cette proposition est prouvée en utilisant la technique des \emph{chaînes d'unions croissantes d'ensembles d'extensions de motifs}. Dans ce qui suit nous présentons essentiellement le même argument, mais transposé dans la technique des résumés ordonnés.


\begin{proof}[Preuve de la Proposition~\ref{prop-kass-maden} :] Nous définissons pour ce shift une famille de résumés ordonnés de la manière suivante. 
Tout d'abord, nous définissons une classe de \emph{motifs simples}: les motifs simples sont l'ensemble des carrés tels que : (i)~ils sont constitués uniquement de lettres noires et blanches (sans lettre rouge), et (ii)~chaque ligne commence par un certain nombre de lettres noires successives, suivies par une séquence de lettres blanches, comme montré dans la Fig.~\ref{f:012-standard}.
Tout motif simple de taille $n\times n$ peut être spécifié par son \emph{profil} --- un n-uplet d'entiers $(k_1,\ldots, k_n)$, où $k_i$ est le nombre de lettres noires présentes dans la $i$-ème ligne du motif. (Ainsi, un motif simple ayant pour profil $(k_1,\ldots, k_n)$ est un carré de taille $n\times n$ où chaque $i$-ème ligne commence par $k_i$ lettres noires suivies de $(n-k_i)$ lettres blanches.)

Soit ${\cal E}_n$ la famille de fonctions qui associent à chaque motif simple son profil, et ne sont pas définies sur les autres motifs. Par exemple, pour le motif $P$ de la Fig.~\ref{f:012-standard} nous avons ${\cal E}_8(P)=(4,3,8,5,4,2,4,6)$.

Nous définissons la famille d'ordre naturel $\preccurlyeq_n$ sur les profils des motifs simples de taille $n\times n$ de la manière suivante : le profil de $P_1$ n'est \emph{pas plus grand} que le profil de $P_2$ si le premier profil n'est pas, coordonnée par coordonnée, plus grand que le second. Par exemple, les profils des deux motifs montrés dans la Fig.~\ref{f:012-patterns} ne sont pas plus grands que le profil du motif de la Fig.~\ref{f:012-standard} (et incomparables entre eux).

Les familles de fonctions et d'ordres ${\cal E}_n$ et $\preccurlyeq_n$ sont clairement calculables, même en temps polynomial. Il reste à montrer que ${\cal E}_n$ et $\preccurlyeq_n$  satisfont la Définition~\ref{d:ordered-epitome} des résumés ordonnés :
\begin{lemma}\label{l:kass-madden}
La famille $({\cal E}_n, \preccurlyeq_n)$ définie ci-dessus forme une famille de résumés ordonnés calculable en temps exponentiel pour le shift considéré.
\end{lemma}
\noindent
Ce lemme est prouvé implicitement dans \cite{kass-madden}. Nous reprenons ici les grandes lignes de la preuve.

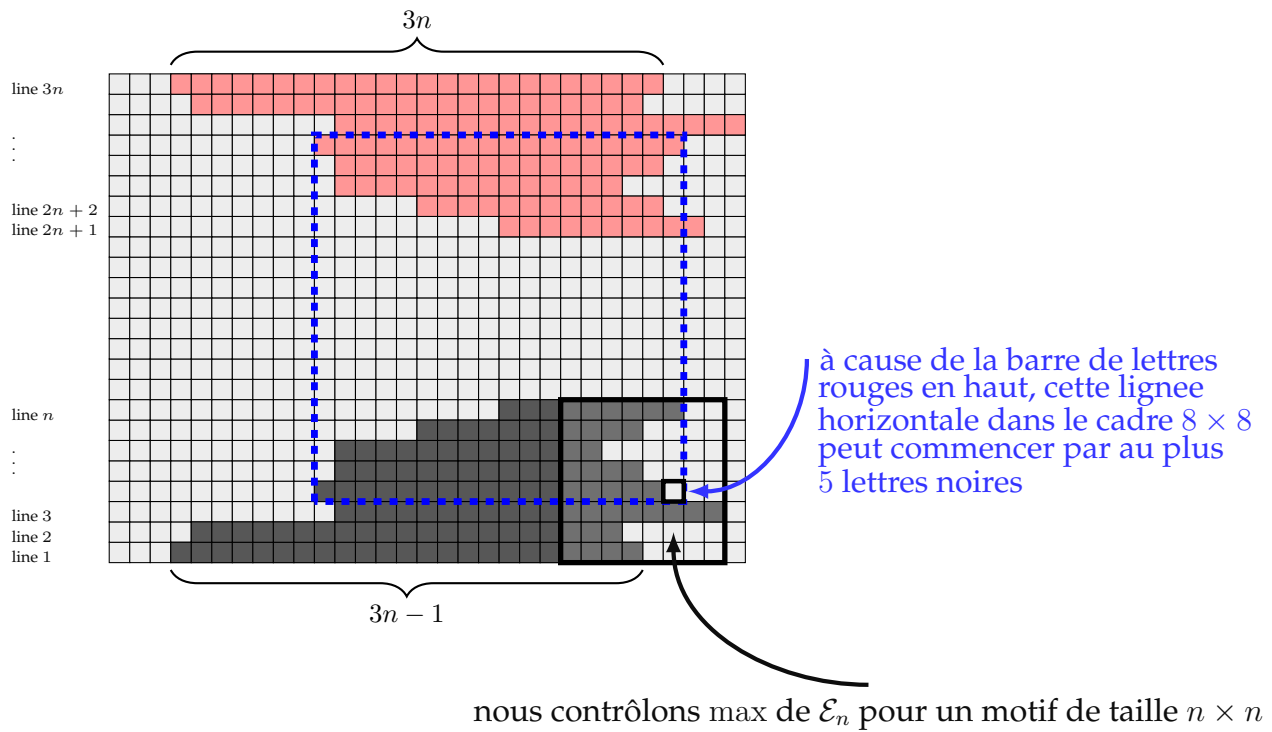
\begin{figure}
\begin{center}
  \begin{tikzpicture}[scale=0.27,x=1cm,baseline=2.125cm]

    \foreach \x in {-22,...,8} \foreach \y in {8,...,15}
    {
          \path[fill=black!7,draw=black] (\x,\y) rectangle ++ (1,1);
    }

    \foreach \x in {-22,...,8} \foreach \y in {16,...,23}
    {
          \path[fill=black!7,draw=black] (\x,\y) rectangle ++ (1,1);
    }

    \foreach \x in {-22,...,-1} \foreach \y in {0,...,7}
    {
          \path[fill=black!7,draw=black] (\x,\y) rectangle ++ (1,1);
    }

     
    \foreach \x in {0,...,8} \foreach \y in {0,...,7}
    {
          \path[fill=black!7,draw=black] (\x,\y) rectangle ++ (1,1);
    }

 \edef\y{0}
 \edef\xmax{3}
 \foreach \x in {0,...,\xmax}
 \path[fill=black!56,draw=black] (\x,\y) rectangle ++ (1,1);

 \edef\xmin{-19}  
 \foreach \x in {\xmin,...,-1} 
 \path[fill=black!66,draw=black] (\x,\y) rectangle ++ (1,1);
 \edef\xmaxtop{4}  
 \edef\ytop{23}  
 \foreach \x in {\xmin,...,\xmaxtop} 
 \path[fill=red!41,draw=black] (\x,\ytop) rectangle ++ (1,1);

 \edef\y{1}
 \edef\xmax{2}
 \foreach \x in {0,...,\xmax}
 \path[fill=black!56,draw=black] (\x,\y) rectangle ++ (1,1);

 \edef\xmin{-18}  
 \foreach \x in {\xmin,...,-1} 
 \path[fill=black!66,draw=black] (\x,\y) rectangle ++ (1,1);
 \edef\xmaxtop{3}  
 \edef\ytop{22}  
 \foreach \x in {\xmin,...,\xmaxtop} 
 \path[fill=red!41,draw=black] (\x,\ytop) rectangle ++ (1,1);

 \edef\y{2}
 \edef\xmax{7}
 \foreach \x in {0,...,\xmax}
 \path[fill=black!56,draw=black] (\x,\y) rectangle ++ (1,1);

\edef\xmin{-11}  
 \foreach \x in {\xmin,...,-1} 
 \path[fill=black!66,draw=black] (\x,\y) rectangle ++ (1,1);
 \edef\xmaxtop{8}  
 \edef\ytop{21}  
 \foreach \x in {\xmin,...,\xmaxtop} 
 \path[fill=red!41,draw=black] (\x,\ytop) rectangle ++ (1,1);

 \edef\y{3}
 \edef\xmax{4}
 \foreach \x in {0,...,\xmax}
 \path[fill=black!56,draw=black] (\x,\y) rectangle ++ (1,1);

\edef\xmin{-12}  
 \foreach \x in {\xmin,...,-1} 
 \path[fill=black!66,draw=black] (\x,\y) rectangle ++ (1,1);
 \edef\xmaxtop{5}  
 \edef\ytop{20}  
 \foreach \x in {\xmin,...,\xmaxtop} 
 \path[fill=red!41,draw=black] (\x,\ytop) rectangle ++ (1,1);

 \edef\y{4}
 \edef\xmax{3}
 \foreach \x in {0,...,\xmax}
 \path[fill=black!56,draw=black] (\x,\y) rectangle ++ (1,1);

\edef\xmin{-11}  
 \foreach \x in {\xmin,...,-1} 
 \path[fill=black!66,draw=black] (\x,\y) rectangle ++ (1,1);
 \edef\xmaxtop{4}  
 \edef\ytop{19}  
 \foreach \x in {\xmin,...,\xmaxtop} 
 \path[fill=red!41,draw=black] (\x,\ytop) rectangle ++ (1,1);

 \edef\y{5}
 \edef\xmax{1}
 \foreach \x in {0,...,\xmax}
 \path[fill=black!56,draw=black] (\x,\y) rectangle ++ (1,1);

\edef\xmin{-11}  
 \foreach \x in {\xmin,...,-1} 
 \path[fill=black!66,draw=black] (\x,\y) rectangle ++ (1,1);
 \edef\xmaxtop{2}  
 \edef\ytop{18}  
 \foreach \x in {\xmin,...,\xmaxtop} 
 \path[fill=red!41,draw=black] (\x,\ytop) rectangle ++ (1,1);

 \edef\y{6}
 \edef\xmax{3}
 \foreach \x in {0,...,\xmax}
 \path[fill=black!56,draw=black] (\x,\y) rectangle ++ (1,1);

\edef\xmin{-7}  
 \foreach \x in {\xmin,...,-1} 
 \path[fill=black!66,draw=black] (\x,\y) rectangle ++ (1,1);
 \edef\xmaxtop{4}  
 \edef\ytop{17}  
 \foreach \x in {\xmin,...,\xmaxtop} 
 \path[fill=red!41,draw=black] (\x,\ytop) rectangle ++ (1,1);

 \edef\y{7}
 \edef\xmax{5}
 \foreach \x in {0,...,\xmax}
 \path[fill=black!56,draw=black] (\x,\y) rectangle ++ (1,1);

\edef\xmin{-3}  
 \foreach \x in {\xmin,...,-1} 
 \path[fill=black!66,draw=black] (\x,\y) rectangle ++ (1,1);
 \edef\xmaxtop{6}  
 \edef\ytop{16}  
 \foreach \x in {\xmin,...,\xmaxtop} 
 \path[fill=red!41,draw=black] (\x,\ytop) rectangle ++ (1,1);


 \path[draw=black,line width=0.7mm] (0,0) rectangle ++ (8,8);



 \path[draw=blue,line width=0.9mm,dashed] (-12,3) rectangle ++ (18,18);
 \path[fill=black!7,draw=black,line width=0.6mm] (5,3) rectangle ++ (1,1);

 \draw[-latex,ultra thick,blue!80](12,10)node[right]{à cause de la barre de lettres }
       to[out=-90,in=0] (6.2,3.5);
\node[right,blue!80] at (12,8.5){rouges en haut, cette lignee  };
\node[right,blue!80] at (12,7.0){ horizontale dans le cadre $8\times 8$  };
\node[right,blue!80] at (12,5.5){peut commencer par au plus};
\node[right,blue!80] at (12,4.0){ $5$ lettres noires};

 \draw[-latex,ultra thick,black!95](15,-6)node[below]{nous contrôlons $\max$ de ${\cal E}_n$ pour un motif de taille $n\times n$}
       to[out=180,in=-90] (5.5,1.5);

\node[right,black!95] at (-27.3,0.3){\tiny line $1$};
\node[right,black!95] at (-27.3,1.3){\tiny line $2$};
\node[right,black!95] at (-27.3,2.3){\tiny line $3$};

\node[right,black!95] at (-27.3,5.3){\tiny  $\vdots$};

\node[right,black!95] at (-27.3,7.3){\tiny line $n$};

\node[right,black!95] at (-27.3,16.3){\tiny line $2n+1$};
\node[right,black!95] at (-27.3,17.3){\tiny line $2n+2$};

\node[right,black!95] at (-27.3,20.7){\tiny  $\vdots$};

\node[right,black!95] at (-27.3,23.3){\tiny line $3n$};

\draw [decorate,decoration={brace,amplitude=10pt},xshift=0pt,yshift=4pt,thick] (-19,24.3) -- (5,24.3) node [black,midway,yshift=17.0] {\footnotesize $3n$};

\draw [decorate,decoration={brace,amplitude=10pt},xshift=0pt,yshift=4pt,thick] (4,-0.5) -- (-19,-0.5) node [black,midway,yshift=-17.0] {\footnotesize $3n-1$};

\end{tikzpicture}
\end{center}
\caption[Un motif $P$ et son motif correspondant $R$.]{Un motif $P$ de taille $n\times n$ avec un voisinage garantissant des résumés ${\cal E}_n$ de valeurs inférieures au maximum désiré.} \label{f:012-enforcing-standard}
\vspace{-5pt}
\end{figure}

Pour tout \emph{motif simple} $P$ de taille $n\times n$ nous devons construire une configuration $R$ sur le complement de $B_n$, telle que 
 \begin{itemize}
 \item[(i)] $P$ et $R$ sont compatibles,
 \item[(ii)] pour tout autre motif simple $P'$ compatible avec $R$ nous avons ${\cal E}_n(P')\preccurlyeq_n {\cal E}_n(P)$.
 \end{itemize}
Pour construire la configuration $R$ requise, nous suivons la construction décrite dans \cite{kass-madden}.

Le motif $R$ sera constitué d'un nombre fini de lettres noires et rouges (les autres lettres seront blanches).

\textbf{Lettres noires de $R$.}
Pour construire $R$, nous étendons chaque bande de lettres noires de $P$ vers la gauche, de telle manière que pour la première ligne nous avons une séquence continue de $(3n-1)$ lettres noires (en comptant les lettres noires appartenant à $P$), dans la seconde ligne une séquence continue de $(3n-3)$  lettres noires, dans la troisième ligne une séquence continue de $(3n-5)$  lettres noires, etc. Dans la  $n$-ème ligne nous obtenons une séquence continue de $(n+1)$ lettres noires, voir la Fig.~\ref{f:012-enforcing-standard}.

\textbf{Lettres rouges de $R$.} 
De manière similaire, nous ajoutons des bandes de lettres rouges à $R$ : une séquence continue de $3n$ lettres rouges à la ligne $3n$, une séquence continue de $(3n-2)$ lettres rouges à la ligne $3n-1$,
 \ldots, une séquence continue de $(n+2)$ lettres rouges à la ligne $(2n+1)$. Nous plaçons ces bandes de lettres rouges de telle manière que pour tout $i=1,\ldots,n$ la lettre rouge la plus à gauche de la ligne $(3n-i+1)$ est verticalement alignée avec la lettre noire la plus à gauche de la ligne $i$, comme montré dans la Fig.~\ref{f:012-enforcing-standard}. 

Toutes les autres lettres autour de $B_n$ sont blanches.

\smallskip

\emph{Fait~1.}  Le motif $R$ construit est compatible avec $P$. 

\smallskip

\emph{Preuve du Fait~1 :}
Ce fait est facile à vérifier : nous avons choisi les longueurs des bandes de lettres noires et rouges de manière à ce qu'ils ne puissent pas former de motifs interdits (voir la Fig.~\ref{f:012-forbidden}), peu importe le placement horizontal de ces bandes. En effet, d'une part les lettres noires de la $i$-ème ligne ne peuvent interférer avec les bandes rouges des lignes $3n, 3n-1, \ldots, 3n-i$, puisque \emph{cette bande noire} est trop courte pour former un motif interdit avec une de ces bandes rouges ; d'autre part, les lettres noires de la $i$-ème ligne ne peuvent interférer avec les bandes rouges des lignes  $3n-i-1, 3n-i-2, \ldots, 2n+1$, puique \emph{ces bandes rouges} sont trop petites. \qed

\smallskip

\emph{Fait~2.} Le motif $R$ construit n'est compatible qu'avec les motifs simples $P'$ tels que ${\cal E}_n(P')\preccurlyeq_n {\cal E}_n(P)$.

\smallskip

\emph{Preuve du Fait~2 : }
Si $R$ est compatible avec un motif  $P'$ de taille $n\times n$, le profil de $P'$ n'est pas déterminé de manière unique. 
En fait, $R$ peut être compatible avec les motifs simples $P'$ dont les profils sont \emph{strictement plus petits} que le profil de $P$ (pour chaque ligne de $P'$ le nombre de lettres noires ne doit pas être plus grand que le nombre de lettres noires pour la ligne correspondante de $P$), voir la Fig.~\ref{f:012-enforcing-ex1}. Par ailleurs, si au moins une ligne de $P'$ contient plus de lettres noires que la même ligne de $P$, alors $P'$ et $R$ sont incompatibles, i.e., l'union de $P'$ et $R$ contient un motif interdit, comme montré dans la Fig.~\ref{f:012-enforcing-ex2}.
\qed

\smallskip

Le lemme découle de Fait~1 et Fait~2. Pour une preuve plus détaillée le lecteur peut se référer à~\cite{kass-madden}.

\begin{remark}
Dans la construction décrite ci-dessus, le motif $R$ ne détermine pas de manière unique les résumés des motifs $P'$ compatibles avec $R$ (ces résumés peuvent être différents, même s'ils doivent tous n'être \emph{pas plus grands} que le résumé du motif initial $P$). C'est pourquoi nous ne pouvons pas utiliser la  Proposition~\ref{p:epitomes}, et que nous devons utiliser la Proposition~\ref{pbis:epitomes}.
\end{remark}


%
\begin{figure}[H]
\begin{center}
  \begin{tikzpicture}[scale=0.24,x=1cm,baseline=2.125cm]

    \foreach \x in {-22,...,8} \foreach \y in {8,...,15}
    {
          \path[fill=black!7,draw=black] (\x,\y) rectangle ++ (1,1);
    }

    \foreach \x in {-22,...,8} \foreach \y in {16,...,23}
    {
          \path[fill=black!7,draw=black] (\x,\y) rectangle ++ (1,1);
    }

    \foreach \x in {-22,...,-1} \foreach \y in {0,...,7}
    {
          \path[fill=black!7,draw=black] (\x,\y) rectangle ++ (1,1);
    }

     
    \foreach \x in {0,...,8} \foreach \y in {0,...,7}
    {
          \path[fill=black!7,draw=black] (\x,\y) rectangle ++ (1,1);
    }

 \edef\y{0}
 \edef\xmax{3}
 \foreach \x in {0,...,\xmax}
 \path[fill=black!56,draw=black] (\x,\y) rectangle ++ (1,1);

 \edef\xmin{-19}  
 \foreach \x in {\xmin,...,-1} 
 \path[fill=black!66,draw=black] (\x,\y) rectangle ++ (1,1);
 \edef\xmaxtop{4}  
 \edef\ytop{23}  
 \foreach \x in {\xmin,...,\xmaxtop} 
 \path[fill=red!41,draw=black] (\x,\ytop) rectangle ++ (1,1);

 \edef\y{1}
 \edef\xmax{1}
 \foreach \x in {0,...,\xmax}
 \path[fill=black!56,draw=black] (\x,\y) rectangle ++ (1,1);
  \path[draw=black,line width=0.5mm,dashed] (0,\y) rectangle ++ (3,1);

 \edef\xmin{-18}  
 \foreach \x in {\xmin,...,-1} 
 \path[fill=black!66,draw=black] (\x,\y) rectangle ++ (1,1);
 \edef\xmaxtop{3}  
 \edef\ytop{22}  
 \foreach \x in {\xmin,...,\xmaxtop} 
 \path[fill=red!41,draw=black] (\x,\ytop) rectangle ++ (1,1);

 \edef\y{2}
 \edef\xmax{2}
 \foreach \x in {0,...,\xmax}
 \path[fill=black!56,draw=black] (\x,\y) rectangle ++ (1,1);
  \path[draw=black,line width=0.5mm,dashed] (0,\y) rectangle ++ (8,1);

\edef\xmin{-11}  
 \foreach \x in {\xmin,...,-1} 
 \path[fill=black!66,draw=black] (\x,\y) rectangle ++ (1,1);
 \edef\xmaxtop{8}  
 \edef\ytop{21}  
 \foreach \x in {\xmin,...,\xmaxtop} 
 \path[fill=red!41,draw=black] (\x,\ytop) rectangle ++ (1,1);

 \edef\y{3}
 \edef\xmax{4}
 \foreach \x in {0,...,\xmax}
 \path[fill=black!56,draw=black] (\x,\y) rectangle ++ (1,1);

\edef\xmin{-12}  
 \foreach \x in {\xmin,...,-1} 
 \path[fill=black!66,draw=black] (\x,\y) rectangle ++ (1,1);
 \edef\xmaxtop{5}  
 \edef\ytop{20}  
 \foreach \x in {\xmin,...,\xmaxtop} 
 \path[fill=red!41,draw=black] (\x,\ytop) rectangle ++ (1,1);

 \edef\y{4}
 \edef\xmax{1}
 \foreach \x in {0,...,\xmax}
 \path[fill=black!56,draw=black] (\x,\y) rectangle ++ (1,1);
  \path[draw=black,line width=0.5mm,dashed] (0,\y) rectangle ++ (3,1);

\edef\xmin{-11}  
 \foreach \x in {\xmin,...,-1} 
 \path[fill=black!66,draw=black] (\x,\y) rectangle ++ (1,1);
 \edef\xmaxtop{4}  
 \edef\ytop{19}  
 \foreach \x in {\xmin,...,\xmaxtop} 
 \path[fill=red!41,draw=black] (\x,\ytop) rectangle ++ (1,1);

 \edef\y{5}
 \edef\xmax{1}
 \foreach \x in {0,...,\xmax}
 \path[fill=black!56,draw=black] (\x,\y) rectangle ++ (1,1);

\edef\xmin{-11}  
 \foreach \x in {\xmin,...,-1} 
 \path[fill=black!66,draw=black] (\x,\y) rectangle ++ (1,1);
 \edef\xmaxtop{2}  
 \edef\ytop{18}  
 \foreach \x in {\xmin,...,\xmaxtop} 
 \path[fill=red!41,draw=black] (\x,\ytop) rectangle ++ (1,1);

 \edef\y{6}
 \edef\xmax{3}
 \foreach \x in {0,...,\xmax}
 \path[fill=black!56,draw=black] (\x,\y) rectangle ++ (1,1);

\edef\xmin{-7}  
 \foreach \x in {\xmin,...,-1} 
 \path[fill=black!66,draw=black] (\x,\y) rectangle ++ (1,1);
 \edef\xmaxtop{4}  
 \edef\ytop{17}  
 \foreach \x in {\xmin,...,\xmaxtop} 
 \path[fill=red!41,draw=black] (\x,\ytop) rectangle ++ (1,1);

 \edef\y{7}
 \edef\xmax{3}
 \foreach \x in {0,...,\xmax}
 \path[fill=black!56,draw=black] (\x,\y) rectangle ++ (1,1);
  \path[draw=black,line width=0.5mm,dashed] (0,\y) rectangle ++ (6,1);

\edef\xmin{-3}  
 \foreach \x in {\xmin,...,-1} 
 \path[fill=black!66,draw=black] (\x,\y) rectangle ++ (1,1);
 \edef\xmaxtop{6}  
 \edef\ytop{16}  
 \foreach \x in {\xmin,...,\xmaxtop} 
 \path[fill=red!41,draw=black] (\x,\ytop) rectangle ++ (1,1);


 \path[draw=black,line width=0.7mm] (0,0) rectangle ++ (8,8);



 \draw[-latex,ultra thick,black!95](15,-6)node[below]{ce motif $P'$ de taille  $n\times n$ est compatible avec le voisinage}
       to[out=180,in=-90] (5.5,1.5);

\end{tikzpicture}
\end{center}
\caption[Un motif $P'$ compatible avec $R$.]{Un motif $P'$ avec ${\cal E}_n(P')\preccurlyeq_n {\cal E}_n(P)$ est compatible avec le voisinage.} \label{f:012-enforcing-ex1}
\vspace{10pt}
\end{figure}

\bigskip

\begin{figure}[H]
\vspace{-5pt}
\begin{center}
  \begin{tikzpicture}[scale=0.24,x=1cm,baseline=2.125cm]

    \foreach \x in {-22,...,8} \foreach \y in {8,...,15}
    {
          \path[fill=black!7,draw=black] (\x,\y) rectangle ++ (1,1);
    }

    \foreach \x in {-22,...,8} \foreach \y in {16,...,23}
    {
          \path[fill=black!7,draw=black] (\x,\y) rectangle ++ (1,1);
    }

    \foreach \x in {-22,...,-1} \foreach \y in {0,...,7}
    {
          \path[fill=black!7,draw=black] (\x,\y) rectangle ++ (1,1);
    }

     
    \foreach \x in {0,...,8} \foreach \y in {0,...,7}
    {
          \path[fill=black!7,draw=black] (\x,\y) rectangle ++ (1,1);
    }

 \edef\y{0}
 \edef\xmax{3}
 \foreach \x in {0,...,\xmax}
 \path[fill=black!56,draw=black] (\x,\y) rectangle ++ (1,1);

 \edef\xmin{-19}  
 \foreach \x in {\xmin,...,-1} 
 \path[fill=black!66,draw=black] (\x,\y) rectangle ++ (1,1);
 \edef\xmaxtop{4}  
 \edef\ytop{23}  
 \foreach \x in {\xmin,...,\xmaxtop} 
 \path[fill=red!41,draw=black] (\x,\ytop) rectangle ++ (1,1);

 \edef\y{1}
 \edef\xmax{2}
 \foreach \x in {0,...,\xmax}
 \path[fill=black!56,draw=black] (\x,\y) rectangle ++ (1,1);

 \edef\xmin{-18}  
 \foreach \x in {\xmin,...,-1} 
 \path[fill=black!66,draw=black] (\x,\y) rectangle ++ (1,1);
 \edef\xmaxtop{3}  
 \edef\ytop{22}  
 \foreach \x in {\xmin,...,\xmaxtop} 
 \path[fill=red!41,draw=black] (\x,\ytop) rectangle ++ (1,1);

 \edef\y{2}
 \edef\xmax{7}
 \foreach \x in {0,...,\xmax}
 \path[fill=black!56,draw=black] (\x,\y) rectangle ++ (1,1);

\edef\xmin{-11}  
 \foreach \x in {\xmin,...,-1} 
 \path[fill=black!66,draw=black] (\x,\y) rectangle ++ (1,1);
 \edef\xmaxtop{8}  
 \edef\ytop{21}  
 \foreach \x in {\xmin,...,\xmaxtop} 
 \path[fill=red!41,draw=black] (\x,\ytop) rectangle ++ (1,1);

 \edef\y{3}
 \edef\xmax{4}
 \foreach \x in {0,...,\xmax}
 \path[fill=black!56,draw=black] (\x,\y) rectangle ++ (1,1);

\edef\xmin{-12}  
 \foreach \x in {\xmin,...,-1} 
 \path[fill=black!66,draw=black] (\x,\y) rectangle ++ (1,1);
 \edef\xmaxtop{5}  
 \edef\ytop{20}  
 \foreach \x in {\xmin,...,\xmaxtop} 
 \path[fill=red!41,draw=black] (\x,\ytop) rectangle ++ (1,1);

 \edef\y{4}
 \edef\xmax{3}
 \foreach \x in {0,...,\xmax}
 \path[fill=black!56,draw=black] (\x,\y) rectangle ++ (1,1);

\edef\xmin{-11}  
 \foreach \x in {\xmin,...,-1} 
 \path[fill=black!66,draw=black] (\x,\y) rectangle ++ (1,1);
 \edef\xmaxtop{4}  
 \edef\ytop{19}  
 \foreach \x in {\xmin,...,\xmaxtop} 
 \path[fill=red!41,draw=black] (\x,\ytop) rectangle ++ (1,1);

 \edef\y{5}
 \edef\xmax{1}
 \foreach \x in {0,...,\xmax}
 \path[fill=black!56,draw=black] (\x,\y) rectangle ++ (1,1);

\edef\xmin{-11}  
 \foreach \x in {\xmin,...,-1} 
 \path[fill=black!66,draw=black] (\x,\y) rectangle ++ (1,1);
 \edef\xmaxtop{2}  
 \edef\ytop{18}  
 \foreach \x in {\xmin,...,\xmaxtop} 
 \path[fill=red!41,draw=black] (\x,\ytop) rectangle ++ (1,1);

 \edef\y{6}
 \edef\xmax{3}
 \foreach \x in {0,...,\xmax}
 \path[fill=black!56,draw=black] (\x,\y) rectangle ++ (1,1);

\edef\xmin{-7}  
 \foreach \x in {\xmin,...,-1} 
 \path[fill=black!66,draw=black] (\x,\y) rectangle ++ (1,1);
 \edef\xmaxtop{4}  
 \edef\ytop{17}  
 \foreach \x in {\xmin,...,\xmaxtop} 
 \path[fill=red!41,draw=black] (\x,\ytop) rectangle ++ (1,1);

 \edef\y{7}
 \edef\xmax{5}
 \foreach \x in {0,...,\xmax}
 \path[fill=black!56,draw=black] (\x,\y) rectangle ++ (1,1);

\edef\xmin{-3}  
 \foreach \x in {\xmin,...,-1} 
 \path[fill=black!66,draw=black] (\x,\y) rectangle ++ (1,1);
 \edef\xmaxtop{6}  
 \edef\ytop{16}  
 \foreach \x in {\xmin,...,\xmaxtop} 
 \path[fill=red!41,draw=black] (\x,\ytop) rectangle ++ (1,1);


 \path[draw=black,line width=0.7mm] (0,0) rectangle ++ (8,8);



 \path[draw=blue,line width=0.9mm,dashed] (-12,3) rectangle ++ (18,18);
 \path[fill=black!66,draw=black,line width=0.4mm] (5,3) rectangle ++ (1,1);

 \draw[-latex,ultra thick,blue!80](12,8)node[right]{en ajoutant une lettre noire }
       to[out=-90,in=0] (5.5,3.5);
\node[right,blue!80] at (12,6.5){supplémentaire,  nous obtenons };
\node[right,blue!80] at (12,5.0){un motif interdit  };

 \draw[-latex,ultra thick,black!95](15,-6)node[below]{ce motif $P''$ de taille $n\times n$ est  incompatible avec le voisinage}
       to[out=180,in=-90] (5.5,1.5);

\end{tikzpicture}
\end{center}
\caption[Un motif $P''$ qui n'est pas compatible avec $R$.]{Un motif $P''$ avec ${\cal E}_n(P'')\not\preccurlyeq_n {\cal E}_n(P)$ n'est pas compatible avec le voisinage.} \label{f:012-enforcing-ex2}
\vspace{-5pt}
\end{figure}

\smallskip

Il reste à observer que pour tout $n$ il y a $(n+1)^n$ motifs simples de taille $n\times n$ (pour chaque ligne la position de la frontière entre les zones noires et blanches varie entre $0$ et $n$). Par conséquent, pour des motifs simples $P$ de taille $n\times n$ la complexité de Kolmogorov de leur profil est supérieure à $n\log (n+1)$, i.e.,  même la complexité de Kolmogorov simple $\CK({\cal E}_n(P))$ est super-linéaire. Nous pouvons appliquer la Proposition~\ref{pbis:epitomes}~(a) et conclure que le shift n'est pas sofique.
\end{proof}

\end{example}

\subsection{La technique des résumés ordonnés et celle des chaînes d'unions croissantes d'ensemble d'extensions de motif}
\label{s:km}

Dans l'Exemple~\ref{ex:km} nous traduisons dans le langage  des résumés une preuve proposée par S.~Kass et K.~Madden dans \cite{kass-madden}.
Dans cette section nous montrons que cet exemple est un cas particulier pouvant être généralisé : la  technique de Kass et Madden est équivalente à un cas particulier de la technique des résumés ordonnés.
Dans ce cas particulier nous supposons que les résumés ont des domaines décidables, et nous utilisons la complexité de Kolmogorov entière (i.e., sans limite de ressources).
Nous commençons par deux définitions issues de \cite{kass-madden}.

\begin{definition} 
Soit $S$ sur $\mathbb{Z}^2$ un shift et $P$ un motif de support $B_n$.
L'ensemble des extensions de $P$, appelé $E_S(P)$, est l'ensemble des motifs $Q$ sur $F_n$ tels que l'union de $P$ et $Q$ forme une configuration de $S$.
La famille des ensembles d'extensions de tous les motifs $P$ de support $B_n$ globalement admissibles est appelée 
\[
Ext_S^n  := \{ E_S(P)\ |\ P \text{ est un motif globalement admissible sur }B_n\}.
\]
\end{definition}

Dans un certain sens, la cardinalité de $Ext_S^n$ correspond au ``flux d'information'' entre un motif de taille $n\times n$ et le reste de la configuration sur $\mathbb{Z}^2$. Intuitivement, ce ``flux d'information'' traverse la bordure, de taille seulement en $O(n)$, entourant un motif de taille $n\times n$. Cette intuition peut être précisée pour les shifts de type fini : pour tout shift de type fini $S$ sur $\mathbb{Z}^2$ la taille de $Ext_S^n$ ne peut excéder $2^{O(n)}$. 
Cependant, cette propriété n'est plus nécessairement vraie pour les shifts sofiques multidimensionnels (dans \cite{kass-madden} cette observation est attribuée à des travaux non publiés de C.~Hoffman,  A.~Quas, and R.~Pavlov). \label{p:extenders-do-not-work}

En réalité, pour certains shifts sofiques sur $\mathbb{Z}^2$ la taille de $Ext_S^n$ peut croître plus rapidement que $2^{O(n)}$. Par exemple, pour le shift sofique de l'Exemple~\ref{ex:semi-mirror-bis2}, la croissance de la taille de $Ext_S^n$ est en $2^{\Omega(n^2)}$. 

Par conséquent, pour montrer qu'un shift n'est pas sofique il ne suffit pas d'étudier la \emph{taille} de la famille des ensembles d'extensions de motifs ; nous devons nous intéresser à des propriétés plus subtiles de ces objets. Plus précisément, nous allons utiliser la structure induite par la relation d'inclusion sur la classe des ensembles d'extensions de motifs. En pratique, nous utilisons les chaînes d'union croissante des ensembles d'extensions de motifs.
\begin{definition}
Soit $S_1,\ldots,S_m$ une séquence finie d'ensembles non vides. Cette séquence est dite une chaîne d'\emph{union croissante} si les unions partielles des ensembles $S_i$ croissent strictement, i.e., 
pour tout $1\le i\le m$
\[
 S_i  \nsubseteq \bigcup\limits_{j=1}^{i-1} S_j .
\]
\end{definition}

Nous pouvons maintenant formuler le théorème de Kass et Madden. 
\begin{theorem}[\cite{kass-madden}]\label{th:km}
Soit $S$ un shift sur $\mathbb{Z}^2$. Supposons que pour tout $M > 0$ donné, il existe un $n > 0$, un $m > M^{n}$, 
et des motifs $P_1,\ldots,P_m$ de support $B_n$ globalement admissibles pour lesquels 
\[
Ext_S(P_1), \ldots,Ext_S(P_m)
\]
est une chaîne d'union croissante. Alors $S$ n'est pas sofique.
\end{theorem}
Intuitivement, ce théorème signifie que l'``information essentielle'' qui peut traverser la bordure d'un motif de taille $n\times n$ (dans un shift sofique)
ne correspond pas au nombre d'ensembles d'extensions de motifs, mais à la longueur des chaînes d'union croissante d'ensembles d'extensions de motifs.
Ce théorème peut être reformulé ainsi :
\begin{corollary}\label{cor:km}
Soit $S$ un shift sur $\mathbb{Z}^2$. Pour tout $n$ nous choisissons une séquence de motifs $P_1^n,\ldots,P_{m(n)}^n$ de support $B_n$ globalement admissibles pour lesquels la séquence d'ensembles d'extensions de motifs
\[
Ext_S(P_1^n), \ldots,Ext_S(P_{m(n)}^n)
\]
est une chaîne d'union croissante. Si $\log m(n) = \omega(n)$, alors  $S$ n'est pas sofique.
\end{corollary}

Nous montrons que chaque preuve qu'un shift n'est pas sofique basée sur le Théorème~\ref{th:km} peut être reformulée dans le langage des résumés ordonnés de la Définition~\ref{d:ordered-epitome}, avec des domaines décidables et la complexité de Kolmogorov entière (pas à ressources bornées), et vice-versa. Nous prouvons cette équivalence en deux étapes : tout d'abord nous montrons qu'un argument exprimé dans le langage des chaînes d'union croissante d'ensembles d'extensions de motifs peut être traduit dans le langage des résumés ordonnés (Proposition~\ref{p:km2e}); ensuite, nous prouvons l'implication inverse en montrons comment un argument expliqué dans le langage des résumés ordonnés peut être reformulé dans le langage des chaînes d'union croissante des ensembles d'extensions de motifs (Proposition~\ref{p:e2km} and Proposition~\ref{p:counting}).

\begin{proposition}\label{p:km2e}
Soit $S$ un shift sur $\mathbb{Z}^2$. Supposons que pour tout $n$ il y a une séquence de motifs $P_1,\ldots,P_{m(n)}$ de support $B_n$ globalement admissibles tels que la séquence des ensembles des extensions de motifs
\[
Ext_S(P_1), \ldots,Ext_S(P_{m(n)})
\]
est une chaîne d'union croissante. Alors il existe une bijection des motifs globalement admissibles de $S$ de support $B_n$ vers les entiers inférieurs à $m$:  
\[
 {\cal E}_n :[\text{motif sur le domaine } B_n] \to \{1,\ldots, m\}
 \]
et un ordre partiel $\preccurlyeq_n$ sur $\{1,\ldots, m(n)\}$ tels que 
$( {\cal E}_n , \preccurlyeq_n)$ satisfait la définition des résumés ordonnés pour $S$.
\end{proposition}
\begin{proof}
La construction des résumés adéquats est directe : soit ${\cal E}_n(P_i) = i$ pour $i=1,\ldots,m(n)$, ${\cal E}_n$ non défini pour les autres motifs. On définit l'ordre linéaire sur les valeurs de ${\cal E}_n$ comme l'inverse de l'ordre naturel sur les entiers :
\[
m(n)\preccurlyeq_n   \ldots  \preccurlyeq_n 2  \preccurlyeq_n 1.
\]
Vérifions que la famille $({\cal E}_n,\preccurlyeq_n)$ obtenue vérifie la définition des résumés. 
Puisque la séquence d'ensembles d'extensions de motifs $Ext_S(P_i)$ est une chaîne d'union croissante, pour tout $i$ nous avons
\[
Ext_S(P_i)  \not\subset  \bigcup\limits_{j=1}^{i-1} Ext_S(P_i).
\]
Cela signifie qu'il existe un motif $R$ de support $F_n$ qui est compatible avec $P_i$ (i.e., appartient à $Ext_S(P_i)$) mais pas avec $P_1,\ldots, P_{i-1}$ (i.e., n'appartient pas à $ \bigcup\limits_{j=1}^{i-1} Ext_S(P_i)$).
Par conséquent, pour tout motif $P'$ de taille $n\times n$ compatible avec $R$, si ${\cal E}_n(P')$ est défini alors ${\cal E}_n(P')\preccurlyeq_n i$, ce qui est exactement la définition des résumés ordonnés.
\end{proof}

Supposons que pour tout $n$ nous avons une séquence de motifs $P^n_1,\ldots,P^n_{m(n)}$ de support $B_n$ globalement admissibles tels que la séquence correspondante d'ensembles d'extension de motifs
\[
Ext_S(P^n_1), \ldots,Ext_S(P^n_{m(n)})
\]
est une chaîne d'union croissante. En appliquant la Proposition~\ref{p:km2e} nous obtenons une famille de résumés ${\cal E}_n$ avec $m(n)$ valeurs. Dans le cas général, cette famille de résumés peut ne pas être calculable. Cependant, elle est trivialement calculable étant donné l'oracle $\cal O$ qui contient la description des séquences $P^n_1,\ldots,P^n_{m(n)}$.

Si le nombre $\log m(n)$ croît de manière super-linéaire (i.e., $m(n) =2^{\omega(n)}$, ce qui est nécessaire pour appliquer le Corollaire~\ref{cor:km}), alors par un argument de comptage nous pouvons dire que la complexité de Kolmogorov de ${\cal E}_n(P^n_i)$ (même avec un accès à l'oracle $\cal O$) n'est pas limitée par une fonction linéaire. Plus précisément, pour tout $\lambda>0$ et pour un $n$ suffisamment grand il existe un motif $P^n_i$ tel que $\CK^{\cal O}({\cal E}_n(P^n_i))  > \lambda n$. Nous pouvons donc appliquer la Proposition~\ref{pbis:epitomes}(c) et conclure que $S$ n'est pas sofique.

Ainsi, si nous pouvons montrer qu'un shift $S$ n'est pas sofique avec la technique de Kass et Madden (Corollaire~\ref{cor:km}), alors grâce à la Proposition~\ref{p:km2e} nous pouvons appliquer la Proposition~\ref{pbis:epitomes}~(c)  (avec la complexité de Kolmogorov relativisée avec l'oracle $\cal O$ choisi ci-dessus) et montrer que $S$ n'est pas sofique avec la technique des résumés.
 
 \medskip
 
Avec les deux propositions suivantes nous montrons que la traduction dans l'autre sens (du langage des résumés ordonnés vers celui des chaîne d'ensembles d'extensions de motifs) est également possible tant que nous considérons des résumés ordonnés avec un domaine décidable et la complexité de Kolmogorov entière.   
 
\begin{proposition}\label{p:e2km}
Soient un shift $S$ , $({\cal E}_n,\preccurlyeq_n)$ une famille de résumés de $S$, et
 $m=m(n)$ le nombre d'images distinctes de ${\cal E}_n$ calculées à partir de motifs globalement admissibles. 
Alors il existe une famille de motifs $P_1,\ldots, P_m$ de support $B_n$ globalement admissibles tels que la séquence d'ensembles d'extensions de motifs
\[
Ext_S(P_1), \ldots,Ext_S(P_m)
\]
est une chaîne d'union croissante.
\end{proposition}
\begin{proof}
Soit $e_1,\ldots,e_m$ les images de ${\cal E}_n$. Sans perte de généralité nous pouvons supposer que pour tout $i$
\begin{equation} \label{eq:km}
e_j \not\preccurlyeq_n e_i \text{ for all } j<i 
\end{equation}
(dans tout ensemble muni d'un ordre partiel nous pouvons arranger les indices des éléments $e_j$ de telle manière que \eqref{eq:km} soit vraie : $e_1$ doit être un des éléments maximaux, $e_2$ doit être un des éléments maximaux parmi les éléments restants, et ainsi de suite). Ensuite, nous fixons les motifs $P_i$ de support $B_n$ globalement admissibles tels que ${\cal E}_n(P_i) = e_i$. Par la définition des résumés, pour chaque $i$ il existe un motif
$R_i$ de support $F_n$ qui est compatible avec $P_i$ mais non compatible avec $P_1,\ldots, P_{i-1}$. 
Par conséquent, $Ext_S(P_i)$ n'est pas contenu dans l'union $\bigcup\limits_{j=1}^{i-1} Ext_S(P_i)$. Ainsi, la séquence d'ensembles d'extensions de motifs $Ext_S(P_i)$ est une chaîne d'union croissante, ce qui conclut la preuve.
\end{proof}

\begin{proposition}\label{p:counting}
Soit $({\cal E}_n,\preccurlyeq_n )$ une famille de résumés ordonnés calculable sur un shift effectif  $S$. 

(a) Supposons que pour tout motif $P$ de taille $n\times n$ globalement admissible
\begin{equation}
\label{eq:ck}
\CK({\cal E}_n(P)) = m,
\end{equation}
et $m\gg \log n$. Alors il existe $2^{\Omega(m)}$ valeurs distinctes de ${\cal E}_n$ correspondant aux motifs de taille $n\times n$
(globalement admissibles ou non admissibles).

(b) Supposons que les domaines de ${\cal E}_n$ sont décidables (il existe un algorithme qui pour tout $n$ et pour tout motif $P$ de taille $n\times n$ peut décider si ${\cal E}_n(P)$ est défini ou non)
et que l'équation \eqref{eq:ck} avec $m\gg \log n$ est vraie pour au moins un motif $P$ de taille $n\times n$ globalement admissible.
Alors il existe $2^{\Omega(m)}$ valeurs distinctes de résumés ${\cal E}_n$ pour les motifs de cette taille.

(c) Les propositions (a) et (b) se relativisent : elles restent vraies pour une complexité de Kolmogorov relativisée avec un oracle $\cal O$.
\end{proposition}
\begin{proof}
(a) Étant donné un entier $n$, nous pouvons appliquer l'algorithme calculant ${\cal E}_n$ pour chaque motif de taille $n\times n$ en effectuant ces calculs en parallèle. Nous pouvons afficher les valeurs de ${\cal E}_n$ au fur et à mesure que ces calculs convergent.
(si le domaine de ${\cal E}_n$ n'est pas décidable, nous ne pouvons pas savoir quand le dernier calcul qui converge s'est terminé, dans la mesure où certains calculs peuvent prendre un temps infini sans jamais renvoyer de résultats.)
Soit $N_n$ le nombre de valeurs qui seront trouvées à un moment ou à un autre. 
Chaque valeur de ${\cal E}_n$ peut être spécifiée par son indice dans la liste (la position dans l'ordre de l'apparition des valeurs) et par la représentation binaire de $n$. L'indice (la position dans la liste) peut être spécifié par $\lceil \log N_n \rceil$ bits.
Par conséquent, pour tout motif $P$ de taille $n\times n$ nous avons
 \[
 \CK({\cal E}_n(P))  = O(\log N_n) +O (\log n) , 
 \]
ce qui prouve l'énoncé~(a).

(b) Soient $N_n$ le nombre de valeurs de ${\cal E}_n(P)$ pour tous les motifs $P$ de taille $n\times n$,
et $N'_n$ le nombre de valeur de ${\cal E}_n(P)$ pour tous les motifs $P$ de la même taille \emph{globalement admissibles}.
Notons $\ell = \lceil \log N_n'\rceil$.

Si $N'_n > N_n/2$, l'énoncé~(b) découle de l'énoncé~(a). Sinon, nous calculons ${\cal E}_n(P)$ pour tous les motifs de taille $n\times n$ (Ceci est possible puisque le domaine des résumés est décidable)
et commençons à énumérer les motifs du shift de cette taille non admissibles. À chaque motif non admissible trouvé, on vérifie pour chaque valeur ${\cal E}_n(P)$ si celle-ci ne correspond plus qu'à des motifs $P$ déjà trouvés ; dans ce cas nous pouvons l'éliminer (à chaque étape les valeurs qui n'ont pas été éliminées correspondent donc à au moins un motif $P$ potentiellement globalement admissible). Nous arrêtons ce processus lorsqu'il ne reste plus qu'exactement $2^\ell$ valeurs de résumées qui n'ont pas été éliminées. 
Un résumé d'un motif $P$ globalement admissible peut alors être spécifié par les nombres $n$, $\ell$ et le nombre ordinal de la valeur ${\cal E}_n(P)$ dans la liste des $2^\ell$ candidats non éliminés. Un tel nombre ordinal peut être représenté par une chaîne de caractère binaire de $\ell_n$ bits. Ainsi,
\[
 \CK({\cal E}_n(P))  = O(\ell)+ O (\log n), 
\]
ce qui conclut la preuve de l'énoncé.

(c) Il est facile de voir que l'argument ci-dessus reste valide pour les algorithmes ayant un accès à un oracle.
\end{proof}

Supposons que nous pouvons prouver qu'un shift $S$ n'est pas sofique avec la technique des résumés ordonnés, en utilisant la Proposition~\ref{pbis:epitomes} et une complexité de Kolmogorov sans limitation des ressources. C'est possible si nous avons pour $S$ une famille de résumés ordonnées $({\cal E}_n,\preccurlyeq_n)$ calculable (éventuellement avec un oracle $\cal O$), et pour laquelle les valeurs $\CK({\cal E}_n(P))$ ne sont pas limitées par une fonction linéaire de $n$. Supposons que l'union des domaines de ${\cal E}_n$ est un ensemble décidable.
Alors en utilisant la Proposition~\ref{p:counting}~(b)  nous pouvons dire que le nombre de valeurs ${\cal E}_n$ pour les motifs globalement admissibles croît en $2^{\omega(n)}$. 
Nous pouvons donc appliquer la Proposition~\ref{p:e2km} et obtenir la liste des motifs $P^n_1,\ldots, P^n_{m(n)}$ de taille $n \times n$ globalement admissibles,  
avec $m(n) = 2^{\omega(n)}$ et telle que la chaînes d'ensembles des extensions de motifs
$
Ext_S(P^n_1), \ldots,Ext_S(P^n_{m(n)})
$
est d'union croissante.
Par conséquent, nous pouvons appliquer le Corollaire~\ref{cor:km} et prouver que $S$ n'est pas sofique avec l'argument de Kass et Madden.
 
\medskip

La Proposition~\ref{p:km2e} et la Proposition~\ref{p:e2km} montrent que la technique des chaînes d'ensembles d'extensions de motifs est étroitement liée à la version de la technique des résumés ordonnés avec un domaine décidable et la complexité de Kolmogorov entière.
Par ailleurs, il semble que l'argument utilisant les résumés simples avec des domaines non décidables
(voir l'Exemple~\ref{ex:busy-beaver}) ne peut être traduit directement dans le langage des chaînes d'union d'ensembles d'extensions de motifs.
Nous pouvons aussi rappeler que la méthode des résumés est plus puissante lorsque nous utilisons la version de la complexité de Kolmogorov à ressources bornées.
En particulier, elle permet de montrer que des shifts avec une complexité par blocs sous-exponentielle ne sont pas sofiques, ce qui ne peut être fait avec le Théorème~\ref{th:km} et son Corollaire~\ref{cor:km}.

\chapter{Des conditions suffisantes pour qu'un shift soit sofique}
\label{c:sofique}

Dans le Chapitre~\ref{c:non-sofique} nous avons défini des outils permettant de montrer qu'un shift n'est pas sofique. En particulier, nous avons montré qu'un shift effectif et non sofique pouvait avoir une complexité par bloc petite, i.e. polynomiale. De manière complémentaire, nous nous intéressons dans ce chapitre à des conditions suffisantes pour qu'un shift soit sofique. Nous verrons ainsi que si la complexité par bloc d'un shift effectif sur l'alphabet des cases noires et blanches est très petite, i.e. le nombre de cases noires dans tout bloc de taille $N \times N$ est sous-linéaire, alors le shift est sofique.

Nous commençons ce chapitre par deux théorèmes formalisant ce lien entre la faible complexité par bloc d'un shift et le fait qu'il soit sofique (Partie~\ref{s:résultats}). Nous introduisons ensuite plusieurs outils qui nous seront nécessaires dans la preuve de ce résultat :
\begin{itemize}
	\item Un modèle de calcul, les automates cellulaires non-déterministes à une dimension, dans la Partie~\ref{s:automate}. Nous utilisons à la fois le parallélisme et le non-déterminisme de ce modèle pour effectuer des calculs de manière efficace (i.e. en temps quasi-linéaire, au lieu d'un temps polynomial pour le cas d'une machine de Turing classique).
	\item Un jeu de tuiles dont les pavages sont auto-similaires. Cette construction servira de base à notre construction finale. Cette Partie~\ref{s:shift-auto-similaire} est basée sur le Chapitre~2 de l'article ``The expressiveness of quasiperiodic and minimal shifts of finite type'' de A. Romashchenko et B. Durand.
	\item Un résultat sur les flots pouvant circuler dans des graphes d'un certain type. Celui-ci est utilisé pour prouver que l'ensemble des pavages de notre construction finale n'est pas vide.
\end{itemize} 

Enfin, dans la dernière partie nous utilisons tous ces outils à la fois pour démontrer les résultats principaux de ce chapitre.

\section{Un shift de très faible complexité par bloc est sofique : résultats formels.}\label{s:résultats} 

Dans cette partie nous formulons les résultats principaux de ce chapitre, en formalisant l'assertion qu'un shift de très faible complexité par bloc est sofique. De plus, nous verrons que les sous-shifts effectifs de tels shifts sont également sofiques. Le reste du chapitre sera dédié à la preuve de ces deux résultats.

\begin{definition}
\label{d:densite}
Soit $\epsilon<1$ un réel.
On dit qu'un shift $S$ sur l'alphabet $\{ \blacksquare,\square \}$ est de densité $\epsilon$ si tout motif $P$ globalement admissible de $S$ de taille $n\times n$ ne contient pas plus de $n^\epsilon$ lettres de couleur $\blacksquare$.

Le shift $S_\epsilon$ sur l'alphabet $\{ \blacksquare,\square \}$, dont les motifs interdits sont exactement les motifs $P$ de taille $n \times n$ contenant plus de $n^\epsilon$ lettres de couleur $\blacksquare$, est appelé le shift de densité $\epsilon$.
\end{definition}

Il est clair que le shift $S_\epsilon$ est effectif, et que tout shift de densité $\epsilon$ est un sous-shift de celui-ci.

\begin{remark}\label{r:calculable}
Dans la suite on suppose que $\epsilon$ est un réel calculable (il existe un algorithme qui pour tout $n$ renvoie la $n$-ème décimale de $\epsilon$), ou du moins que $\epsilon$ peut être approximé par en haut (il existe un programme qui énumère l'ensemble des rationnels supérieurs à $\epsilon$). Dans ce cas, il est clair que le shift de densité $\epsilon$ est un shift effectif.
\end{remark}

\begin{theorem}\label{th:sparse1}
Pour tout $\epsilon<1$ calculable par en haut, le shift de densité $\epsilon$ est sofique.
\end{theorem}

De plus, chaque sous-shift effectif d'un shift de densité $\epsilon$ est également sofique :

\begin{theorem}\label{th:sparse2}
Pour tout $\epsilon<1$ calculable par en haut, tout sous-shift effectif du shift de densité $\epsilon$ est sofique.
\end{theorem}

\begin{remark}
Le Théorème~\ref{th:sparse2} reste vrai pour plusieurs couleurs, dans le cas où toutes les couleurs sauf une (la couleur majoritaire) sont de densité $\epsilon$ (dans tout motif de taille $N \times N$, le nombre de cases qui ne sont pas de la couleur majoritaire est au plus $N^\epsilon$).
\end{remark}

Le reste de ce chapitre est dédié à la preuve du Théorème~\ref{th:sparse1} et du Théorème~\ref{th:sparse2}. La preuve combine plusieurs techniques différentes, et malheureusement il est difficile de diviser l'argumentation en une série de lemmes indépendants (cela est souvent le cas pour les preuves basées sur les constructions de pavages à point-fixe). Pour aider le lecteur a suivre cette longue preuve, nous ébauchons ici (jusque la fin de cette partie) les idées principales, les différentes difficultés et les principales nouveautés introduites dans notre construction.

D'un point de vue de haut niveau, l'idée de la preuve est assez commune. Nous avons un shift $S$ de densité $\epsilon$ (ou un sous-shift effectif de $S$) avec une description explicite de ses motifs interdits. Nous voulons construire un shift de type fini dont la projection coordonnée par coordonnée nous donnera exactement $S$. Un tel shift de type fini doit associer à une configuration valide de $S$ une sorte de ``certificat'' prouvant que cette configuration appartient bien à $S$. Pour implémenter ce certificat, nous encodons dans le shift de type fini des calculs énumérant les motifs interdits et vérifiant qu'aucun d'eux ne peut apparaître dans la projection.

Plus précisément, nous allons construire un shift de type fini possédant une structure hiérarchique : chaque configuration valide est une grille de motifs carrés spéciaux (nous appelons ces carrés des ``super-tuiles'' ; nous en donnerons une définition précise ultérieurement). Chacun de ces carrés contient un diagramme espace-temps simulant un calcul (fini). Ces super-tuiles sont regroupées dans des carrés de taille plus importante (des ``super-super-tuiles'') simulant un calcul sur un temps plus long, et ainsi de suite. À chaque niveau de la hiérarchie, nous pouvons voir la configuration entière comme une grille consistant en des super-tuiles de rang $n$. Chaque super-tuile de rang $n$ est construite avec des super-tuiles de rang $(n-1)$, et sert à son tour de ``brique'' dans une plus grande super-tuile de rang $(n+1)$. Chaque super-tuile de rang $n$ contient une ``zone de calcul'' où est représenté le diagramme espace-temps d'un calcul (a priori il peut s'agir du diagramme espace-temps d'une machine de Turing, d'un automate cellulaire, ou d'un autre modèle de calcul similaire).
Ces calculs reçoivent en entrée une description complète de certains motifs du shift $S$ (les motifs qui seront obtenus par la projection de cette super-tuile et de ses voisines), et vérifie la consistance de ces motifs (si ils sont admissibles ou non). Nous pouvons noter que comme le diagramme espace-temps est encodé dans une super-tuile, le ``temps de calcul'' et l'``espace disponible'' sont limités par la taille de la super-tuile.

Le plan résumé dans le paragraphe précédent est très standard. L'idée d'encoder un calcul dans un shift de type fini remonte à Berger (voir \cite{berger}) et Robinson (voir \cite{robinson}). Des constructions similaires sont utilisées, par exemple, dans \cite{hochman2009dynamics,aubrun-sablik,dls} et bien d'autres travaux. Une autre approche peut être empruntée pour construire des shifts de type fini, ayant une structure hiérarchique, et encodant des calculs, que celle de Berger et Robinson. 

Cette approche a été développée en se basant sur des pavages à point-fixe (également nommés pavages auto-similaires), voir \cite{fixed-point}, ou, plus récemment, \cite{point-fixe}). Nous choisissons la méthode des pavages à point-fixe, et dans ce qui suit nous justifions ce choix.

De manière schématique, dans notre preuve nous avons à résoudre trois problèmes~:
\begin{itemize}
\item \emph{Stockage de l'information} : nous devons encoder dans la zone de calcul de chaque super-tuile de rang $n$ une description complète d'un motif ``clairsemé'' de taille comparable à celle de la super-tuile ;
\item \emph{Transmission de l'information} : nous devons délivrer à la zone de calcul d'une super-tuile de rang $n$ l'information de chaque lettre noire dans sa ``zone de responsabilité'' ;
\item \emph{Calcul de l'information} : nous devons simuler dans chaque super-tuile de rang $n$ un calcul vérifiant que le motif correspondant est valide.
\end{itemize}

Le Théorème~\ref{th:sparse1} et le Théorème~\ref{th:sparse2} peuvent être prouvés avec des constructions plus standards (par exemple, avec une adaptation assez directe de la construction de \cite{fixed-point}) pour des shifts \emph{extrêmement} clairsemés, par exemple pour des shifts sur l'alphabet $\{ \square, \blacksquare \}$ où chaque motif de taille $N \times N$ ne contient au plus que $\poly(\log N)$ lettres noires (les autres lettres étant blanches). Nous ne connaissons pas de preuve n'utilisant pas les shifts auto-similaire.

Pour démontrer le résultat pour des configurations ``modérément clairsemées'' avec $N^\epsilon$ lettres noires apparaissant dans les motifs de taille $N\times N$ (avec un $\epsilon<1$ arbitraire). C'est cela qui cause plusieurs difficultés techniques.

La difficulté pour le \emph{stockage de l'information} : la ``mémoire'' d'un calcul simulé par une super-tuile de rang $k$ doit contenir la description d'une configuration ``clairsemée'' d'une taille comparable avec celle de la super-tuile entière. Si une super-tuile de rang $k$ est de taille $N_k \times N_k$, il faut pouvoir stocker une liste de $N_k^\epsilon$ coordonnées de lettres noires. La taille de la liste est seulement légèrement plus petite que la taille de la super-tuile. Pour que cela soit possible, nous avons besoin de super-tuiles de rang $k$ dont la taille croît très vite (de manière double exponentielle, comme nous le verrons plus tard) en fonction de $k$. Cela peut être réalisé avec la construction des pavages à point-fixe. Nous expliquons cette partie de la construction dans la Partie~\ref{s:shift-auto-similaire}.

La difficulté avec le \emph{transfert d'information} : puisque la taille des super-tuiles de rang $k$ croît très vite, une super-tuile de rang $k$ n'a pas assez d'espace pour se ``souvenir'' de la liste de toutes les lettres noires du voisinage d'une super-tuile de rang $(k+1)$. Cela rend le transfert d'information depuis les super-tuiles de rang $k$ vers les super-tuiles de rang $(k+1)$ plus complexe ; en particulier, nous ne pouvons pas utiliser les techniques classiques de ``transmission de l'information'' utilisées dans \cite{drs,fixed-point,point-fixe} et dans la plupart des autres travaux basés sur les pavages à point-fixe. C'est pourquoi nous utilisons une construction très différente (et nécessairement non déterministe) pour contrôler le flux d'information. Intuitivement, chaque morceau d'information ``devine'' de manière non déterministe son propre parcours depuis la source jusqu'à la destination, et nous n'avons qu'à vérifier localement que les flux d'information sont cohérents. Nous démontrons qu'il est possible de répartir les ``canaux de communication'' d'une façon cohérente en utilisant des techniques de flots sur les graphes (le principe ``flot maximal/coupe minimale''). Cette partie de la construction est expliquée dans la Partie~\ref{s:flots}.
 
La difficulté dans le \emph{calcul de l'information} : chaque super-tuile de rang $k$ doit réaliser un calcul avec des données de taille seulement légèrement plus petite que la taille de la super-tuile elle-même. Cela signifie que nous devons simuler un calcul qui effectue certaines opérations sur des entrées de taille $N_k^\epsilon$ (où $\epsilon$ peut être proche de 1) en un temps inférieur à $N_k$. Par conséquent, nous ne pouvons pas nous permettre des calculs s'effectuant en temps polynomial (sans spécifier les polynômes) ; nous devons au contraire réaliser ces calculs en temps quasi-linéaire. Une telle performance ne peut pas être réalisée avec une machine de Turing classique. C'est pourquoi nous utilisons un modèle de calcul moins conventionnel, incluant du parallélisme et du non déterminisme (une sorte d'automate cellulaire non déterministe). Ce modèle de calcul est présenté dans la Partie~\ref{s:automate}.

Enfin, dans la Partie~\ref{s:preuve} nous combinons tous les ingrédients ensemble pour obtenir la preuve des Théorèmes~\ref{th:sparse1} et~\ref{th:sparse2}.

\section{Un modèle de calcul en parallèle : les automates cellulaires non déterministes en une dimension}\label{s:automate}

Dans la Partie~\ref{s:preuve} nous construirons un jeu de tuiles $\tau$, dont la projection des pavages nous donnera le shift de densité $\epsilon$ (ou un sous-ensemble effectif de ce shift). Pour cela, nous utiliserons des shifts auto-similaires, qui seront présentés dans la partie suivante. Ceux-ci sont fortement basés sur la possibilité de simuler une machine de Turing par un pavage. Nous verrons que dans notre cas, la simulation d'une machine de Turing classique ne sera pas suffisante : nous aurons également besoin de simuler une machine de Turing à deux têtes de lecture sur un ruban, et un automate cellulaire. Plus précisément, il s'agira d'un automate cellulaire non déterministe unidimensionnel, avec un voisinage de rayon 1 (l'état d'une cellule à l'étape suivante dépend de l'état de sa voisine de gauche, de son état et de l'état de sa voisine de droite).

\subsection{Simulation par un pavage}

La manière dont nous entendons simulation est restrictive : il s'agit de simuler une machine de Turing ou un automate cellulaire uniquement au sein d'un rectangle fini $R$. Plus précisément, $R$ représente une partie finie d'un diagramme espace-temps, sa ligne du bas représentant les données d'entrée (le temps va vers le haut). En particulier, nous pouvons choisir la position pour la ou les têtes de lecture d'une machine de Turing au niveau de ces données d'entrée.

\begin{remark}
Le fait que nous simulons un diagramme espace-temps d'une machine de Turing seulement sur un espace fini nous évite de devoir prendre en compte certaines considérations, qui sont nécessaires quand la simulation du diagramme espace-temps se fait sur un espace infini. En effet, dans ce dernier cas il peut ne pas y avoir de tête de lecture, et qu'aucun calcul ne soit effectué dans le pavage. De plus, comme il existe des zones arbitrairement grandes où aucun calcul ne se fait, par compacité cela ne peut être évité. Il faut aussi vérifier dans ce cas que, outre les têtes de lecture désirées, il n'existe pas de tête de lecture surnuméraire.
\end{remark} 

La Proposition~\ref{p:shift-tuiles} nous garantit qu'à tout shift $S$ nous pouvons associer un jeu de tuiles $\tau$ tel que les configurations de $S$ sont isomorphes aux pavages de $\tau$. Par conséquent, nous allons expliquer comment simuler ces modèles de calcul par un shift $S$, l'explication étant plus directe qu'avec un jeu de tuiles. La manière de construire $S$ est essentiellement la même pour les trois modèles de calcul que nous utilisons (une machine de Turing standard, une machine de Turing à deux têtes de lecture sur un ruban, ou un automate cellulaire non-déterministe). En effet, pour connaître l'état (ou les états possibles) d'une cellule à l'étape suivante, il suffit de connaître l'état de la cellule et de ses deux voisines à l'étape donnée. Par exemple, nous pouvons décider que le support des motifs interdits de $S$ sera un rectangle $3\times 2$.

Commençons par expliquer comment simuler par un shift $S$ une machine de Turing classique $M$ (une seule tête et un seul ruban) dans une partie rectangulaire finie $R$ d'un pavage. Chaque cellule du diagramme espace-temps de $M$ contient un symbole appartenant à l'alphabet des symboles écrits sur le ruban. Une cellule par ligne, où se trouve la tête de lecture, contient en plus de son symbole l'état interne de la machine. Il suffit alors d'interdire tous les rectangles de taille $3\times 2$ ne correspondant pas à la table de transition de $M$ pour simuler le comportement de la tête de lecture ; nous interdisons également les rectangles où le symbole d'une cellule ne contenant pas de tête de lecture est modifié.

Pour une machine de Turing $M'$ à deux têtes de lecture sur un ruban, la situation est légèrement plus complexe. En effet, il faut que les deux têtes de lecture connaissent l'état interne de $M'$, mais également, outre son propre symbole, le symbole de l'autre tête de lecture. Pour cela, nous décidons que chaque cellule du diagramme espace-temps connaît, outre son symbole, l'état interne de $M'$ et les deux symboles pointés par les deux têtes de lecture (les têtes de lecture vérifient chacune que leur symbole respectif est cohérent avec ces symboles partagés). Nous pouvons alors de nouveau interdire tous les rectangles de taille $3\times 2$ ne correspondant pas à la table de transition de $M'$ et ceux où le symbole d'une cellule sans tête de lecture est modifié pour obtenir $S$.

Enfin, pour un automate cellulaire non-déterministe avec un voisinage de taille 1, il n'y a pas de difficulté particulière : nous interdisons tous les rectangles de taille $3\times 2$ pour lesquels la cellule au milieu en haut du rectangle ne correspond pas à une transition de l'automate pour les trois cellules du bas du rectangle.

\subsection{Calculer avec un automate cellulaire non déterministe}

Dans cette section nous allons montrer un lemme technique nous assurant que des opérations sur des listes peuvent être effectuées de manière très rapide par un automate cellulaire non déterministe.

Nous verrons que ce résultat nous sera utile dans la Section~\ref{s:preuve} pour démontrer le résultat principal de cette partie.

Dans la preuve du résultat principal nous avons besoin d'encoder dans un shift de type fini des calculs implémentant la procédure suivante. Nous recevons en entrée un motif de taille $N \times N$ sur l'alphabet constitué de points noirs ou blancs, et ce motif est décrit par la liste complète de ses points noirs. Nous supposons que le motif est donné comme une liste de points noirs, où chaque point noir est représenté par ses coordonnées. Nous supposons que le nombre de points noirs $L$ est au plus $\leq N^{\epsilon}$, avec $\epsilon < 1$. Nous recevons également en entrée une liste de motifs interdits. Nous supposons que cette liste est très petite. Disons qu'elle contient $\ell$ éléments, et que $\ell \ll L$. La tâche est de vérifier qu'aucun de ces motifs interdits n'apparaît dans le motif de taille $N \times N$. La tâche est facile à implémenter en temps $\poly (N)$, mais nous voulons qu'elle soit effectuée de manière très efficace, en utilisant un temps pas plus grand que $O(L)$. Dans ce but, nous utilisons un modèle de calcul non-conventionnel, permettant un usage important du parallélisme et du non-déterminisme.

D'un point de vue formel, un tel modèle peut être représenté par un automate cellulaire non-déterministe. De manière moins formelle, nous parlerons d'une machine de Turing possédant des têtes de lecture ``classiques'', dites principales, et de nombreuses têtes de lecture ``secondaires'' (voir ci-dessous). Dans cette partie nous donnons les grandes lignes de la description des opérations nécessaires. Nous expliquons que ce modèle nous permet de manipuler de manière très efficace de longues listes de ``points''. Nous verrons que nous pouvons vérifier si un ou plusieurs éléments appartiennent à une liste, et localiser ces éléments dans la liste (voir Lemme~\ref{l:automate}). Le lecteur ayant une expérience de ``programmation'' à l'aide d'automate cellulaire peut facilement reconstruire les constructions nécessaires. Par commodité pour le lecteur nous donnons les grandes lignes de la preuve technique de ce lemme.

Nous allons détailler dans un premier temps le ``fonctionnement général'' de l'automate, puis les symboles de ses cellules. Ensuite, nous expliquerons comment utiliser ce fonctionnement général pour réaliser efficacement certains calculs. Nous nous appuierons sur ceux-ci dans la Partie~\ref{s:preuve} pour définir le comportement exact de l'automate $A$. 

Intuitivement, l'automate $A$ consiste en une machine de Turing multi-têtes ``améliorée''. Soit une machine de Turing $M$ multi-têtes, possédant $k$ têtes de lecture ($k$ est un \emph{entier constant}). L'automate $A$ peut facilement simuler la machine $M$ : chaque cellule de l'automate $A$ possède un \emph{symbole}, l'automate possède un \emph{état interne} (connu par chaque cellule), et pour chaque ligne, un nombre $k$ de cellules, dites ``principales'', jouent le rôle de têtes de lecture (les symboles de ces cellules sont également connus par chaque cellule de la ligne). Dans une machine de Turing multi-têtes de lecture, les différentes têtes sont ``discernables'' : nous pouvons associer à chacune un nombre entre 1 et $k$, et le comportement de chacune des différentes têtes de lecture doit être spécifié. De même, à chaque cellule principale de l'automate est associé un nombre différent entre 1 et $k$. À chaque étape, l'automate décide, en fonction des symboles pointés par les $k$ cellules principales et de son état interne, de son nouvel état interne et du comportement de chaque cellule principale : celles-ci peuvent écrire un symbole, et se déplacer vers la gauche, vers la droite ou rester immobiles. Le symbole des cellules dites ``inactives'' n'est pas modifié. En particulier, toute tâche pouvant être effectuée par une machine de Turing multi-têtes peut être effectuée par un tel automate cellulaire.

Le principal ajout est celui de cellules dites ``secondaires''. Celles-ci jouent le rôle de têtes de lecture ``supplémentaires'', et peuvent être en \emph{nombre non borné}. Ces cellules secondaires connaissent l'état interne de l'automate, les symboles des cellules principales et, pour chacune d'elles, le symbole qu'elle pointe. Ainsi, une cellule secondaire est la seule à connaître le symbole qu'elle pointe. À chaque étape, chacune de ces cellules secondaires décident, de manière indépendante et en fonction de leur symbole, de l'état interne et des symboles des cellules principales, soit de devenir inactive, soit d'écrire un symbole et de se diriger vers la gauche, vers la droite ou de rester immobile.

Une fois initialisées, ces cellules secondaires opèrent indépendamment, sans communiquer les unes avec les autres et sans envoyer d'information aux cellules principales (par contre leur comportement peut dépendre des symboles des cellules principales). Par conséquent, avec un nombre non borné de cellules secondaires, nous pouvons initialiser des calculs indépendants s'exécutant en parallèle sur différentes parties du ruban. Cependant, nous pouvons réaliser bien plus que des calculs indépendants si nous autorisons des interactions plus subtiles entre les cellules principales et secondaires. Dans ce but, nous supposons qu'une cellule principale, à laquelle est associée le numéro $i$ (avec $1\leq i \leq k$), peut également envoyer un ``signal spécial'', vers la gauche ou vers la droite. Dans ce cas, la cellule devient par défaut secondaire (elle peut éventuellement devenir inactive) et la première cellule secondaire à ``capter'' le signal devient principale (il s'agit de la première cellule secondaire située respectivement à gauche ou à droite de la cellule principale), avec $i$ comme numéro associé. Cette opération ressemble à un ``saut'' de la position de cette cellule principale vers la gauche ou la droite, jusqu'à la cellule où, dans la direction choisie, la cellule secondaire la plus proche est localisée. Remarquons que la longueur potentielle d'un tel saut n'est pas bornée.

Pour envoyer ces signaux spéciaux, nous utilisons essentiellement le fait que notre automate cellulaire soit non-déterministe. Plus précisément, si les informations d'une cellule et de ses deux voisines ne sont pas cohérentes (i.e., elles ne respectent pas les règles décrites ci-dessus, par exemple deux cellules inactives voisines encodent des signaux différents), alors l'automate entre dans un ``état d'erreur''. 

Par ailleurs, quand l'état interne change (chaque cellule connaît alors l'ancien et le nouvel état interne), une cellule principale peut décider de devenir secondaire ou inactive, une cellule secondaire peut décider de devenir inactive, et une cellule inactive peut devenir principale ou secondaire.  Nous devons ainsi préciser, pour chaque état $e$, pour chaque état $e'$ pouvant succéder à $e$, pour chaque symbole et pour chaque ``activité'' des cellules (principale, secondaire ou inactive), le comportement adopté. Une fois ces changements effectués, le nombre de cellules principales doit toujours être égal à $k$, et les nombres entre 1 et $k$ associés à chaque cellule principale doivent être différents : les choix de ces comportements doivent garantir cela (au moins sur les entrées ``valides'' de l'automate). Bien que l'automate soit non déterministe, son évolution est entièrement déterminée par ses données d'entrée.

Comme c'est le cas pour tout automate cellulaire (déterministe ou non-déterministe), un diagramme espace-temps valide de l'automate est déterminé par un nombre fini de motifs interdits. Nous avons vu que nous pouvions décrire un automate ayant un voisinage de taille 1 par des motifs interdits dont les supports sont inclus dans des rectangles de taille $3\times 2$ ; nous interdisons en outre les motifs d'une seule case correspondant à l'état d'erreur de l'automate. Nous pouvons donc simuler par un shift de type fini un fonctionnement correct de l'automate (sans état d'erreur) sur un rectangle fini $R$. D'après la Proposition~\ref{p:shift-tuiles}, nous pouvons également simuler l'automate ({i.e., représenter son diagramme espace-temps) par un jeu de tuiles. 

Nous considérons que chaque symbole de l'automate est une paire constituée d'un bit d'information et éventuellement d'une ``marque spéciale''. Ces marques spéciales nous permettront d'initialiser des calculs massivement parallèles sur un nombre arbitrairement grand d'éléments, et seront détaillées dans la Partie~\ref{s:preuve}. 

Dans le lemme suivant nous supposons qu'un \emph{élément} est une suite de bits, encodés de telle manière qu'une machine de Turing puisse déterminer où commence et finit l'élément. Nous supposons que les données d'entrée de l'automate consistent en des éléments groupés dans plusieurs listes. Pour indiquer où commencent et finissent les données d'entrée, chaque liste et chaque élément d'une liste, nous ajoutons des marques spéciales supplémentaires. Plus précisément, une paire de marques labellisent le premier et le dernier bit des données d'entrée. Pour chaque liste, une paire de marques labellisent le premier et le dernier élément de la liste. Enfin, pour chaque liste, une paire de marques labellisent, pour chaque élément, le premier et le dernier bit de l'élément.

\begin{lemma}\label{l:automate}
Soit un élément $e$ formé d'une chaîne binaire de $q$ bits, et un nombre $k$ de listes $L_1, \ldots, L_k$, dont les éléments sont constitués d'une chaîne binaire de $q'\geq q$ bits. Alors nous pouvons vérifier si $e$ est le préfixe d'exactement un élément $e'$ présent dans exactement une liste, ou si $e$ est le préfixe d'exactement deux éléments $e_1$ et $e_2$ présents dans exactement deux listes distinctes. Nous pouvons également localiser ces éléments $e'$ ou $e_1$ et $e_2$. Ces opérations se font en temps $O(q)$.
\end{lemma}

\begin{proof}
Initialement, la première cellule de l'élément $e$ est principale, et les premières cellules de tous les éléments des $k$ listes $L_1, \ldots, L_k$ sont secondaires. La cellule de fin $f$, située après les $k$ listes, est également secondaire.
La cellule principale parcourt les $q$ cellules de l'élément $e$ ; les cellules secondaires parcourent leur élément respectif, et si le bit d'information lu ne correspond pas au bit d'information de la cellule principale elles deviennent inactives (excepté pour la cellule $f$, qui reste secondaire et immobile). Quand la dernière cellule de l'élément $e$ est principale, la cellule principale effectue alors un premier saut vers la droite ; si la cellule $f$ devient principale, c'est qu'aucun élément des $k$ listes n'a pour préfixe l'élément $e$, et l'automate entre dans un état d'erreur. Sinon, la cellule principale effectue un deuxième saut vers la droite. Si la cellule de fin $f$ est principale, c'est que l'élément $e$ est préfixe d'exactement un élément $e'$ dans exactement une liste. La cellule principale $f$ effectue alors un saut à gauche, puis se déplace vers la gauche jusqu'à atteindre le premier bit de l'élément $e'$.

Dans le cas contraire, c'est qu'au moins 2 éléments des $k$ listes ont pour préfixe $e$. Soit $e_1$ et $e_2$ respectivement le premier et le deuxième élément trouvé. Nous devons vérifier que $e_1$ et $e_2$ n'appartiennent pas à la même liste. Pour cela, toutes les premières cellules des premiers éléments des $k$ listes deviennent secondaires. La cellule principale (située au niveau du $q$-ème bit de $e_2$) effectue un saut vers la gauche. Si la cellule principale est située au niveau de $e_1$, c'est que les deux éléments appartiennent à la même liste et l'automate entre dans un état d'erreur. Sinon, la cellule principale est située sur le premier bit du premier élément d'une liste. Toutes les cellules secondaires situées sur les premiers bits des premiers éléments d'une liste deviennent inactives. La cellule principale effectue un saut vers la droite (et devient inactive et non secondaire), et se situe de nouveau dans $e_2$. La cellule principale effectue alors un saut vers la droite. Si la cellule principale n'est pas $f$, alors il existe au moins trois éléments qui sont préfixes de $e$, et l'automate entre dans un état d'erreur. Sinon, nous avons trouvé nos deux éléments ; il reste à effectuer un saut vers la gauche, puis à se déplacer vers la gauche jusqu'à atteindre le premier bit de $e_2$. Pour accéder à l'élément $e_1$, il suffira d'effectuer de nouveau un saut vers la gauche, puis de se déplacer vers la gauche jusqu'à atteindre le premier bit de $e_1$.
\end{proof}

\section{Shifts auto-similaires}\label{s:shift-auto-similaire}

Historiquement, la notion d'auto-similarité est introduite par John Von Neumann pour les automates cellulaires dans \cite{von1963general}. Ces idées sont reprises par Peter G\'acs pour construire un automate cellulaire tolérant aux erreurs, voir \cite{gacs1986reliable,gacs2001reliable} (l'article \cite{gacs2001reliable} est commenté dans \cite{gray2001reader}). Plus tard, cette technique d'auto-simulation a été utilisée dans le contexte des shifts multidimensionnels de type fini, voir \cite{fixed-point}, ou, plus récemment, \cite{surface-entropies} et \cite{point-fixe} (une explication à destination d'un public plus large est également donnée dans \cite{durand2009foundations}). Mentionnons également une approche très formelle sur les shifts auto-similaires dans \cite{zinoviadis2016hierarchy}, et un survol reprenant ces notions d'auto-simulation pour les shift de type fini et pour les automates cellulaires dans \cite{torma2021fixed}.

Considérons à nouveau le shift constitué de carrés noirs de tailles différentes de \cite{westrick2017seas}, évoqué page~\pageref{p:ex-westrick} dans la Sous-Section~\ref{ss:2-dimension} ; la preuve que celui-ci est sofique est basée sur la technique d'auto-simulation de \cite{fixed-point}.

Dans notre preuve, nous ne pouvons pas utiliser les \emph{résultats} démontrés dans ces papiers, mais nous devons plutôt réutiliser la \emph{méthode} générale développée pour obtenir des constructions auto-similaires et l'adapter à nos besoins. Comme cette construction est assez complexe, nous commençons par une présentation générale de la technique permettant d'obtenir un jeu de tuiles auto-similaire. 
Dans notre présentation de cette technique, nous suivons essentiellement la version de la construction par point-fixe présentée dans \cite{point-fixe}.
Pour le lecteur familier avec \cite{point-fixe}, nous pouvons indiquer les principales différences : (1)~nous adaptons la construction pour le cas où les ``super-couleurs'' sur les bords des ``super-tuiles'' encodent beaucoup d'information (soit un nombre de bits comparable avec la taille d'un côté de la super-tuile, et non avec le logarithme de cette taille), et (2)~dans les ``zones de calcul'' des ``super-tuiles'', nous utilisons, en plus d'une machine de Turing classique, une machine de Turing à plusieurs têtes (avec deux têtes de lecture sur un même ruban) et un automate cellulaire non déterministe, ce qui nous permet de traiter ces très longues ``super-couleurs'' en temps linéaire. Dans la Partie~\ref{s:preuve} nous ajusterons cette construction générale en la combinant avec les autres ingrédients de notre preuve.  

Décrivons le plan de cette partie plus en détail. Dans la première sous-partie, nous commençons par expliquer comment un jeu de tuiles peut en simuler un autre. La deuxième sous-partie est la plus technique : c'est là que nous reprenons la construction d'un jeu de tuiles auto-similaire de \citep{point-fixe}, avec les différences indiquées précédemment. Cette construction repose sur la simulation d'une machine de Turing au sein des pavages de ce jeu de tuiles. Elle est ``hiérarchique'', et possède une infinité de ``rangs''. La simulation de la machine de Turing a lieu à chaque ``rang'' de la hiérarchie. La machine de Turing est simulée sur un espace et durant un temps fini ; cet espace et ce temps sont les mêmes à tous les rangs de la hiérarchie. Nous modifions alors ce jeu de tuiles dans la troisième sous-partie pour obtenir une construction plus flexible ; une machine de Turing y est toujours simulée à chaque rang de la hiérarchie, mais cette fois l'espace et le temps disponibles croissent avec le rang de la hiérarchie. Ces trois premières sous-parties sont directement basées sur la Partie~2 de \cite{point-fixe}, à quelques modifications près (celles-ci nous permettront d'adapter la construction pour obtenir la construction finale dans la Partie~\ref{s:preuve}).

Nous obtenons ainsi, à la fin de cette partie, le ``squelette'' pour la construction principale de la Partie~\ref{s:preuve}, squelette que nous enrichirons pour prouver le résultat principal de ce Chapitre~\ref{c:sofique}.
Notre construction (basée sur la simulation d'une machine de Turing) garantit que chaque pavage valide possède une structure hiérarchique. Par la suite nous utiliserons cette structure pour encoder une ``charge utile''. Plus précisément, notre construction implique que chaque pavage valide contient la représentation du diagramme espace-temps d'un automate cellulaire. Dans la Partie~\ref{s:shift-auto-similaire}, cela peut paraître inutile : cet automate n'est pas utilisé par les mécanismes garantissant l'auto-simulation, et en fait cet automate n'effectue aucun calcul. Nous définirons plus précisément cet automate dans la Partie~\ref{s:preuve} où il joue un rôle important dans la preuve de notre résultat principal. Dans la Partie~\ref{s:flots} nous démontrons un résultat sur les flots qui nous garantira que le shift défini n'est pas vide.

\subsection{Simuler un shift quelconque}\label{ss:simulation-shift}

Nous commençons par montrer comment un jeu de tuiles $\tau$ peut simuler un autre jeu de tuiles $\rho$.

\begin{definition}\label{d:simule-jeu-tuiles}
On dit qu'un jeu de tuiles $\tau$ simule un jeu de tuiles $\rho$ pour un grossissement $N > 1$ si : 
\begin{itemize}
	\item il existe ${\cal L}$, un sous-ensemble des motifs globalement admissibles de $\tau$ de taille $N\times N$. Ces motifs carrés sont appelés les super-tuiles de $\tau$. Pour une super-tuile $T$, la super-couleur d'un bord de $T$ (gauche, droit, haut ou bas) est la concaténation des couleurs du bord des tuiles situées sur ce bord de $T$ ;
	\item il existe $\phi$ une fonction injective de ${\cal L}$ dans $\rho$ ;
	\item pour tout pavage ${\cal C}'$ de $\tau$, il existe une unique manière de partitionner ${\cal C}'$ en une grille de carrés de ${\cal L}$. En appliquant $\phi$ sur chacun des carrés de la grille, nous obtenons un pavage ${\cal C}$ de $\rho$. On note alors $\phi({\cal C'})={\cal C}$ ;
	\item pour tout pavage ${\cal C}$ de $\rho$, il existe un pavage ${\cal C'}$ de $\tau$ tel que $\phi({\cal C'})={\cal C}$. 
\end{itemize}
\end{definition}

\begin{remark}
Si une tuile de $\rho$ apparaît au moins une fois dans un pavage de $\rho$, alors une unique super-tuile de ${\cal L}$ lui est associée (d'après les points 2 et 4 de la définition). Autrement dit, $\phi$ est une bijection entre ${\cal L}$ et les tuiles de $\rho$ apparaissant au moins une fois dans au moins un pavage de $\rho$.
Si deux tuiles de $\rho$ peuvent être placées l'une à côté de l'autre (l'une à gauche de l'autre, ou l'une au-dessus de l'autre), i.e., elles ont la même couleur sur leur bord en contact (les bords gauche /droite ou ceux haut/bas), alors les super-tuiles correspondantes (si elles existent) doivent aussi avoir la même super-couleur sur leur bord en contact.
\end{remark}

\begin{proposition}
Soit $\rho$ un jeu de tuiles, et $N$ un entier suffisamment grand. Alors il existe un jeu de tuiles $\tau$ qui simule $\rho$ avec un grossissement $N$.
\end{proposition}

\begin{proof}
Nous allons décrire une procédure générale permettant, étant donné un jeu de tuiles $\rho$ et un entier suffisamment grand $N$, de construire un jeu de tuiles $\tau$ simulant $\rho$ avec un grossissement $N$. Le nombre de tuiles de $\tau$ est en $O(N^2)$, et la constante dans le grand $O$ ne dépend pas de $\rho$ ; par contre, la valeur de $N$ à partir de laquelle la construction de $\tau$ est possible dépend de $\rho$.

Soient $\rho$ un jeu de tuiles et $N$ un entier suffisamment grand. Nous supposons que les couleurs des bords des tuiles de $\rho$ sont des chaînes binaires de $q'$ bits (i.e., l'ensemble des couleurs $C\subset \{0,1\}^{q'}$) et que le jeu de tuiles $\rho \subset C^4$ peut être représenté par un prédicat $P(c_1,c_2,c_3,c_4)$ (le prédicat est vrai si et seulement si le quadruplet $(c_1,c_2,c_3,c_4)$ correspond à une tuile de $\rho$).

Nous supposons également que nous avons une machine de Turing $M$ qui calcule $P$ (cela peut paraître excessif d'utiliser une machine de Turing pour calculer un prédicat sur un domaine fini, mais nous verrons que cette façon de faire se révélera utile par la suite). Plus précisément, étant donné $4q'$ bits écrits à la suite, $M$ est capable de déterminer s'ils correspondent à 4 couleurs d'une tuile de $\rho$. Pour cela nous supposons, quitte à augmenter un peu la valeur de $q'$, qu'une machine de Turing parcourant ces $4q'$ bits peut déterminer où commence et où finit chacune des 4 parties correspondant à une couleur de $\rho$. Nous allons construire en parallèle un jeu de tuiles $\tau$, simulant $\rho$ avec un facteur de grossissement $N$, et un ensemble ${\cal L}$ de super-tuiles de $\tau$ de taille $N\times N$ correspondant à la définition de la simulation, voir Définition~\ref{d:simule-jeu-tuiles}. 

\begin{figure}[H]
	\centering
	\includegraphics[scale=1]{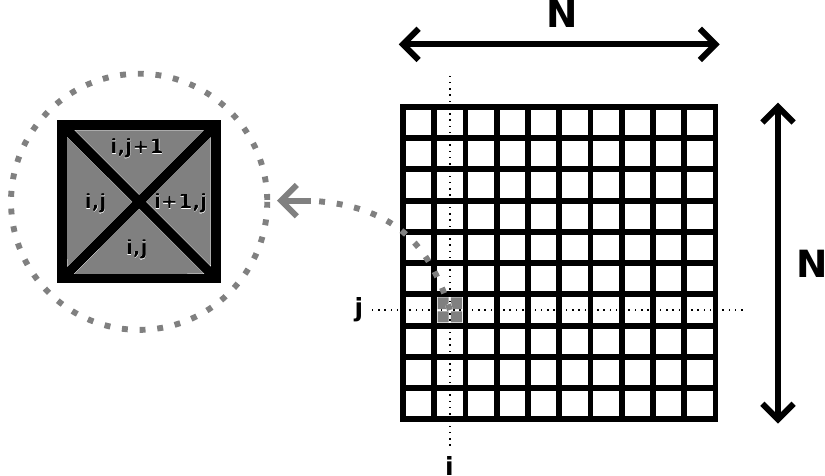}
	\caption{Détail d'une tuile à la position $(i,j)$.}\label{f:tuile}
\end{figure}

Nous voulons que dans un pavage de $\tau$, chaque tuile ``connaisse'' ses coordonnées modulo $N$, i.e. qu'à chaque tuile soit associée une paire d'entier de l'intervalle $[0,N-1]$. En pratique, cette information est encodée au niveau des couleurs des bords des tuiles. Plus précisément, pour une tuile de coordonnée $(i,j)\mod{N}$, les couleurs des bords gauche et bas doivent encoder $(i,j)$, celle du bord droit $(i+1 \mod{N},j)$ et celle du haut $(i,j+1 \mod{N})$, voir Fig.~\ref{f:tuile}. Cela nous garantit qu'un pavage de $\tau$ peut être découpé de manière unique en blocs (en super-tuiles) de taille $N\times N$, dans lesquels les tuiles en bas à gauche et en haut à droite ont pour coordonnées respectives $(0,0)$ et $(N-1,N-1)$. Intuitivement, chaque tuile ``connaît'' sa position au sein de la super-tuile.
Pour les tuiles situées sur les bords d'une super-tuile (i.e., une de leurs coordonnées au moins est égale à $0$ ou $N-1$), nous leur ajoutons une valeur binaire supplémentaire au niveau de leur couleur située au bord de la super-tuile. Pour une super-tuile de $\tau$ (de taille $N\times N$), les 4 super-couleurs correspondantes peuvent donc être représentées comme 4 chaînes binaires de taille $N$.

Une première différence entre la construction de \cite{point-fixe} et la nôtre est que dans \cite{point-fixe} le nombre $N$ de bits encodés dans une super-couleur est excessivement large, et seul un nombre $q\ll N$ de ces bits sont utilisés (plus précisément, $q=O(\log N)$). 
Au contraire, dans notre construction, la situation est différente ; nous voulons que le nombre de bits $q$ encodés dans une super-couleur soit important, même comparé à $N$. Techniquement, nous choisissons $q=\frac{N}{16}$ (nous verrons que ce choix de $q$ sera adéquat pour notre construction). Une petite partie de ces bits, soit $q'$ bits, correspond aux symboles ``réellement'' utilisés dans la simulation de $\rho$. L'écart entre $q$ et $q'$ correspond à des bits qui sont encore inutilisés à cette étape, que nous réservons pour des ajustements futurs de notre construction. Nous supposons que les super-couleurs sont encodées de telle manière qu'une machine de Turing puisse, lorsqu'elle parcourt les $q$ bits, détecter où la partie composée des $q'$ bits servant à simuler $\rho$ commence et finit.

Pourquoi avons-nous besoin d'encoder autant de bits``inutilisés'' dans les super-couleurs ?
Intuitivement, nous aurons besoin dans la Partie~\ref{s:preuve} que deux super-tuiles voisines puissent s'échanger, via leurs super-couleurs, une quantité d'information encodée par un nombre de bits égal à une fraction de $N$. Nous préférons donc expliquer, dès maintenant, la ``géométrie'' d'une super-tuile de $\tau$, en prévoyant déjà la possibilité d'encoder des bits supplémentaires dans ses super-couleurs. La fraction exacte choisie devra être suffisamment petite et nous la préciserons plus tard, dans la Partie~\ref{s:preuve}.

Les bits utilisés pour la simulation de $\rho$ sont situés en bas pour les bords gauche et droit d'une super-tuile ; pour les bords haut et bas, ils sont situés à une distance $q$ du bord gauche. Les $q'$ bits effectivement utilisés sont situés, pour chacun des bords, au début des $q$ bits réservés (en partant du haut ou de la gauche de ces bits). Les $N-q$ bits des super-couleurs restants sont ``triviaux'' : ils sont, et resteront, ``inutilisés''. Formellement, nous leur affectons la valeur 0, voir Fig.~\ref{f:tuile-bords}.

\begin{figure}[H]
	\centering
	\includegraphics[scale=1]{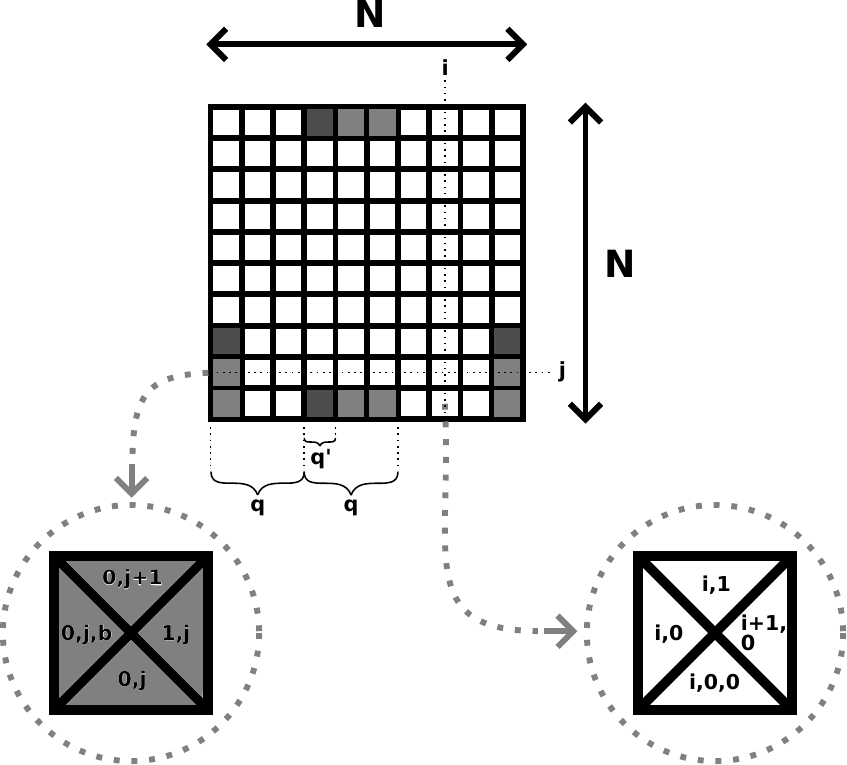}
	\caption[Les bits encodés sur les bords d'une super-tuile.]{Nous voulons qu'à chaque tuile sur les bords d'une super-tuile de taille $N\times N$ soit associé un bits supplémentaire. Par exemple, au bord gauche de la $j$-ème tuile du bord gauche d'une super-tuile est associée un triplet $(0,j,b)$, où la paire $(0,j)$ représente les coordonnées de la tuile dans la super-tuile, et $b$ le bit de la tuile. Ainsi, chaque super-couleur d'un bord d'une super-tuile peut maintenant être représentée par une chaîne binaire de taille $N$, encodée dans les couleurs des bords des tuiles correspondantes situées sur ce bord. Nous supposons que pour chaque bord, ces bits supplémentaires ne sont pas triviaux pour seulement $q$ tuiles, $q$ correspondant à une fraction constante de $N$. Parmi ces $q$ tuiles, les $q'$ premières sont utilisées dans la simulation de $\rho$. Les $q$ tuiles sont montrées en gris clair (i.e. elles encodent un bit non trivial), excepté pour les $q'$ premières tuiles montrées en gris foncé (i.e. elles encodent un bit qui est utilisé pour simuler $\rho$).}\label{f:tuile-bords}
\end{figure}

Les autres tuiles du bord de la super-tuile ne contribuent pas à la super-couleur -- leur bit est toujours égal à 0. Ainsi, une super-couleur est en réalité déterminée par une chaîne binaire de taille $q$.

Les bits représentés au niveau des super-couleurs sont ``transmis'' au sein de la super-tuile. Dans ce qui suit nous détaillons comment cette information est traitée au sein d'une super-tuile. Nous pouvons noter que, de même que dans la construction de \cite{point-fixe}, $q'$ est bien plus petit que $N$ (i.e., $q'=O(\log N)$). Même si le nombre $q$ de bits ``transmis'' au sein de la super-tuile est important (égal à une fraction de $N$), le nombre $q'$ de bits à traiter pour simuler $\rho$ est petit.  

Ajoutons des contraintes supplémentaires sur le jeu de tuiles $\tau$ pour garantir que les super-couleurs des super-tuiles ``simulent'' bien le prédicat $P$. Pour cela, assurons-nous que les bits des tuiles dont le rôle est d'encoder les super-couleurs, situés sur les bords de la super-tuile, sont transférés jusqu'à la zone du milieu de la super-tuile. Plus précisément, nous voulons que les 4 chaînes binaires encodées par les 4 super-couleurs de la super-tuile soient présentes au niveau de la  ligne du bas de la zone du milieu. Nous nous assurons ensuite que, dans cette zone du milieu, est vérifié que pour ces 4 chaînes binaires le prédicat $P$ est vrai. La taille totale de ces chaînes binaires est $4q$ ; néanmoins, pour vérifier la véracité du prédicat $P$ nous n'avons qu'à considérer $4q'$ bits. Nous fixons quelles tuiles dans la super-tuile jouent le rôle de ``câbles de communication'', et plus précisément quelle sorte de cellule du câble de communication (horizontale, verticale, ou un des 4 angles droits) la tuile joue. En particulier, les $4q$ tuiles impliquées dans le codage effectif d'une super-couleur sont des cellules d'un câble de communication. 

Nous demandons que pour ces tuiles jouant le rôle de câble de communication, les couleurs de \emph{deux} bords (le choix de ces bords étant déterminé par le type de la cellule du câble joué par la tuile) encodent le même bit. Par exemple, pour une tuile jouant le rôle d'une cellule horizontale d'un câble de transmission, il s'agit des couleurs des bords gauche et droit de la tuile. Dans ce cas, le bit encodé dans le bord droit de la tuile voisine de gauche est égal au bit encodé dans le bord gauche de la tuile voisine de droite, voir Fig.\ref{f:tuile-complète}.

\begin{figure}[H]
	\centering
	\includegraphics[scale=0.38]{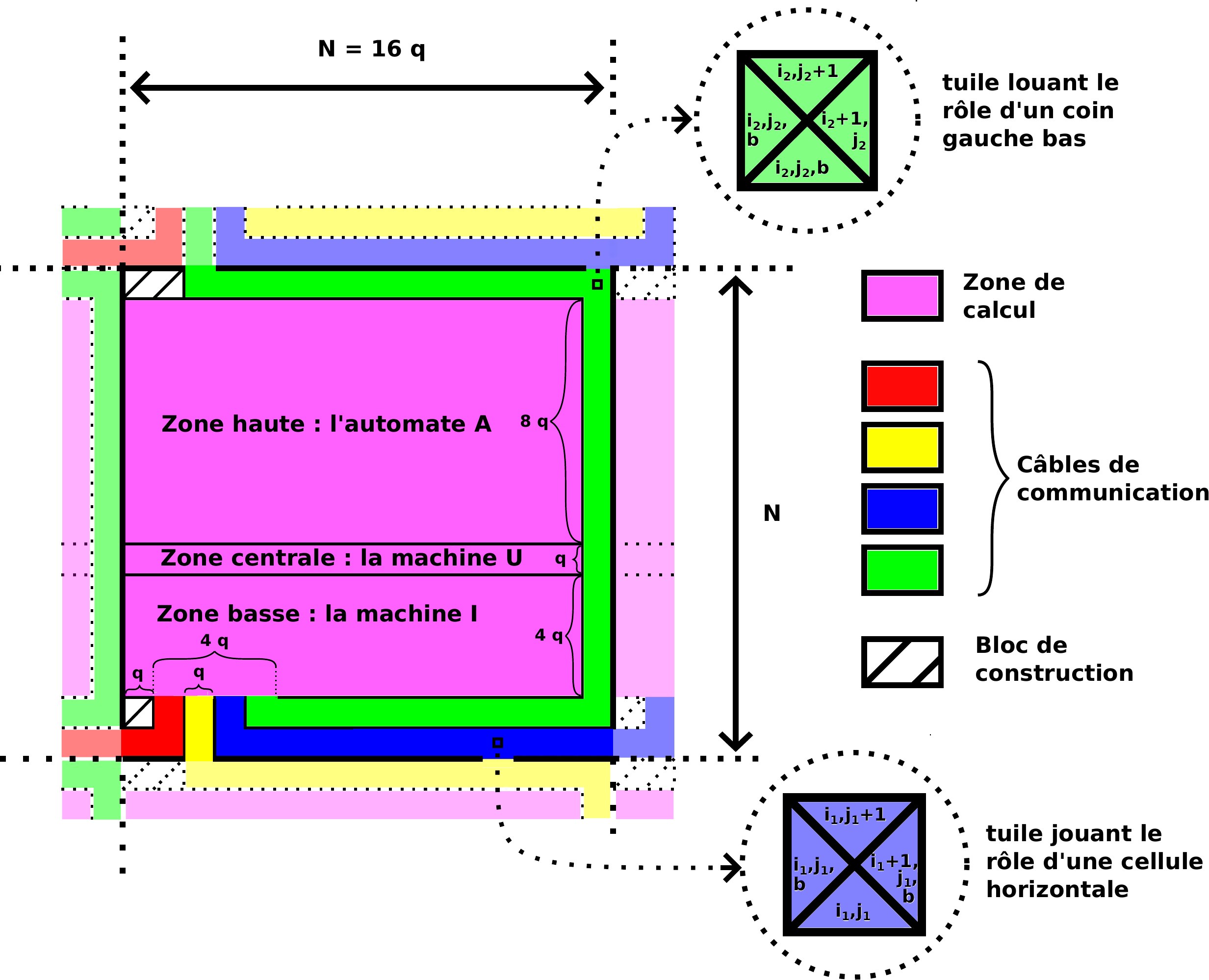}
	\caption[Les câbles de communication et les différentes parties de la zone de calcul d'une super-tuile.]{Au niveau de chacun des 4 bords d'une super-tuile nous encodons dans les bords de $q$ tuiles les $q$ bits représentant une super-couleur. Ces tuiles font partie des ``câbles de communication'', i.e. les tuiles chargées de transférer ces données jusqu'à la zone du milieu. Puisque $q$, la largeur des câbles rouge, vert, bleu et gris dans la figure, est égale à une fraction de $N$, la manière dont nous disposons ces câbles doit être plus précise que dans la construction de \cite{point-fixe}. Nous les organisons comme montré dans la figure. Par ailleurs, nous détaillons les couleurs de deux bords de deux tuiles impliquées dans les ``câbles de communication''. La tuile de coordonnées $(i_1,j_1)$ appartient au câble de communication bleu, chargé du transfert de la super-couleur de gauche jusqu'à la zone du milieu. Le rôle de cette tuile est celui d'une cellule horizontale d'un câble de communication : les bits de ses bords gauche et droit sont identiques, et elle ``transfère'' le bit du bord droit de sa voisine de gauche au bord gauche de celle de droite. Pareillement, la tuile de coordonnées $(i_2,j_2)$ appartient au câble de communication vert, qui est chargé de transférer la super-couleur du haut jusqu'à la zone du milieu. Le rôle de cette tuile est celui d'un ``coin gauche bas'' d'un câble de communication. En tant que tuile jouant ce rôle, le bit encodé dans son bord gauche est égal à celui encodé dans son bord bas : elle ``transfère'' donc le bit présent sur le bord droit de sa voisine de gauche jusqu'au bord haut de sa voisine du bas. Nous pouvons noter que les coordonnées et le bit d'une tuile sont encodés au niveau des \emph{couleurs} des bords de celle-ci. Les couleurs indiquées à l'intérieur d'une tuile (rouge, vert, bleu ou gris dans la figure) ne sont là qu'à titre illustratif pour montrer le rôle que joue la tuile dans la super-tuile. Nous pouvons également noter que bien que la notion de ``transfert'' corresponde bien à l'intuition, l'information véhiculée par les câbles de communication n'est pas réellement ``orientée'' : en réalité, les câbles nous garantissent juste que les bits composant les super-couleurs sont les mêmes que ceux présents au niveau de la ligne du bas de la zone du milieu.}\label{f:tuile-complète}
\end{figure}

La zone du milieu de la super-tuile, que nous appelons ``la zone de calcul'', est de taille $l\times h$, avec $l=N-q=\frac{15N}{16}$ et $h=N-3q=\frac{13N}{16}$ : c'est dans celle-ci qu'est vérifié que les couleurs de $\rho$ représentées par les 4 super-couleurs correspondent bien à une tuile de $\rho$ (i.e. que le prédicat $P$ est valide pour les 4 chaînes binaires encodées dans les super-couleurs). La communication entre cette zone de calcul et le ``reste du monde'' est restreinte à sa ligne du bas. Il s'agit même plus précisément d'une partie de cette ligne composée de $4q$ tuiles, dont les bits doivent être cohérents avec les 4 super-couleurs de la super-tuile. Nous pouvons noter qu'il existe pour l'instant également $q$ tuiles inutilisées au début de la première ligne de la zone de calcul.

Dans la construction de \cite{point-fixe}, la zone du milieu de la super-tuile représente juste le diagramme espace-temps de la machine de Turing $M$ (le ruban est horizontal, et le temps va vers le haut). Cette machine $M$ reçoit en entrée 4 chaînes binaires de taille $O(\log N)$, correspondant aux 4 super-couleurs de la super-tuile.

Dans notre construction, nous subdivisons les calculs simulés dans la \emph{zone de calcul} en trois parties ; la zone de calcul est coupée horizontalement deux fois, pour obtenir une zone ``basse'', une zone ``centrale'' et une zone ``haute'', voir Fig.\ref{f:tuile-complète}.

La première étape, ayant lieu dans la zone basse, est réalisée par la machine de Turing $I$. Son rôle est très simple : elle se contente de rassembler à un endroit du ruban les parties des données d'entrée utilisées pour la simulation de $\rho$. Cependant, la tâche n'est pas complètement triviale ; la difficulté réside en la collecte de petits blocs de données qui sont au départ éparpillés sur le ruban, et fortement éloignés les uns des autres. Nous devons réaliser cette tâche efficacement, en temps linéaire. Par conséquent, nous ne pouvons pas ``transporter'' les bits d'information utile un par un, en parcourant à chaque fois une partie importante des données d'entrée. La tâche ne peut être réalisée suffisamment rapidement seulement si nous utilisons une machine de Turing à plusieurs têtes de lecture ; c'est pourquoi la machine $I$ possède deux têtes de lecture. Nous pouvons noter que les $q'$ bits utilisés pour la simulation sont situés au début des $q$ bits correspondant aux super-couleurs de gauche et du bas, et à la fin pour les super-couleurs de droite et du haut. 

La seconde étape, dans la partie centrale, est réalisée par la machine $M$ (que nous avons présentée ci-dessus). Cette machine réalise l'essentiel des calculs complexes et non-triviaux. Il est important de noter que pour mener ceux-ci à bien elle n'a à considérer qu'une très petite partie des données d'entrée. Par conséquent, nous n'avons pas besoin de spécifier dans le détail l'algorithme de la machine $M$ ; il est suffisant de dire qu'elle effectue des opérations classiques (i.e., des opérations arithmétiques sur des entiers naturels) en temps polynomial.

La troisième étape, dans la partie haute, sera spécifiée dans la Partie~\ref{s:preuve} et contiendra des opérations qui requerront une utilisation massive du parallélisme. Dans ce cas, une machine de Turing à plusieurs têtes de lecture ne sera pas un modèle de calcul adéquat, et nous utiliserons un automate cellulaire non déterministe $A$.

Détaillons maintenant l'étape ``d'initialisation'' réalisée par la machine $I$. Elle consiste à ``récupérer'', parmi l'ensemble des $4q$ bits présents au niveau de la première ligne de la zone de calcul (représentant les 4 super-couleurs), les $4q'$ bits représentant 4 couleurs de $\rho$. La machine de Turing $I$ a deux têtes de lecture et un ruban, les symboles du ruban étant des bits. 

La première tête de lecture $t_1$ de $I$ commence au début de la zone de calcul, et la deuxième tête de lecture $t_2$ commence au début des $4q$ bits encodant les 4 super-couleurs. La seconde tête de lecture va parcourir l'intégralité de ces $4q$ bits, en un nombre $4q$ d'étapes de calcul. L'encodage de ces $4q$ bits lui permet de savoir quand elle est arrivée au début et à la fin d'un bloc de bits, de taille $q'$, utilisé dans la simulation de $\rho$. Quand elle passe au-dessus d'un tel bloc, la première tête de lecture recopie ce bloc de bits sur le ruban (un bit à la fois). Une fois cette tâche effectuée, le ruban contient $4q'$ bits à la suite, situés au début de la première ligne de la zone de calcul (voir Fig.\ref{f:tuile-zone-basse}. Ils représentent les parties des 4 super-couleurs utilisées pour simuler $\rho$. Ces bits correspondent exactement aux données d'entrée de la zone de calcul de la construction de \cite{point-fixe} ; nous pouvons donc exécuter la machine $M$ sur ceux-ci (i.e., la partie centrale de la zone de calcul représente le diagramme espace-temps de $M$). Nous pouvons noter que la machine $I$ ne dépend ni de $N$ ni de $\rho$. 

\begin{figure}[H]
	\center
	\includegraphics[scale=0.7]{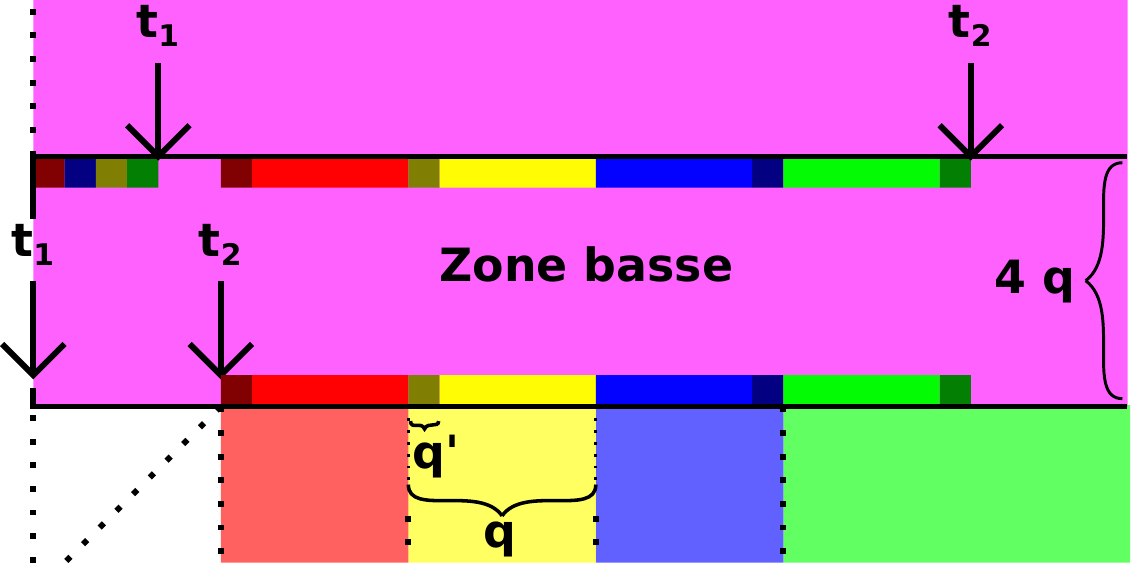}
	\caption{Regroupement des bits utilisés pour la simulation du shift au niveau de la zone basse de la zone de calcul.}\label{f:tuile-zone-basse}
\end{figure}

\begin{remark}
La machine $I$ se termine après $4q=\frac{N}{4}$ étapes de calcul. Ainsi, nous connaissons précisément la taille de la zone basse de la zone de calcul, qui est de $\frac{15N}{16} \times \frac{N}{4}$.
\end{remark}

Intéressons-nous maintenant à la partie centrale, où s'exécute la machine $M$. Nous devons nous assurer que nous avons assez de temps et d'espace dans la zone de calcul, une fois exécutée la machine $I$, pour que la machine $M$ puisse mener ses calculs à bien. De plus, nous voulons garder un espace suffisant pour la partie haute. Nous choisissons d'allouer un temps $q=\frac{N}{16}$ à $M$ ; il faut donc que tous les calculs de $M$ sur des entrées valides se terminent en au plus $q$ étapes de calcul, en ayant utilisé un espace d'au plus $N-q$ cases du ruban. Une fois que la machine $M$ s'est arrêtée, le ruban n'est plus modifié lors des éventuelles étapes restantes de la zone centrale (et $M$ reste dans son état final). L'entier $q$ étant une fraction de $N$, il est nécessaire pour que cette construction fonctionne que l'entier $N$ (i.e. la taille d'une super-tuile) soit suffisamment grand. 

Nous avons vu dans la Partie~\ref{s:automate} que nous pouvons simuler une machine de Turing classique par un jeu de tuiles. Cependant, pour l'instant la quantité d'information supplémentaire à encoder dans les couleurs des tuiles de la zone centrale dépend du choix de la machine $M$. Nous voulons éviter cette dépendance, et comme dans \cite{point-fixe} nous remplaçons $M$ par une machine de Turing universelle $U$ fixée, qui exécute un programme simulant $M$. Nous choisissons $U$ telle que $U$ simule $M$ de manière efficace, i.e. un calcul effectué par $M$ en temps polynomial doit être simulé par $U$ également en temps polynomial (le polynôme peut être différent). 

L'article~\cite{woods2009complexity} nous garantit l'existence d'une machine de Turing universelle a un ruban, pouvant simuler n'importe quelle machine de Turing à un ruban, avec une perte en temps polynomial. Nous pouvons noter que pour les machines de Turing à plusieurs rubans, un résultat plus fort existe : l'article~\cite{hennie1966two} montre qu'il existe une machine de Turing universelle à deux rubans pouvant simuler n'importe quelle machine de Turing à plusieurs rubans, avec une perte en temps seulement logarithmique. Par ailleurs, une machine de Turing à un ruban peut simuler une machine de Turing à plusieurs rubans avec une perte en temps quadratique (voir~\cite{hopcroft2001introduction}). Avec ce résultat, combiné avec l'article de~\cite{woods2009complexity}, nous obtenons une machine de Turing universelle à un ruban pouvant simuler n'importe quelle machine de Turing multi-rubans, avec une perte en temps polynomiale.   

Par ailleurs, nous voulons que la machine $U$ puisse déterminer où se termine le code de la machine qu'elle simule.

En pratique, nous allons encoder le code de la machine $M$ (pour la machine universelle $U$) au début du ruban de la machine $I$ (au niveau des $q$ tuiles réservées situées sur la ligne du bas de la zone de calcul).

La machine de Turing $I$ a donc deux ``champs'' en entrée : (i) le code de $M$, soit une chaîne de caractères de l'alphabet de $U$, situé au début des données d'entrée, et (ii) $4q$ bits, représentant les 4 super-couleurs, situés après les $q$ premières tuiles de la zone de calcul. Nous pouvons noter que contrairement aux tuiles du champ (ii), qui sont déterminées par les super-couleurs de la super-tuile, les tuiles du champ (i) sont uniquement déterminées par leurs coordonnées, i.e. les tuiles de $\tau$ dont les coordonnées correspondent au champ (i) n'existent qu'en un seul exemplaire (autrement dit, la lettre qui leur est associée est déterminée par leurs coordonnées).

Nous ajustons le code de $I$ pour prendre en compte ce nouveau champ : la première tête de lecture (recopiant les $4q'$ bits correspondant aux 4 couleurs de $\rho$) est désormais initialisée juste après le dernier symbole du premier champ, et non plus au début de la première ligne de la zone de calcul ; la deuxième tête de lecture est toujours initialisée au début du champ (ii), et les deux têtes de lectures se comportent toujours de la même manière. 
Quand la machine $I$ se termine, les $q$ premières cellules du ruban sont composées de deux champs : (i) le code de $M$ et (iii) la représentation des 4 couleurs de $\rho$. Nous pouvons noter que le premier champ n'est pas modifié par $I$. Ces deux champs constituent les données d'entrée de la machine universelle $U$. La tête de lecture de $U$ est restreinte aux $q$ premières tuiles des lignes de la zone de calcul ; en particulier, $U$ n'a pas accès aux $4q$ bits représentant les 4 super-couleurs. Gardons à l'esprit que nous avons choisi une machine universelle $U$ et un encodage de l'information dans les $4q'$ bits de telle manière qu'une machine de Turing puisse déterminer où commencent et où finissent le code de $M$ (le champ (i)) et chacune des 4 parties (de taille $q'$) des $4q'$ bits.  
Une autre restriction imposée à la machine $U$ concerne le premier champ (i) : celui-ci pourra être lu, mais pas modifié. Ainsi, les valeurs du ruban correspondant à ce champ ne changeront jamais durant les calculs effectués par $U$ (donc les cases d'une colonne du diagramme espace-temps de $U$ correspondant à ce champ partagent une même lettre, qui reste inchangée lors de la totalité des étapes de calcul de $U$).

Ainsi, dans la zone centrale dédiée à $U$ d'une super-tuile, est représenté un calcul se terminant dans un état final acceptant, pour le programme $M$. Par conséquent, nous supposons désormais que les données d'entrée du programme simulé $M$ sont constituées de 4 chaînes binaires encodant 4 couleurs de $\rho$, cohérentes avec les 4 super-couleurs de la super-tuile (cette information est ``transférée'' depuis les bords de la super-tuile jusqu'à la ligne du bas de la zone de calcul par les câbles de communication, puis ``triée'' par la machine $I$ dans la zone basse de la zone de calcul).

Maintenant que nous avons décrit les machine $I$ et $U$, nous utilisons la construction de la Partie~\ref{s:automate} pour les machines de Turing à deux têtes de lecture et pour les machines de Turing classiques pour obtenir deux jeux de tuiles $i$ et $u$ simulant les machines $I$ et $U$. Comme les tuiles de $\tau$ ``connaissent'' leurs coordonnées, à chaque tuile de $i$ et de $u$ doivent correspondre plusieurs tuiles de $\tau$. Ainsi, chaque tuile de la zone basse (excepté la première ligne) peut à priori correspondre à n'importe quelle tuile du jeu de tuiles $i$, et nous devons ajouter à $\tau$, pour chaque tuile de $i$, exactement $(\frac{15N}{16}) \times (\frac{N}{4}-1)$ tuiles. Autrement dit, dans les couleurs de ces tuiles sont encodées des coordonnées et une quantité d'information garantissant qu'elles jouent le rôle d'une cellule du diagramme espace-temps de $I$. 

Soit $k$ la taille du code de la machine $M$ pour $U$. Nous ajoutons également aux tuiles de $\tau$ les tuiles dont les positions correspondent à la ligne du bas de la zone de calcul (de longueur $N-q$). Pour les $q$ premières positions de la ligne, les tuiles correspondantes n'existent qu'en un seul exemplaire (les $k$ premières encodent les lettres du code de la machine $M$, et les $q-k$ suivantes le bit de valeur nulle, correspondant à une case vierge du ruban de $I$). Les $4q$ positions suivantes correspondent aux tuiles d'``arrivée'' des câbles de transmission (i.e., leur tuile voisine du bas est une cellule d'un câble de transmission), et deux tuiles doivent être ajoutées à $\tau$ par position (encodant la valeur 0 ou 1 du bit supplémentaire). Pour les $N-6q$ positions restantes, une seule tuile est ajoutée à $\tau$ (correspondant également à une case vierge du ruban). En particulier, l'unique tuile correspondant à la $k+1$-ème position de la ligne du bas contient la première tête de lecture de $I$, et les deux tuiles correspondant à la $q+1$-ème position contiennent la deuxième tête de lecture de $I$. Nous ajoutons donc à $\tau$ un nombre de tuiles égal à $\frac{N}{16}+2\frac{N}{4}+\frac{10N}{16}$. 

De même, la zone centrale étant de taille $\frac{15N}{16}\times \frac{N}{16}$, chaque tuile du jeu de tuiles $u$ correspond à $\frac{15N}{16} \cdot (\frac{N}{16}-1)$ tuiles de $\tau$ (les couleurs des bords de ces tuiles encodent des coordonnées et une information supplémentaire garantissant qu'elles jouent le rôle d'une cellule du diagramme espace-temps de $U$). De nouveau, il existe $\frac{15N}{16}$ tuiles particulières permettant l'interface entre la zone basse et la zone centrale. Nous pouvons associer à chacune de celles-ci deux tuiles de $\tau$ (encodant le bit de valeur 0 ou 1). En particulier, les deux tuiles de $\tau$ correspondant à la première position de cette ligne contiennent la tête de lecture de $U$.

Les tuiles restantes de la zone de calcul sont dans la partie haute (de taille $(N-8q)\times (N-q)$), et jouent le rôle de cellule du diagramme espace-temps de l'automate cellulaire $A$. Nous avons vu comment simuler un automate cellulaire par un jeu de tuiles dans la Partie~\ref{s:automate}. Pour l'instant, l'automate cellulaire $A$ de la zone haute est trivial et n'effectue aucun calcul, les différentes lignes de cette zone sont donc identiques. Nous préciserons le code de l'automate $A$ dans la Partie~\ref{s:preuve}. Là encore, des tuiles spéciales permettent l'interface entre la zone centrale et la zone haute (comme pour l'interface entre la zone basse et la zone haute, les tuiles dont la position correspond à cette ligne existent en deux exemplaires, selon le bit qu'elles encodent).

Les tuiles de $\tau$ peuvent donc être un des 6 types possibles de cellules des câbles de communication (et éventuellement participer au codage d'une super-couleur), une cellule du diagramme espace-temps de $I$, de $U$ ou de $A$, ou n'avoir aucun rôle particulier (nous disons qu'elles jouent le rôle de ``bloc de construction''). Le rôle d'une tuile est entièrement déterminé par ses coordonnées.

Ceci termine la description de la construction explicite d'un jeu de tuiles $\tau$, possédant $O(N^2)$ tuiles, et simulant $\rho$. Cette construction marche pour tous les $N$ suffisamment grands. Nous pouvons noter que la plupart des étapes de la construction de $\tau$ ne dépendent pas du programme simulé $M$ : le jeu de tuiles $\tau$ dépend seulement de $M$ au niveau de la zone basse de la zone de calcul. Néanmoins, cette dépendance est très limitée. Le programme simulé (et implicitement, le prédicat $P$) affecte seulement les règles pour les tuiles utilisées au niveau de la ligne du bas de la zone de calcul (précisément les tuiles du premier champ (i) de la ligne, où est encodé $M$). Les couleurs des bords de toutes les autres tuiles sont génériques, et ne dépendent pas du jeu de tuiles simulé $\rho$.

Le caractère explicite de la construction peut être compris formellement de la façon suivante : il existe un algorithme qui prend en entrée deux entiers, $N$ et $q'$, et le code d'une machine de Turing $M$ pour la machine universelle $U$ (calculant le prédicat $P$, qui décrit le jeu de tuiles $\rho$). Cet algorithme retourne la liste des tuiles du jeu de tuiles $\tau$ décrit ci-dessus, simulant $\rho$ avec un grossissement $N$. L'algorithme s'arrête avec un message d'erreur approprié si les entiers $N$ ou $q'$ sont trop petits pour notre construction. De plus, pour tout quadruplet de couleurs des bords de $\tau$, nous pouvons vérifier en temps polynomial (en temps $\poly(\log N)$) si ce quadruplet correspond à une tuile de $\tau$ ou pas.
\end{proof}

\subsection{Un shift auto-similaire élémentaire}

Nous avons vu dans la sous-partie précédente comment un jeu de tuiles $\tau$ pouvait simuler un autre jeu de tuiles $\rho$. Dans cette partie, nous cherchons à ce que le jeu de tuiles simulé soit le même que le jeu de tuiles qui le simule (i.e., que $\tau = \rho$).

\begin{definition}
Si un jeu de tuiles $\tau$ se simule lui-même (selon la Définition~\ref{d:simule-jeu-tuiles}) pour un grossissement $N$, alors $\tau$ est appelé un jeu de tuiles auto-similaire de grossissement $N$.
\end{definition}

\begin{remark}
Si un jeu de tuiles $\tau$ est auto-similaire de grossissement $N$, alors par définition il existe un ensemble ${\cal L} \subset \tau^{N^2}$ et une fonction injective $\phi : {\cal L} \rightarrow \tau$ tels que définis dans la Définition~\ref{d:simule-jeu-tuiles}.

Soit ${\cal P}_1$ un pavage  de $\tau$. Il peut être découpé de manière unique en une grille de super-tuiles de $\tau$ de taille $N \times N$ ; nous disons que ces super-tuiles sont de \emph{rang} 1. En appliquant $\phi$ sur ${\cal P}_1$ (i.e., nous appliquons $\phi$ sur chacune des super-tuiles de la grille) nous obtenons un nouveau pavage ${\cal P}_2$ de $\tau$. ${\cal P}_2$ se découpe également de manière unique en une grille d'éléments de ${\cal L}$. Pour chacun de ces motifs carrés $P$ de taille $N \times N$, nous considérons l'ensemble des tuiles de ${\cal P}_1$ qui sont projetées par $\phi$ sur $P$, soit un carré de taille $N^2 \times N^2$. Nous appelons ces carrés des super-tuiles de $\tau$ de rang 2. Clairement, le pavage ${\cal P}_1$ peut se découper de manière unique en une grille de super-tuiles de rang 2. Remarquons qu'une super-tuile de rang 2 est composée de $N \times N$ super-tuiles de rang 1. Nous pouvons itérer cet argument : appliquer $\phi$ sur ${\cal P}_2$ pour obtenir un pavage ${\cal P}_3$, découper celui-ci en une grille unique de super-tuiles, considérer les tuiles de ${\cal P}_1$ envoyées par $\phi^2$ sur chacune de ces super-tuiles de ${\cal P}_3$. Les super-tuiles de rang 3 sont composées de $N \times N$ super-tuiles de rang 2. Nous obtenons ainsi une \emph{hiérarchie} de super-tuiles de $\tau$ de rangs $1, 2, \ldots$, voir Fig~\ref{f:tuiles-hiérarchie}.

\begin{figure}[H]
	\center
	\includegraphics[scale=0.1]{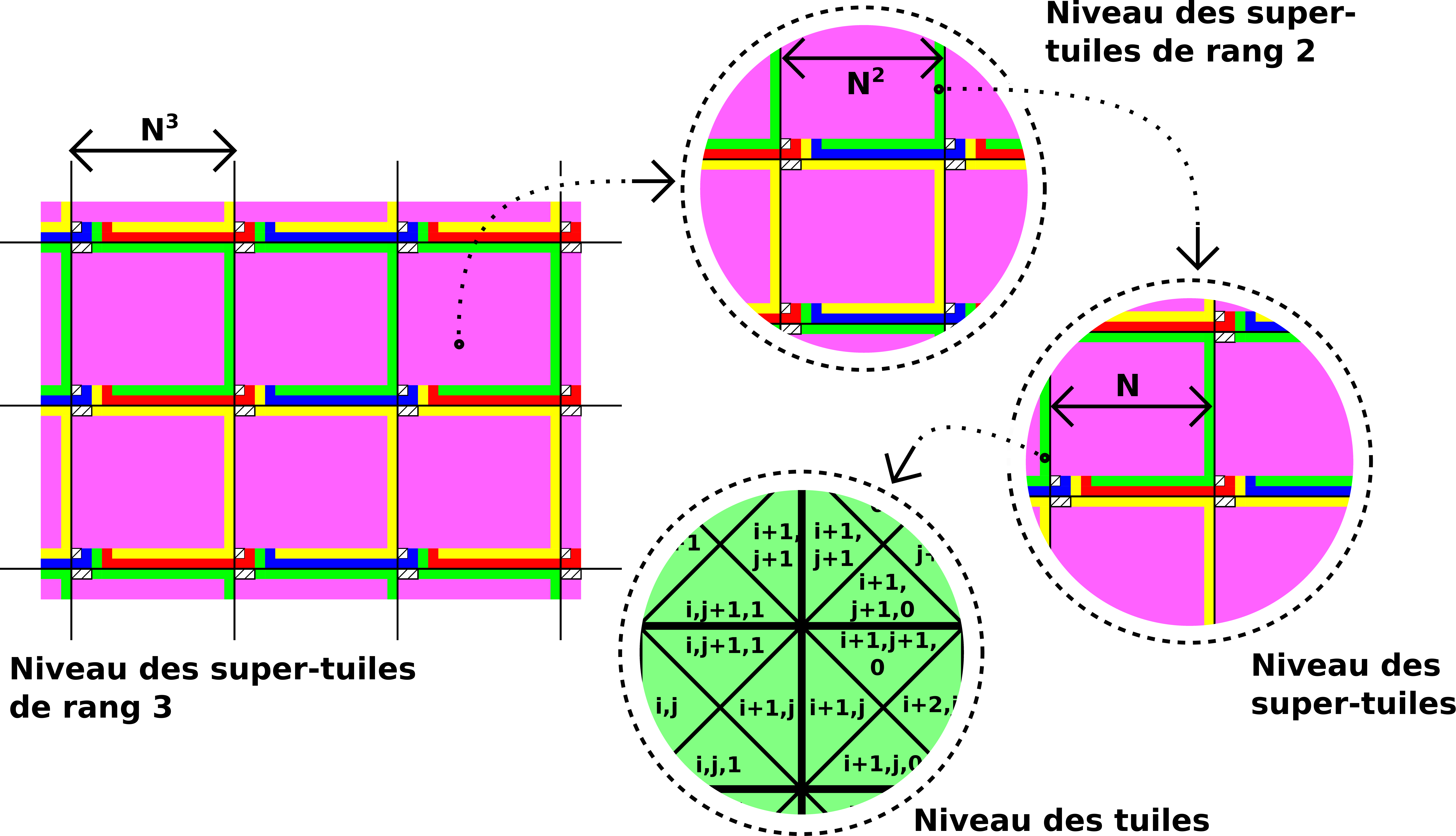}
	\caption[La structure hiérarchique des super-tuiles.]{La structure hiérarchique des super-tuiles : les super-tuiles de rang $k$ sont des blocs pour les super-tuiles de rang $k+1$, les super-tuiles de rang $k+1$ sont des blocs pour les super-tuiles de rang $k+2$, etc. À tous les rangs de la hiérarchie, la structure des super-tuiles est similaire.}\label{f:tuiles-hiérarchie}
\end{figure}

Introduisons quelques notations qui nous seront utiles par la suite. Si une super-tuile $T$ de rang $k$ appartient à une super-tuile $M$ de rang $k+1$, nous disons que $M$ est la super-tuile \emph{mère} de $T$, et que $T$ est une super-tuile \emph{fille} de $M$. Deux super-tuiles $T$ et $T'$ ayant la même super-tuile mère sont dites sœurs. Pour une super-tuile $T$ de rang $k$, les super-tuiles de rang $l>k$ auxquelles elle appartient sont appelées les ancêtres de $T$.
\end{remark}

\begin{proposition}
Pour tout $N$ suffisamment grand, il existe un jeu de tuiles $\tau$ auto-similaire de grossissement $N$.
\end{proposition}

\begin{proof}
Pour la construction de ce jeu de tuiles nous utilisons une idée venant de la preuve du théorème de récursion de Kleene. Intuitivement, nous utilisons l'idée qu'un programme peut, d'une certaine manière, accéder à son propre code et utiliser les symboles de celui-ci pour effectuer des calculs.
À la fin de la sous-partie précédente nous avons vu que le jeu de tuiles $\tau$, simulant le jeu de tuiles $\rho$, ne dépend que très peu du programme $M$ simulé par $U$ (et donc très peu de $\rho$).

Nous fixons les paramètres $q'$ et $N$ (nous supposons que $q' \ll q$, $q$ étant toujours une fraction de $N$), et appliquons partiellement la procédure décrite dans la sous-partie précédente, pour obtenir la liste $L$ des tuiles de $\tau$, exceptées celles correspondant aux $q$ premières tuiles de la ligne du bas de la zone de calcul d'une super-tuile de $\tau$. Dans notre construction, à chacune de ces $q$ positions ne correspond qu'une seule tuile de $\tau$ ; il ne ``manque'' donc que $q$ tuiles pour obtenir la totalité du jeu de tuiles $\tau$. La liste $L$ ne dépend pas du jeu de tuiles $\rho$ que simule $\tau$. Ainsi, soit un jeu de tuiles $\rho$, possédant au plus $2^{q'}$ couleurs, et pour lequel le grossissement $N$ est assez grand. Plus précisément, soit $M$ le programme correspondant au prédicat associé à $\rho$. $N$ doit être assez grand pour que d'une part nous ayons une place suffisante pour écrire le code de $M$ pour $U$ sur les $q-4q'$ premières tuiles de la première ligne de la zone de calcul. D'autre part, il faut également que nous ayons l'espace et le temps nécessaires pour que $U$ puisse se terminer sur toutes les entrées valides, i.e. les entrées correspondant à des couleurs représentant une tuile de $\rho$. Dans ce cas, il est possible de compléter les $L$ tuiles de $\tau$ par $q$ tuiles, de façon à ce que $\tau$ simule $\rho$. 

Nous allons maintenant ``oublier'' le fait que la liste $L$ peut être complétée de \emph{nombreuses} façons différentes en des jeux de tuiles $\tau$ simulant des jeux de tuiles $\rho$. Nous allons au contraire compléter la liste $L$ d'une façon \emph{spécifique}, de manière à obtenir un jeu de tuiles $\tau$ auto-similaire. Nous pouvons noter que le programme $\pi$ que simulera $U$ différera de celui décrit dans la Sous-Partie~\ref{ss:simulation-shift} précédente : alors que la machine $M$ correspondait au calcul d'un prédicat, $\pi$ correspondra à des calculs plus complexes, dépendant entre autre de $N$. Il est donc nécessaire que l'entier $N$ appartienne aux données d'entrée de $\pi$. Pour cela, nous ajoutons, après le premier champ contenant le code de $\pi$ pour $U$, un second champ aux données d'entrée de $I$, qui contient une représentation binaire de l'entier $N$ permettant à une machine de Turing de déterminer où commence et finit cette représentation ; ces deux champs sont situés au début de la première ligne de la zone de calcul. La machine $I$ ne modifie pas ce champ, et celui-ci est donc également présent au niveau des données d'entrée de la machine $U$ (elles sont situées sur la première ligne de la zone centrale) ; la première tête de lecture de $I$ est désormais initialisée au niveau de la tuile située après ce deuxième champ (la deuxième tête de lecture est toujours initialisée au début des bits représentant les super-couleurs, qui sont maintenant le champ (iii) des données d'entrée de $I$). Ces modifications dans la construction de $\tau$ n'affectent que des tuiles parmi les $q$ premières tuiles de la ligne du bas de la zone de calcul, et donc ne modifient pas la liste $L$ définie ci-dessus. 

Pour que la liste $L$ puisse être complétée en un jeu de tuiles auto-similaire, il est nécessaire que $q'>2 \log N + O(1)$, afin que nous puissions encoder $O(N^2)$ couleurs dans les chaînes binaires de taille $q'$. Autrement dit, il faut que $q'$, le nombre de bits encodés dans une super-couleur, soit assez grand pour pouvoir encoder des coordonnées de super-tuiles au sein de la super-tuile mère de rang 2 (ce qui nécessite $2 \log N$ bits), ainsi que éventuellement $O(1)$ bits d'information supplémentaire (si la super-tuile joue le rôle d'un câble de communication ou d'une cellule du diagramme espace-temps de $I$, de $U$ ou de $A$ dans sa super-tuile mère). De plus, il est nécessaire qu'une machine de Turing puisse déterminer où commencent et finissent les bits de la première coordonnée, les bits de la seconde coordonnée et les bits supplémentaires. Nous avons donc $q'=O(\log N)$ bits.

Bien que la définition de $\tau$ ne soit pas terminée, nous pouvons déjà dire qu'un pavage de $\tau$ (s'il en existe) est constitué d'une grille de super-tuiles de taille $N \times N$, les super-tuiles ayant des chaînes binaires de taille $q \gg q'$ encodées dans les super-couleurs de leurs bords, ainsi que des câbles de communication de largeur $q$, et d'une zone de calcul de taille $(N-q) \times (N-3q)$, comme montré dans la Fig.~\ref{f:tuile-complète}. Cette zone de calcul est coupée horizontalement en trois : dans la partie basse s'exécute la machine de Turing $I$ à deux têtes de lecture initialisant les données, dans la partie centrale la machine de Turing universelle $U$ simulant un programme $\pi$ (une machine de Turing) qu'il reste à définir, et dans la partie haute un automate cellulaire $A$ (pour l'instant trivial). Le programme $\pi$ reçoit en entrée son propre code (selon la machine de Turing universelle $U$), la représentation binaire de $N$, et $4q'$ bits représentant, pour chacune des 4 super-couleurs, les parties de celles-ci utilisées pour l'auto-simulation.

Détaillons le programme $\pi$, qui doit effectuer les vérifications nécessaires (ces vérifications sont effectuées dans la zone centrale de chaque super-tuiles). Pour les super-tuiles qui ``représentent'' des tuiles incluses dans $L$ (i.e. les super-tuiles qui sont envoyées par $\phi$ sur des tuiles de $L$), les vérifications nécessaires sont directes (voir les commentaires à la fin de la sous-partie précédente). Il reste à implémenter (et encoder dans le programme $\pi$) les vérifications pour les tuiles qui ne sont pas encore définies. Gardons à l'esprit que pour compléter les tuiles de $\tau$ encore manquantes il nous suffit de connaître $\pi$, le code du programme. Cela peut sembler un paradoxe : nous devons écrire le code d'un programme qui est chargé de vérifier les tuiles qui encodent ce programme. D'une part, nous devons connaître ces tuiles pour écrire le programme ; d'autre part, nous devons connaître le code du programme pour déterminer les tuiles. Cependant, nous pouvons compléter la description de $\pi$ sans connaître les tuiles manquantes de $\tau$. Puisque dans notre construction d'une super-tuile, le programme voulu $\pi$ (la liste des instructions interprétées par la machine de Turing universelle $U$) est écrit au début du ruban de la machine universelle, nous pouvons demander au programme d'accéder aux lettres de son propre code : si les coordonnées de la super-tuile indiquent qu'elle représente la $j$-ème tuile du premier champ des données d'entrée de $I$, la machine universelle $U$ vérifie que la lettre de la case du ruban encodée au niveau des super-couleurs de la super-tuile correspond bien à la $j$-ème lettre de son code (i.e. à la $j$-ème lettre du premier champ de $U$).

C'est le moment crucial de la construction. L'algorithme $\pi$ simulé par $U$ peut être expliqué de la manière suivante. L'algorithme obtient en entrée une représentation binaire de $N$ et les bits encodant les parties des 4 super-couleurs d'une super-tuile qui sont utiles pour l'auto-simulation. Nous supposons qu'une super-tuile représente une tuile de $\tau$, et nous savons qu'au niveau des couleurs de chaque tuile de $\tau$ est encodée une paire de coordonnées $(i,j)\mod{N}$. Nous pouvons donc extraire les coordonnées $(i,j)$ des super-couleurs de la super-tuile. Si la ligne $j$ ne correspond pas à la première ligne du diagramme espace-temps de $I$, notre tâche est facile : nous utilisons l'algorithme décrit à la fin de la sous-partie précédente. Sinon, la machine $U$ détermine la taille $t$ de son premier champ, où est encodé $\pi$ ; nous avons choisi $U$ de telle manière qu'elle puisse déterminer là où se termine le programme qu'elle simule. Là encore, si $i>t$, nous utilisons l'algorithme décrit à la fin de la sous-partie précédente. Dans le cas contraire, la tâche est plus subtile : dans les couleurs de la tuile de $\tau$ représentée doit être encodée une lettre du code de $\pi$. 

Comment vérifier que cette lettre est correcte ? Où trouvons-nous la valeur ``correcte'' de celle-ci ? La réponse est directe : la machine de Turing universelle simulant $\pi$ doit déplacer sa tête dans sa propre zone de calcul, à exactement $i$ cases du début du ruban, et lire la valeur de la lettre requise. Notons que nous ne sommes pas confrontés au paradoxe de l’œuf et de la poule, car nous pouvons écrire les instructions de $\pi$ \emph{avant} de connaître l'intégralité de son code. Ainsi, nous obtenons l'intégralité du code (i.e. la liste des instructions) du programme $\pi$.
C'est exactement le programme qui doit être simulé par la machine de Turing universelle $U$ dans la zone de calcul de chaque super-tuile, et donc c'est le programme qui doit être écrit dans le premier champ des données d'entrée de la zone de calcul. Ce programme nous fournit la part manquante du jeu de tuiles $\tau$ (nous ajoutons à la liste $L$ les tuiles de $\tau$ présentes dans le premier champ des données d'entrée, qui représentent le code de $\pi$, puis à la suite les $O(\log N)$ tuiles de la représentation binaire de $N$, et complétons les $q$ premières tuiles de la première ligne du diagramme espace-temps de $I$ par des tuiles n'encodant pas d'information).

Il reste à choisir le paramètre $N$. Il doit être de valeur suffisamment grande pour que les calculs décrits ci-dessus (qui reçoivent des données d'entrée de taille $O(\log N)$) puissent s'effectuer dans la zone centrale de taille $(N-q) \times q$. Les calculs sont assez élémentaires (polynomiaux en la taille des entrées, i.e., polynomiaux en $O(\log N)$), et nous savons que $U$ les simule de manière efficace (i.e., également en temps polynomial). Comme le temps et l'espace disponibles sont des fractions constantes de $N$, pour $N$ assez grand ces calculs ont le temps et l'espace nécessaires pour être menés à terme. Ceci termine la construction de notre jeu de tuiles auto-similaire. Il est alors facile de vérifier que le jeu de tuiles construit (i) permet de paver le plan, et (ii) est auto-similaire. La construction décrite ci-dessus fonctionne pour tout grossissement $N$ suffisamment grand. En d'autres termes, pour tout $N$ suffisamment grand nous obtenons un jeu de tuiles auto-similaire $\tau_N$ de grossissement $N$, et tous ces jeux de tuiles $\tau_N$ ont une structure très similaire, avec des super-tuiles telles que montrées dans Fig.~\ref{f:tuile-complète}.
\end{proof}

\subsection{Un shift auto-similaire à grossissement variable}\label{ss:g-variable}

Dans cette sous-partie nous reprenons les idées développées dans \cite{point-fixe}, expliquant comment obtenir à partir d'un jeu de tuiles auto-similaire basique une construction plus flexible, en les adaptant à notre construction. 

Plus précisément, alors que dans la construction basique la zone de calcul a la même taille à tous les niveaux de la hiérarchie, dans la construction présentée ici la taille de la zone de calcul va croître avec le rang. Cette propriété nous permettra de simuler l'exécution d'un algorithme au niveau de la zone de calcul durant un temps et sur un espace de plus en plus importants. Typiquement, cet algorithme ne se termine pas, et à chaque niveau de la hiérarchie nous ne pouvons simuler qu'un nombre fini d'étapes de l'algorithme. Néanmoins, pour tout nombre d'étapes $n$, il existe un rang $k$ à partir duquel, dans la zone de calcul des super-tuiles de rang supérieur ou égal à $k$, l'algorithme est simulé durant au moins $n$ étapes.

Le jeu de tuiles $\tau$, une fois modifié pour avoir cette propriété, ne sera plus auto-similaire. Au contraire, nous obtiendrons non un unique jeu de tuiles mais une suite infinie de jeu de tuiles, chaque jeu de tuiles de la suite simulant le jeu de tuiles suivant dans la suite. 

Plus précisément, pour une large classe de suites d'entiers $n_k$ ayant de ``bonnes propriétés'', nous pouvons construire une famille de jeux de tuiles $\tau_k$ ($k = 0, 1, \ldots$) telle que $\tau_{k-1}$ simule le jeu de tuiles suivant $\tau_k$ avec un grossissement $n_k$ (et, par conséquent, que $\tau_0$ simule chaque jeu de tuiles $\tau_k$ avec un grossissement $N_k = n_1 \cdot n_2 \ldots n_k$).
L'idée est de réutiliser la construction élémentaire de la partie précédente et de varier les tailles des super-tuiles (i.e. les grossissements) pour les différents niveaux de la hiérarchie.
Tandis que dans la construction basique les super-tuiles (construites à partir de $N \times N$ tuiles), les ``super-super-tuiles'' (construites à partir de $N\times N$ super-tuiles), les ``super-super-super-tuiles'' (construites à partir de $N \times N$ super-super-tuiles), et ainsi de suite, se comportent exactement de la même manière, dans la construction revisitée nous avons des super-tuiles construites à partir de $n_1 \times n_1$ tuiles de base, des super-super-tuiles faites de $n_2 \times n_2$ super-tuiles, des super-super-super-tuiles faites de $n_3 \times n_3$ super-super-tuiles, et ainsi de suite. Dans cette construction, nous avons un isomorphisme seulement entre les super-tuiles de rang $k$ et les tuiles de $\tau_k$, et l'idée d'``auto-similarité'' doit être entendue de manière moins littérale.

Pour implémenter cette idée, nous n'avons besoin que de modifier légèrement la construction de la sous-partie précédente. Comme dans notre construction auto-similaire basique, chaque tuile de $\tau_k$ ``connaît'' sa position modulo $n_k$ dans le pavage : les couleurs de gauche et du bas doivent encoder $(i,j)$, la couleur de droite $(i+1 \mod {n_k},j)$, et la couleur du haut $(i, j+1 \mod{n_k})$. Par conséquent, chaque pavage de $\tau_k$ peut être découpé de manière unique en blocs (super-tuiles) de taille $n_k \times n_k$, où les coordonnées des cellules vont de $(0,0)$ pour le coin en bas à gauche à $(n_{k-1},n_{k-1})$ pour le coin en haut à droite, telles que montrées dans la Fig.\ref{f:tuile}.  
De nouveau, intuitivement, chaque super-tuile de rang $k$ ``connaît'' sa position dans sa super-tuile mère de rang $k+1$. Pour tout $k$, les $n_k \times n_k$ super-tuiles (faites de tuiles de $\tau_k$) doivent avoir la structure montrée dans la Fig.\ref{f:tuile-complète}, avec des câbles de communication, une zone de calcul où s'exécute une phase d'initialisation (la machine $I$), suivie d'une phase de calculs auto-référentiels (la machine $U$), suivie de l'exécution de l'automate cellulaire trivial. À chaque rang $k$ de la hiérarchie, le nombre de tuiles $q_k$ encodant une super-couleur est égal à la largeur des câbles de communication ; la zone de calcul est de taille $(n_k-q_k)\times (n_k-3q_k)$, et la ligne du bas de la zone de calcul commence par $q_k$ tuiles, suivies de $4q_k$ tuiles représentant les 4 super-couleurs de la super-tuile. Pour tout rang $k$, $q_k$ est égal à une fraction constante $q$ de $n_k$ (on suppose que pour tout $k$, $n_k$ est divisible par $q$).

La différence avec la construction basique est que maintenant les calculs simulés par une super-tuile de rang $k$ n'ont plus comme entrée $N$ mais la valeur $k$ ; le grossissement $n_k$ est désormais calculé par une fonction dépendant de $k$. Dans ce qui suit, nous supposons toujours que $n_k$ peut être facilement calculé en fonction de $k$ (disons, en temps $\poly (\log n_k)$). Techniquement, nous supposerons à partir de maintenant qu'au niveau de la première ligne de la zone de calcul, les données d'entrée contiennent les champs suivants :

(i) le programme de la machine de Turing $\pi$ qui vérifie si un quadruplet de super-couleurs correspond à une super-tuile valide,

(ii) une représentation binaire de $k$, l'entier représentant le rang de la super-tuile dans la hiérarchie ; nous choisissons une représentation permettant de déterminer où celle-ci commence et finit,

(iii) les bits encodant les super-couleurs : dans chaque super-couleur sont ainsi encodés une position dans la super-tuile mère de rang $k+1$ (deux coordonnées modulo $n_{k+1}$), et éventuellement $O(1)$ bits encodant l'information supplémentaire assignée aux super-couleurs (un bit si la super-tuile joue le rôle d'une super-couleur ou d'un câble de transmission, ou $O(1)$ bits si elle joue le rôle d'une cellule du diagramme espace-temps de $I$, de $U$ ou de $A$). Cela représente un nombre de bits très inférieurs à $q_k$ ; les bits restants sont réservés pour un usage futur (voir Partie~\ref{s:preuve}).

\begin{figure}[H]
	\centering
	\includegraphics[scale=0.7]{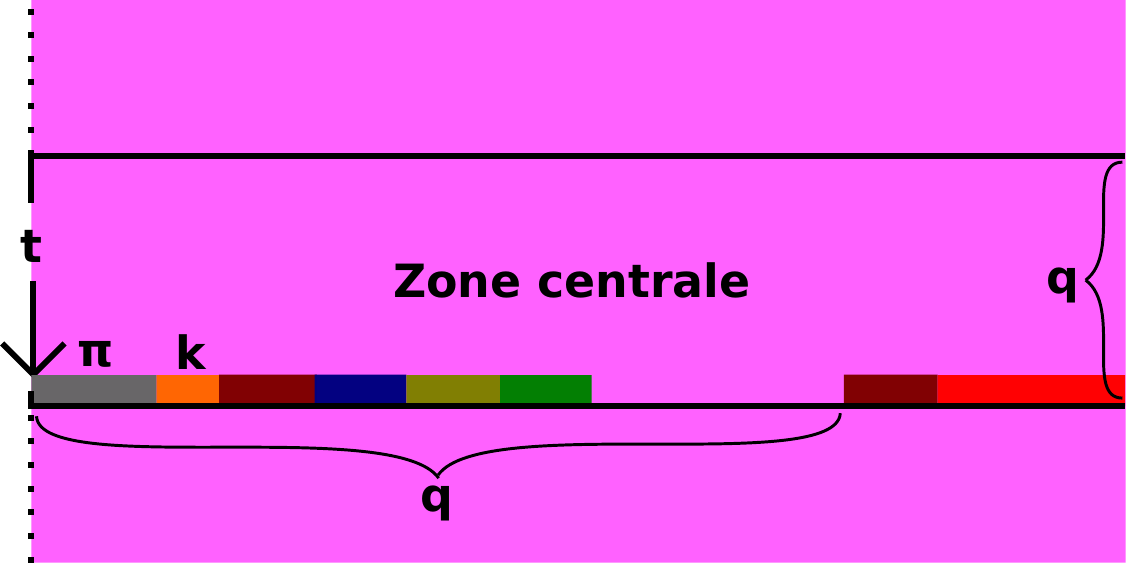}
	\caption[La partie centrale de la zone de calcul d'une super-tuile.]{La partie centrale représente la machine de Turing universelle $U$ qui simule un programme $\pi$, recevant en entrée le code de $\pi$, la représentation binaire du rang $k$ et les bits encodés dans les super-couleurs utilisés pour assurer la construction hiérarchique.}\label{f:tuile-zone-centrale}
\end{figure}

Remarquons que désormais le grossissement n'est plus fourni explicitement comme un champ des données d'entrée. Au contraire, nous avons la représentation binaire de $k$, à partir de laquelle une machine de Turing peut calculer la valeur $n_k$, voir Fig.~\ref{f:tuile-zone-centrale}. La différence par rapport à la Fig.~\ref{f:tuile-complète} est que les calculs effectués dans une super-tuile de rang $k$ reçoivent en entrée le rang $k$ au lieu du grossissement universel $N$.

Comme précédemment, la machine de Turing $I$ parcourt le champ (iii), et recopie les bits servant aux calculs garantissant la structure hiérarchique à la suite du champ (ii) (la représentation de $k$). La machine universelle $U$ s'exécute alors, en étant restreinte aux $q_k$ premières cases de son ruban. Nous demandons à ce que les calculs simulés se terminent dans un état acceptant si les super-couleurs de la super-tuiles forment un quadruplet valide (sinon, aucun pavage correct ne peut être formé). Les calculs simulés garantissent qu'il y a un isomorphisme entre les super-tuiles de rang $k$ et les tuiles de $\tau_{k+1}$.

Remarquons qu'à chaque rang $k$ de la hiérarchie, nous simulons dans les super-tuiles une exécution  d'une seule et unique machine de Turing $\pi$. Seules les données d'entrée reçues par la machine (incluant la représentation binaire de $k$) varient d'un rang à l'autre. Cette construction fonctionne si $n_k$ ne croît ni trop lentement (pour que les super-tuiles de rang $k$ aient suffisamment de place pour écrire la représentation binaire de $k$), ni trop rapidement (pour que dans la zone de calcul des super-tuiles de rang $k$ puissent s'effectuer des opérations arithmétiques élémentaires, sur des données d'entrée de taille $\log n_{k+1}$. 

Dans ce qui suit nous supposons que pour tout $k$, nous avons $N_k = 2^{C^k}$ pour une constante $C$ suffisamment large. Le choix exact de $C$ sera explicité dans la Partie~\ref{s:preuve}. 

Le facteur de grossissement $n_k$ permet d'encoder une ``charge utile'' dans la zone de calcul : des calculs ``utiles'', qui n'ont rien à voir avec l'auto-similarité, mais affectent les propriétés du pavage. Plus précisément, dans \cite{point-fixe} il est supposé que le programme $\pi$ (dont la simulation par une machine de Turing universelle est encodée dans chaque super-tuile) effectue deux tâches différentes : la première tâche est de vérifier la consistance des 4 super-couleurs de la super-tuile, comme décrit ci-dessus ; la deuxième tâche est d'exécuter un certain algorithme ``utile'' $B$. Cette algorithme $B$ varie selon les propriétés des pavages souhaitées ; celui-ci est explicitement codé dans le programme $\pi$, et, par conséquent, implicitement encodé dans le jeu de tuiles construit.

Dans notre construction, il n'est pas suffisant de modifier $\pi$, la ``deuxième tâche'' nécessitant un usage massif du parallélisme, et ne pouvant être effectuée de manière suffisamment efficace par une machine de Turing. Par conséquent, nous ne modifierons pas le programme $\pi$ ; cette deuxième tâche sera réalisée par l'automate cellulaire $A$, qui est exécuté dans la partie haute.

Nous avons vu que le code de cet automate sera détaillé dans la Partie~\ref{s:preuve}. Il utilisera comme données d'entrée l'entier $k$ représentant le rang (en effet ce champ n'a été modifié ni par $I$ ni par $U$), et les $q_k$ bits représentant les super-couleurs de la super-tuile (ceux-ci n'ont pas été modifiés par $I$, et la machine $U$ étant restreinte aux $q_k$ premières tuiles, elle n'a pas pu les modifier non plus). Nous supposons que dans chaque super-tuile la simulation de $A$ est effectuée en respectant les limites d'espace et de temps disponibles (soit respectivement $(n_k-q_k)$ et $(n_k-8q_k)$), et que la simulation s'arrête si l'automate atteint une de ces limites. Bien que toutes les super-tuiles (à tous les rangs de la structure hiérarchique) simulent un et un seul automate $A$, l'espace et le temps disponibles dépendent du rang de la super-tuile. Puisque le grossissement augmente avec le rang, nous pouvons allouer à la simulation de cet automate $A$ de plus en plus d'espace et de temps, au fur et à mesure que le rang des super-tuiles augmente.

\section{Flots d'information}
\label{s:flots}

Dans cette partie nous introduisons une autre technique que nous utilisons dans la preuve principale. Dans notre construction principale nous allons décrire un shift de type fini avec des ``flux d'information'' ayant quelques propriétés spécifiques. Essentiellement, ces ``flux d'information'' peuvent être pensés comme des flots multi-commodités à valeur entière sur un graphe particulier (nous nous intéressons à des graphes constitués d'une partie finie de la grille bidimensionnelle, avec des restrictions sur les sommets sources et les sommets puits). Dans ce qui suit nous définissons un graphe adéquat et démontrons certaines propriétés sur les flots de ce graphe. Les arguments développés dans cette partie n'ont rien à voir avec la dynamique symbolique. Les résultats peuvent être formulés en termes de flots entiers sur des graphes orientés. Nous expliquons les preuves en utilisant l'idée standard du rapport entre une coupe minimale et un flot maximal sur un graphe.

Nous commençons avec la définition standard d'un flot sur un graphe :

\begin{definition}\label{d:graphe-flot}
On dit qu'un graphe orienté $G$ est un \emph{graphe de flot} si :
\begin{itemize}
	\item il possède deux sommets distincts : une source notée $s$, et un puits noté $p$. Le sommet $s$ n'a pas d'arc entrant, et le sommet $p$ n'a pas d'arc sortant ;
	\item à chaque arc $(u,v)$ de $G$ est associé un entier positif appelé la capacité de l'arc, notée $c(u,v)$. La capacité d'un arc représente le flot maximal pouvant passer par cet arc. Si $(u,v) \notin G$, alors $c(u,v)=0$. 
\end{itemize}
\end{definition}

\begin{definition}\label{d:flot}
Soit $G$ un graphe de flot. Un \emph{flot} $f$ sur $G$ est une fonction associant à chaque arc $(u,v)$ de $G$ un entier naturel ou nul, représentant la quantité de flot passant par l'arc, notée $f(u,v)$. Ce flot doit respecter les propriétés suivantes :
\begin{itemize}
	\item contrainte de capacité : pour tout arc $(u,v)$ de $G$, $f(u,v) \leq c(u,v)$
	\item conservation du flot : pour tout sommet de $G$ autre que $s$ et $p$, la somme des flots des arcs entrants est égale à la somme des flots des arcs sortants.
\end{itemize}  
La quantité de flot sortant de $s$ est alors égale à la quantité de flot entrant dans $p$, et est appelée la valeur du flot.
En particulier, un flot de valeur maximale est appelé un flot maximal.
\end{definition}

\begin{remark}
Dans la définition ci-dessus nous considérons que le flot associé à un arc de $G$ est un entier naturel ou nul (et non un réel positif ou nul). Nous ne considérons que ce type particulier de flot, qui nous suffit pour nos besoins. Cela nous permet en outre de ne pas avoir à nous soucier des flots non entiers, ce qui simplifie les propositions et les preuves de cette partie.
\end{remark}

\begin{definition}\label{d:flot-super-tuile}
Pour les besoins de la construction du shift de densité $\epsilon$, nous considérons les graphes de flot $G$ suivants, appelés \emph{graphes d'une super-tuile}, ayant trois entiers naturels en paramètres ($N$, $F$ et $\rho$) :
\begin{itemize}
	 \item[1] les sommets de $G$ qui ne sont pas la source ou le puits forment une grille de taille $N \times N$ (il y en a donc $N^2$). On dit que deux sommets $u,v$ sont reliés s'il existe un arc de $u$ vers $v$ et un arc de $v$ vers $u$.  Chacun des sommets de la grille est relié à ses sommets voisins du haut, du bas, de gauche et de droite (excepté pour les sommets situés sur les bords ou les coins de la grille, reliés respectivement à trois et deux sommets voisins). Tous les arcs de la grille ont pour capacité un même entier $\rho$ ;
 	\item[2] un nombre d'arcs compris entre $0$ et $N^2$ partent de la source $s$ vers les sommets de la grille. Ces arcs ont pour capacité un entier compris entre $1$ et $\rho$. Nous notons $F$ la somme des capacités des arcs sortant de $s$ ;
 	\item[3] un nombre $F$ d'arcs partent des sommets de la grille vers le puits $p$. Ces arcs ont tous pour capacité $1$ ;
 	\item[4] pour tout sous-ensemble de sommets de la grille formant un carré $C$ de taille $c \times c$, nous imposons que la somme des capacités des arcs de $s$ vers les sommets de $C$ n'est pas plus grande que $\rho \cdot c$, et que la somme des capacités des arcs de $C$ vers $p$ n'est pas plus grande que $\rho \cdot c$.
\end{itemize}
\end{definition}

\begin{figure}[H]
	\center
	\includegraphics[scale=0.5]{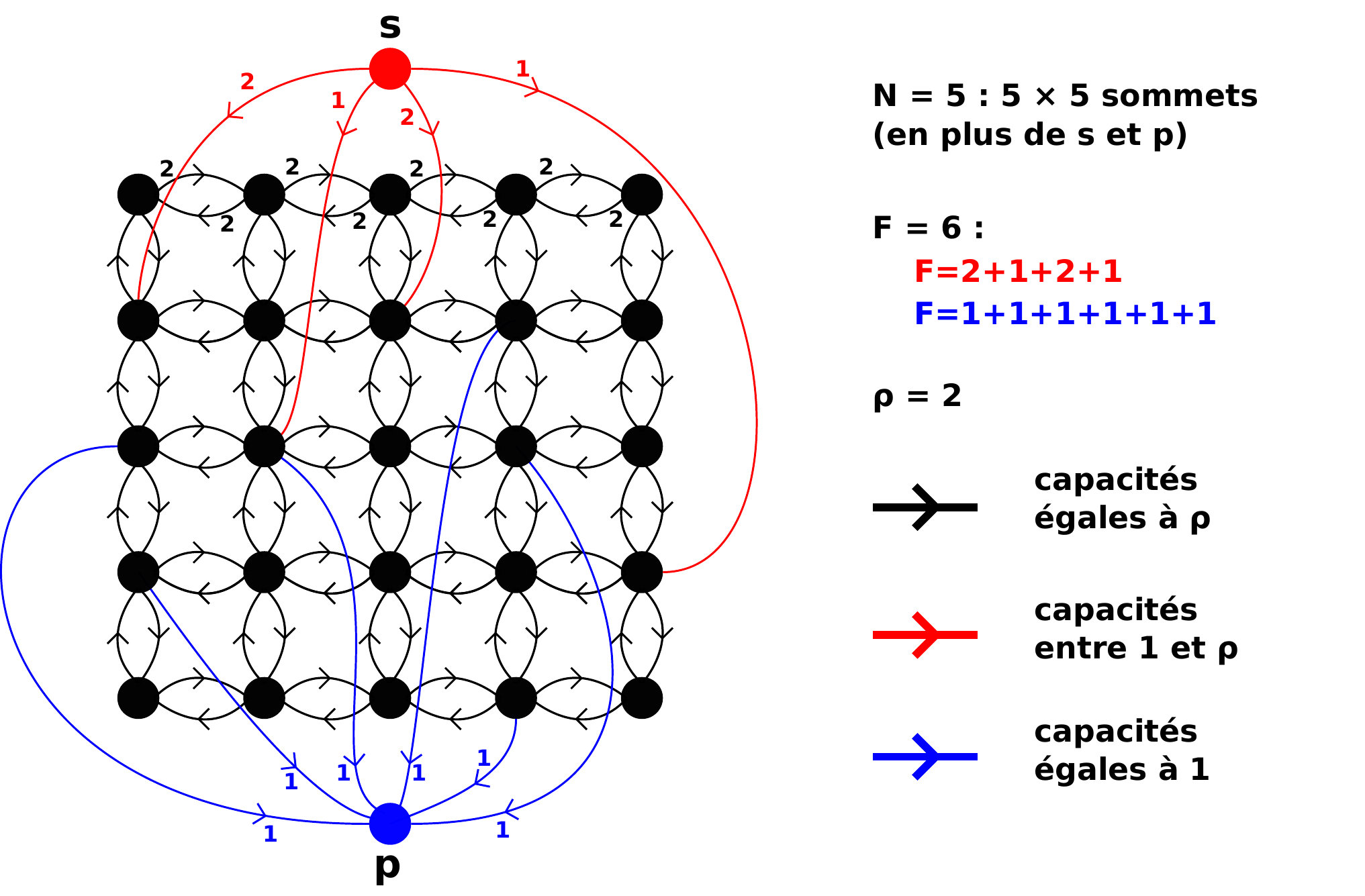}
	\caption{Image du graphe d'une super-tuile.}\label{f:graphe-tuile}
\end{figure}

Nous pouvons noter que pour que les propriétés du graphe soient respectées, il est nécessaire que les trois paramètres $N$, $F$ et $\rho$ soient compatibles entre eux (par exemple, il faut que $F \leq N^2$).

Maintenant nous rappelons l'idée de la décomposition d'un flot en une combinaison de chemins orientés et de cycles.

\begin{definition}
Soit $G$ un graphe de flot. Nous appelons un chemin élémentaire un flot $f$ de valeur $1$ tel que :
\begin{itemize}
 \item le flot $f$ dans un arc de $G$ est $0$ ou $1$ ;
 \item les arcs pour lesquels le flot $f$ est $1$ forment un chemin orienté (sans cycle) allant de la source $s$ au puits $p$.
\end{itemize}

Nous appelons un cycle élémentaire un flot $f$ de valeur nulle tel que :
\begin{itemize}
 \item le flot $f$ dans un arc de $G$ est $0$ ou $1$ ;
 \item les arcs pour lesquels le flot $f$ est $1$ forment un cycle orienté.
\end{itemize}
\end{definition}

\begin{figure}[H]
	\center
	\includegraphics[scale=0.37]{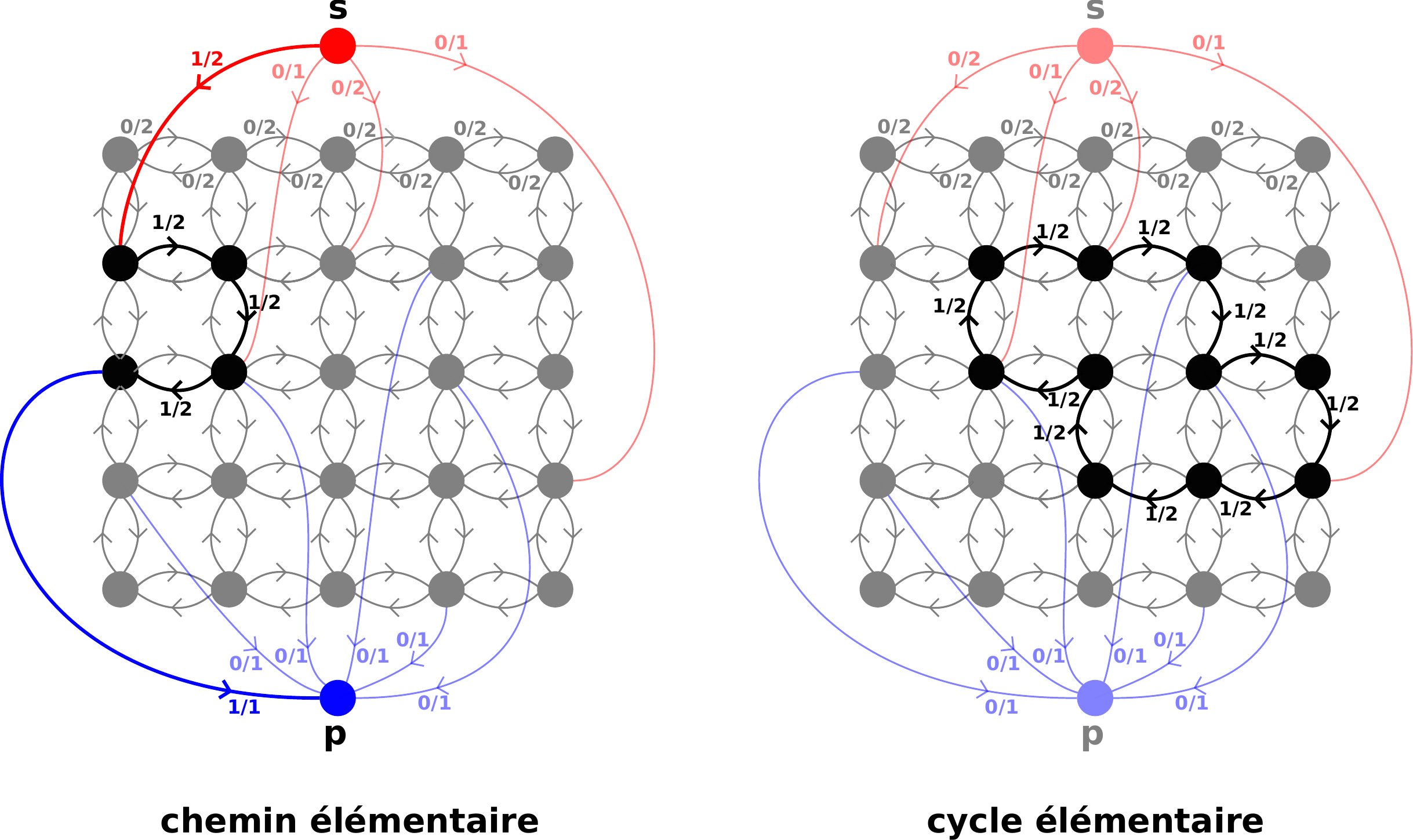}
	\caption{Un chemin et un cycle élémentaires.}\label{f:chemin-cycle}
\end{figure}

\begin{lemma}\label{l:flot-chemin}
Soient $G$ un graphe de flot et $f$ un flot sur $G$, de valeur $F$. 

Alors $f$ peut être représenté par un ensemble $A$ de chemins élémentaires de cardinalité $F$, et un ensemble $B$ de cycles élémentaires. Plus précisément, le flot $f$ d'un arc $(u,v)$ de $G$ est égal à la somme de deux entiers : le nombre de chemins élémentaires de $A$ passant par $(u,v)$, et le nombre de cycles élémentaires de $B$ passant par $(u,v)$.
\end{lemma}

\begin{proof}
Soit $G$ un graphe de flot. Pour tout flot $f$ sur $G$, on note $S$ la valeur totale de $f$ : il s'agit de la somme, pour chaque arc $a$ de $G$, de la valeur du flot $f$ associé à $a$ ($S$ ne doit pas être confondue avec la valeur du flot $f$). Nous allons démontrer ce lemme par récurrence forte sur $S$, la valeur totale du flot $f$ de $G$. 

Si $S=0$, le résultat est immédiat : pour tous les arcs de $G$, le flot $f$ est nul. On peut donc prendre $A$, l'ensemble des chemins élémentaires, et $B$, l'ensemble des cycles élémentaires, tous deux égaux à l'ensemble vide.

Supposons le résultat vrai pour tous les flots sur $G$ pour lesquels $S\leq n$, et considérons un flot $f$ sur $G$ de valeur $F$ pour lequel $S=n+1$.

Si $F$ est non nulle, nous initialisons une suite $L$ de sommets avec uniquement la source ($L = (s)$). Comme $F$ est non nulle (et donc $F \geq 1$), il existe un arc $(s,u_0)$ sortant de $s$ pour lequel $f(s,u_0)\geq 1$. 

Si $u_0$ est égal à $p$, nous avons trouvé un chemin élémentaire $c$ constitué d'un seul arc $(s,p)$. Nous définissons un nouveau flot $f'$ sur $G$ obtenu en soustrayant $c$ à $f$. Plus précisément, nous entendons par là que le flot $f'$ de l'arc $(s,p)$ est égal au flot $f$ de cet arc décrémenté de $1$, et que le flot $f'$ de chacun des autres arcs de $G$ est égal au flot de $f$ associé à eux. Nous avons $F'$, la valeur de $f'$, égale à $F-1$. La somme $S'$ du flot de $f'$ sur chacun des arcs de $G$ est égale à $n$, et nous pouvons donc appliquer l'hypothèse de récurrence sur $f'$. Le flot $f'$ peut donc être représenté par un ensemble $A'$ de chemins élémentaires de cardinalité $F'$, et un ensemble $B'$ de cycles élémentaires. Le flot $f$ est alors représenté par un ensemble $A$ de chemins élémentaires de cardinalité $F$ obtenu en ajoutant à $A'$ le chemin élémentaire $c$, et par un ensemble $B$ de cycles élémentaires égal à $B'$.

Si $u_0 \neq p$, nous ajoutons $u_0$ à la suite $L$. Comme il y a conservation du flot en $u_0$, la  somme $f(u_i,u_0)$ pour tous les $u_i \neq u_0$ est égale à la somme $f(u_0,u_j)$ pour tous les $u_j \neq u_0$. Il existe donc un arc $(u_0,u_1)$ avec $f(u_0,u_1)\geq 1$. De même que précédemment, si $u_1 = p$ nous avons trouvé un chemin élémentaire $c$, et nous appliquons la même procédure (définir un flot $f'=f-c$, appliquer l'hypothèse de récurrence sur $f'$, ajouter $c$ à $A'$ pour obtenir une représentation en chemins et cycles élémentaires de $f$). Dans le cas contraire, nous ajoutons $u_1$ à $L$. Nous pouvons itérer cette procédure d'ajout de sommets à $L$. Comme le nombre de sommets de $G$ est fini, un des deux cas suivants finit par se produire :
\begin{itemize}
	\item soit il existe un entier $k$ tel que $u_k = p$ ($L=(s,u_0, \ldots, u_{k-1},p)$). Comme précédemment, cela correspond au cas où un chemin élémentaire $c$ est trouvé et nous appliquons la procédure habituelle ;
	\item soit il existe deux entiers $i < k$ avec $u_i = u_k$ ($L=(s,\ldots, u_i, \ldots u_k)$. Dans ce cas nous avons trouvé un cycle élémentaire $c=(u_i,\ldots,u_{k-1})$. Nous définissons un flot $f'$, de valeur $F$, obtenu en soustrayant $c$ à $f$ ; la somme $S'$ des flots de $f'$ sur $G$ est inférieure ou égale à $n$. Nous pouvons donc appliquer l'hypothèse de récurrence à $f'$, et obtenir les ensembles $A'$ et $B'$. Le flot $f$ est alors représenté par l'ensemble $A=A'$ de chemins élémentaires de cardinalité $F$, et par l'ensemble $B=B'\cup \{c\}$ de cycles élémentaires. 
\end{itemize}

Étudions maintenant le cas où $F$ est nulle. Comme $S \geq 1$, il existe un arc $(u_0,u_1)$ tel que $f(u_0,v_1)\geq 1$. Comme $F=0$, on a $u_0$ différent de $s$ et $u_1$ différent de $p$. Nous définissons la suite $L=(u_0,u_1)$. Comme pour le cas où $F$ est non nulle, il y a conservation du flot en $u_1$, et il existe donc un arc $(u_1,u_2)$ avec $f(u_1,u_2)\geq 1$ ($u_2 \neq p$) ; nous ajoutons $u_2$ à $L$. Nous itérerons là aussi cette procédure d'ajout de sommets à $L$. Le nombre de sommets de $G$ étant fini et les sommets $u_i$ ajoutés étant différents de $p$, le cas suivant finit par se produire : il existe deux entiers $i,k$ tels que $u_i = u_k$ et $L=(u_0,\ldots, u_i, \ldots u_k)$. Nous avons alors trouvé un cycle élémentaire, et nous appliquons la même procédure que dans le cas où un cycle élémentaire est trouvé et que $F$ est non nulle. 

Cela termine le raisonnement par récurrence, ainsi que la preuve.
\end{proof}

Dans ce qui suit nous déterminons le flot maximal du graphe d'une super-tuile en utilisant l'égalité entre un flot maximal et une coupe minimale.

\begin{proposition}\label{p:flot-max}
Soient $G$ un graphe d'une super-tuile (voir Définition~\ref{d:flot-super-tuile}) et $F$ la somme des capacités des arcs sortant de la source $s$. Alors il existe un flot $f$ sur $G$ de valeur $F$.  
\end{proposition}

Ce flot est maximal, puisqu'il correspond à la somme des capacités des arcs sortant de $s$, ou encore à la somme des capacités des arcs entrant dans $p$.

Nous pouvons maintenant énoncer le résultat principal de cette partie :
\begin{corollary}\label{cor:flot-chemin}
Soit $G$ un graphe d'une super-tuile et $F$ la somme des capacités des arcs sortant de $s$. Alors il existe un flot $f$ sur $G$ de valeur $F$, qui peut être représenté par une paire d'ensembles $(A,B)$ : un ensemble $A$ de chemins élémentaires de cardinalité $F$, et un ensemble $B$ de cycles élémentaires.
\end{corollary}

\begin{proof}
Le résultat s'obtient en appliquant le Lemme~\ref{l:flot-chemin} au résultat obtenu par la Proposition~\ref{p:flot-max}.
\end{proof}

Pour démontrer la Proposition~\ref{p:flot-max}, nous allons introduire le concept de coupe dans un graphe de flot, et utiliser le Théorème~\ref{th:max-min}.

\begin{definition}
Soit $G$ un graphe de flot. Nous appelons \emph{coupe s-p} de $G$, notée ${\cal C}$, un couple de sous-ensembles de sommets $(S,P)$ disjoints d’union $G$, tels que $s\in S$ et $p\in P$. La capacité d'une coupe ${\cal C}$, notée $|{\cal C}|$, est la somme des capacités respectives des arcs allant de $S$ à $P$.
En particulier, une coupe s-p de capacité minimale est appelée une coupe s-p minimale.
\end{definition}

\begin{theorem}\label{th:max-min}
Théorème (\cite{flot-max-coupe-min}) \emph{flot-max/coupe-min} : Pour tout graphe de flot $G$, la valeur d'une coupe s-p minimale est égale à la valeur d'un flot maximal. 
\end{theorem}

\begin{remark}
Le problème de maximisation du flot peut être formulé en termes de programmation linéaire, comme les auteurs le faisaient déjà remarquer dans \cite{flot-max-coupe-min}. Plus précisément, les problèmes de maximisation du flot et de minimisation d'une coupe s-p peuvent être formulés comme les versions primale et duale d'un même programme linéaire, et le théorème est alors une conséquence directe du théorème de dualité forte de l'optimisation linéaire (voir par exemple \cite{ford1957simple}). La preuve de \cite{flot-max-coupe-min} est cependant intéressante car elle fournit un schéma, qui sera implémenté par différents algorithmes, pour calculer le flot maximal plus efficacement qu'en utilisant la programmation linéaire.
\end{remark}

Soit $G$ un graphe d'une super-tuile, et ${\cal C}=(S,P)$ une coupe s-p de $G$. Le singleton contenant la source $s$ et celui contenant le puits $p$ sont des composantes fortement connexes respectivement de $S$ et de $P$. À chaque autre composante fortement connexe $C$ de $S$ ou de $P$, on associe le plus petit sous-ensemble de sommets de la grille formant un rectangle $R$ de taille $l \times h$ et contenant $C$.

\begin{figure}[H]
	\center
	\includegraphics[scale=0.35]{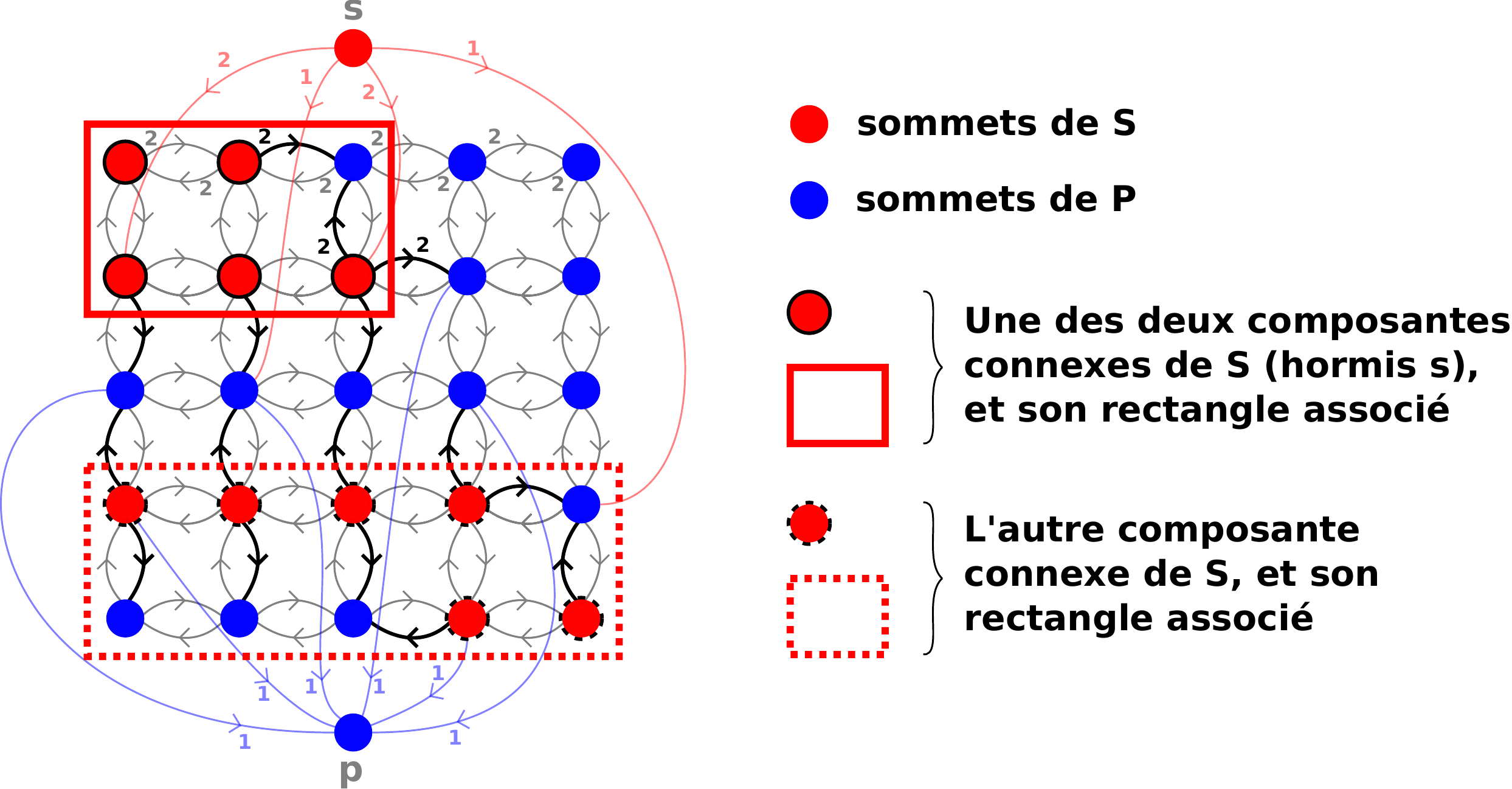}
	\caption[$(S,P)$ une coupe s-p.]{$(S,P)$ une coupe s-p. Les ensemble $S$ et $P$ ont tous deux 2 composantes connexes (hormis respectivement $s$ et $p$).}\label{f:coupe-s-p}
\end{figure}

\begin{claim}\label{c:flot-rectangle}
Si un rectangle $R$ associé à une composante fortement connexe $C$ de $S$ est égal à l'ensemble de la grille, alors aucun rectangle $R'$ associé à une composante fortement connexe $C'$ de $P$ n'est égal à la grille.
\end{claim}

\begin{proof}
Supposons que le rectangle $R$ associé à une composante fortement connexe $C$ de $S$ est égal à l'ensemble de la grille. Dans ce cas, il existe au moins un sommet de $C$ présent sur le bord gauche de la grille, et au moins un autre sur le bord droit (il en est de même pour les bords haut et bas). Autrement dit, il existe une chaîne $c$ de sommets de $C$ où le premier sommet est situé sur le bord gauche de la grille, et le dernier est situé sur le bord droit ; cette chaîne $c$ ``sépare'' en deux la grille, voir Fig.~\ref{f:coupe-s-p-grand-carré} (il y a la partie au-dessus de la chaîne, et la partie au-dessous ; une de ces parties peut éventuellement être vide). 

Soit $R'$ le rectangle associé à une composante connexe $C'$ de $P$. Il ne peut y avoir de chaîne de sommets de $C'$ reliant les bords haut et bas de la grille (on ne peut pas ``traverser'' la chaîne $c$ de sommets de $C$ reliant les bords gauche et droit de la grille). 

\begin{figure}[H]
	\center
	\includegraphics[scale=0.35]{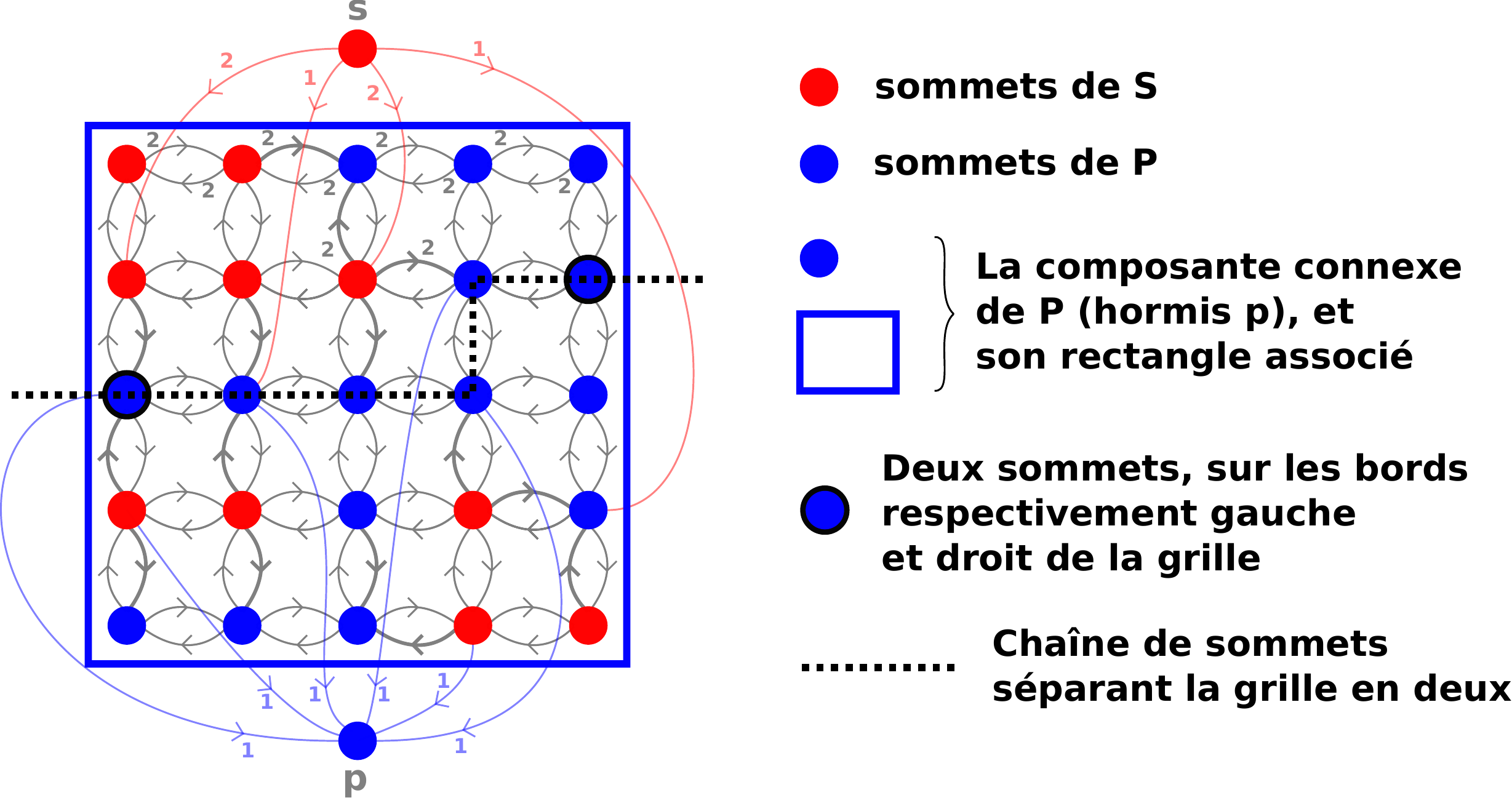}
	\caption[Une composante connexe de l'ensemble $P$ dont le rectangle associé est l'ensemble de la grille.]{$(S,P)$ une coupe s-p. L'ensemble $P$ a une composante connexe dont le rectangle associé est l'ensemble de la grille.}\label{f:coupe-s-p-grand-carré}
\end{figure}

Ce résultat, intuitif, repose en réalité sur le Théorème de Jordan, dans sa version discrète. Ce théorème a été énoncé par Jordan en 1887, mais avec une preuve fausse ; une preuve correcte a été donnée dans~\cite{jordan}.

Par conséquent $R'$ ne peut être égal à la grille.
\end{proof}

\begin{claim}\label{c:flot-connexe}
Soit ${\cal C}=(S,P)$ une coupe s-p. Si $C$ est une composante connexe de $S$ (autre que le singleton $\{s\}$) telle que le rectangle associé à $C$ ne soit pas égal à la grille, alors ${\cal C'} = (S-C,P\cup C)$ est une coupe s-p et $|{\cal C'}| \leq |{\cal C}|$.
De même, si $C$ est une composante connexe de $P$ (autre que le singleton $\{p \}$) telle que le rectangle associé ne soit pas la grille, alors ${\cal C'} = (S\cup C,P - C)$ est une coupe s-p et $|{\cal C'}| \leq |{\cal C}|$.
\end{claim}

\begin{proof}
Soit ${\cal C}=(S,P)$ une coupe s-p, et $C$ une composante connexe de $S$ (autre que le singleton $\{s\}$) telle que le rectangle $R$ de taille $l \times h$ associé à $C$ ne soit pas égal à la grille de taille $N\times N$. Nous supposons sans perte de généralité que $h \leq l$. Alors nous avons $h<N$ (si $h=N$, alors comme $l \geq h$ on aurait $l=N$ et donc $R$ égal à la grille). Comme $h<N$, au moins un des bords latéraux de $R$ n'est pas au bord de la grille ; nous notons ce bord latéral $b$. À chaque sommet du bord $b$, correspond un arc de $S$ vers $P$ de capacité $\rho$ (celui-ci se situe sur la même ligne que le sommet, voir Fig.~\ref{f:coupe-s-p-inférieure}). 

\begin{figure}[H]
	\center
	\includegraphics[scale=0.34]{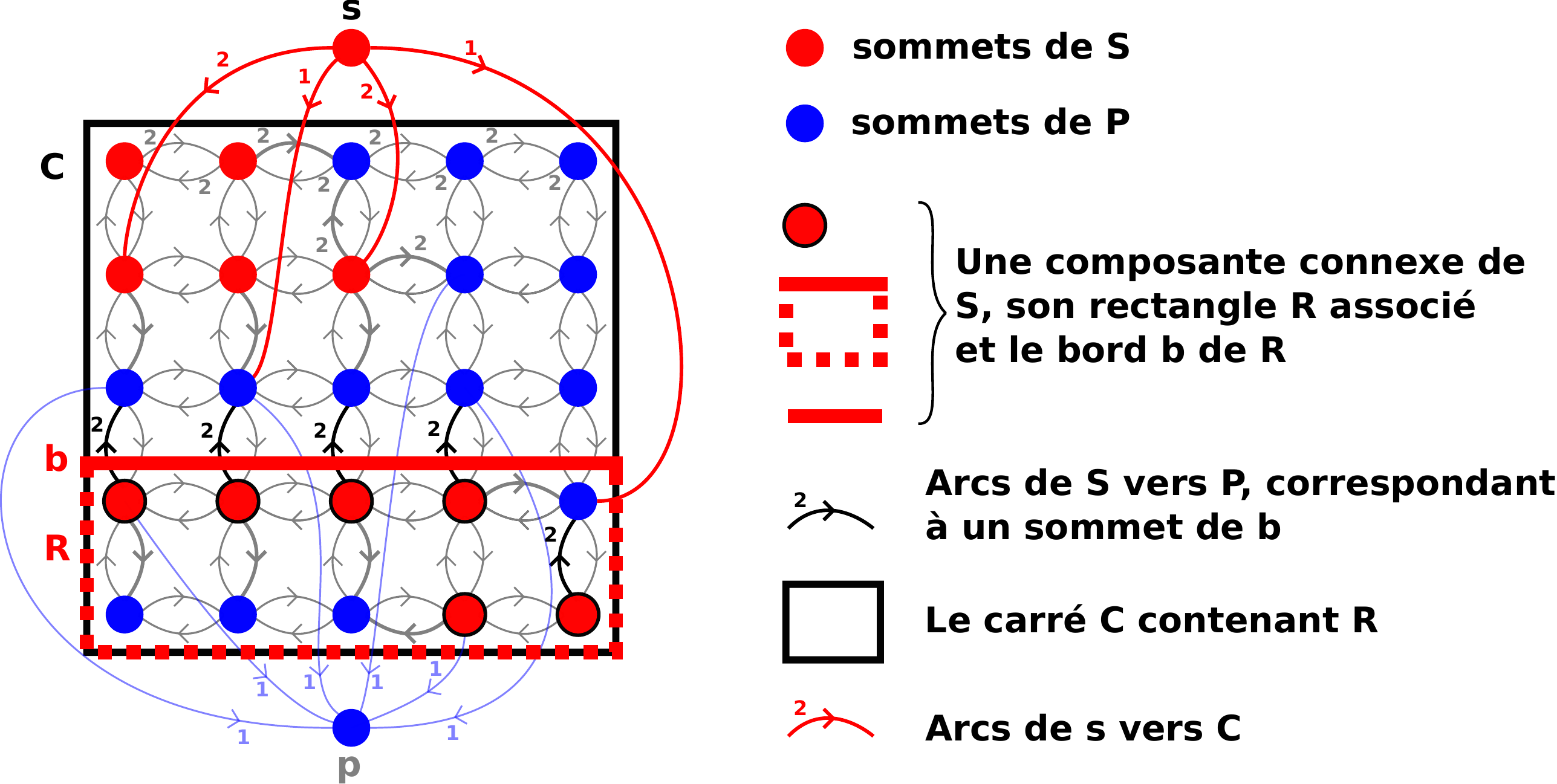}
	\caption[Une composante connexe de l'ensemble $S$ dont le rectangle associé n'est pas l'ensemble de la grille.]{Une composante connexe de $S$ a pour rectangle associé $R$ de taille $l \times h$, avec $l \geq h$. Il existe un bord horizontal $b$ (de taille $l$) de $R$ qui n'est pas situé au bord de la grille À chaque sommet de $b$ correspond un arc de $S$ vers $P$ de capacité $\rho$. En faisant passer les sommets de cette composante connexe du côté de $P$, la valeur de la coupe diminue donc d'au moins $l \cdot \rho$. De plus, considérons le carré $C$ de taille $l\times l$ contenant $R$. La somme des capacités des arcs de $s$ vers $C$ est au plus $l \cdot \rho$. Ainsi, la valeur de la coupe augmente d'au plus $l \cdot \rho$. Par conséquent, la valeur de la coupe va diminuer.}\label{f:coupe-s-p-inférieure}
\end{figure}

Lorsque nous déplaçons les sommets de $C$ (qui sont du côté de la source $s$) du côté du puits $p$, ces arcs ne sont plus comptabilisés dans la valeur de la nouvelle coupe ${\cal C'}$, alors qu'ils y contribuaient à hauteur de $l \cdot \rho$. Par conséquent, la somme totale des capacités des arcs qui reliaient $S$ à $P$ et qui ne sont plus comptabilisés est supérieure à $l \cdot \rho$.

Comme $C$ est une composante connexe de $S$, les seuls arcs de $S-C$ vers $C$ sont ceux provenant de $s$. La propriété (4) de la Définition~\ref{d:flot-super-tuile} des graphes de super-tuiles nous permet de garantir que la somme des capacités des arcs de $s$ vers un des carrés de taille $l\times l$ contenant $R$ n'est pas plus grande que $l \cdot \rho$. À fortiori, la somme des capacités des arcs de $s$ vers $R$, et donc celle des capacités des arcs de $s$ vers $C$, est également majorée par $l \cdot \rho$. Lorsque nous déplaçons les sommets de $C$ du côté de $p$, la somme des capacités des nouveaux arcs de $S$ vers $P$ n'est donc pas plus grande que $l \cdot \rho$.

En réunissant ces deux résultats, nous obtenons que la valeur de la nouvelle coupe est inférieure ou égale à l'ancienne : $|(S - C,P \cup C)| \leq |(S,P)|$.

Comme la propriété (4) de la Définition~\ref{d:flot-super-tuile} garantit également que la somme des capacités des arcs d'un carré de taille $l \times l$ vers $p$ n'est pas plus grande que $l \cdot \rho$, nous pouvons tenir le même raisonnement dans le cas d'une composante connexe $C$ de $P$ dont le rectangle associé n'est pas la grille, pour obtenir $|(S \cup C,P - C)| \leq |(S,P)|$.
\end{proof}

Nous pouvons maintenant utiliser les Faits~\ref{c:flot-rectangle} et \ref{c:flot-connexe} précédents pour prouver la Proposition~\ref{p:flot-coupe-minimale} suivante :

\begin{proposition}\label{p:flot-coupe-minimale}
Les coupes s-p avec d'un côté le singleton $\{s\}$ et de l'autre $G-\{s\}$, ou avec d'un côté $G-\{p\}$ et de l'autre le singleton $\{p\}$ sont toutes deux des coupes s-p minimales de valeur $F$. 
\end{proposition}

\begin{proof}
Soit $G$ un graphe d'une super-tuile, ${\cal C}=(S,P)$ une coupe s-p de $G$. 

Si les rectangles associés aux composantes connexes de $S$ sont tous différents de la grille, nous appliquons le Fait~\ref{c:flot-connexe} sur chaque composante connexe de $S$ l'une après l'autre, pour obtenir une série de coupes s-p de valeurs (non strictement) décroissantes. On obtient finalement une coupe s-p avec le singleton $\{s\}$ d'un côté, et le reste des sommets de $G$ de l'autre : la coupe est de valeur $F$.

Sinon, le Fait~\ref{c:flot-rectangle} nous garantit que les rectangles associés aux composantes connexes de $P$ sont tous différents de la grille. Dans ce cas nous appliquons là aussi le Fait~\ref{c:flot-connexe} sur chaque composante connexe de $P$ l'une après l'autre ; nous obtenons également une suite de coupes s-p de valeurs (non strictement) décroissantes, et finalement une coupe s-p avec d'un côté $G-\{p\}$ et de l'autre le singleton $\{p\}$, de valeur $F$.

Dans les deux cas, nous avons montré que la valeur de la coupe ${\cal C}$ n'est pas plus petite que $F$. Donc $F$ est la valeur des coupes s-p minimales.
\end{proof}

Nous pouvons maintenant prouver la Proposition~\ref{p:flot-max} :
\begin{proof}
Il suffit d'appliquer le Théorème~\ref{th:max-min} au résultat de la Proposition~\ref{p:flot-coupe-minimale}.
\end{proof}

En conclusion nous présentons une reformulation alternative du résultat démontré. Puisque les capacités du graphe d'une super-tuile sont des entiers naturels, nous pouvons substituer chaque arc de capacité $\rho$ par $\rho$ arcs parallèles de capacité $1$. Alors notre résultat principal peut être reformulé de la manière suivante : il existe $F$ chemins élémentaires dont les arcs sont disjoints (deux chemins peuvent utiliser des arcs parallèles reliant la même paire de sommets, mais ils ne peuvent avoir en commun un même arc). En d'autres termes, nous pouvons acheminer $F$ ``commodités'' de valeur $1$ de $s$ à $p$ par des chemins disjoints, et ils arrivent à destination sans être entrés en collision l'un avec l'autre.
Ce résultat est essentiellement le théorème de Menger sur les multigraphes orientés appliqué au graphe d'une super-tuile, voir Théorème~7.3.1 in \cite{menger}.

Nous pouvons également remarquer qu'au lieu d'avoir un unique sommet source $s$ et un unique sommet puits $p$, nous pouvons avoir plusieurs sources et plusieurs puits dans la grille $N\times N$ et distribuer les $F$ unités ``produites'' entre les sommets sources et les $F$ unités ``consommées'' entre les sommets puits. Dans cette formulation, notre résultat principal garantit que nous pouvons trouver des chemins orientés transportant les $F$ unités ``produites'' jusqu'aux sommets où elles sont ``consommées''. Il est crucial que tous les sommets puits soient indiscernables : nous sommes autorisés à transporter n'importe quelle unité provenant de n'importe quel sommet source jusqu'à n'importe quel sommet puits disponible.

\section{Utilisation des différents outils pour prouver le Théorème~\ref{th:sparse1} et le Théorème~\ref{th:sparse2}}\label{s:preuve}

\subsection{Plan de la preuve du théorème}

Les preuves du Théorème~\ref{th:sparse1} et du Théorème~\ref{th:sparse2} que nous proposons dans cette partie sont très similaires, et sont basées sur la même construction. Nous les expliquerons en parallèle. Malheureusement, la construction est relativement complexe, du fait entre autre de l'utilisation de la technique des pavages à point-fixe. Le plan général est décrit ci-dessous. 

Pour tout shift $S_\varepsilon$ du Théorème~\ref{th:sparse1} ou $S_\varepsilon'$ du Théorème~\ref{th:sparse2} (l'alphabet de ces shifts est, par définition, constitué des deux lettres ``noire'' et ``blanche''), nous construisons un shift de type fini à deux dimensions $S$. Nous préférons décrire $S$ comme un jeu de tuiles $\tau$ ; voir page~\ref{p:shift-tuiles} pour l'équivalence entre shifts de type fini et jeux de tuiles. Le nombre de tuiles de $\tau$ (l'alphabet de $S$) est très important (bien plus grand que deux, la taille de l'alphabet de $S_\varepsilon$ et $S_\varepsilon'$). Toutes les tuiles de $\tau$ sont divisées en deux classes, les tuiles ``noires'' et les tuiles ``blanches'' (i.e., $\tau = \tau_{\text{noir}} \sqcup \tau_{\text{blanc}}$). Le but est de construire $\tau$ tel que la projection naturelle des pavages de $\tau$ (quand nous gardons pour chaque tuile uniquement sa ``couleur'', qui peut être noire ou blanche) donne respectivement exactement le shift $S_\varepsilon$ ou exactement le shift $S_\varepsilon'$. En d'autres termes, quand nous prenons l'ensemble des pavages de $\tau$, et que pour chaque tuile nous ne gardons que sa couleur (noire ou blanche), nous devons obtenir l'ensemble des configurations du shift $S_\varepsilon$ ou du shift $S_\varepsilon'$, et seulement celles-ci.

Pour cela, nous partons de la construction d'un jeu de tuiles $\rho$ auto-similaire à grossissement variable (voir Sous-Partie~\ref{ss:g-variable}), où chaque pavage du plan de $\rho$ a une structure hiérarchique : une configuration peut être subdivisée d'une manière unique en super-tuiles de taille $n_1\times n_1$, ces super-tuiles pouvant à leur tour être regroupées (de nouveau, de manière unique) en blocs de taille $n_2\times n_2$  (qui sont de taille $(n_1\cdot n_2 \times n_1\cdot n_2)$ si on les mesure en nombre de tuiles individuelles), et ainsi de suite. Au $k$-ème niveau de la hiérarchie nous avons des super-tuiles de rang $k$ qui consistent en $n_k\times n_k$ super-tuiles de rang $(k-1)$, qui sont des carrés de taille $(n_1\cdot \ldots \cdot n_k) \times (n_1 \cdot \ldots \cdot n_k)$ mesurés en tuiles individuelles. Nous supposons que $N_k= 2^{C^k}$ (le choix de $C$ sera explicité à la fin de la Sous-Partie~\ref{ss:construction-shift}).

De plus, chaque super-tuile de $\rho$ (de tout rang) possède une zone de calcul divisée horizontalement en trois zones ; dans la zone du bas une machine de Turing initialise les données, dans la zone intermédiaire une machine de Turing universelle $U$ garantit la structure hiérarchique du pavage, et dans la zone du haut s'exécute un automate cellulaire non-déterministe à une dimension dont le comportement sera explicité dans cette partie. 

Pour qu'une super-tuile $M$ de rang $k$ puisse connaître la liste des points noirs qu'elle contient, il est nécessaire que, d'une manière ou d'une autre, cette information connue par les $n_k^2$ super-tuiles filles de $M$ soit acheminée jusqu'aux données d'entrée de la zone de calcul de $M$. Cela ne peut pas être fait de manière naïve (toutes les super-tuiles filles de $M$ partageraient cette information) : le nombre de bits nécessaires pour encoder cette information est supérieur à la taille (en nombre de super-tuiles de taille $k-2$) des super-tuiles filles de $M$ de rang $k-1$. Chaque super-tuile fille transmet seulement quelques points noirs présents dans $M$ à ses voisines. Un résultat sur les flots, démontré dans la Partie~\ref{s:flots}, nous garantit alors que la totalité des points noirs présents dans les filles de $M$ sont bien acheminés jusqu'aux données d'entrée de la zone de calcul de $M$.   

\begin{figure}[H]
	\center
	\includegraphics[scale=0.18]{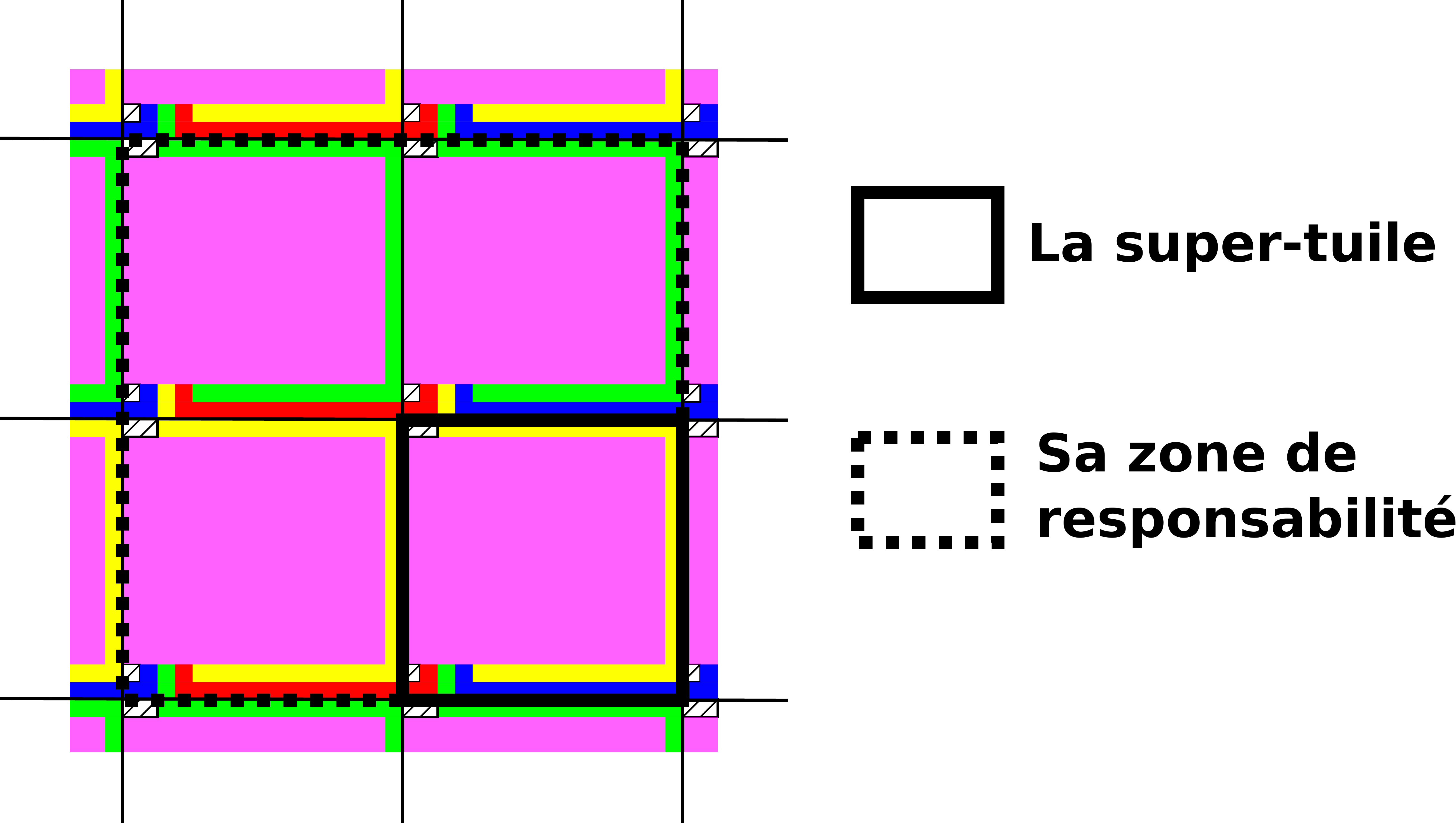}
	\caption{Une super-tuile et sa zone de responsabilité.}\label{f:tuile-zone-responsabilité}
\end{figure}
  
Nous disons que chaque super-tuile de rang  $k$ (qui est un carré de taille $N_k\times N_k$) possède sa ``zone de responsabilité'', qui est un carré de taille $(2N_k) \times (2N_k)$. Cette zone inclut la super-tuile elle-même, sa super-tuile voisine de même rang de gauche, sa super-tuile voisine de même rang du haut et une super-tuile de même niveau supplémentaire, qui est diagonalement en haut à gauche, voir  Fig.~\ref{f:tuile-zone-responsabilité}. Parmi d'autres vérifications effectuées dans la zone de calcul, nous allons y simuler un algorithme énumérant les motifs interdits (nous en énumérons un nombre fini, étant donné que le temps et l'espace disponibles pour effectuer ces calculs sont bornés), et vérifier que la ``zone de responsabilité'' ne contienne pas de motif interdit parmi l'ensemble des motifs interdits trouvés.
 
La taille de la zone de calcul augmente quand le rang $k$ de la super-tuile croît. Ainsi, chaque motif interdit est énuméré par les calculs effectués dans les super-tuiles de rangs assez grands. Comme chaque motif de la configuration est vérifié par des super-tuiles de rangs arbitrairement larges (i.e. le motif appartient à leur zone de responsabilité), nous pouvons garantir que toutes les configurations valides ne contiennent aucun motif interdit. Leur projection appartient donc respectivement au  shift $S_\varepsilon$ et $S_\varepsilon'$.

Réciproquement, à une configuration ${\cal C}$ de $S_\varepsilon$ ou de $S_\varepsilon'$, nous pouvons associer un pavage de $\tau$ qui, une fois projeté, nous donne cette configuration ${\cal C}$. 

\subsection{Construction d'un shift de type fini}\label{ss:construction-shift}

\paragraph{Hiérarchie des motifs interdits.} Dans le Théorème~\ref{th:sparse1} nous nous intéressons au shift $S_\varepsilon$ où les motifs interdits sont des carrés de taille $n\times n$ avec plus de $n^\epsilon$ points noirs. Comme $\epsilon$ est un nombre calculable par en haut, ces motifs interdits peuvent facilement être énumérés. Fixons un algorithme ${\cal A}$ qui énumère ces motifs un à la fois. Il doit calculer des approximations de $\epsilon$ qui convergent vers ce dernier, et, en accord avec celles-ci, énumérer des motifs carrés de tailles arbitrairement grandes qui contiennent trop de points noirs (le reste des points de ces carrés étant de couleur blanche). L'énumération des motifs interdits est infinie, l'algorithme ${\cal A}$ ne s'arrête jamais.

Dans notre preuve ci-dessous, nous aurons besoin d'une procédure lente et graduelle énumérant les motifs interdits. Pour cela, nous fixons une séquence de nombres $\ell_k$ qui croît extrêmement lentement. Dans un souci de précision, nous fixons $\ell_k = \log k$, bien que ce choix particulier de $\ell_k$ n'a que peu d'importance. Nous disons qu'un motif est de rang $k$, si l'algorithme fixé ${\cal A}$ produit ce motif en au plus $\ell_k$ étapes de son calcul infini (le motif fait donc partie des $\ell_k$ premiers motifs produits par $A$), s'il contient au plus $\ell_k$ points noirs, et que la taille de ce motif est au plus $\ell_k\times \ell_k$. Clairement, pour tout motif interdit $P$ produit par ${\cal A}$, il existe un entier $k_0$ tel que $P$ est un \emph{motif interdit de rang $k$} pour tout $k\geq k_0$. \label{ell-forbidden-patterns}

De manière similaire, pour les shifts $S_\varepsilon'$ du Théorème~\ref{th:sparse2} (qui sont des sous-shifts effectifs de $S_\varepsilon$), nous pouvons fixer un algorithme énumérant les motifs interdits de $S_\varepsilon'$ et imposer une hiérarchie de motifs interdits telle que chaque motif de rang $k$ est obtenu en au plus $\ell_k$ étapes de calcul, ne contient au plus que $\ell_k$ points noirs et est de taille au plus $\ell_k\times \ell_k$ (pour tout motif interdit $P$ produit par $A$, il existe un entier $k_0$ tel que ce motif $P$ est un motif interdit de rang $k$ pour tout $k\ge k_0$).

\paragraph{Calculs effectués dans les super-tuiles : rappel de la construction de la Partie~\ref{s:shift-auto-similaire}.}
Rappelons les caractéristiques principales garanties par la construction auto-référentielle décrite dans la Sous-Partie~\ref{ss:g-variable}.
Chaque pavage du plan valide pour notre jeu de tuiles a une structure hiérarchique : une configuration peut être divisée (de façon unique) en super-tuiles de taille $n_1\times n_1$, ces super-tuiles peuvent à leur tour être regroupées (de nouveau, de manière unique) en blocs de taille $n_2\times n_2$  (qui sont de taille $(n_1\cdot n_2 \times n_1\cdot n_2)$ s'ils sont mesurés en tuiles individuelles), et ainsi de suite. Au $k$-ème rang de cette hiérarchie, nous avons des super-tuiles de rang $k$ formées de  $n_k\times n_k$ super-tuiles de rang $(k-1)$, ce qui représente un carré de taille $(n_1\cdot \ldots \cdot n_k) \times (n_1 \cdot \ldots \cdot n_k)$ en tuiles individuelles. 

Chaque super-tuile de rang $k$ peut jouer plusieurs rôles dans sa super-tuile mère de rang $k+1$ à laquelle elle appartient : celui d'une cellule des ``câbles de communication'' (en contribuant éventuellement à une super-couleur), celui d'une cellule dans la ``zone de calcul'' ou le rôle de bloc de construction, voir Fig.~\ref{f:tuile-complète}.

Gardons également à l'esprit que la zone de calcul d'une super-tuile $T$ est découpée horizontalement en trois zones. La première zone, en bas, effectue un tri sur les données d'entrée, en regroupant les parties des données d'entrée garantissant la structure hiérarchique des pavages (à savoir le code d'un programme $\pi$, la représentation binaire du rang $k$ et l'information présente dans les 4 super-couleurs, indiquant pour chacune des coordonnées et des bits d'information supplémentaires) ; elle correspond au diagramme espace-temps d'une machine de Turing $I$ à deux têtes de lecture. La zone centrale exploite alors ces données, et garantit la structure hiérarchique ; elle correspond au diagramme espace-temps d'une machine de Turing universelle $U$, simulant le programme $\pi$. Enfin, la partie haute correspond aux vérifications garantissant que les pavages ont la propriété souhaitée : i.e., que leur projection appartient respectivement aux shifts $S_\varepsilon$ et $S_\varepsilon'$, et que toute configuration de ces shifts peut être obtenue par une telle projection.

Regardons plus en détail les données d'entrée d'une machine $\pi$ présente dans une super-tuile $T$ de rang $k$. D'après la Partie~\ref{s:shift-auto-similaire}, elles sont constituées de trois champs : (i) le code de $\pi$ pour la machine $U$, (ii) la représentation binaire du rang $k$ et (iii) les parties des 4 super-couleurs nécessaires pour garantir la structure hiérarchique. Pour chacune des 4 parties du champ (iii), nous avons un ``sous-champ'' encodant des coordonnées (soit $O(\log n_{k+1})$ bits) et un sous-champ encodant des bits supplémentaires (soit $O(1)$ bits). Nous avons vu que ces champs et sous-champs sont encodés de manière à ce qu'une machine de Turing soit capable de déterminer où commence et où finit chacun d'eux.
Il est naturel de vouloir considérer d'une part les sous-champs encodant des coordonnées, et d'autre part les sous-champs encodant les bits supplémentaires. Nous pouvons alors dire que la machine universelle $U$ reçoit en entrée les données suivantes, que nous appellerons Champ~1, $\ldots$, Champ~4 :

\begin{enumerate}
\item[1] le code de la machine de Turing $\pi$ pour $U$ ;
\item[2] la représentation binaire de l'entier $k$ (le rang de la super-tuile dans la hiérarchie)~;
\item[3] les coordonnées (la représentation binaire de ces coordonnées) de 4 super-tuiles filles de rang $k$ à l'intérieur de sa super-tuile mère $M$ de rang $k+1$, qui sont deux nombres entiers de l'intervalle $[0, n_{k+1}-1]$. Il s'agit des coordonnées de la super-tuile $T$ et de ses voisines du bas, de droite et du haut ;
\item[4] les $O(1)$ bits supplémentaires, présents sur chacune des 4 super-couleurs, nécessaires pour vérifier que la super-tuile $T$ joue correctement son rôle vis-à-vis de ses 4 voisines. Si $T$ joue le rôle d'un bloc de construction, les 4 sous-champs sont vides. Si la super-tuile joue le rôle d'une cellule des câbles de communication, alors deux de ses sous-champs sont constitués d'un bit, représentant l'information transmise par le câble, et les deux autres sont vides. Si la super-tuile joue le rôle d'une cellule de la zone de calcul, alors les 4 sous-champs contiennent $O(1)$ bits (le nombre exact de bits dépend de si $T$ simule une cellule d'un diagramme espace-temps de $I$, de $U$ ou de $A$).
\end{enumerate}

Les vérifications effectuées par $U$ au niveau de la zone centrale de la zone de calcul de $T$ sont les suivantes : 
\begin{itemize}
	\item[(A)] les coordonnées de $T$ sont cohérentes avec celles de ses voisines (Champ~3). En particulier, un calcul à partir de $k$ (Champ~2) nous permet de savoir si $T$ est située au niveau d'un bord de sa super-tuile mère (les coordonnées étant modulo $n_{k+1}$)~;
	\item[(B)] les bits supplémentaires garantissant la structure hiérarchique sont cohérents entre eux (Champ~4). Pour cela, nous déterminons le rôle joué par $T$ au sein de sa super-tuile mère $M$ via un calcul à partir de ses coordonnées (Champ~3). Si $T$ n'est pas située au niveau des Champs~1 ou 2 de $M$, la vérification est immédiate. Sinon, la tâche est plus subtile : si $T$ est située dans le Champ~1 de $M$, $U$ doit récupérer un bit sur son propre ruban, situé dans le Champ~1 de $T$. Précisément, si $T$ correspond à la $j$-ème cellule du Champ~1 de $M$, il s'agit du bit encodé à la position $j$ dans le Champ~1 de $T$. La machine $U$ vérifie alors que ce bit est bien encodé dans le Champ~4 de $T$. Si $T$ appartient au Champ~2 de sa super-tuile mère, $U$ doit récupérer l'entier $k$ (soit le Champ~2 de $T$, qui est écrit sur son ruban), incrémenter $k$ de 1, puis vérifier que le bit de la représentation binaire de $k+1$ correspondant à la position de $T$ dans le Champ~2 de $M$ est bien encodé dans le Champ~4 de $T$.
\end{itemize}

Il est bien connu que beaucoup d'algorithmes standards (i.e., les opérations arithmétiques avec des nombres entiers encodés par leur représentation binaire) peuvent être effectués par une machine de Turing en temps polynomial. De manière plus générale, la thèse de Church-Turing polynomiale affirme que tout calcul effectué par un ordinateur en un temps polynomial peut également être effectué en temps polynomial par une machine de Turing.  

En outre, nous avons choisi $U$ de telle manière qu'un calcul polynomial soit simulé par $U$ également en temps polynomial (le polynôme peut être différent). Par conséquent, il n'est pas nécessaire de décrire plus précisément les algorithmes effectuant les tâches $A$ et $B$ en détail. Il suffit en effet de remarquer que les entrées sont de tailles logarithmiques (en $O(\log n_{i+1})$). Les calculs effectués par $U$ se terminent donc après un nombre d'étapes en $poly(n_{i+1}) \ll q$ pour $N$ assez grand ($q$ étant une fraction de $N$).

\paragraph{Les données d'entrée de l'automate cellulaire.}

Maintenant que nous avons rappelé les caractéristiques principales de la construction de la Sous-Partie~\ref{s:shift-auto-similaire}, décrivons les modifications à y apporter pour obtenir la construction finale. Tout d'abord, nous allons ajouter des informations au niveau des super-couleurs des super-tuiles, en utilisant les bits inutilisés que nous avions réservés dans la Partie~\ref{s:shift-auto-similaire}. Cette information se retrouvera au niveau des données d'entrée  de la zone de calcul. Les machines $I$ et $\pi$ ne seront pas modifiées : les vérifications garantissant que la projection du shift appartient respectivement à $S_\varepsilon$ ou $S_\varepsilon'$ seront effectuées par l'automate, dont la description conclura la preuve.

Intuitivement, chaque super-tuile ``échange'' avec ses 4 super-tuiles voisines un certain nombre de champs. Concrètement, pour deux super-tuiles voisines $T$ et $T'$, dans la super-couleur qu'elles partagent sont encodés à la fois les champs de la super-tuile $T$ et ceux de la super-tuile $T'$. Nous choisissons un encodage permettant à une machine de Turing de déterminer où commence et finit chaque champ, et pour les champs composés d'une liste où commence et finit chaque élément de la liste. Détaillons ces différents champs pour une super-tuile $T$ de rang $k$ ; nous les appellerons Champ~5, Champ~6, Champ~7 et Champ~8 :

\begin{enumerate}
\item[5] la liste de \emph{tous} les ``points noirs'' présents dans la super-tuile de rang $k$. Nous pouvons noter que la taille de la liste n'est pas plus grande que $(N_k) ^{\epsilon} = (n_1\cdot\ldots \cdot n_k)^\epsilon$ (chaque point est représenté par ses coordonnées, en partant du coin en bas à gauche de la super-tuile, i.e. par deux nombres de l'intervalle $[0,N_{k}-1]$) ; 

\item[6] la liste de \emph{quelques} ``points noirs'' de sa super-tuile mère. Chaque point de la liste est représenté par ses coordonnées (en partant du coin en bas à gauche de la super-tuile mère, soit par deux entiers de l'intervalle $[0, N_{k+1}-1]$) et par un entier compris entre 1 et 20, représentant une ``flèche'' (voir ci-dessous).
Nous pouvons noter qu'il n'y a pas assez de place dans la zone des données d'entrée de la super-tuile de rang $k$ pour représenter \emph{toutes} les positions des points noirs présents dans sa super-tuile mère de rang $(k+1)$. C'est pourquoi nous ne pouvons garder dans cette liste que \emph{quelques} points noirs de la super-tuile mère.
La taille de la liste doit être en $O(N_k^\epsilon)$, soit ne pas être plus grande (à une constante multiplicative près) que la taille de la liste des points noirs présents dans cette super-tuile (ce qui peut être beaucoup moins que le nombre total de points noirs présents dans la super-tuile mère). 
En outre, à chaque point noir de cette liste est associé une flèche. Intuitivement, chacun de ces points noirs provient d'une super-tuile initiale (il est présent dans les Champs~5 et 6 de cette super-tuile), transite de proche en proche à travers des super-tuiles voisines (il est présent dans le Champ~6 de celles-ci), et enfin arrive à une super-tuile finale (il est présent dans le Champ~6 de celle-ci, et les bits représentant ses coordonnées sont dans le Champ~7 de cette dernière). Dans la super-tuile initiale on associe au point dans le Champ~6 une flèche sortante vers la gauche, le haut, la droite ou le bas (soit 4 possibilités). Pour une super-tuile de transit, on associe au point une flèche allant d'un des 4 bords (le bord entrant) à un des 3 autres (le bord sortant), soit 12 possibilités. Enfin, pour une super-tuile finale on associe au point une flèche entrante (soit de nouveau 4 possibilités). Ainsi, à chaque point noir de la liste constituant un Champ~6 est associé une des 20 flèches possibles (un entier entre 1 et 20), indiquant la direction de circulation du point noir (voir Fig.~\ref{f:parcours-point}) ;

\begin{figure}[H]
	\center
	\includegraphics[scale=0.8]{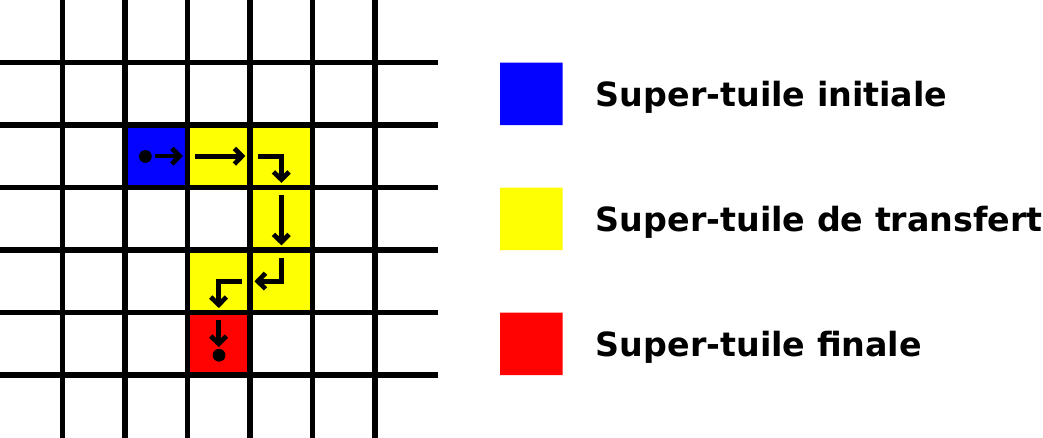}
	\caption[Chaîne de super-tuiles ayant un point $p$ présent dans leur Champs~6.]{Chaîne de super-tuiles ayant un point $p$ présent dans leur Champs~6. Les flèches associées au point $p$ dans chacun des Champs~6 sont indiqués.}\label{f:parcours-point}
\end{figure}

\item[7] éventuellement, une partie des bits encodant les coordonnées d'\emph{un} point noir de la super-tuile mère. En effet, une super-tuile $T$ de rang $k$ ne peut encoder (au niveau de son Champ~4) qu'un seul bit du Champ~5 de sa super-tuile mère. Ainsi, le premier bit des coordonnées d'un point noir d'une super-tuile finale sera encodé par $T$, et le reste des bits sera communiqué à la super-tuile voisine de droite (qui encodera le second bit des coordonnées du point noir dans son Champ~4 et transmettra les bits restants à sa propre voisine de droite), et ainsi de suite. Ce champ représente donc un nombre de bits compris dans l'intervalle $[0,2 \lceil \log(N_{k+1}-1) \rceil]$ ;

\item[8]  la liste de tous les points noirs de la \emph{zone de responsabilité} de la super-tuile, i.e. ceux présents dans la super-tuile, dans sa voisine de gauche, dans celle du dessus et dans celle au dessus de la voisine de gauche, soit un bloc de taille $(2N_k)\times (2N_k)$. Le nombre de points noirs dans la liste est donc au plus $(2N_k)^\epsilon$ . Ils sont représentés par leurs coordonnées dans ce bloc, en partant du coin en bas à gauche de la super-tuile voisine de gauche, soit par deux entiers de l'intervalle $[0, 2N_k -1]$.
\end{enumerate}

\begin{remark}
Nous pouvons observer que de même que les Champs~1 et 2 (correspondant au code de la machine de Turing $\pi$ pour $U$, et à $k$ le rang de la super-tuile), et contrairement aux autres Champs, le Champ~5 d'une super-tuile de rang $k$ n'est pas obtenu via des données échangées avec des super-tuiles voisines de même rang, mais est déterminé par ses super-tuiles filles de rang $k-1$.
\end{remark}

Outre ces 4 champs, l'automate cellulaire doit également connaître le rang $k$ de la super-tuile $T$ (représenté dans le Champ~2), les coordonnées de $T$ dans sa super-tuile mère (représentées dans le Champ~3 de ses super-couleurs de gauche et du bas) et les bits supplémentaires présents dans le Champ~4.
Les données d'entrée de l'automate cellulaire $A$ sont donc constituées du Champ~2 et des super-couleurs de gauche, du bas, de droite et du haut (chaque super-couleur est constituée des Champs~3 et 4, des Champs~5 à 8 de la super-tuile, et des Champs~5 à 8 de la super-tuile voisine correspondante). Le Champ~1 (le code du programme $\pi$) n'est pas utilisé par l'automate cellulaire.

Par ailleurs, la taille de tous ces champs est en $O(N_k^{\epsilon}) \cdot \poly (\log N_{k+1})$.

Nous commençons par expliquer dans les grandes lignes les calculs effectués par l'automate, puis nous détaillerons ceux-ci. Enfin, dans un troisième temps nous verrons comment implémenter ceux-ci de manière suffisamment efficace avec un automate cellulaire.

\paragraph{Présentation générale des calculs effectués par l'automate cellulaire.}
Au niveau de la zone haute, représentant le diagramme espace-temps de l'automate cellulaire $A$, les vérifications suivantes ont lieu séquentiellement :
\begin{itemize}
	\item[(C)] les données envoyées par la super-tuile à ses 4 super-tuiles voisines sont identiques ;
	\item[(D)] les flux d'information sont propagés correctement :
	\begin{itemize}
		\item[(D1)] les points noirs de la super-tuile (son Champ~5) sont ajoutés à son flux (son Champ~6),
		\item[(D2)] il y a conservation du flux : le flux de la super-tuile est cohérent avec les flux de ses voisines (leur Champ~6),
		\item[(D3)] un point noir peut arriver à destination : le flux de la super-tuile (son Champ~6) est cohérent avec d'une part le bit encodé dans le Champ~4 de $T$, et d'autre part avec les bits des coordonnées du point restant à encoder (son Champ~7). Cela est possible si $T$ appartient au Champ~5 de sa super-tuile mère $M$. De plus, $T$ doit soit se situer au début du Champ~5 de $M$, soit que sa voisine de gauche encode le dernier bit des coordonnées d'un point noir de $M$ (i.e, la taille du Champ~7 de sa voisine de gauche est égale à 1) ;
	\end{itemize}	
	\item[(E)] un point noir arrivé à destination est correctement représenté dans le Champ~5 de sa super-tuile mère ; 
	\item[(F)] le Champ~8 est égal à la liste des points de la zone de responsabilité : la liste des points noirs contenus dans $T$ (Champ~5) et celle de sa zone de responsabilité (Champ~8) sont cohérentes avec les listes des super-tuiles voisines de $T$ ;
	\item[(G)] la liste des points noirs de la zone de responsabilité de la super-tuile de rang $k$ est correcte, dans le sens où elle ne contient pas de motif interdit (parmi ceux énumérés par les super-tuiles de rang $k$) : cette vérification a besoin des Champs~2 (représentant l'entier $k$) et 8 en entrée.
\end{itemize}
Dans ce qui suit nous discutons de ces vérifications plus en détail.

\paragraph{Propriété~C : les données envoyées sont identiques.} Cette vérification est directe. Dans les données d'entrée se trouvent encodées les 4 super-couleurs. Comme la super-tuile envoie ses Champs~5 à 8 à ses 4 super-tuiles voisines, cette information doit se trouver en quadruple au niveau des données d'entrée (une fois par super-couleur) : nous vérifions qu'il s'agit bien 4 fois des mêmes données. 

\paragraph{Propriété~D : les flux d'information sont correctement propagés.} Commençons par la propriété $(D1)$. Pour chaque point noir $p$ du Champ~5 d'une super-tuile $T$, nous vérifions qu'il apparaît une et une seule fois dans son Champ~6 (après avoir ajouté les coordonnées de la super-tuile $T$ vis-à-vis de sa super-tuile mère $M$ à celles de $p$, pour obtenir les coordonnées de $p$ vis-à-vis de $M$). 

Dans ce cas, la super-tuile $T$ est appelée la super-tuile \emph{initiale} de ce point, et une flèche sortante (on note le bord sortant $s$) est associée au point dans ce Champ~6. Le point $p$ doit apparaître une et une seule fois dans le Champ~6 d'exactement une super-tuile $V$, sœur et voisine de $T$ ($T$ et $V$ partagent la même super-tuile mère). Les flèches doivent être ``cohérentes'' : si la flèche associée à $p$ dans le Champ~6 de $T$ ``sort'' par un bord $b$, alors $T$ et $V$ doivent être en contact au niveau de $b$, et la flèche associée au point $p$ du Champ~6 de $V$ doit ``entrer'' par $b$. 

Pour chaque autre point $p'$ dans la liste du Champ~6 de la super-tuile $T$, nous vérifions que : 
\begin{itemize}
\item soit ``il passe à travers lui''. Cela correspond à la propriété $(D2)$. Dans ce cas, le même point $p'$ doit apparaître dans les Champs~6 d'exactement deux super-tuiles sœurs et voisines de $T$ (une et une seule fois dans chaque champ). La super-tuile $T$ est dans ce cas une super-tuile de \emph{transit} pour ce point. 

Là aussi, les flèches doivent être cohérentes. La flèche associée à $p'$ dans le Champ~6 de $T$ entre par un des 4 bords (on note ce bord entrant $b_1$), et sort par un des 3 autres bords (on note ce bord sortant $b_2$). Une des deux super-tuiles voisines, appelée $V_1$, doit être en contact avec $T$ au niveau du bord $b_1$ de $T$ ; dans $V_1$, la flèche associée à $p'$ doit avoir $b_1$ comme bord sortant ($V_1$ peut donc être une super-tuile initiale ou une super-tuile de transit). La seconde de ces super-tuiles voisines de $T$, appelée $V_2$, doit être en contact avec $T$ au niveau du bord $b_2$ de $T$ ; dans $V_2$, la flèche associée au point $p'$ doit avoir $b_2$ comme bord entrant ($V_2$ peut donc être une super-tuile de transit ou une super-tuile ``finale'', voir ci-dessous). 

\item soit ``il est arrivé à destination'', ce qui correspond au cas $(D3)$. Dans ce cas, le même point doit apparaître une et une seule fois dans le Champ~6 d'exactement une super-tuile $V$ sœur et voisine de $T$ ; la super-tuile $T$ étant appelée la super-tuile finale de $p'$. La flèche associée à $p'$ dans la super-tuile $T$ est alors une flèche entrante depuis un des 4 bords $b$, et les flèches doivent être cohérentes ($V$ et $T$ sont en contact au niveau de $b$, et la flèche associée au point $p'$ dans le Champ~6 de $V$ doit sortir par $b$).
Ce cas est possible uniquement si cette super-tuile de rang $k$ joue le rôle d'une cellule du Champ~5 des données d'entrée de sa super-tuile mère $M$ de rang $k+1$. Nous devons être plus précis ici. En effet, le Champ~5 de $M$ est représenté 4 fois dans les données d'entrée (une fois par super-couleur), et il faut décider dans lequel de ceux-ci pourront ``arriver à destination'' les points noirs. Nous pouvons par exemple choisir le Champ~5 correspondant à la super-couleur de gauche. Nous ne mettons pas de contrainte supplémentaire sur les Champs~5 correspondant aux trois autres super-couleurs ; en effet, la Propriété~C nous garantit qu'ils seront égaux au Champ~5 de la super-couleur de gauche, contenant les points noirs des super-tuiles filles de $M$, et partant les points noirs de $M$.
\end{itemize}

\paragraph{Propriété~E : les coordonnées d'un point noir ``arrivé à destination'' au niveau d'une super-tuile de rang $k$ sont correctement représentées dans sa super-tuile mère.} 
Soit $p$ un point noir arrivé à destination au niveau d'une super-tuile $T$. En pratique, $p$ est représenté de manière standard par ses coordonnées vis-à-vis de $M$, la super-tuile mère de $T$. Cette représentation consiste ainsi en une paire de nombres entiers de l'intervalle $[0,N_{k+1}-1]$, autrement dit par $m=O(\log N_{k+1})$ bits (là encore, une machine de Turing parcourant l'élément doit pouvoir déterminer là où celui-ci commence et finit). Ce point noir $p$ doit donc être représenté, au niveau du Champ~5 de $M$ de rang $k+1$, par $m$ super-tuiles filles de $M$ de rang $k$ (chacune d'elles codant un bit). Celles-ci doivent être disposées les unes à la suite des autres, la première étant $T$, la super-tuile finale de $p$. En pratique, ces bits sont encodés dans les Champs~4 de ces $m$ super-tuiles.
Pour mener à bien cette représentation de $p$ dans la super-tuile mère $M$, il est nécessaire que les $m$ super-tuiles filles aient connaissance des bits représentant $p$ (ou du moins des bits de $p$ non encore encodés). Nous utilisons pour cela le Champ~7 des données d'entrée. Le Champ~7 de $T$, la super-tuile finale de $p$, est égal à la totalité des coordonnées du point $p$. Considérons $V$ la super-tuile voisine \emph{de droite} de $T$. Alors $V$ ``connaît'' le Champ~7 de $T$, via la super-couleur qu'elles partagent ; $V$ vérifie alors que son propre Champ~7 est égal à celui de $T$, une fois ôté le premier bit du champ. La super-tuile $V$ peut alors vérifier que le bit encodé dans son Champ~4 correspond au premier bit de son Champ~7. De manière générale, soit $T'$ une super-tuile située dans le Champ~5 de sa super-tuile mère qui n'est la super-tuile finale d'aucun point, et soit $V$ sa voisine de \emph{gauche}. Si le Champ~7 de $V$ est vide, alors celui de $T$ l'est aussi. Sinon le Champ~7 de $T'$ est égal au Champ~7 de $V$ une fois ôté le premier bit. Dans ce cas, si le Champ~7 de $T'$ est non vide, son Champ~4 encode le premier bit de son Champ~7.

Il nous reste à ajouter une contrainte pour qu'une super-tuile $T$ puisse être finale pour un point $p$. Nous avons vu que $T$ devait être située dans le Champ~5 de sa super-tuile mère $M$. Alors :
\begin{itemize}
	\item soit la super-tuile $T$ est située au début du Champ~5 de $M$ ;
	\item soit la super-tuile voisine de gauche $V$ de $T$ encode le dernier bit d'un autre point noir $p'$. Dans ce cas-là le Champ~7 de $V$ ne contient qu'un seul bit.
\end{itemize}
Cette contrainte supplémentaire nous permet de nous assurer que d'une part une super-tuile ne doit encoder qu'un seul bit du Champ~5 de sa super-tuile mère $M$, et que d'autre part dans le Champ~5 de $M$ les représentations des points noirs de $M$ sont les unes à la suite des autres (sans ``trou'', voir Fig.\ref{f:parcours-points}).

\begin{figure}[H]
	\center
	\includegraphics[scale=0.6]{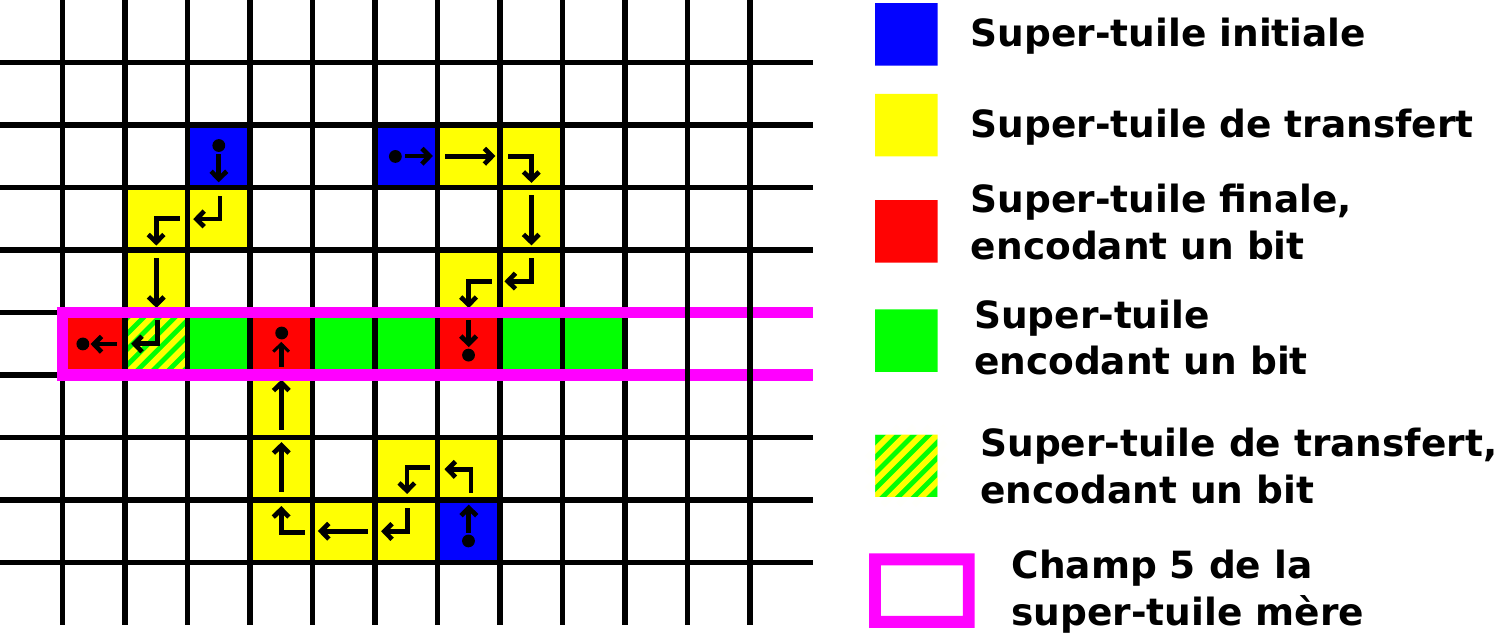}
	\caption[Encodage des coordonnées de points dans le Champ~5 de la super-tuile mère.]{Trois chaînes de super-tuiles ayant chacune un même point présent dans leur Champs~6. Chacun des trois bits des coordonnées de ces points est encodé au niveau du Champ~5 de la super-tuile mère (les premiers bits sont encodés au niveau des carrés rouges, les deuxièmes et troisièmes bits au niveau des carrés verts).}\label{f:parcours-points}
\end{figure}

Les contraintes expliquées ci-dessus garantissent que dans un pavage valide, il y ait un ``flux d'information'' qui passe par chaque super-tuile de rang $k$ (éventuellement, pour certaines super-tuiles ne contenant pas de point noir, ce flux peut être nul). Ce flux consiste en ``points noirs'' (leurs coordonnées par rapport à la super-tuile mère de rang $k+1$). Chaque point commence son trajet au niveau de la super-tuile de niveau $k$ à laquelle il appartient (celle-ci joue le rôle de super-tuile initiale pour ce point ; c'est là que l'information du Champ~5 est transférée au Champ~6). Ensuite, l'information est transférée à travers une suite de super-tuiles voisines (le cas où ``il passe à travers'', ces super-tuiles jouant le rôle de super-tuiles de transit pour ce point) jusqu'au moment où ``il est arrivé à destination'', qui est une des cellules des données d'entrée de la zone de calcul de sa super-tuile mère, jouant le rôle de super-tuile finale pour ce point (c'est là que les coordonnées du point passent du Champ~6 au Champ~7). Une chaîne de super-tuiles voisines, commençant au niveau de la super-tuile finale et allant vers la droite, encode alors ce point noir dans le Champ~5 de la super-tuile mère.

Ces propriétés sur ces ``flux d'information'' sont suffisantes pour avoir le résultat suivant :

\begin{claim}\label{claim:info-flow}
Dans un pavage valide, tous les points noirs du Champ~5 d'une super-tuile de niveau $k$ apparaissent également dans le Champ~5 de sa super-tuile mère de rang $(k+1)$. Réciproquement, tous les points noirs apparaissant dans le Champ~5 d'une super-tuile de rang $(k+1)$ doivent appartenir au Champ~5 d'une de ses super-tuiles filles de rang $k$.
\end{claim}

\begin{proof}
Notre construction garantit que chaque point noir $p$ présent dans le Champ~5 d'une super-tuile $T$ de rang $k$ est également représenté dans le Champ~6 de celle-ci ($T$ est la super-tuile initiale de $p$) ; par ailleurs, $p$ doit également être représenté dans le Champ~6 d'exactement une super-tuile voisine. On considère alors la chaîne des super-tuiles sœurs voisines de même rang, $T_0, T_1, T_2, \ldots T_l \ldots $, pour lesquelles $p$ est représenté dans leurs Champs~6.

Cette chaîne ne peut pas boucler sur elle-même : ni au niveau de $T_0$, puisque cette super-tuile vérifie que $p$ est contenu dans le Champ~6 d'exactement une super-tuile voisine (ici $T_1$) ; ni au niveau d'une super-tuile $T_l$ pour $1 \leq l$, car la super-tuile $T_l$ vérifie que $p$ est contenu dans les Champs~6 d'exactement deux super-tuiles voisines (ici $T_{l-1}$ et $T_{l+1}$). Par ailleurs, la taille de la chaîne est bornée par le nombre de super-tuiles de rang $k$ composant une super-tuile de rang $k+1$, soit $n_{k+1}^2$, puisque toutes les super-tuiles de la chaîne sont sœurs (elles partagent la même super-tuile mère). La chaîne se termine donc au niveau d'une super-tuile $T_f$ ; celle-ci n'a alors qu'une seule super-tuile voisine pour laquelle $p$ est représenté dans le Champ~6 (les super-tuiles $T_i$ de la chaîne, $1\leq i \leq f-1$, étant des super-tuiles de transit pour $p$). De plus, la flèche associée à $p$ dans $T_f$ est entrante, cette super-tuile est donc la super-tuile finale de $p$. 
Par conséquent, le point noir $p$ arrive à destination au niveau de la super-tuile $T_f$, qui ne peut être qu'un emplacement dans les données d'entrée de la zone de calcul de la super-tuile mère. Plus précisément, cet emplacement est soit situé au début du Champ~5 de la super-tuile mère, soit ailleurs dans le Champ~5 mais alors à droite d'une super-tuile voisine encodant le dernier bit d'un autre point noir $p'$. Dans les deux cas, le point $p$ est représenté par une chaîne de super-tuiles voisines commençant par $T_f$ et allant vers la droite. 

Il est clair que le premier point noir présent dans la super-tuile mère $T$ est situé au début du Champ~5 de $T$. En effet, si $T$ n'a pas de point noir le résultat est trivialement vrai. Sinon, on considère un point noir $p_1$ de $T$. Celui-ci est présent dans le Champ~5 de $T$. S'il n'est pas au début du Champ~5, alors il existe un autre point noir $p_2$ encodé à gauche de la représentation de $p_1$ dans le Champ~5 de $T$. On peut réitérer l'argument~; comme le nombre de points noirs de $T$ est fini, on finit par trouver un point $p_l$ situé au début du Champ~5 de $T$. Par ailleurs, nous avons vu que tous les points noirs du Champ~5 de $T$ étaient encodés les uns à la suite des autres, sans ``trou'' : le Champ~5 de $T$ est donc constitué d'une liste $L$ de points noirs, et tous les points noirs présents dans $T$ y appartiennent.

Réciproquement, soit $M$ une super-tuile de rang $k+1$. Nous pouvons noter que la taille des représentations de tous les points noirs de $M$ est identique, et nous notons $m$ cette valeur.
Considérons $T_0$ la super-tuile fille de $M$ de rang $k$, située au début de son Champ~5. Si le Champ~7 de $T_0$ est nul, alors le Champ~7 de $T_1$, sa super-tuile voisine de droite, est également nul : $T_1$ ne peut être la super-tuile finale d'aucun point noir puisque d'une part elle n'est pas située au début du Champ~5 de $M$, et d'autre part que la taille du Champ~7 de $T_0$, sa voisine de gauche, n'est pas égale à 1. Par ailleurs, le Champ~7 de $T_1$ ne peut être égal au Champ~7 de $T_0$ moins le premier bit, puisque celui-ci est nul. Nous pouvons réitérer l'argument de proche en proche pour montrer que les Champs~7 de toutes les super-tuiles du Champ~5 de $M$ sont nuls, et donc que leur Champ~4 n'encodent pas d'information. Le Champ~5 de $M$ ne contient dans ce cas aucun point noir (et trivialement tous les points noirs de $M$ appartiennent à un des Champs~5 des super-tuiles filles de $M$).

Si le Champ~7 de $T_0$ n'est pas nul, alors il est de taille $m$ : comme la super-tuile voisine de gauche de $T_0$ n'appartient pas au Champ~5 de la super-tuile mère $M$, son Champ~7 est nul ; nous sommes donc dans le cas où $T_0$ est une super-tuile finale pour un point noir $p_0$. Dans ce cas, nous avons vu que $p_0$ est encodé dans le Champ~5 de $M$, par $T_0$ et les $m-1$ super-tuiles à droite de $T_0$. Par ailleurs, comme $T_0$ est la super-tuile finale de $p_0$, le Champ~6 de $T_0$ contient également $p_0$ (et la flèche associée est entrante). Nous pouvons tenir le même raisonnement que ci-dessus pour montrer qu'il existe une chaîne finie de super-tuiles de rang $k$ pour lesquelles $p_0$ apparaît dans leur Champ~6 (la chaîne ne boucle pas et est circonscrite aux super-tuiles filles de $M$). Ces super-tuiles sont des super-tuiles de transit pour $p_0$, exceptée la dernière, $T_0'$, qui est la super-tuile initiale de $p_0$ (cela est garantie par les flèches associées à $p_0$). Le premier point noir du Champ~5 de $M$ apparaît donc dans le Champ~5 de $T_0'$.

Nous considérons alors $T_1$, la $(m+1)$-ème super-tuile du Champ~7 de $M$. Nous pouvons tenir essentiellement le même raisonnement que pour $T_0$ : soit le Champ~7 de $T_1$ est nul, et dans ce cas $p_0$ est le seul point noir de $M$ ; sinon le Champ~7 est de taille $m$, et $T_1$ et les $(m-1)$ super-tuiles à droite de $T_1$ représentent un point noir $p_1$, qui apparaît dans les Champs~5 de $M$ et de $T_1'$, où $T_1'$ est une super-tuile fille de $M$. Nous pouvons réitérer cette argumentation jusqu'à soit trouver une super-tuile de rang $k$, appartenant au Champ~5 de $M$, pour laquelle le Champ~7 est nul ; soit avoir parcouru la totalité des super-tuiles de rang $k$ appartenant au Champ~5 de $M$. 
\end{proof}

Nous pouvons alors prouver par induction que chaque super-tuile de niveau $k$ ``connaît consciemment'' tous les points noirs contenus en elle (ils sont contenus dans le Champ~5 des données d'entrée de sa zone de calcul), et seulement ceux-ci (les points noirs contenus dans le Champ~5 correspondent tous à des points noirs ``réels'' présents dans la super-tuile).

Nous pouvons observer que notre construction autorise également des boucles, constituées de super-tuiles de transit pour un point $p$. Ce point $p$ peut ne pas correspondre à un point réel ; en effet, il ne vient d'aucune source (d'aucune super-tuile initiale), voir Fig.~\ref{f:parcours-boucle}. 

\begin{figure}[H]
	\center
	\includegraphics[scale=0.6]{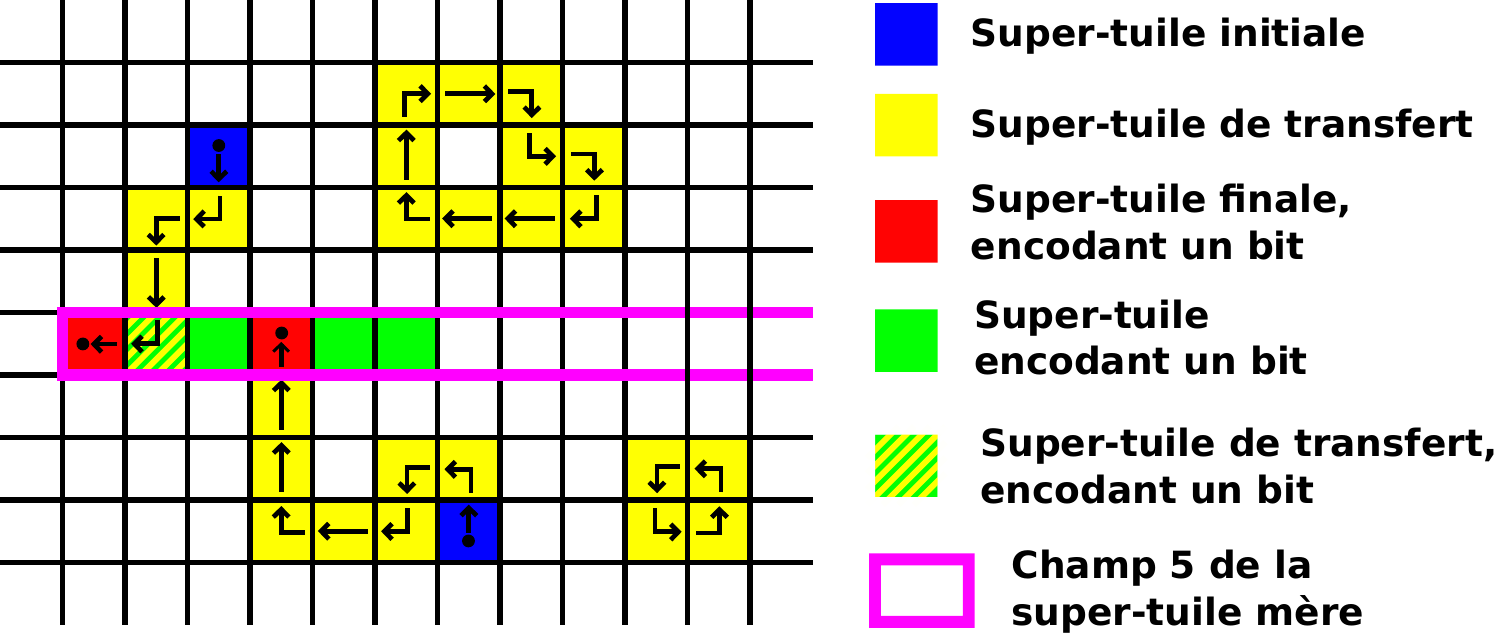}
	\caption[Des chaînes de super-tuiles de transit qui bouclent.]{Deux chaînes de super-tuiles de transit, ayant chacune un point présent dans leur Champs~6, qui bouclent. Ces deux points peuvent ne pas correspondre à des points réels (ils ne proviennent pas du Champ~5 d'une super-tuile de rang $k$), mais ne seront pas encodés au niveau du Champ~5 de la super-tuile mère de rang $(k+1)$.}\label{f:parcours-boucle}
\end{figure}

Cependant, il n'arrive à aucune destination (aucune super-tuile finale). En d'autres termes, les Champs~6 des super-tuiles de rang $k$ peuvent contenir des ``flux parasites'' de points, mais ces derniers n'apparaissent jamais dans les Champs~5 des super-tuiles de rang $k$ ni dans celui de leur super-tuile mère.

\paragraph{Propriété~F : la liste des points noirs de la zone de responsabilité.}
La zone de responsabilité de la super-tuile $T$ de rang $k$ est un bloc de super-tuiles de taille $2 \times 2$, ou de taille $(2N_k)\times(2N_k)$ en tuiles individuelles. Un point noir présent dans la zone sera donc représenté par deux entiers de l'intervalle $[0,2N_k-1]$, en partant du coin en bas à gauche du bloc (i.e. le coin en bas à gauche de la voisine de gauche de $T$). 
Pour les points noirs situés dans $T$, dans sa voisine de gauche et dans sa voisine du dessus, les vérifications sont assez directes : 
\begin{itemize}\label{propriety-f}
	\item les points noirs du Champ~8 de $T$ d'abscisses et d'ordonnées comprises dans l'intervalle $[0,N_k-1]$ (i.e. la zone du bloc en bas à gauche) doivent correspondre exactement au Champ~5 de la super-tuile voisine de gauche ;
	\item ceux d'abscisses comprises dans $[N_k, 2N_k-1]$ et d'ordonnées comprises dans $[0,N_k-1]$ (i.e. la zone du bloc en bas à droite) doivent correspondre exactement au Champ~5 de la super-tuile $T$ (une fois ajouté l'entier $N_k$ aux abscisses des points noirs du Champ~5);
	\item ceux d'abscisses et d'ordonnées comprises dans l'intervalle $[N_k, 2N_k-1]$ (i.e. la zone du bloc en haut à droite) doivent correspondre exactement au Champ~5 de la super-tuile voisine du haut (une fois ajouté l'entier $N_k$ aux abscisses et ordonnées des points noirs de ce Champ~5).   
\end{itemize}

Pour les points noirs de la zone de responsabilité de $T$ présents dans la super-tuile au dessus de la super-tuile voisine de gauche de $T$, la vérification est effectuée au niveau de $V$, la super-tuile voisine de gauche de $T$.

Pour cela, $V$ va vérifier que les points noirs du Champ~8 de sa super-tuile voisine de droite (i.e. du Champ~8 de $T$) d'abscisses dans $[0,N_k-1]$ et d'ordonnées dans $[N_k, 2N_k-1]$ (i.e. dans la zone du bloc en haut à gauche) correspondent exactement au Champ~5 de sa super-tuile voisine du dessus $V'$ (une fois ajouté l'entier $N_k$ aux ordonnées des points noirs du Champ~7 de $V'$). Cette super-tuile $V'$ correspond bien à la super-tuile en diagonale, en haut à gauche, de $T$. 

\paragraph{Propriété~G : vérification de l'absence de motif interdit.} 

Nous devons vérifier que la ``zone de responsabilité'' de la super-tuile de rang $k$ ne contient pas de motif interdit de rang $k$. Nous effectuons ce calcul en trois étapes.

\paragraph{Étape G.1 : énumération des motifs interdits de rang $k$.} 
Cette étape commence avec le calcul de la liste des motifs interdits de rang $k$, voir p.~\pageref{ell-forbidden-patterns}. Par définition, cette tâche prend un temps assez court $\ell_k$ (qui est négligeable par rapport au temps disponible dans la zone de calcul de la super-tuile de rang $k$). La liste de tous les motifs interdits de rang $k$ possède au plus $\ell_k$ éléments, chaque élément contenant au plus $\ell_k$ points noirs dont les coordonnées appartiennent à l'intervalle $[0,\ell_k-1]$. La taille de la représentation binaire de la liste (au plus $\ell_k^2 \times 2 \times \log \ell_k$) est donc négligeable par rapport à la largeur de la zone de calcul.

Le Champ~8 contient la liste des ``points noirs'' de la zone de responsabilité de la super-tuile, et nous désignons sa taille par $t_k$. Dans une configuration valide, la taille de cette liste (le nombre de points noirs dans un carré de taille $(2N_k) \times (2N_k)$) n'est pas plus grande que $(2N_k)^\epsilon$.  

\paragraph{Étape G.2 : recherche des motifs interdits.} C'est l'étape la plus coûteuse en temps, où nous vérifions que la zone de responsabilité de la super-tuile de rang $k$ ne contient pas de motif interdit de rang $k$. 

Il y a au plus $\ell_k$ motifs, chacun d'eux est de taille au plus $\ell_k \times \ell_k$, et ne contient au plus que $\ell_k$ points noirs. 

\begin{lemma}\label{un-point-noir}
Soient $\epsilon<1$ et $S$ un sous-shift du shift de densité $\epsilon$. Alors si $S$ est non-vide, chaque motif interdit de $S$ contient au moins un point noir.
\end{lemma} 

\begin{proof}
Soient $\epsilon<1$ et $S$ un sous-shift du shift de densité $\epsilon$ possédant un motif interdit $M$ entièrement blanc. Montrons que $S$ est vide. Soit $C$ un carré de taille $c \times c$ contenant le support du motif $M$. Supposons par l'absurde que $S$ soit non vide, et considérons une configuration ${\cal C}$ de $S$. Soit un entier $n$ tel que $n>c$. Considérons un carré $D$ de taille $(n \cdot c)\times (n \cdot c)$ de ${\cal C}$ ; il est constitué de $n^2$ carrés de taille $c \times c$. Chacun de ces carrés doit contenir au moins un point noir (sinon le motif interdit $M$ y apparaîtrait). Par conséquent, $D$ possède au moins $n^2 > n \cdot c > (n \cdot c)^\epsilon$ points noirs, et est donc un motif interdit du shift de densité $\epsilon$, ce qui est absurde. Par conséquent $S$ est vide.
\end{proof}
 
Chaque motif interdit contient au moins un point noir d'après le lemme ci-dessus. Ainsi, pour chaque motif interdit $M$ de rang $k$ et pour chaque point noir $p$ de la liste, il est suffisant de vérifier que chaque translation de $M$ \emph{couvrant} $p$ n'est pas un motif de la configuration donnée. En d'autres termes, pour chaque point noir $p=(x,y)$ du Champ~8, pour chaque motif interdit $M$ de taille $l\times h$ (soit une liste de points noirs $L_M$, dont les coordonnées des points sont dans $[0,l-1]\times[0,h-1]$), pour chaque point noir $p'=(x',y')$ de $L_M$, nous devons vérifier qu'il existe au moins une paire $(i,j)\in [0,l-1]\times[0,h-1]$, telle que :
\begin{enumerate}
	\item soit le point de coordonnées $(x-x'+i,y-y'+j)$ appartient au Champ~8, et le point de coordonnées $(i,j)$ n'appartient pas à $L_M$ (il y a au moins un point noir apparaissant dans la configuration et pas dans le motif interdit), 
	\item soit le point de coordonnées $(x-x'+i,y-y'+j)$ n'appartient pas au Champ~8 et le point de coordonnées $(i,j)$ appartient à $L_M$ (il y a au moins un point noir n'apparaissant pas dans la configuration et apparaissant dans le motif interdit). 
\end{enumerate}
 
Comme la taille de la liste $L_M$, le nombre de points noirs par élément de la liste, les nombres $l$ et $h$ valent au plus $\ell_k$, cela nécessite de tester l'appartenance d'un point au Champ~8 un nombre au plus $\ell_k^4$ de fois par point du Champ~8.

Comme précédemment, si la vérification échoue le calcul doit créer une  ``erreur'' : il ne faut pas qu'il puisse y avoir dans ce cas de pavage valide dans la zone de calcul. Cela se produit si pour au moins un point $p$ du Champ~8 et au moins un point $p'$ d'au moins un motif interdit $M$ de taille $l\times h$, les $l \cdot h$ paires de points vérifiés (respectivement dans le Champ~8 et dans le motif interdit $M$) concordent toutes. 

Maintenant que nous avons détaillé les calculs effectués par l'automate cellulaire, nous allons voir comment ceux-ci peuvent être implémentés de manière efficace.

\paragraph{L'algorithme implémenté par l'automate cellulaire}

Contrairement aux implémentations d'algorithmes par une machine de Turing qui sont standards, leur implémentations par un automate cellulaire sont plus inhabituelles. Nous décrivons avec plus de détails techniques comment les procédures décrites ci-dessus peuvent être réalisées de manière très rapide (en temps quasi-linéaire), bien que nous n'explicitons pas la programmation de bas niveau en détail.

Pour pouvoir initialiser et mener à bien des calculs massivement parallèles, nous avons besoin d'ajouter de l'information dans certaines des cellules de l'automate sous forme d'un grand nombre de marques spéciales différentes. Intuitivement, ceux-ci nous permettront de commencer un même calcul au niveau de toutes les cellules ayant la même marque spéciale (le nombre de ces cellules n'étant pas borné quand le rang de la super-tuile augmente). La première étape des calculs de l'automate consiste ainsi en une étape d'initialisation durant laquelle sont ajoutées ces marques. Plus précisément, la tête de lecture principale va parcourir l'ensemble des données d'entrée (en un nombre $5q$ d'étapes) et ajouter les marques spéciales au fur et à mesure (nous avons vu que les champs sont encodés de telle manière qu'une machine de Turing, et à fortiori un automate cellulaire, puisse savoir où commence et finit chaque champ et chaque élément d'un champ). Deux marques spéciales différentes sont ajoutées au début et à la fin de chaque champ des données d'entrée (i.e., nous ajoutons cette information à la première et à la dernière cellule du champ), et, pour les trois champs qui sont des listes d'éléments (une paire d'entiers pour les Champs~5 et 8, et un triplet d'entiers pour le Champ~6), nous ajoutons deux autres marques spéciales au début et à la fin de chaque élément (excepté pour le début du premier élément et pour la fin du dernier élément, ces cellules ayant déjà reçu leur marque spéciale). Les marques sont différentes pour les champs correspondant à des super-couleurs différentes. 

Enfin, parmi les Champs~5 de la super-tuile $T$, celui correspondant à la super-couleur de gauche (contenant les super-tuiles filles de $T$ qui sont finales pour les points noirs de $T$) joue un rôle particulier. Ainsi, outre les quatre marques spéciales associées au début et à la fin du champ, et au début et à la fin de chaque élément, nous associons une autre marque spéciale aux restes des cellules du champ.  

Nous obtenons :
\begin{enumerate}
	\item $2$ marques spéciales pour le Champ~2 représentant $k$ le rang de $T$, et situées dans les $q$ premières tuiles des données d'entrée ;
	\item $2\times 2 \times 4$ marques pour les Champs~3 et 4, représentant les coordonnées et les bits supplémentaires de $T$ ; il y a $2$ marques par champ et par super-couleur ;
	\item $4\times 4+1$ marques pour le Champ~5, représentant les points noirs de $T$. Il y a $4$ marques pour le début du champ, la fin du champ, le début d'un élément et la fin d'un élément par super-couleurs ; plus une marque supplémentaire pour le champ correspondant à la super-couleur de gauche ;
	\item $4\times 4$ marques pour le Champ~6, représentant le flot de points noirs passant par $T$. Il y a $4$ marques par super-couleurs ;
	\item $2\times 4$ marques pour le Champ~7, utilisées pour encoder un point noir dans le Champ~5 de la super-tuile mère de $T$. Il y a $2$ marques par super-couleurs car ce champ n'est pas constitué d'une liste d'élément ;
	\item $4\times 4$ marques pour le Champ~8, représentant les points noirs de la zone de responsabilité de $T$ ;
\end{enumerate}
soit $75=O(1)$ marques spéciales.

Détaillons maintenant l'implémentation des vérifications garantissant les Propriétés~C à G. Le nombre d'étapes de chaque procédure pourra dépendre du \emph{nombre} d'éléments présents dans les Champs~5, 6 et 8 ; puisque ces trois listes d'éléments sont entièrement représentées par au plus $q$ bits (le nombre de bits encodés dans une super-couleur), nous pouvons borner la taille de ces listes par $q$. Cette borne est très large, mais elle suffira pour prouver que l'automate cellulaire a suffisamment de temps et d'espace pour mener à bien ses calculs. Ce nombre d'étapes dépendra également de la taille des Champs~2, 3, 4 et 7, ou de la \emph{taille} des éléments des Champs~5, 6 et 8. Nous pouvons borner ces tailles par $\poly(\log q)$, un polynôme dépendant de $\log q$.

Les vérifications des Propriétés~C à G sont effectuées séquentiellement. Détaillons maintenant comment sont implémentées les vérifications des Propriétés~C à G. Nous appelons la zone située à droite des $5q$ premières cellules de la zone de calcul, de largeur $N-6q = \frac{10N}{16}$, la ``mémoire'', que nous utiliserons pour les calculs intermédiaires. En particulier, les informations encodées dans les $5q$ premières cellules ne seront pas modifiées par l'automate cellulaire $A$.

\begin{itemize}
	\item[(C)] Cela revient à vérifier que 4 suites de bits de taille en $O(N_k^{\epsilon}) \cdot \poly (\log N_{k+1})$, toutes les 4 délimitées par des marques spéciales différentes, sont identiques. Cela peut être vérifié en $O(N_k^{\epsilon}) \cdot \poly (\log N_{k+1})$ étapes de calcul par une machine de Turing multi-têtes ;
	\item[(D)] Dans une super-tuile $T$ de niveau $k$, l'automate commence à déterminer $d= \lceil \log N_{i+1} \rceil - \lceil \log N_i \rceil$, la différence de taille entre une coordonnée dans $T$ et dans sa super-tuile mère $M$. Il recopie ensuite le Champ~5 de $T$ dans la zone mémoire en ajoutant $d$ bits initialisés à 0 devant chaque coordonnée, puis ajoute les coordonnées de $T$ (le Champ~3 de $T$) à chaque élément du Champ~5 recopié pour obtenir un Champ~5' (ce sont les coordonnées des points noirs de $T$ vis-à-vis de $M$). Cela peut être réalisé en temps $O(N_k^{\epsilon}) \cdot \poly (\log N_{k+1})$ par une machine de Turing multi-têtes, sans utiliser le parallélisme. Il recopie également le Champ~6 de $T$ dans la zone mémoire, que nous appelons Champ~6'. 
	
Pour chaque élément $p=(x,y)$ du Champ~5', nous vérifions s'il existe un unique élément $(x',y',f)$ du Champ~6' qui lui correspond (i.e., $(x,y)=(x',y')$). $T$ est la super-tuile initiale du point $p$. Nous vérifions ensuite que parmi les 4 Champs~6 des super-tuiles voisines, il en existe un et un seul dans lequel ce point de coordonnées $(x,y)$ apparaît une et une seule fois ; si c'est le cas, nous vérifions alors que les flèches associées à $p$ dans le Champ~6' et dans le Champ~6 de la voisine correspondante sont cohérentes. De plus, nous indiquons via un bit spécifique associé à la première case de l'élément dans le Champ~6' que celui-ci a été traité. La seule opération nécessitant le parallélisme est celle consistant à vérifier qu'un élément donné du Champ~5' apparaît une et une seule fois dans le Champ~6', et de localiser cet élément dans le Champ~6'. Cela peut être effectué en temps $O(\log N_{k+1})$ (la taille d'un élément de ces champs) d'après le Lemme~\ref{l:automate}.
	
Pour chaque point $p$ non déjà traité du Champ~6', nous vérifions que soit $p$ apparaît une et une seule fois dans \emph{un} des Champs~6 des super-tuiles voisines ($T$ est la super-tuile finale de $p$), soit que $p$ apparaît une et une seule fois dans \emph{deux} Champs~6 distincts ($T$ est une super-tuile de transit pour $p$). Comme dans le cas précédent où $T$ était la super-tuile initiale du point $p$, nous vérifions que les flèches correspondantes sont cohérentes. Là encore, le Lemme~\ref{l:automate} nous garantit que nous pouvons vérifier qu'un élément du Champ~6' apparaît exactement une ou deux fois dans les Champs~6, et localiser ceux-ci, en temps $O(\log N_{k+1})$.

Si $T$ est la super-tuile finale de $p$, il faut également nous assurer que $T$ est une cellule du Champ~5 de sa super-tuile mère $M$, correspondant à la super-couleur gauche de $M$. Pour cela, nous récupérons la marque spéciale (si elle existe) correspondant à $T$ à partir des Champs~4 de $T$, et vérifions en temps constant que celui-ci correspond à une des 5 marques spéciales associées au Champ~5 d'une super-couleur gauche ;
	\item[(E)] Comme la taille du Champ~7 est en $O(\log N_k)$, nous pouvons utiliser des processus standards, en temps polynomiaux ;  
	\item[(F)] Pour implémenter les 3 opérations de l'étape $F$ (voir p.~\pageref{propriety-f}), il est nécessaire de pouvoir parcourir le Champ~8 et de recopier dans la mémoire uniquement les éléments dont les coordonnées $(x,y)$ sont comprises dans deux intervalles $X$ et $Y$, pour obtenir un Champ~8'. Cela nécessite $\poly (\log N_k)$ étapes de calcul par élément, soit au plus $(2N_k)^{\epsilon} \cdot \poly (\log N_k)$ étapes (il n'est pas nécessaire d'utiliser le parallélisme). Il est également nécessaire de recopier un Champ~5 (possédant au plus $N_i^{\epsilon}$ éléments) dans la zone de calcul (en ajoutant un nombre de zéros supplémentaires égale à $\lceil \log(2N_i) \rceil - \lceil \log N_i \rceil$ au début de chaque coordonnée), et éventuellement d'ajouter à chaque coordonnée l'entier $N_i$ pour obtenir un Champ~5' (soient $N_k \cdot \poly (\log N_k)$ étapes de calcul, toujours sans utiliser le parallélisme). Nous devons alors vérifier que les Champs~5' et 8' contiennent les mêmes éléments. Pour cela, nous parcourons les éléments du Champ~5', et vérifions pour chacun d'eux s'il apparaît dans le Champ~8'. Nous parcourons ensuite le Champ~8' en vérifiant que chaque élément apparaît dans le Champ~5'. Comme précédemment, le Lemme~\ref{l:automate} nous garantit que la recherche d'un élément se fait en $\poly (\log N_k)$ étapes de calcul, ce qui nous donne un total de $6 N_k^{\epsilon} \cdot \poly(\log N_k)$ étapes de calcul pour les 3 opérations ;
	\item[(G)]
	
	\begin{itemize}
		\item[(G1)] Cette première sous-étape ne nécessite pas le parallélisme, et par définition s'effectue en un petit nombre d'étapes (i.e. en $\ell_k=\log k$),
		\item[(G2)] La plupart des opérations effectuées dans cette sous-étape ne nécessitent pas le parallélisme. Le seul moment où celui-ci est nécessaire est pour vérifier si un point appartient ou non au Champ~8, ce qui peut être effectué en temps $O(\log N_k)$ (toujours d'après le Lemme~\ref{l:automate}). Cette opération est effectuée au plus $\ell_k^4 \cdot 2(2N_k)^{\epsilon}$ fois ; cette sous-étape s'effectue ainsi en temps $2(2N_k)^{\epsilon} \cdot \poly(\log N_k)$. 
	\end{itemize}
\end{itemize}

Par conséquent, la totalité des opérations de vérifications effectuées dans la zone haute nécessite un temps en $O(N_k^{\epsilon}) \cdot \poly (\log N_{k+1})$.
Nous pouvons maintenant discuter du choix de $C$, avec $N_k=2^{C^k}$. Nous avons donc $n_k = \frac{N_k}{N_{k-1}}=2^{C^k-C^{k-1}}$.

Il faut que d'une part il y ait assez de place dans les câbles de communication pour permettre l'échange des informations entre super-tuiles (soit $O(N_k^{\epsilon} \cdot \log N_{k+1})$) , et d'autre part qu'il y ait assez de place dans la zone de calcul centrale pour les calculs de $U$ (nécessitant un temps en $poly(\log N_{k+1})$) et dans la zone de calcul haute pour les calculs de $A$ (nécessitant un temps en $O(N_k^{\epsilon}) \cdot poly (\log N_{k+1})$).   

Nous obtenons donc les inégalités suivantes : 
\begin{itemize}
	\item $\frac{n_k}{16} \geq c_1 \cdot N_k^{\epsilon} \cdot \log N_{k+1}$,
	\item $\frac{n_k}{16} \geq P_1(\log N_{k+1})$,
	\item $\frac{n_k}{4} \geq c_2 \cdot N_k^{\epsilon} \cdot P_2(\log N_{k+1})$,
\end{itemize}
avec $c_1$ et $c_2$ des entiers et $P_1$ et $P_2$ des polynômes. Quitte à prendre $\epsilon < \epsilon' < 1$, nous obtenons l'inégalité suivante : $n_k \geq c_3 \cdot N_k^{\epsilon'}$ avec $c_3$ un entier. D'où :

\[
\begin{array}{l l l}
2^{C^{k-1}(C-1)} \geq c_3 \cdot (2^{C^k})^{\epsilon'} & \Longleftarrow & C^{k-1}(C-1) \geq \log c_3 \cdot \epsilon' \cdot C^k\\
& \Longleftarrow & C^{k-1} \cdot (C-\epsilon' \cdot C' -1) \geq \log c_3 \\
& \Longleftarrow & C-\epsilon' \cdot C -1\geq \log c_3\\
& \Longleftarrow & C(1-\epsilon') \geq \log c_3 +1\\
& \Longleftarrow & C \geq \frac{\log c_3 +1}{1-\epsilon'}
\end{array}
\]

Par conséquent, la construction est valide pour tout $C$ suffisamment grand, ce qui termine notre construction.

\subsection{La projection naturelle des shifts de type fini donne les shifts des Théorèmes~\ref{th:sparse1} et \ref{th:sparse2}}

La construction garantit que le lemme suivant est vrai.
\begin{lemma}\label{lemma:la1}
Dans la projection d'un pavage valide, tout bloc de taille $2\times 2$ de super-tuiles de chaque rang ne contient pas de motif interdit de rang $k$.
\end{lemma}

\begin{figure}[H]
	\center
	\includegraphics[scale=0.3]{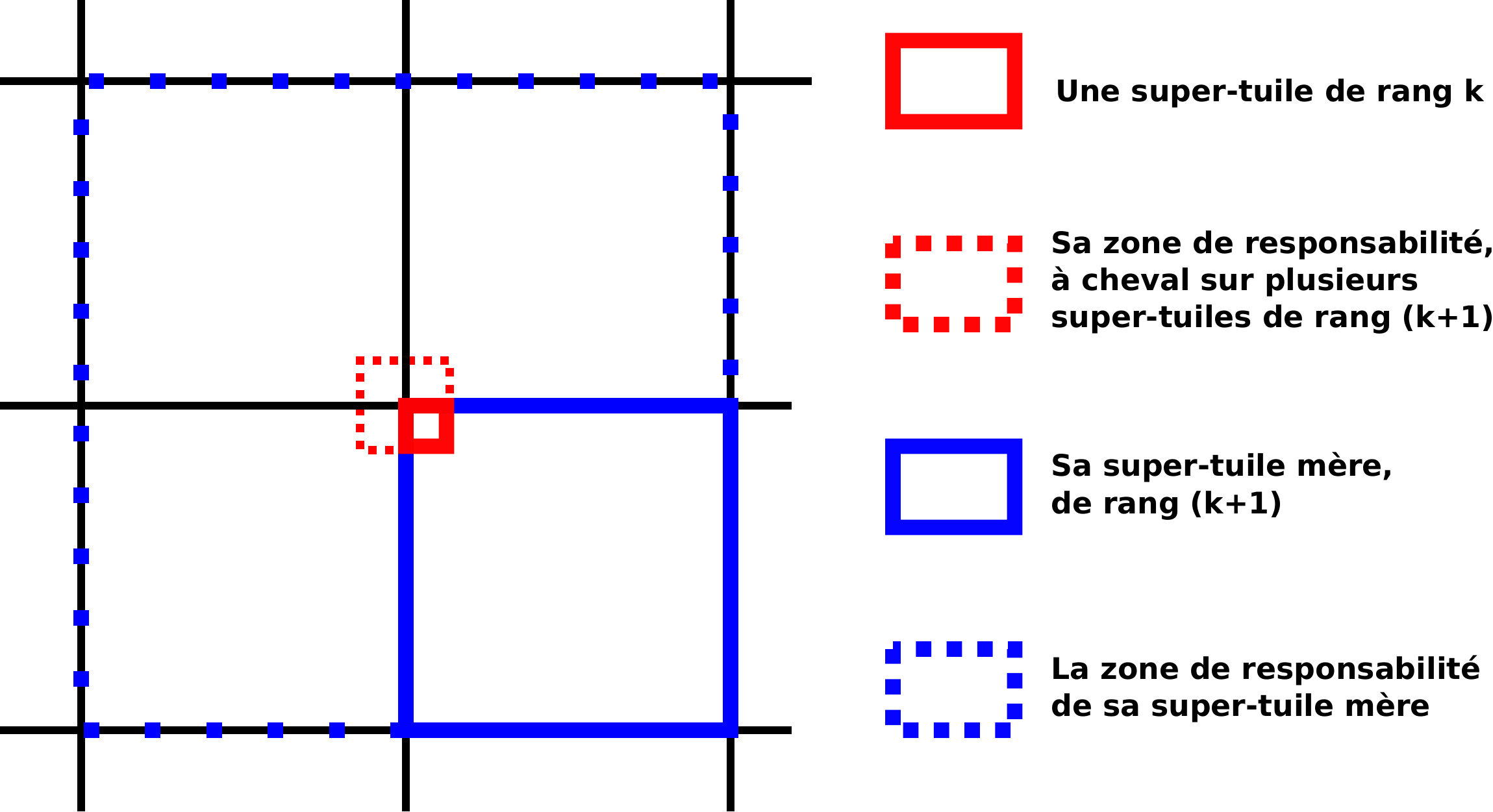}
	\caption[Un bloc de taille $2\times 2$ de super-tuiles appartient à une infinité de super-tuiles.]{Un bloc de taille $2\times 2$ de super-tuiles appartient à une infinité de super-tuiles de rang $k, (k+1), (k+2), \ldots$}\label{f:zones-responsabilité}
\end{figure}

\begin{proof}
Dans une configuration valide, un bloc de taille $2\times 2$ de super-tuiles de rang $k$ correspond à la zone de responsabilité d'une super-tuile $T$ de rang $k$ (précisément, la super-tuile en bas à droite de ce bloc). Ce bloc appartient également à la zone de responsabilité de sa super-tuile mère de rang $l=k+1$, et en réalité à la zone de responsabilité de tous les ``ancêtres'' de rangs $l\geq k+1$ de $T$, comme montré dans la Fig.\ref{f:zones-responsabilité}. La projection de ce bloc est constituée des points noirs qui appartiennent à ce motif de taille $(2N_k) \times (2N_k)$). Ceux-ci apparaissent dans les données d'entrée des zones de calcul de super-tuiles de rangs $l=k, (k+1), (k+2), \ldots$, voir Fait~\ref{claim:info-flow}. La construction garantit (voir Étape~G.2 ci-dessus) que le motif de points noirs ne contient pas de motif interdit de rang $l=k, (k+1), (k+2), \ldots$. Par conséquent, il ne contient pas de motif interdit du tout.
\end{proof}

\begin{lemma}\label{lemma:la2}
Dans un pavage valide, chaque motif fini de taille $L\times L$  est couvert par un bloc de super-tuiles de rang $k$ de taille  $2\times 2$, pour tout $k$ suffisamment grand.
\end{lemma}

\begin{figure}[H]
	\center
	\includegraphics[scale=0.3]{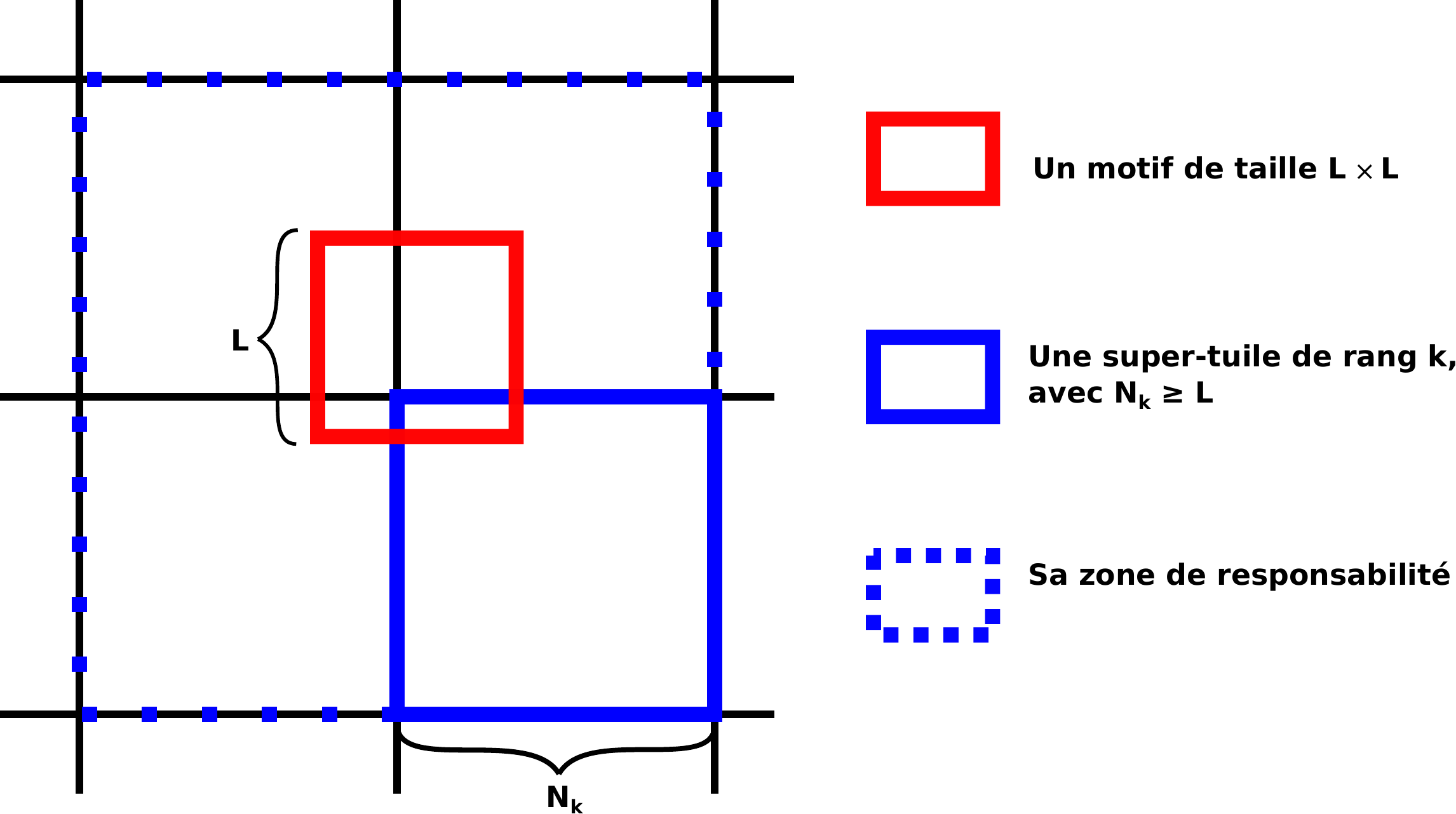}
	\caption[Un motif inclus dans la zone de responsabilité d'une super-tuile de rang $k$.]{Un motif de taille $L \times L$ inclus dans la zone de responsabilité d'une super-tuile de rang $k$.}\label{f:motif-zone-responsabilité}
\end{figure}

\begin{proof}
Nous ne pouvons pas garantir qu'un tel motif soit couvert par une unique super-tuile de rang $k$,
mais nous pouvons affirmer cela pour un bloc de super-tuiles de taille $2\times 2$. Cette affirmation est trivialement vraie pour tout $k$ tel que $N_k\ge L$, voir Fig.~\ref{f:motif-zone-responsabilité}.
\end{proof}

Des Lemmes~\ref{lemma:la1} et \ref{lemma:la2}, il s'ensuit que la projection de chaque pavage valide ne contient pas de motif interdit d'aucun rang $k$ et, par conséquent, aucun motif interdit du tout. Cela signifie que la projection appartient au shift $S_\epsilon$ (ou respectivement, $S_\epsilon'$).

Il reste à prouver que chaque configuration ${\cal C}$ du shift $S_\epsilon$ (ou, respectivement, du shift $S_\epsilon'$) peut être obtenue par la projection d'une configuration du shift de type fini. En effet, nous pouvons superposer la structure hiérarchique des super-tuiles (avec des positions horizontale et verticale de la grille arbitrairement choisies pour chaque étape $N_1, N_2, \ldots$), et imposer que les données d'entrée des zones de calcul de toutes les super-tuiles contiennent la position réelles des ``points noirs'' de ${\cal C}$. 

Nous montrons que cela est possible par induction sur le rang des super-tuiles. Pour les super-tuiles de rang $0$, il suffit de prendre les points noirs présents dans la zone de responsabilité des super-tuiles et d'encoder les représentations binaires des coordonnées de ces points noirs les unes à la suite des autres dans le Champ~5 de ces super-tuiles (le choix de $N_0$ nous garantit que nous avons assez de place pour cela). Supposons maintenant que les Champs~5 des données d'entrée des zones de calcul de toutes les super-tuiles de rang $k$ contiennent les positions réelles des ``points noirs'' de ${\cal C}$. En particulier, nous considérons, pour une super-tuile $M$ de rang $k+1$, toutes les super-tuiles de rang $k$ filles de $M$. Notons $F$ le nombre des points noirs présents dans $M$, et $l = O(\log N_{k+1})$ le nombre de bits nécessaires pour encoder les coordonnées d'un point noir vis-à-vis de $M$. Soit $E$ l'ensemble des super-tuiles de rang $k$, présentes dans le Champ~5 de $M$, constitué de la $1$-ère super-tuile du champ, de la $(1+l)$-ème, $\ldots $, de la $(1+(F-1) \cdot l)$-ème super-tuile du champ (le choix de $N_{k+1}$ nous garantit que nous avons assez de place). Alors nous pouvons voir qu'il est possible d'associer à chaque point $p$ apparaissant dans le Champ~5 d'une super-tuile fille de $M$ un chaîne de super-tuiles filles de $M$ pour lesquelles $p$ apparaît dans les Champs~6, se terminant par une super-tuile fille pour laquelle $p$ apparaît également dans le Champ~7 (et que le pavage obtenu soit valide) si et seulement s'il existe un flot (voir Définition~\ref{d:flot}) de valeur $F$, pouvant être représenté comme un ensemble de chemins élémentaires de cardinalité $F$ et un ensemble de circuits élémentaires, sur le graphe de flot $G$ (voir Définition~\ref{d:graphe-flot}) suivant :
\begin{itemize}
	\item $\rho$ correspond à $N_k$ ;
	\item à chaque super-tuile fille est associé un sommet. Entre les deux sommets correspondant à deux super-tuiles voisines et sœurs il existe deux arcs de capacité $\rho$ (de direction opposée) ;
	\item pour chaque super-tuile fille $T$ ayant $l$ points noirs en elle, il existe un arc de capacité $l$ de la source $s$ vers le sommet correspondant à $T$ (on peut noter que $l<N_k^\epsilon<N_k=\rho$) ;
	\item pour chaque super-tuile fille $T$ de $E$, il existe un arc de capacité $1$ du sommet correspondant à $T$ vers le puits $p$.
\end{itemize}

Cela est vrai d'après le Corollaire~\ref{cor:flot-chemin}. 

Puisque ${\cal C}$ ne contient pas de motif interdit (et, donc, pas de motif interdit d'aucun rang $i$), les calculs effectués dans toutes les super-tuiles se terminent avec succès, sans créer d'erreur.

Cela conclut la preuve des Théorèmes~\ref{th:sparse1} et \ref{th:sparse2}.

\chapter*{Conclusion}
\setcounter{chapter}{4}
\addcontentsline{toc}{chapter}{Conclusion}

Le principal but de ce travail était d'étudier la frontière, pour les shifts multidimensionnels, entre les shifts sofiques et les shifts non sofiques mais effectifs. Nous avons montré que la complexité algorithmique (avec différentes implémentations formelles) était un outil utile pour capturer la différence entre shifts sofiques et non sofiques. Le point de départ de cette recherche était une conjecture naïve :

\begin{itemize}
\item si un shift à deux dimensions est sofique, alors il est effectif et l'information ``essentielle'' dans chaque motif globalement admissible de taille $N\times N$ doit avoir une description algorithmique courte et efficace ;
\item si un shift à deux dimensions est effectif et que pour l'information ``essentielle'' dans chaque motif globalement admissible de taille $N\times N$ nous avons une description algorithmique suffisamment courte et efficace, alors ce shift est sofique.
\end{itemize}
Avec les techniques connues nous ne pouvons pas démontrer la conjecture. De plus, nous ne pouvons pas non plus suggérer une version plus formelle de celle-ci explicitant l'intuition d'\emph{information essentielle} et de \emph{description suffisamment courte et efficace}.
Cependant, nous avons fait quelques progrès en approchant par les deux directions la caractérisation voulue (dans la partie ``si'' et dans la partie ``seulement si''). 

Dans le Chapitre~2, la notion d'``information essentielle'' est formalisée dans nos résultats par différentes versions des ``résumés'', et celle de ``représentation suffisamment courte et efficace'' est formalisée par une complexité de Kolmogorov (avec ou sans ressources bornées) linéaire en $N$ pour un motif de taille $N\times N$. 

Dans le Chapitre~3, nous n'utilisons pas la complexité de Kolmogorov ; au contraire, nous utilisons une description naïve d'un motif ``clairsemé'', à savoir une liste des points noirs du motif avec leurs coordonnées, la quantité d'information dans le motif étant alors essentiellement égale à la longueur de la liste.

Nous pensons que ce travail peut être continué, et que l'écart entre les conditions nécessaires et les conditions suffisantes puisse être resserré. Nous pouvons mentionner deux questions ouvertes plus spécifiques :

\begin{question}
Dans le Chapitre~3 nous avons prouvé qu'un shift sur l'alphabet $\{\square, \blacksquare\}$ dont les configurations sont ``clairsemées''  (dans tout motif de taille $N \times N$ il y a au plus $N^{1-\epsilon}$ lettres noires, et les autres sont blanches) est sofique. Ce théorème est similaire dans un certain sens au résultat de L.B.~Westrick (voir \cite{westrick2017seas}), qui a prouvé que le shift $S$, dont les configurations sont constituées de carrés noirs sur une mer de cases blanches, carrés dont les tailles sont distinctes, est sofique. Dans notre argument et dans celui de Westrick, le point clé de la preuve est que chaque motif admissible de taille $N\times N$ a une petite description (la taille d'une telle description étant bien plus petite que $N$). Dans ces deux résultats, un certain langage \textit{ad hoc} pour ces descriptions a été utilisé. Il est instructif d'observer que chaque motif globalement admissible de taille $N\times N$ du shift suggéré par Westrick a une description naturelle de taille $N^{2/3} \cdot \poly(\log N)$. Cela est la conséquence d'une borne supérieure sur le nombre de carrés disjoints, de tailles deux à deux distinctes, qui peuvent appartenir à un motif de taille $N\times N$ (voir \cite{westrick2017seas} pour les détails).

Les preuves de ces théorèmes ne peuvent être reformulées directement en termes de complexité de Kolmogorov, simple ou à ressources bornées. Cependant, il est possible que nous ayons besoin ici d'une version de la complexité algorithmique plus adaptée. Une question demeure donc : ces deux théorèmes sont-ils des cas spéciaux d'une propriété plus générale, pour une mesure de la complexité algorithmique plus générale et flexible ? 
\end{question}

\begin{question}
Dans la preuve de nos résultats, dans le Chapitre~3, il était crucial que la taille de la description de tout motif globalement admissible de taille $N\times N$ fût inférieure à $N$. Notre technique ne s'applique pas, par exemple, pour le shift $S$ dont les motifs interdits sont les motifs de taille $N \times N$ contenant plus de $N^{1.5}$ lettres noires. Par ailleurs, il semble qu'il n'y ait pas de manière directe d'appliquer notre technique du Chapitre~2 pour prouver que ce shift n'est pas sofique. Une question se pose alors : ce shift est-il sofique ou non ? 
\end{question}

\newpage
\addcontentsline{toc}{chapter}{Bibliographie générale}

\renewcommand{\bibname}{Bibliographie générale}
\bibliography{bibliographie}

\begin{thebibliography}{1}

\bibitem{stacs}
J.~Destombes and A.~Romashchenko.
\newblock Resource-bounded kolmogorov complexity provides an obstacle to
  soficness of multidimensional shifts.
\newblock In {\em 36th International Symposium on Theoretical Aspects of
  Computer Science (STACS 2019)}, volume 126, pages 23:1--23:17. Schloss
  Dagstuhl--Leibniz-Zentrum fuer Informatik, 2019.

\bibitem{stacs-submitted}
J.~Destombes and A.~Romashchenko.
\newblock Resource-bounded kolmogorov complexity provides an obstacle to
  soficness of multidimensional shifts.
\newblock {\em soumis à un journal}, arXiv:1805.03929, 37~pp., 2021.

\end{thebibliography}


\begin{thebibliography}{10}

\bibitem{allender1989some}
E.W. Allender.
\newblock Some consequences of the existence of pseudorandom generators.
\newblock {\em Journal of Computer and System Sciences.}, 39(1):101--124, 1989.

\bibitem{mirror}
N.~Aubrun, S.~Barbieri, and E.~Jeandel.
\newblock About the domino problem for subshifts on groups.
\newblock In {\em Sequences, Groups, and Number Theory}, pages 331--389.
  Springer, 2018.

\bibitem{aubrun-sablik}
N.~Aubrun and M.~Sablik.
\newblock Simulation of effective subshifts by two-dimensional subshifts of
  finite type.
\newblock {\em Acta Applicandae Mathematicae}, 126(1):35--63, 2013.

\bibitem{pumping-lemma}
J.D.~Ullman. A.V.~Aho, J.E.~Hopcroft.
\newblock {\em The Design and Analysis of Computer Algorithms.}
\newblock Addison-Wesley, 1974.

\bibitem{menger}
J.~Bang-Jensen and G.~Gutin.
\newblock {\em Digraphs Theory, Algorithms and Applications.}
\newblock Springer, 2007.

\bibitem{berger}
R.~Berger.
\newblock The undecidability of the domino problem.
\newblock {\em Mem. Amer. Math. Soc.}, 66:1--72, 1966.

\bibitem{surface-entropies}
A.~Callard and P.~Vanier.
\newblock Computational characterization of surface entropies for
  $\mathbb{Z}^2$ subshifts of finite type.
\newblock {\em ICALP}, 122:1-122:20, 2021.

\bibitem{chaitin1969simplicity}
G.J. Chaitin.
\newblock On the simplicity and speed of programs for computing infinite sets
  of natural numbers.
\newblock {\em Journal of the ACM (JACM)}, 16(3):407--422, 1969.

\bibitem{delorme1998cellular}
M.~Delorme and J.~Mazoyer.
\newblock {\em Cellular Automata: a parallel model.}, volume 460.
\newblock Springer Science \& Business Media, 1998.

\bibitem{dls}
B.~Durand, L.~Levin, and A.~Shen.
\newblock Complex tilings.
\newblock {\em The Journal of Symbolic Logic}, 73(2):593--613, 2008.

\bibitem{point-fixe}
B.~Durand and A.~Romashchenko.
\newblock The expressiveness of quasiperiodic and minimal shifts of finite
  type.
\newblock {\em Ergodic Theory and Dynamical Systems}, 41(4):1086--1138, 2021.

\bibitem{durand2009foundations}
B.~Durand, A.~Romashchenko, and A.~Shen.
\newblock Fixed point theorem and aperiodic tilings.
\newblock {\em Bulletin of the EATCS}, 97:126--136, 2009.

\bibitem{drs}
B.~Durand, A.~Romashchenko, and A.~Shen.
\newblock Effective closed subshifts in {1D} can be implemented in {2D}.
\newblock {\em Fields of logic and computation}, 6300:208--226, 2010.

\bibitem{fixed-point}
B.~Durand, A.~Romashchenko, and A.~Shen.
\newblock Fixed-point tile sets and their applications.
\newblock {\em Journal of Computer and System Sciences}, 78(3):731--764, 2012.

\bibitem{flot-max-coupe-min}
L.R. Ford and D.R. Fulkerson.
\newblock Maximal flow through a network.
\newblock {\em Canadian Journal of Mathematics}, 8:399--404, 1956.

\bibitem{ford1957simple}
L.R. Ford and D.R. Fulkerson.
\newblock A simple algorithm for finding maximal network flows and an
  application to the hitchcock problem.
\newblock {\em Canadian Journal of Mathematics}, 9:210--218, 1957.

\bibitem{furusawa2017uniform}
H.~Furusawa.
\newblock Uniform continuity of relations and nondeterministic cellular
  automata.
\newblock {\em Theoretical Computer Science}, 673:19--29, 2017.

\bibitem{gacs1986reliable}
P.~G{\'a}cs.
\newblock Reliable computation with cellular automata.
\newblock {\em Journal of Computer and System Sciences}, 32(1):15--78, 1986.

\bibitem{gacs2001reliable}
P.~G{\'a}cs.
\newblock Reliable cellular automata with self-organization.
\newblock {\em Journal of Statistical Physics}, 103(1):45--267, 2001.

\bibitem{gray2001reader}
L.F. Gray.
\newblock A reader's guide to gacs's “positive rates” paper.
\newblock {\em Journal of Statistical Physics}, 103(1):1--44, 2001.

\bibitem{hartmanis1983generalized}
J.~Hartmanis.
\newblock Generalized kolmogorov complexity and the structure of feasible
  computations.
\newblock In {\em 24th Annual Symposium on Foundations of Computer Science
  (sfcs 1983)}, pages 439--445. IEEE, 1983.

\bibitem{hennie-stearns}
F.C. Hennie and R.E. Stearns.
\newblock Two tape simulation of multitape {T}uring machines.
\newblock {\em Journal of the ACM}, 13:533--546, 1966.

\bibitem{hennie1966two}
F.C. Hennie and R.E. Stearns.
\newblock Two-tape simulation of multitape {T}uring machines.
\newblock {\em Journal of the ACM (JACM)}, 13(4):533--546, 1966.

\bibitem{hochman2009dynamics}
M.~Hochman.
\newblock On the dynamics and recursive properties of multidimensional symbolic
  systems.
\newblock {\em Inventiones mathematicae}, 176(1):131, 2009.

\bibitem{Turing-degree-spectra}
M.~Hochman and P.~Vanier.
\newblock Turing degree spectra of minimal subshifts.
\newblock In {\em Computer Science -- Theory and Applications}, pages 154--161.
  Springer International Publishing, 2017.

\bibitem{hopcroft2001introduction}
J.E. Hopcroft, R.~Motwani, and J.D. Ullman.
\newblock Introduction to automata theory, languages, and computation.
\newblock {\em Acm Sigact News}, 32(1):60--65, 2001.

\bibitem{jeandel-hdr}
E.~Jeandel.
\newblock Propriétés structurelles et calculatoires des pavages.
\newblock {\em Habilitation thesis, Universit\'e Montpellier~2}, 2011.

\bibitem{minimal-aperiodic-tile-set}
E.~Jeandel and M.~Rao.
\newblock An aperiodic set of 11 wang tiles.
\newblock {\em CoRR}, abs/1506.06492, 2015.

\bibitem{Pi01}
E.~Jeandel and P.~Vanier.
\newblock {T}uring degrees of multidimensional sfts.
\newblock {\em Theoretical Computer Science}, 505:81--92, 2013.

\bibitem{moutot-kari}
J.~Kari and E.~Moutot.
\newblock Decidability and periodicity of low complexity tilings.
\newblock {\em 37th International Symposium on Theoretical Aspects of Computer
  Science (STACS 2020)}, 154:14:1--14:12, 2020.

\bibitem{kass-madden}
S.~Kass and K.~Madden.
\newblock A sufficient condition for non-soficness of higher-dimensional
  subshifts.
\newblock {\em Proceedings of the American Mathematical Society},
  141(11):3803--3816, 2013.

\bibitem{ko1986notion}
K.I. Ko.
\newblock On the notion of infinite pseudorandom sequences.
\newblock {\em Theoretical Computer Science}, 48:9--33, 1986.

\bibitem{kolmogorov-three-approaches}
A.~N. Kolmogorov.
\newblock Three approaches to the quantitative definition of information.
\newblock {\em Problems of information transmission}, 1(1):1--7, 1965.

\bibitem{kutrib2012non}
M.~Kutrib.
\newblock Non-deterministic cellular automata and languages.
\newblock {\em International Journal of General Systems}, 41(6):555--568, 2012.

\bibitem{di2020topological}
P.~Di Lena.
\newblock Topological dynamics of nondeterministic cellular automata.
\newblock {\em Information and Computation}, 274:104532, 2020.

\bibitem{di2014nondeterministic}
P.~Di Lena and L.~Margara.
\newblock Nondeterministic cellular automata.
\newblock {\em Information Sciences}, 287:13--25, 2014.

\bibitem{li2008introduction}
M.~Li and P.~Vit{\'a}nyi.
\newblock {\em An introduction to Kolmogorov complexity and its applications.},
  volume~3.
\newblock Springer, 2008.

\bibitem{sft}
D.~Lind and B.~Marcus.
\newblock An introduction to symbolic dynamics and coding.
\newblock {\em Cambridge University Press}, 1995.

\bibitem{symbolic-dynamics-1938}
M.~Morse and G.A. Hedlund.
\newblock Symbolic dynamics.
\newblock {\em Amer. J. Math.}, 60:815--866, 1938.

\bibitem{symbolic-dynamics-1940}
M.~Morse and G.A. Hedlund.
\newblock Symbolic dynamics {II}: {S}turmian trajectories.
\newblock {\em Amer. J. Math.}, 62:1--42, 1940.

\bibitem{moutot}
E.~Moutot.
\newblock {\em Around the Domino Problem - Combinatorial Structures and
  Algebraic Tools.}
\newblock PhD thesis, Université de Lyon and Turun yliopisto, 2020.

\bibitem{von1963general}
J.~Von Neumann.
\newblock {\em The general and logical theory of automata.}, volume~5.
\newblock Pergamon Press Ltd., Oxford, 1963.

\bibitem{nivat}
M.~Nivat.
\newblock New challenges for theoretical computer science.
\newblock {\em TAPSOFT '97: Theory and Practice of Software Development}, pages
  11--14, 1997.

\bibitem{ormes-pavlov}
N.~Ormes and R.~Pavlov.
\newblock Extender sets and multidimensional subshifts.
\newblock {\em Ergodic Theory and Dynamical Systems}, 36(3):908--923, 2016.

\bibitem{ozhigov1999computations}
Y.~Ozhigov.
\newblock Computations on nondeterministic cellular automata.
\newblock {\em Information and Computation}, 148(2):181--201, 1999.

\bibitem{pavlov}
R.~Pavlov.
\newblock A class of nonsofic multidimensional shift spaces.
\newblock {\em Proceedings of the American Mathematical Society}, 141:987--996,
  2013.

\bibitem{word-computational}
R.~Pavlov and P.~Vanier.
\newblock The relationship between word complexity and computational complexity
  in subshifts.
\newblock {\em Discrete and Continuous Dynamical Systems}, 41(4):1627--1648,
  2021.

\bibitem{richardson1972tessellations}
D.~Richardson.
\newblock Tessellations with local transformations.
\newblock {\em Journal of Computer and System Sciences}, 6(5):373--388, 1972.

\bibitem{robinson}
R.M. Robinson.
\newblock Undecidability and nonperiodicity for tilings of the plane.
\newblock {\em Invent. Math.}, 12:177--209, 1971.

\bibitem{rumyantsev-ushakov}
A.~Rumyantsev and M.~Ushakov.
\newblock Forbidden substrings, {K}olmogorov complexity and almost periodic
  sequences.
\newblock {\em In Proc. Annual Symposium on Theoretical Aspects of Computer
  Science}, pages 396--407, 2006.

\bibitem{shen2017kolmogorov}
A.~Shen, V.A. Uspensky, and N.~Vereshchagin.
\newblock {\em Kolmogorov complexity and algorithmic randomness.}, volume 220.
\newblock American Mathematical Soc., 2017.

\bibitem{sipser1983complexity}
M.~Sipser.
\newblock A complexity theoretic approach to randomness.
\newblock In {\em Proceedings of the fifteenth annual ACM symposium on Theory
  of computing}, pages 330--335, 1983.

\bibitem{solomonoff1960preliminary}
R.J. Solomonoff.
\newblock A preliminary report on a general theory of inductive inference.
\newblock Technical report, United States Air Force, Office of Scientific
  Research, 1960.

\bibitem{solomonoff1964formalI}
R.J. Solomonoff.
\newblock A formal theory of inductive inference. {P}art {I}.
\newblock {\em Information and control}, 7(1):1--22, 1964.

\bibitem{solomonoff1964formalII}
R.J. Solomonoff.
\newblock A formal theory of inductive inference. {P}art {II}.
\newblock {\em Information and control}, 7(2):224--254, 1964.

\bibitem{torma2021fixed}
I.~T{\"o}rm{\"a}.
\newblock Fixed point constructions in tilings and cellular automata.
\newblock In {\em 27th IFIP WG 1.5 International Workshop on Cellular Automata
  and Discrete Complex Systems (AUTOMATA 2021)}. Schloss
  Dagstuhl-Leibniz-Zentrum fur Informatik, 2021.

\bibitem{jordan}
O.~Veblen.
\newblock Theory of plane curves in non-metrical analysis situs.
\newblock {\em Trans. Am. Math.Soc.}, 6:83--98, 1905.

\bibitem{wang}
H.~Wang.
\newblock Proving theorems by pattern recognition — {II}.
\newblock {\em Bell System Technical Journal}, 40(1):1--41, 1961.

\bibitem{weiss}
B.~Weiss.
\newblock Subshifts of finite type and sofic systems.
\newblock {\em Monatsh. Math.}, 77:462--474, 1973.

\bibitem{westrick2017seas}
L.B. Westrick.
\newblock Seas of squares with sizes from a {$\Pi_1^0$} set.
\newblock {\em Israel Journal of Mathematics}, 222(1):431--462, 2017.

\bibitem{woods2009complexity}
D.~Woods and T.~Neary.
\newblock The complexity of small universal {T}uring machines: A survey.
\newblock {\em Theoretical Computer Science}, 410(4-5):443--450, 2009.

\bibitem{yaku1976surjectivity}
T.~Yaku.
\newblock Surjectivity of nondeterministic parallel maps induced by
  nondeterministic cellular automata.
\newblock {\em Journal of Computer and System Sciences}, 12(1):1--5, 1976.

\bibitem{zinoviadis2016hierarchy}
C.~Zinoviadis.
\newblock Hierarchy and expansiveness in two-dimensional subshifts of finite
  type.
\newblock {\em arXiv preprint arXiv:1603.05464}, 2016.

\bibitem{zvonkin1970complexity}
A.K. Zvonkin and L.A. Levin.
\newblock The complexity of finite objects and the development of the concepts
  of information and randomness by means of the theory of algorithms.
\newblock {\em Russian Mathematical Surveys}, 25(6):83, 1970.

\end{thebibliography}

\newpage
\addcontentsline{toc}{chapter}{Bibliographie personnelle}

\bibliographyMe{bibliographie}

\end{document}